\documentclass[acmsmall,screen,nonacm]{acmart}

\usepackage{amsmath}
\usepackage{mathtools}
\usepackage{physics} 
\usepackage{geometry}
\usepackage{xcolor}
\usepackage{hyperref}
\usepackage{cleveref}
\usepackage{stmaryrd} 
\usepackage[shortlabels]{enumitem}
\usepackage{csquotes}
\usepackage{mathpartir}
\usepackage{booktabs}

\usepackage{mathpartir} 
\usepackage{mdframed} 
\usepackage{xcolor} 
\usepackage{array}
\usepackage{longtable}
\usepackage{tikz}
\usepackage{tikz-cd}

\newtheorem{theorem}{Theorem}[section]
\newtheorem{lemma}{Lemma}[section]
\newtheorem{proposition}{Proposition}[section]
\newtheorem{corollary}[theorem]{Corollary}
\newtheorem{definition}{Definition}[section]

\newlength{\mycustomparindent}
\AtBeginDocument{\setlength{\mycustomparindent}{\parindent}} 

\makeatletter
\g@addto@macro\@acmplainbodyfont{\setlength{\listparindent}{\mycustomparindent}\setlength{\parindent}{\mycustomparindent}}
\g@addto@macro\@acmdefinitionbodyfont{\setlength{\listparindent}{\mycustomparindent}\setlength{\parindent}{\mycustomparindent}}
\makeatother

\makeatletter
\newtheoremstyle{acmremark}%
  {.5\baselineskip\@plus.2\baselineskip\@minus.2\baselineskip}
  {.5\baselineskip\@plus.2\baselineskip\@minus.2\baselineskip}
  {\normalfont\setlength{\listparindent}{\mycustomparindent}\setlength{\parindent}{\mycustomparindent}}
  {}
  {\itshape}
  {.}
  {.5em}
  {}
\makeatother

\theoremstyle{acmremark}
\newtheorem{remark}{Remark}

\newcommand{\cB} {\mathcal{B}}
\newcommand{\cC} {\mathcal{C}}
\newcommand{\cD} {\mathcal{D}}
\newcommand{\cE} {\mathcal{E}}

\newcommand{\cH} {\mathcal{H}}

\newcommand{\cK} {\mathcal{K}}

\newcommand{\cQ} {\mathcal{Q}}

\newcommand{\cT} {\mathcal{T}}

\newcommand{\bN} { {\mathbb{N}}}   
\newcommand{\bC} { {\mathbb{C}}}  
\newcommand{\bQ} { {\mathbb{Q}}}   
   
\newcommand{\bR} { {\mathbb{R}}}

\newcommand{\QF}{\mathfrak{t}}                
\newcommand{\Dom}{{\operatorname{Dom}}}              
\newcommand{\Pred}{\mathbb{P}}                
\newcommand{\sem}[1]{\llbracket #1 \rrbracket}
\newcommand{\sqleq}{\sqsubseteq}              
\newcommand{\sqgeq}{\sqsupseteq}
\newcommand{\<}{{\langle}}
\renewcommand{\>}{{\rangle}}
\newcommand{\Ran} {\operatorname{Ran}}
\newcommand{\mul} {\operatorname{Mul}}
\newcommand{\sP} {\mathsf{P}}
\newcommand{\bind}{\mathbf{bind}}
\renewcommand{\vec}[1]{\ensuremath{\mathbf{#1}}}
\newcommand{\seq}[1]{\ensuremath{\overline{#1}}}
\renewcommand{\u}{\vec{u}}
\renewcommand{\v}{\vec{v}}
\newcommand{\w}{\vec{w}}
\newcommand{\x}{\vec{x}}
\newcommand{\y}{\vec{y}}
\newcommand{\z}{\vec{z}}
\renewcommand{\tr}{\operatorname{tr}}
\renewcommand{\Tr}{\operatorname{Tr}}
\renewcommand{\wp}{\operatorname{wp}}
\newcommand{\wlp}{\operatorname{wlp}}
\newcommand{\ii}{\mathrm{i}}

\newcommand{\ert}{\operatorname{ert}}
\let\geq=\geqslant
\let\leq=\leqslant
\let\epsilon=\varepsilon

\newcommand{\density}[1]{%
  \if\relax\detokenize{#1}\relax
    \mathcal{D}
  \else
    \mathcal{D}(#1)
  \fi
}
\newcommand{\pardensity}[1]{\density{#1}}

\title{Reasoning about Continuous-Variable Quantum Systems}
\author{Tianshi Yu}
\affiliation{%
  \institution{Institute of Software, Chinese Academy of Sciences}
  \city{Beijing}
  \country{China}
}
\email{yuts@ios.ac.cn}

\author{Gilles Barthe}
\affiliation{%
  \institution{Max Planck Institute for Security and Privacy (MPI-SP)}
  \city{Bochum}
  \country{Germany}
}
\affiliation{%
  \institution{IMDEA Software Institute}
  \city{Madrid}
  \country{Spain}
}
\email{gilles.barthe@mpi-sp.org}

\author{Minbo Gao}
\affiliation{%
  \institution{Institute of Software, Chinese Academy of Sciences}
  \city{Beijing}
  \country{China}
}
\affiliation{
  \institution{University of Chinese Academy of Sciences}
  \city{Beijing}
  \country{China}
}
\email{gmb17@tsinghua.org.cn}

\author{Mingsheng Ying}
\affiliation{%
  \department{Centre for Quantum Software and Information}
  \institution{University of Technology Sydney}
  \city{Sydney}
  \country{Australia}
}
\email{Mingsheng.Ying@uts.edu.au}

\author{Li Zhou}
\affiliation{%
  \institution{Institute of Software, Chinese Academy of Sciences}
  \city{Beijing}
  \country{China}
}
\email{zhouli@iscas.ac.cn}

\begin{document}

\begin{abstract}
Continuous-variable quantum computing (CVQC) is a computing paradigm in which measurements yield values over a continuous domain. CVQC is both a convenient computational framework for modeling physical quantum systems, and a good abstraction for hardware platforms based on quantum optics. Yet, the semantic foundations of CVQC remain underdeveloped. To address this gap, we develop a formal semantics for a core CV quantum programming language, and sound verification methods for program correctness. 

A main contribution of this work is to isolate a well-behaved quantitative predicate domain that achieves sufficient expressiveness to accommodate unbounded values as they arise in the infinite-dimensional, continuous setting. Specifically, we choose closed positive quadratic forms as semantic predicates, representing finite expectations, domains of finiteness, and infinite penalties in one ordered object. We validate our choice by establishing that our semantic predicates satisfy desirable closure properties including the definition of weakest preconditions.

We validate our design with two case studies, including an example based on the celebrated GKP error-correcting code, for which we establish a second moment bound.

\end{abstract}

\keywords{continuous-variable quantum programs, quantum program verification,
quantitative Hoare logic, weakest preconditions, unbounded quantum predicates,
continuous-outcome measurements.}

\maketitle




\section{Introduction}

\label{sec:introduction}

Continuous-variable (CV) quantum computation and communication protocols have been extensively studied within the field of quantum information processing \cite{lloyd1999quantum,braunstein2005quantum}.
For example, it plays a central role in photonic quantum computing, where optical systems provide a natural hardware platform, and in optical quantum communication, where many protocols are naturally formulated in CV terms. It is also used in quantum sensing and metrology, quantum simulation, 
and fault-tolerant quantum computation. Beyond these platform-level applications, CV methods are notably used in error-correction procedures such as the Gottesman--Kitaev--Preskill (GKP) code~\cite{gottesman2001encoding}. 



Making CV quantum computation and communication protocols  correct is even more difficult than for discrete ones. Formal verification is useful for this, but it is also more demanding because their correctness is often quantitative rather than merely Boolean. A CV protocol may combine transformations, measurements, outcome-dependent commands, and loops while maintaining bounds on position, resource usage, or cost. Error correction gives a representative example: a measurement returns a real number, the next command uses that number to choose a correction, and the postcondition is not a boolean-valued assertion but a numerical bound on the remaining error, such as a variance bound under finite-moment assumptions. 

The aim of this paper is to develop a \textit{program logic} for reasoning about the correctness of CV quantum algorithms and protocols by treating them as \textit{abstract programs}, rather than programs that are actually executable on (digital) quantum computers.  Our logic supports reasoning about expected cost, hard constraints, and quantitative error bounds, without first replacing these CV  programs by a finite-dimensional cutoff.

\paragraph*{Challenges} Existing quantum weakest-precondition calculus, and Hoare, expected-runtime, and cost logics~\cite{dhondt2006quantum,Ying11,feng2020,liu2025quantum,avanzini2022quantum} provide the foundations for our work, but their usual bounded predicates do not directly express the unbounded expectations needed in reasoning about CV quantum computation and communication. The first challenge for our logic is denotational. Exact verification uses the program's own denotation, rather than a finite-dimensional approximation chosen for simulation. In the finite-dimensional case, a measurement command has finitely many branches. In a CV program, a measurement may return a real number, and the following command may depend on that number. Thus the syntax contains an outcome-indexed bind rather than a finite case split. Giving this construct a denotation is not just replacing a sum by an integral. The semantic clauses record when branch states and continuation maps vary measurably with the outcome, when the resulting integral is a positive trace-class operator, and when the induced state transformer is completely positive and trace-nonincreasing. The same account also covers sequencing, measurement-guarded loops, and cost, and supports a backward weakest-precondition transformer dual to forward execution.

The second challenge is the predicate. In finite-dimensional quantum logics, predicates are often bounded positive operators. These predicates describe event probabilities and bounded expected values, but they do not express quantities that can be arbitrarily large, such as second moments, accumulated cost, or mean-square error in CV quantum systems. Moreover, treating an arbitrary unbounded operator as a predicate is not compositional: backwards semantics may be undefined, infinite branch sums and outcome integrals may not have closed forms, and loop preconditions may only arise as suprema of finite approximants. A related issue is that in practice, assertions often need to combine numerical bounds with hard constraints, for example that an input lies in the finiteness domain of a given quantity. The infinite value captures violation of such a constraint inside the predicate itself. The predicate domain therefore combines finite values, finiteness domains, and infinite penalties in one ordered structure closed under the backward clauses. Defining an appropriate formalism that captures such assertions is challenging.

\paragraph*{Contributions}

This paper develops a logic for CV quantum programs by overcoming the above challenges. The main contributions are as follows.

\begin{itemize}

  \item \textbf{Semantics.} We define a quantum while-language whose variable types are interpreted by separable Hilbert spaces and whose commands include initialization, unitary commands, sequencing, binary measurement-guarded while loops, and an admissible continuous-outcome measurement bind. The semantics of programs denote completely positive, trace-nonincreasing superoperators on trace-class partial density operators; continuous bind is interpreted by integrating measurable families of continuation-applied branch states under explicit measurability and admissibility conditions for continuous quantum instruments~\cite{davies1976quantum,ozawa1984quantum}.

  \item \textbf{Quantum predicates and their transformations.} We use closed positive quadratic forms, equivalently positive self-adjoint linear relations, as expectation predicates for unbounded expectations and hard domain constraints~\cite{LICS2026LinearRelation}. We define weakest preconditions at the level of forms: Kraus-specified quantum operations are pulled back by the corresponding Kraus sums, countable branching uses countable sums of forms, admissible continuous bind uses pointwise integration of extended branch-form values, loops use suprema of finite unrollings, and costs are added as positive forms. We prove closure of the predicate domain under these clauses and establish an exact extended trace-duality identity with the forward Schrödinger semantics. On this basis, we define a cost-parametric weakest precondition $\wp_c$: $c=0$ recovers zero-cost partial-correctness reasoning, while $c>0$ supports expected-cost reasoning and finiteness obligations. 
  
  \item \textbf{Program logic.} Based on the above results, we present a cost-parameterized QHL-style proof system for the resulting calculus.
  This proof system unifies both zero-cost partial correctness and expected-cost reasoning through the cost parameter.
  Its distinctive \textsc{While} rule is a potential-function rule: the loop predicate must cover the exit postcondition and pay for each guard evaluation and, on the continuing branch, one body execution that re-establishes the potential.
  We prove the soundness and relative completeness of this proof system via the exact characterization of cost-parametric weakest preconditions.

\item \textbf{Applications.} We apply our logic to two case studies. The first case study pertains to a discrete, infinite-dimensional setting: a quantum symmetric walk over an $\ell^2(\mathbb Z)$ position register separates almost-sure termination from finite expected cost within the same predicate transformer. The second case study is continuous: a one-round GKP-style correction program uses continuous measurement bind and unbounded variance predicates to prove a variance bound of output state under physical rational assumptions. \textit{A particularly interesting observation arising from these applications is that logical tools developed by the programming languages community can also facilitate reasoning in physics.} 

\end{itemize}

\paragraph{Organization} 
 The rest of this paper is organized as follows. 
\Cref{sec:Preliminaries} provides the minimal mathematical background of the paper. 
The technical core is divided into three sections. To make the boundary between prior foundations and our extensions explicit, each section begins with a review of the relevant existing results. \Cref{sec:syntax_semantics} recalls the \textbf{qwhile} language, then extends it with admissible continuous-outcome bind and develops its parameterized denotational semantics. \Cref{sec:unbounded_predicates} reviews bounded quantum predicate transformers and the unbounded predicate domain, then constructs form-level predicate transformers, including the transformer for continuous-outcome bind. \Cref{sec:unified_proof_system} reviews expected-runtime transformers, then presents the cost-parametric Hoare logic and proves soundness and semantic relative completeness via the cost-parametric weakest precondition. \Cref{sec:random_walk} and~\Cref{sec:gkp_verification} present the quantum random-walk and GKP case studies, respectively. We then discuss related work and conclude in \Cref{sec:related-work,sec:conclusion}.

\section{Preliminaries}
\label{sec:Preliminaries}
This section collects the mathematical conventions used in the rest of the paper. We only recall the parts of operator theory, bounded Schr\"odinger--Heisenberg duality, and measure-theoretic integration that are needed for the constructions below.

\subsection{Hilbert Spaces and Operators}
\label{subsec:hilbert-operators}
This subsection fixes the operator-theoretic notation used in this paper.  The goal is only to record the
infinite-dimensional analogue of the finite-dimensional matrix notation.

A Hilbert space is a complete complex inner-product space.  We write
$\langle \u,\v\rangle$ for the inner product and use the convention that it is
linear in the second argument and conjugate-linear in the first. All Hilbert spaces in this paper are separable, i.e., each has a finite or countably infinite orthonormal basis. Thus 
\(\ell^2(\mathbb Z)\) and \(L^2(\mathbb R)\) are both valid spaces in this paper.
Unbounded predicates are introduced later; this subsection only concerns states, bounded operators, and bounded observables.

A linear map \(A:\cH\to\cH\) is bounded if
\(\|A\|\triangleq \sup_{\|\u\|\le 1}\|A\u\|<\infty\).  We write \(\cB(\cH)\) for
the algebra of bounded operators on \(\cH\).  In finite dimension this is just
the algebra of matrices.  For \(A\in\cB(\cH)\), its adjoint \(A^\dagger\) is
defined by \(\langle A^\dagger\u,\v\rangle=\langle \u,A\v\rangle\).  A unitary is
an operator \(U\in\cB(\cH)\) satisfying \(U^\dagger U=UU^\dagger=I\), where $I$ represents the identity operator.

For Hilbert spaces \(\cH_1,\ldots,\cH_n\), we write
\(\cH_1\otimes\cdots\otimes\cH_n\) for their Hilbert tensor product.  If
\(A_i\in\cB(\cH_i)\), then \(A_1\otimes\cdots\otimes A_n\) is the bounded
operator determined by its action on simple tensors. In finite dimension this is just the Kronecker product of matrices. 

An operator \(A\in\cB(\cH)\) is self-adjoint if \(A=A^\dagger\).  A self-adjoint
operator \(A\) is positive, written \(A\sqgeq 0\), if
\(\langle \u,A\u\rangle\ge 0\) for every \(\u\in\cH\).  For self-adjoint \(A,B\),
we write \(A\sqleq B\) if \(B-A\sqgeq 0\).  This is the usual L\"owner order on
Hermitian matrices.  We write \(\cB(\cH)_+\) for the cone of positive bounded
operators.

A bounded operator \(T\in\cB(\cH)\) is trace-class if its trace norm 
\(\|T\|_1\triangleq \tr(|T|)<\infty\), where
\(|T|\triangleq (T^\dagger T)^{1/2}\).  Equivalently, its singular values are
summable.  We write \(\cT(\cH)\) for the trace-class operators and
\(\cT(\cH)_+\) for the positive trace-class operators.  For positive
\(T\in\cT(\cH)\), the trace is
\(\tr(T)\triangleq \sum_i\langle \u_i,T\u_i\rangle\), where \((\u_i)_i\) is any
orthonormal basis; the value is independent of the basis.  The trace of a
general trace-class operator is obtained by linearity.  If \(A\in\cB(\cH)\) and
\(T\in\cT(\cH)\), then \(AT,TA\in\cT(\cH)\), and \(\tr(AT)=\tr(TA)\).

Finally we introduce some notations for quantum case. A partial density operator is a positive trace-class operator \(\rho\) with
\(\tr(\rho)\le 1\).  We denote the set of partial density operators as $\pardensity{\cH}$.
A bounded observable is a bounded self-adjoint operator \(A\in\cB(\cH)\), and
its expected value on state \(\rho\in\cD(\cH)\) is \(\tr(A\rho)\).  A bounded
effect is an operator \(A\in\cB(\cH)\) with \(0\sqleq A\sqleq I\).  
See \cite{reed1980methods} for the details.

\subsection{Quantum Operations and Schr\"odinger-Heisenberg Duality
}
On the Schr\"odinger picture, a superoperator is a bounded linear map \(\cE:\cT(\cH)\to\cT(\cH)\).  It is
positive if \(\cE(\cT(\cH)_+)\subseteq\cT(\cH)_+\).  It is completely positive
if, for every \(n\in\bN^+\), the amplification
\(\cE\otimes I_n\) maps \(\cT(\cH\otimes\mathbb C^n)_+\) into itself.  It is
trace-nonincreasing if \(\tr(\cE(\rho))\le \tr(\rho)\) for every
\(\rho\in\cT(\cH)_+\).  A completely positive trace-nonincreasing map is often called
a CPTN map for short, also there exist some textbooks that call it quantum operation or CPTNI.  CPTN maps send partial density operators to partial density
operators.

On the Heisenberg side, a linear map
\(\Phi:\cB(\cH)\to\cB(\cH)\) is positive if
\(\Phi(\cB(\cH)_+)\subseteq\cB(\cH)_+\).  It is completely positive if, for
every \(n\ge 1\), the amplification
\(\Phi\otimes I_n:\cB(\cH\otimes\mathbb C^n)\to
\cB(\cH\otimes\mathbb C^n)\) is positive.  It is subunital if
\(\Phi(I)\sqleq I\).  Finally, \(\Phi\) is normal if it preserves suprema of
bounded increasing nets of positive operators: whenever
\(0\le A_\alpha\uparrow A\) in \(\cB(\cH)\), one has
\[
    \Phi(A)=\sup_\alpha \Phi(A_\alpha).
\]

Now we can state the  Schr\"odinger-Heisenberg duality. For every CPTN map \(\cE:\cT(\cH)\to\cT(\cH)\) in Schr\"odinger picture, the trace pairing determines
a unique Heisenberg-picture map \(\cE^\dagger:\cB(\cH)\to\cB(\cH)\) satisfying
\[
    \tr(A\cE(\rho))=\tr(\cE^\dagger(A)\rho).
\]
This dual map is normal, completely positive, and subunital.  Conversely, normal completely positive subunital
maps on \(\cB(\cH)\) are exactly the Heisenberg-picture maps that have
trace-class CPTN preduals. There exist completely positive
subunital maps on \(\cB(\cH)\) that are not normal; such maps do not arise as
duals of trace-class Schr\"odinger-picture evolutions and are not used as program denotations in this paper. 

Furthermore, every CPTN map $\cE$ admits a Kraus representation:
there are \textbf{finite or countable} operators \(K_m\in\cB(\cH)\) such that
\(\sum_m K_m^\dagger K_m\sqleq I\) in the sense that
\[
    \sum_m \|K_m\u\|^2 \le \|\u\|^2
    \qquad(\forall\u\in\cH),
\]
and
\[
    \cE(\rho)=\sum_m K_m\rho K_m^\dagger .
\]
The state-side series converges in trace norm.  The corresponding Heisenberg formula is written
\[
    \cE^\dagger(A)=\sum_m K_m^\dagger A K_m ,
\]
but this equality is read through the trace pairing: for every
\(A\in\cB(\cH)\) and every \(\rho\in\cT(\cH)\),
\[
    \tr(\cE^\dagger(A)\rho)
    =
    \lim_{N\to\infty}
    \tr\!\left(
        \left(\sum_{m=1}^N K_m^\dagger A K_m\right)\rho
    \right).
\]
In finite dimension, one may always choose a Kraus representation with finitely
many operators; with such a choice, the displayed sums are finite matrix sums
and are exactly the familiar Kraus formulas for quantum operations and their
Heisenberg-picture action on observables. See \cite{kraus1971general,WolfQIPNotes,holevo2001statistical} for the details.

\subsection{Measure Spaces and Integration}
\label{subsec:measure-integration}

We fix the measure-theoretic conventions used for continuous families and
operator-valued integration.  Only standard measure theory is needed here.

A measurable space is a pair \((\Omega,\Sigma)\), where \(\Omega\) is a set and
\(\Sigma\) is a \(\sigma\)-algebra of subsets of \(\Omega\).  A measure space is a
triple \((\Omega,\Sigma,\mu)\), where \(\mu:\Sigma\to[0,\infty]\) is a countably
additive measure.  A measure space is \(\sigma\)-finite if
\(\Omega=\bigcup_{n\ge 0}\Omega_n\) for measurable sets \(\Omega_n\in\Sigma\) with
\(\mu(\Omega_n)<\infty\).  Throughout the paper, all measure spaces are assumed to be
\(\sigma\)-finite.  This is a common hypothesis for integration and product
measures; it includes finite or countable discrete spaces with counting measure
and Euclidean spaces such as \((\bR^n,\operatorname{Borel}(\bR^n),\lambda^n)\), where
\(\operatorname{Borel}(\bR^n)\) is the Borel \(\sigma\)-algebra and \(\lambda^n\) is
Lebesgue measure.

For finitely many measure spaces \((\Omega_i,\Sigma_i,\mu_i)\), their product is
\[
    \left(\prod_i \Omega_i,\ \bigotimes_i\Sigma_i,\ \bigotimes_i\mu_i\right).
\]
Here \(\bigotimes_i\Sigma_i\) is the product \(\sigma\)-algebra generated by
measurable rectangles, and \(\bigotimes_i\mu_i\) is the product measure.  Finite
products of \(\sigma\)-finite measure spaces are again \(\sigma\)-finite. In this paper, product \(\sigma\)-algebras are not completed.

If \((\Omega,\Sigma_\Omega)\) and \((Y,\Sigma_Y)\) are measurable spaces, a function
\(f:\Omega\to Y\) is measurable if \(f^{-1}(E)\in\Sigma_\Omega\) for every
\(E\in\Sigma_Y\).  When \(Y\) is a metric or topological space, we always use its
Borel \(\sigma\)-algebra.  Thus, a scalar function \(f:\Omega\to\bR\), a Hilbert-space
valued function \(f:\Omega\to\cH\), and a trace-class-valued function
\(f:\Omega\to\cT(\cH)\) are measurable when they are measurable with respect to the Borel
structures induced respectively by the usual norm, the Hilbert norm, and the
trace norm.

For non-negative scalar functions \(f:\Omega\to[0,\infty]\), the integral
\(\int_\Omega f\,d\mu\) is the usual extended Lebesgue integral and may take the value
\(+\infty\).  We use the standard monotone convergence theorem and Tonelli's
theorem for non-negative measurable functions without further comment.

For trace-class-valued functions, we use the Bochner
integral.  Intuitively, this is the operator-valued analogue of taking the
expectation of an integrable random variable.  Thus, if
\(f:\Omega\to\cT(\cH)\) is trace-norm measurable and
\(\int_\Omega \|f(x)\|_1\,d\mu(x)<\infty\), then its Bochner integral
\(\int_\Omega f(x)\,d\mu(x)\) is an element of \(\cT(\cH)\). For more details see \cite{cohn2013measure,dunford1958linear,diestel1977vector,hytonen2016analysis}.


\section{Syntax and Semantics of Continuous-Variable Quantum Programs}
\label{sec:syntax_semantics}

This section fixes the language model used throughout. We first recall the qwhile language, then extend it with the bind construct for modeling quantum measurement with  continuous outcomes, and finally present the admissibility conditions that make the denotational semantics well defined.

\subsection{Review of qwhile Language}
\label{subsec:_qwhile}
In this subsection, 
we briefly review the \textbf{qwhile} language  \cite{Ying11}, with its syntax and semantics.

\paragraph{Basic setting} Let $\mathsf{QVar}$ be a finite set of quantum variables. Each variable $q$ has an
associated finite or countably infinite-dimensional Hilbert space $\cH_q$. A  quantum register is a list of distinct quantum variables, usually denoted by $\seq q=q_1,\ldots,q_n$. Its Hilbert space is the tensor product
\[
    \cH_{\seq q}=\bigotimes_{i=1}^n \cH_{q_i}.
\]
We write $\cH_{\mathrm{all}}$ for the global Hilbert space of all the quantum variables, which is the tensor product of all $\cH_q$. By convention, if needed, every operator on a register $\seq q$ can be lifted to global $\cH_{\mathrm{all}}$ by using its cylinder extension, i.e., tensoring the identity operator $I_{\mathsf{QVar}\backslash \seq q}$ on the remaining variables.

\paragraph{Syntax} The \textbf{qwhile} programs are generated by the following syntax:
\[
\begin{aligned}
S ::= {}&
    \mathbf{skip} \mid \mathbf{abort} 
    \mid S_1;S_2
    \mid q:=\ket 0
    \mid \seq q := U[\seq q]
\\ &\mid
    \mathbf{if}\ 
      \bigl(\square_{m}\; M[\seq q]=m \to S_m\bigr)\
    \mathbf{fi}
\\ &\mid
    \mathbf{while}\ M[\seq q]=1\ \mathbf{do}\ S\ \mathbf{od}.
\end{aligned}
\]
\textbf{qwhile} is an extensions of classical \textbf{while} language with commands that characterize the common quantum transformations:
\(q := |0\rangle\) resets or initializes quantum variable \(q\) to the basis state \(|0\rangle\); \(\seq q := U[\seq q]\) applies the unitary
transformation \(U\) to the register \(\seq q\);
for conditional $\mathbf{if}\ 
      \bigl(\square_{m}\; M[\seq q]=m \to S_m\bigr)\
    \mathbf{fi}$, the guard is replaced by a 
quantum measurement, which means to first perform the measurement \(M\) on register \(\seq q\) with probabilistic outcome \(m\), and then execute subprogram \(S_m\). Here, the measurement and its outcome are usually assumed to be finite.
In the loop \(\mathbf{while}\ M[\seq q]=1\ \mathbf{do}\ S\ \mathbf{od}\), the measurement in the guard has only two outcomes \(0,1\) and it executes as: first perform measurement \(M\) on \(\seq q\) and, if the outcome is \(0\), the loop terminates, otherwise it executes the loop body \(S\) and enters the loop again.

\paragraph{Semantics}
Program states are modeled by partial density operators, i.e., positive 
operators of trace at most $1$. We write $\pardensity{\cH_{\mathrm{all}}}$ for the set of all program states on global space.
The denotational semantics of a program $S$, written 
\(\sem S:\pardensity{\cH_{\mathrm{all}}}\to \pardensity{\cH_{\mathrm{all}}}\), is a state transformer defined structurally by the \textbf{qwhile} clauses:
\begin{itemize}
    \item \(\sem{\mathbf{abort}}(\rho)
    \triangleq
    0\);
    \item \(\sem{\mathbf{skip}}(\rho)
    \triangleq
    \rho\);
    \item \(\sem{q:=\ket 0}(\rho)
    \triangleq
    \sum_k
    \ket 0_q\!\bra k\cdot
    \rho\cdot
    \ket k_q\!\bra 0\), where $\{|k\>\}_k$ is the computational basis of $\cH_{q}$;
    \item \(\sem{\seq q:=U[\seq q]}(\rho)
    \triangleq
    U\rho U^\dagger\);
    \item \(\sem{S_1;S_2}(\rho)
    \triangleq
    \sem{S_2}(\sem{S_1}(\rho))\)
    \item \(\sem{
  \mathbf{if}\ 
    (\square_{m}\; M[\seq q]=m\to S_m)\
  \mathbf{fi}
}(\rho)
    \triangleq
    \sum_{m}
    \sem{S_m}(M_m\rho M_m^\dagger)\), where $\{M_m\}$ are measurement operators of \(M\);
    \item \(\sem{
  \mathbf{while}\ M[\seq q]=1\ \mathbf{do}\ S\ \mathbf{od}
}(\rho)
    \triangleq
    \bigvee_{n=0}^\infty \sem{
  \mathbf{while}^{(n)}
}(\rho)\),
where $\mathbf{while}^{(n)}$ is the $n$-th syntactic approximation of the loop, i.e., inductively defined by $\mathbf{while}^{(0)} \triangleq \mathbf{abort}$, and $\mathbf{while}^{n+1} \triangleq \mathbf{if}\ 
    M[\seq q]=0\to \mathbf{skip}\ \square\ 1 \to S; \mathbf{while}^{(n)} \ 
  \mathbf{fi}$.
\(\bigvee\) stands for the least upper bound of partial density operators according to the L\"owner order.
\end{itemize}


\subsection{Continuous-outcome bind and admissible programs}
\label{subsec:continuous_bind}

The syntax recalled above is adequate for programs whose measurements have
finite outcomes. 
However, many physically realistic measurements, such as measurements on position,
 return a value in a non-discrete outcome space, typically a value in $\mathbb R$. 
To model this kind of measurement for CV quantum programs, we extend the conditional command with an explicit bind, written as
\[
    \bind(M[\seq q],x.S).
\]
It performs the measurement \(M\) on the register \(\seq q\),
binds the classical outcome to the symbol \(x\), and continues as \(S\) which is parameterized by \(x\).
When the outcome set is finite, such parameterization of \(S\) can be enumerated as branches \(S_m\) and therefore reduces to the ordinary conditional command reviewed in \Cref{subsec:_qwhile}. 
In probabilistic programming, the continuous-outcome sampling construct assumes a bound variable
\(x\) to range over a measurable outcome space.
We adopt the same idea here, coincident to the measurable structure for continuous quantum instruments~\cite{davies1976quantum}.

Formally, the extended syntax of CV quantum programs is:
\[
\begin{aligned}
S ::= {}&
    \mathbf{skip}
    \mid \mathbf{abort}
    \mid S_1;S_2
    \mid q:=\ket 0
    \mid \seq q:=U[\seq q]
\\ &\mid
    \bind(M[\seq q],x.S)
    \mid
    \mathbf{while}\ M[\seq q]=1\ \mathbf{do}\ S\ \mathbf{od}.
\end{aligned}
\]
The commands inherited from \textbf{qwhile} have the same meaning and semantics as
before.
One notable thing is the introduction of $x$ in $\mathbf{bind}$ command, which allows for parameterizing branches using $x$. For simplicity, we view $x$ as a classical meta-variable rather than a classical program variable, and therefore, the unitary transformation or measurement can also be parameterized by these meta-variables; we use \(\mathrm{mv}(U)\) or \(\mathrm{mv}(M)\) to denote the set of meta-variables appearing in $U$ or $M$.

\paragraph{Well-formed programs}
Formally, we introduce the classical meta-variable context $\Gamma$, which collects the allowed meta-variables and will be essential when we handle the combinations of measurable parameters.
We write
\[
    \Gamma\vdash S
\]
to mean that \(S\) is well-formed with respect to the finite context
\(\Gamma\), i.e., all free meta-variables used in $S$ are contained in $\Gamma$.
Note that the context need not be minimal: unused parameters are allowed, and all judgments are closed under weakening.
The well-formedness is defined structurally as follows:
\begin{figure}
    \centering
\begin{mathpar}
\mprset{flushleft}
\inferrule*[right=\textsc{Skip}]
    { }
    { \Gamma \vdash \mathbf{skip} }

\inferrule*[right=\textsc{Abort}]
    { }
    { \Gamma \vdash \mathbf{abort} }
\\
\inferrule*[right=\textsc{Init}]
    { }
    { \Gamma \vdash q := \ket{0} }

\inferrule*[right=\textsc{Unitary}]
    { \mathrm{mv}(U)\subseteq \Gamma}
    { \Gamma \vdash q := U[q] }
\\
\inferrule*[right=\textsc{Seq}]
    { \Gamma_1 \vdash S_1 \\ \Gamma_2 \vdash S_2 }
    { \Gamma_1\cup\Gamma_2 \vdash S_1;S_2 }
    
\inferrule*[right=\textsc{Bind}]
    { \mathrm{mv}(M)\subseteq \Gamma
      \qquad
      \Gamma,x \vdash S }
    { \Gamma \vdash \mathbf{bind}(M[\seq q],x.S) }
\\
\inferrule*[right=\textsc{While}]
    { \mathrm{mv}(M)\subseteq \Gamma_1
      \qquad
      \Gamma_2 \vdash S }
    { \Gamma_1\cup \Gamma_2 \vdash
        \mathbf{while}\ M[\seq q]=1\ \mathbf{do}\ S\ \mathbf{od}
    }

\inferrule*[right=\textsc{Weak}]
    { \Gamma \vdash S \qquad\Gamma \subseteq \Delta }
    { \Delta \vdash S }
\end{mathpar}
\caption{Well-formedness}
    \label{fig:well-formedness}
     \Description{well-formedness}
\end{figure}


\paragraph{Admissible programs} We introduce the concept of admissible programs to check if the parameterized program satisfies certain measurability over the meta-variable context. 
In our programs, each meta-variable \(x\) ranges over its associated measure
space \((\Omega_x,\Sigma_x,\mu_x)\).  Thus, for a finite context \(\Gamma\),
the valuation space is the product measure space
\[
    (\Omega_\Gamma,\Sigma_\Gamma,\mu_\Gamma)\triangleq
    \left(
        \prod_{x\in\Gamma}\Omega_x,\,
        \bigotimes_{x\in\Gamma}\Sigma_x,\,
        \bigotimes_{x\in\Gamma}\mu_x
    \right).
\]
We write \(\mathcal V_\Gamma\) for the same set when its elements are viewed as
valuations:
\[
    \mathcal V_\Gamma
    \triangleq
    \{\omega \mid \omega(x)\in\Omega_x \text{ for every } x\in\Gamma\}
    \cong
    \Omega_\Gamma .
\]
Given a valuation \(\omega\in\mathcal{V}_\Gamma\),
i.e., an assignment of meta-variables, and $\Gamma\vdash S$, 
we write \(S_\omega\) for the instantiated program of \(S\), whose free meta-variables are all fixed by value according to $\omega$.

The involved commands include unitary transformation, while loop and bind, since both $U$ and $M$ are parameterized.
For a unitary command \(\seq q:=U[\seq q]\), it is admissible if
its dependence on the parameter valuation $\omega$ is strongly measurable: for every
\(\u\in\cH_{\seq q}\), the map \(\omega\mapsto U_{\omega}\u\) is measurable; for a $\mathbf{while}$ command, it is admissible if the two-outcome loop guard $\{M_{\omega,0},M_{\omega,1}\}$ satisfying the same strongly measurable condition as unitary command.

For the bind command, its admissible conditions are more complicated: 
    \begin{itemize}
    \item for each valuation $\omega\in\mathcal V_\Gamma$, its measurement $M$ bind the outcome to a fresh meta variable $x$, whose associated measure space is $(\Omega_M,\Sigma_M,\mu_M)$, and $M_{\omega}$ admits a (not necessarily countable) Kraus density:
    \[
    \{M_{\omega,\nu}\}_{\nu\in\Omega_M}\subseteq\cB(\cH_q);
    \]
    \item for each valuation $\omega\in\mathcal V_\Gamma$, the Kraus density of measurement $M$ satisfies the normalization 
    \[
    \int_{\Omega_M}\|M_{\omega,\nu}\u\|^2\,d\mu_M(\nu)=\|\u\|^2,
    \qquad \forall \u\in\cH_{\seq q},
    \]
    which is the continuous counterpart of
    \(\sum_m M_m^\dagger M_m=I\) in the discrete syntax;
    \item \((\omega,\nu)\mapsto M_{\omega,\nu}\u\) is joint measurable for
    every \(\u\in\cH_{\seq q}\).
    \end{itemize}

\subsection{Denotational semantics with parameters}
\label{subsec:denotational_semantics}

We now define the denotational semantics of well-formed programs.  Since the
language has an explicit classical parameter layer, denotations are first
defined after fixing a parameter valuation.  

If \(\Gamma\vdash S\) and
\(\omega\in\Omega_\Gamma\), then all parameter expressions in \(S\) are
evaluated at \(\omega\), and the program denotes a mapping:
\[
\sem{S_\omega}:
\pardensity{\cH_{\mathrm{all}}}\to\pardensity{\cH_{\mathrm{all}}}.
\]
There is also a second, compositional requirement: for each fixed input state
\(\rho\), the output \(\sem{S_\omega}(\rho)\) should depend measurably on
\(\omega\).  This is the quantum analogue of the kernel-measurability condition
in probabilistic semantics with sampling; the additional technical point is
that the values here are partial density operators rather than scalar
probabilities.

For the commands inherited from \textbf{qwhile}, the denotational clauses are the
 ones from Section~\ref{subsec:_qwhile}, with all 
unitaries and measurements instantiated by the current parameter valuation.
The command \(\mathbf{abort}\) denotes the zero state transformer.  Thus the
only new semantic clause is bind. Since the denotation of condition command in \textbf{qwhile} sums over
measurement outcomes, intuitively, the continuous bind replaces this sum by an integral:
\[
    \sem{\bind(M[\seq q],x.S)_{\omega}}(\rho)
    \triangleq
    \int_{\Omega_M}
        \sem{S_{\omega,x\coloneqq\nu}}
        \bigl(M_{\omega,\nu}\rho M_{\omega,\nu}^{\dagger}\bigr)
    \,d\mu_M(\nu).
\]
Here \(M_{\omega,\nu}\) is the Kraus density selected by the current
parameter valuation and outcome \(\nu\), and \(\sem{S_{\omega,x\coloneqq\nu}}\) is denotation of the continuation
whose contexts are evaluated at $\omega\in\Gamma$ and $x=\nu$.  The expression inside the integral is a
trace-class branch state.

The main question is the meaning of this integral. For bind we need more than outcome-wise scalar probabilities: after integrating over outcomes, the result must be a trace-class state that can be passed to the next command. Our admissibility assumptions in \Cref{subsec:continuous_bind} ensure that the branch-state map is trace-class-valued measurable and integrable. Thus the displayed integral is a Bochner integral (a type of operator-valued Lebesgue integral); when tested against bounded observables it agrees with the physical quantum-instrument model.

The proof of parameter measurability proceeds by structural induction on
programs. The cases of basic commands  ($\mathbf{skip},\,\mathbf{abort},\,q:=\ket{0}$) are trivial; $\seq q:=U[\seq q]$ use the admissibility assumptions on
parameterized unitary.  Sequencing uses composition of
measurable maps, and loops use measurable finite approximants together with
monotone convergence in the trace-class state domain.  The bind case is the only
substantive one: the induction hypothesis for the continuation gives
measurability in the extended context \(\Gamma,x\), and the Bochner
integral then integrates over \(x\) while preserving measurability in the
remaining parameters.  Thus the definition of bind and the measurability
statement are proved by a simultaneous structural induction, not by a circular
argument.  The measure-theoretic details are routine but lengthy and are
deferred to the supporting material.

In summary, we can have the following:
\begin{theorem}[Well-defined denotational semantics]
\label{THM:semantics_well_defined}
If \(\Gamma\vdash S\) and $S$ is admissible, then for every valuation
\(\omega\in\mathcal{V}_\Gamma\), the clauses above define a CPTN map
\[
    \sem{S_\omega}:\pardensity{\cH_{\mathrm{all}}}\to\pardensity{\cH_{\mathrm{all}}}.
\]
Moreover, for every fixed \(\rho\in\pardensity{\cH_{\mathrm{all}}}\), the map
\[
    \omega\longmapsto \sem{S_\omega}(\rho)
\]
is measurable from \(\mathcal{V}_\Gamma\) to $\pardensity{\cH_{\mathrm{all}}}$ equipped with the
trace-norm Borel structure.
\end{theorem}

Consequently, after fixing classical parameters, every well-formed and admissible program has
the same kind of Schr\"{o}dinger-picture denotation as in  \textbf{qwhile} Language: a
CPTN mapping on trace-class states.  
The continuous bind construct
does not change the class of denotations; it only replaces finite outcome sums by
Bochner integrals whose measurability is guaranteed by the formation judgment. 

In the rest of the paper, all semantic constructions are considered pointwise
in the valuation.  That is, we first fix an arbitrary
\(\omega\in\mathcal V_\Gamma\), instantiate the program to \(S_\omega\), and then
study the induced quantum state transformer and its predicate transformers.
For readability, we suppress this fixed valuation from the notation: unless
otherwise stated, \(S\), \(\sem{S}\), \(\wlp(S,Q)\), \(\wp_c(S,Q)\) should be read as
\(S_\omega\), \(\sem{S_\omega}\), \(\wlp(S_{\omega},Q)\), \(\wp_c(S_\omega,Q)\), respectively, for an
arbitrary fixed \(\omega\in\mathcal V_\Gamma\).

%


\section{Assertions and Predicate Transformers for Unbounded Quantities}
\label{sec:unbounded_predicates}

We now turn to the assertion layer and backward reasoning. Quantum predicates between the zero and identity operators employed in   \cite{dhondt2006quantum} are not sufficient for our target validity specifications, such as those about energy, expected running time, etc., not being bounded effects. 
To address this, we adopt the positive closed quadratic forms~\cite{LICS2026LinearRelation} as predicates to capture the unbounded behaviors, and then develop weakest-precondition transformers, in particular addressing the healthiness for continuous measurement.

\subsection{Review of Bounded Predicate Transformers and Unbounded Predicate Domain}
\label{subsec:from_bounded_to_unbounded}

We start from the standard predicate-transformer view of quantum programs.  In \cite{dhondt2006quantum},
predicates are bounded quantum effects ordered by the L\"owner order.  
A predicate transformer is regarded as healthy when it satisfies the following conditions\footnote{\cite{dhondt2006quantum} additionally requires healthy conditions to contain monoidality: that is, a local predicate transformer can be lift to a global one; in our setting we always work on the global Hilbert space, so we do not use the local-to-global formulation explicitly. }.

\begin{itemize}
    \item \emph{Linearity.}
    It respects linear combinations of bounded predicates.  

    \item \emph{Monotonicity.}
    It preserves the L\"owner order on predicates.

    \item \emph{$\omega$-Continuity.}
    It preserves suprema of increasing predicate chains.

\end{itemize}

For bounded predicates, CPTN state transformers induce healthy predicate
transformers \cite{dhondt2006quantum,WolfQIPNotes}.  Equivalently, every CPTN map $\cE$ has a Heisenberg dual
$\cE^\dagger$ on bounded predicates, where $\cE^\dagger$ has countable Kraus representation $\cE^{\dagger}(\cdot)=\sum_{m=1}^{\infty}K_m^{\dagger}(\cdot) K_m$ when $\cE(\rho)=\sum_{m=1}^{\infty}K_m \rho K_m^{\dagger}$, and $\cE^{\dagger}$ is characterized by the trace pairing
\[
    \tr(\cE^\dagger(A)\rho)
    =
    \tr(A\,\cE(\rho)).
\]
Thus the bounded quantum weakest precondition equals the Heisenberg dual:
for the program $S$, its denotational semantics $\sem{S}$ is a CPTN map, thus the Heisenberg dual $\sem{S}^{\dagger}$ satisfies the above trace pairing identity, so the predicate transformer is defined as
\[
    \wp(S,Q)\ \triangleq\ \sem{S}^{\dagger}(Q).
\]
This means that backward reasoning computes the same expectation as forward execution.

This bounded predicate domain covers effects, projections, and bounded
probabilities, but it cannot express the unbounded assertions
that arise with continuous variables. Continuous-variable
programs use infinite-dimensional state spaces, and the bind command in
\Cref{subsec:continuous_bind} allows later commands to depend on real-valued
outcomes.  Natural expectation-style assertions---
variances, energy, expected running time---are
often unbounded.  The predicate domain must also represent hard constraints,
such as staying inside a subspace or inside the finite-energy domain of an
observable, where violation behaves as an infinite penalty.
Intuitively one may use unbounded operators as the predicate domain. However, raw unbounded operators do not give such a predicate domain: weakest
preconditions must be closed under bounded pullback, branch sums, continuous
bind, and loop suprema, but ordinary operator products, sums, and integrals may
have domain mismatches.

To treat unbounded quantities and hard constraints uniformly, recent
work~\cite{LICS2026LinearRelation} proposed using linear relations / closed
quadratic forms as predicates.  We follow this viewpoint, but take the positive
closed quadratic forms as our predicate domain, denoted by $\Pred(\cH)$, and use positive self-adjoint
linear relations as their algebraic representation. 

Formally, a predicate \(P\) is a closed positive quadratic form
\[
    \QF_P:\cH\to[0,+\infty].
\]
Its finite-value domain is
\[
    \Dom(\QF_P)\triangleq\{\u\in\cH:\QF_P[\u]<+\infty\}.
\]
Closedness means lower semicontinuity: whenever \(\u_n\to\u\) in \(\cH\),
\(\QF_P[\u]\leq \liminf_{n\to\infty}\QF_P[\u_n]\).
This condition is essential: it is what makes the predicate domain closed under
increasing suprema and under the backward predicate-transformer operations used
below.  Every bounded predicate (i.e., positive bounded operator) \(A\) is viewed as an element in $\Pred(\cH)$ by the inner-product formula
\[
    \QF_A[\u]\triangleq\langle \u,A\u\rangle .
\]
Values \(+\infty\) outside the finite-value domain record the hard-constraint
part of the predicate.

We recall the following properties of this predicate domain, which are introduced in \cite{LICS2026LinearRelation}.
Let $P,Q\in\Pred(\cH)$ with the associated quadratic form denoted by $\QF_P,\QF_Q$; and let
$\rho\in\pardensity{\cH}$.

\begin{itemize}
    \item \emph{Extended trace.}
    If $\rho=\sum_k p_k |\u_k\rangle\langle\u_k|$ is a spectral
    decomposition, define
    \[
        \Tr(P\rho)\triangleq\sum_k p_k\,\QF_P[\u_k]\in[0,+\infty].
    \]
    This value is independent of the decomposition and agrees with the usual trace
    $\tr(A\rho)$ when $P$ is represented by a bounded positive operator $A$.

    \item \emph{Extended L\"owner order.}
    We write $P\sqleq Q$ iff $\QF_P[\u]\leq\QF_Q[\u]$ for every
    $\u\in\cH$, equivalently iff
    $\Tr(P\rho)\leq\Tr(Q\rho)$ for every partial density operator $\rho$.
    Domain inclusion is contravariant: larger predicates may have smaller
    finite-value domains.

    \item \emph{$\omega$-CPO property:}
    Every increasing chain $P_0\sqleq P_1\sqleq\cdots$ has a supremum
    $P_\infty\in\Pred(\cH)$, characterized by
    \[
        \Tr(P_\infty\rho)=\sup_n\Tr(P_n\rho).
    \]

    \item \emph{Bounded approximation.}
    Every predicate $P\in\Pred(\cH)$ is the supremum of an increasing sequence of bounded positive operators. In particular, we can take the spectral truncation to approximate $P$.
\end{itemize}

\subsection{Unbounded Predicate Transformer}
\label{subsec:Unbounded_Predicate_Transformer}


With the facts stated above, we can describe the intended predicate transformer
through the Kraus representation.  Suppose
\[
    \cE(\rho)=\sum_{m=1}^{\infty}K_m\rho K_m^\dagger .
\]
For an unbounded predicate \(P\), we write \(\cE^\dagger(P)\) (by abuse of notation) for the predicate whose quadratic form is intended to satisfy
\begin{equation}\label{eq:def-of-predicate-trans}
    \QF_{\cE^\dagger(P)}[\u]
    =
    \sum_{m=1}^{\infty}\QF_P[K_m\u].
\end{equation}
This is formally the same formula as in the bounded case.  Indeed, when
\(P\) is represented by a bounded positive operator \(A\), we have
\[
    \QF_{\cE^\dagger(A)}[\u]=
    \sum_{m=1}^{\infty} \QF_A[K_m\u]
    =
    \sum_{m=1}^{\infty} \langle  K_m\u,A K_m\u\rangle
    =
    \left\langle
    \u,\Bigl(\sum_{m=1}^{\infty} K_m^\dagger A K_m\Bigr)\u
    \right\rangle,
\]
which is exactly the quadratic form of the bounded predicate transformer
\(\cE^\dagger(A)=\sum_m K_m^\dagger A K_m\).
The right-hand side of \Cref{eq:def-of-predicate-trans} is allowed to be \(+\infty\), so it records both the
quantitative pre-expectation and the pulled-back hard constraint.   In the next subsection we will show that this construction indeed yields
a predicate and satisfies the exact duality
\[
    \Tr(\cE^\dagger(P)\rho)=\Tr(P\,\cE(\rho)).
\]

In this paper, for a bounded operator \(K\in\cB(\cH)\) and a predicate
\(P\in\Pred(\cH)\), we write
\[
    K^\dagger P K\in\Pred(\cH)
\]
for the quadratic form defined by
\[
    \QF_{K^\dagger P K}[\u]\triangleq\QF_P[K\u].
\]
This is a quadratic-form-level notation, not multiplication of unbounded operators.

The next theorem gives the math foundation of turning the Kraus-level intuition into a well-defined healthy
predicate transformer.

\begin{theorem}[Unbounded Heisenberg Dual]
\label{thm:unbounded_heisenberg_transformer}
Let $\cE$ be a CPTN map.  For every $P\in\Pred(\cH)$, there exists a
unique predicate $\cE^\dagger(P)\in\Pred(\cH)$ such that, for every partial
density operator $\rho$,
\begin{equation}\label{eq:trace-dual-eq}
    \Tr(\cE^\dagger(P)\rho)=\Tr(P\,\cE(\rho)).
\end{equation}
If
\(\cE(\rho)=\sum_{m=1}^{\infty} K_m\rho K_m^\dagger
\)
is any Kraus representation, then
\[
    \QF_{\cE^\dagger(P)}[\u]
    =
    \sum_{m=1}^{\infty} \QF_P[K_m\u],
\]
and this form is independent of the chosen Kraus representation.  For bounded
$P$, this agrees with the ordinary Heisenberg dual on bounded observables.

Moreover, 
\[
    \cE^\dagger:\Pred(\cH)\to\Pred(\cH)
\]
satisfies the healthy conditions: it is linear,
monotone and $\omega$-continuous. 
\end{theorem}

\begin{proof}[Proof Sketch]
We can approximate the unbounded postcondition \(P\) from below by bounded positive
predicates, apply the bounded Heisenberg dual at each stage, and take the increasing supremum; $\omega$-CPO property ensures the supremum lies in \(\Pred(\cH)\).

The trace identity \Cref{eq:trace-dual-eq} is inherited from the bounded trace identities by the
\emph{Extended trace} and the \emph{\(\omega\)-CPO Property}, and uniqueness
follows because predicates are determined by their expectations on pure states.

The Kraus formula is the pure-state version of the trace identity; independence of Kraus representation and the healthiness conditions
follow from the same trace characterization and monotone-limit argument.
\end{proof}

Therefore, we can define the unbounded predicate transformer of the program as follows:
\begin{definition}[Weakest liberal precondition]
\label{def:unbounded_wlp}
Let \(S\) be a well-formed program with Schr\"odinger-picture denotation
\(\sem{S}\), and let \(Q\in\Pred(\cH)\) be a postcondition.  The
\emph{weakest liberal precondition} of \(S\) with respect to \(Q\) is the
predicate
\[
    \wlp(S,Q)
    \triangleq
    \sem{S}^{\dagger}(Q).
\]
Equivalently, \(\wlp(S,Q)\) is the unique predicate satisfying
\[
    \Tr(\wlp(S,Q)\rho)
    =
    \Tr(Q\,\sem{S}(\rho))
\]
for every partial density operator \(\rho\).  If
\(\sem{S}(\rho)
    =
    \sum_{m=1}^{\infty}K_m\rho K_m^\dagger
\)
is a Kraus representation of \(\sem{S}\), then \(\wlp(S,Q)\) is characterized
at the quadratic-form level by
\(
    \QF_{\wlp(S,Q)}[\u]
    =
    \sum_{m=1}^{\infty}\QF_Q[K_m\u]
\).
\end{definition}

We call this transformer \emph{liberal} because it does not imply any terminability.  In Section~\ref{sec:unified_proof_system} it will appear as the
zero-cost instance \(\wp_0\) of the cost-parametric transformer; under the
validity convention used there, this case carries no termination obligation.

\subsection{Predicate Transformers for Continuous Measurement}
\label{subsec:continuous_branch_predicate_transformers}
The previous construction gives the transformer for a single CPTN map.
Bind command requires one further closure principle: the postcondition may
depend on a measurement outcome, and the backward predicate must combine all
branches.  For a finite or countable measurement this is the usual branch sum
\[
    \sum_\nu M_\nu^\dagger P_\nu M_\nu .
\]
For a continuous outcome space, the same operation should be read as a
continuous branchwise predicate transformer.  Since the branch predicates
\(P_\nu\) may be unbounded, we do not integrate unbounded operators.  Instead,
we integrate the scalar form values
\[
    \QF_{P_\nu}[M_\nu u]
\]
for each input vector \(u\).

\begin{definition}[Continuous-Measurement predicate transformer]
\label{def:continuous_branch_form_transformer}
A continuous measurement branch consists of measurement operators
\(M_\nu\), indexed by outcomes \(\nu\in\Omega\).  Suppose that branch \(\nu\)
has postcondition \(P_\nu\).  The backward predicate is the continuous analogue
of the finite branch sum
\(\sum_i M_i^\dagger P_i M_i \).

For an input vector \(\u\), branch \(\nu\) contributes the nonnegative value
\[
    \QF_{P_\nu}[M_\nu\u].
\]
When these branch values form a measurable function of \(\nu\), we say the family $\{P_\nu\}_{\nu\in\Omega_M}$ is \textbf{$M$-admissible}, and define the
continuous-branch transformer by averaging these values:
\[
    \QF_P[\u]
    \triangleq
    \int_\Omega \QF_{P_\nu}[M_\nu\u]\,d\mu(\nu).
\]
We write this predicate formally as
\[
    P
    \triangleq
    \int_\Omega M_\nu^\dagger P_\nu M_\nu\,d\mu(\nu).
\]
This notation denotes the quadratic-form-level construction above; it is not an integral of unbounded operators.
\end{definition}

\begin{proposition}[Closure and trace identity for continuous branching]
\label{prop:continuous_branch_trace_identity}
If the branch family \((P_\nu)_{\nu\in\Omega}\) is $M$-admissible, then the construction above defines a predicate \(P\in\Pred(\cH)\). Moreover, for every partial density operator \(\rho\),
\[
    \Tr(P\rho)
    =
    \int_\Omega
        \Tr(P_\nu\,M_\nu\rho M_\nu^\dagger)
    \,d\mu(\nu).
\]
\end{proposition}

\begin{proof}[Proof sketch]
This is the same construction as the finite branching rule: each branch assigns a nonnegative predicate to the input vector, and the precondition is obtained by adding these branch predicates.
The only difference is that a continuous measurement has a continuum of branches, so the finite sum is replaced by an integral over the measurement outcome.

Because scaling and parallelogram identities hold in every branch and are
preserved by averaging, the result is again a positive quadratic form; because
branch costs cannot decrease abruptly under limits, by Fatou's Lemma their average cannot decrease
abruptly either, i.e., the form is lower semi-continuous, that is, the form is closed.

For pure states the trace identity is exactly the definition of the averaged
branch predicate, and for mixed states it follows by writing the state as a
probabilistic mixture of pure states and exchanging the two nonnegative averages by Tonelli's theorem:
one over mixture components and one over measurement outcomes.
\end{proof}

The above proposition gives the weakest liberal precondition for the bind command.
Let
\[
    B \equiv \bind(M[\seq q],x.S),
\]
and fix a postcondition \(Q\in\Pred(\cH)\).  For each outcome
\(\nu\in\Omega_M\), write \(S_\nu\) for the continuation obtained by setting
\(x\triangleq\nu\).  Then
\[
    \wlp(B,Q)
    \triangleq
    \int_{\Omega_M}
        M_\nu^\dagger\,\wlp(S_\nu,Q)\,M_\nu
    \,d\mu_M(\nu),
\]
meaning, at the quadratic-form level,
\[
    \QF_{\wlp(B,Q)}[\u]
    =
    \int_{\Omega_M}
        \QF_{\wlp(S_\nu,Q)}[M_\nu\u]
    \,d\mu_M(\nu).
\]
The well-formedness and admissibility assumptions on bind ensure that the family
\(\{\wlp(S_\nu,Q)\}_{\nu\in\Omega_M}\) is \(M\)-admissible, so the integral
defines a predicate in \(\Pred(\cH)\).  By the trace identity above and the
definition of \(\wlp(S_\nu,Q)\), this predicate satisfies
\[
    \Tr(\wlp(B,Q)\rho)
    =
    \Tr(Q\,\sem{B}(\rho)),
\]
and hence is exactly the weakest liberal precondition of the bind command.

When \(\Omega_M\) is countable with counting measure, this rule reduces to the
usual measurement-branch sum.  Thus continuous bind is handled by the similar
backward rule as discrete measurement, with sums replaced by form-level integrals.

\section{A Cost-Parametric Hoare Logic}
\label{sec:unified_proof_system}

The previous section gives the unbounded predicate transformer for CV quantum programs. 
Now we are ready to present a cost-parametric quantum Hoare logic, which provides the proper extension that reflect the intuitive meaning of partial correctness and total correctness when unbounded predicates are involved.

\subsection{Review: Expected Runtime Transformers}
\label{subsec:review_expected_runtime_transformers}

We aim to develop a Hoare-style system based on compositional quantitative predicate transformers, treating assertion transformation and cost accumulation uniformly. Under parametric cost perspective, changing the cost interpretation affects only primitive-command costs, while rules for sequencing, bind, and loops remain fixed. In backward reasoning, we start with a quantitative requirement after the command. The backward transformer then computes the corresponding requirement before the command, including the cost of executing it.
Costs and continuations inhabit a common quantitative assertion domain. This yields expected-runtime transformer calculi, where runtime is internalized as a quantitative assertion rather than an auxiliary state variable~\cite{kaminski2018weakest,batz2021relatively}. 
Probabilistic ERT uses extended-real-valued expectations, while its finite-dimensional quantum analogue uses positive bounded operators related, under almost-sure termination, to forward state semantics by trace pairing~\cite{liu2025quantum}. Our semantics extends this principle to a quadratic-form-valued assertion domain.

\paragraph*{Probabilistic programs.}
For probabilistic programs, assertions are expectations
\(f:\Sigma\to[0,\infty]\). The expected-runtime transformer
\(\ert\llbracket C\rrbracket(t)\) takes a continuation expectation \(t\) and
returns the expected value of the accumulated runtime plus the continuation
after executing \(C\). Its relationship with the weakest
preexpectation is
\[
    \ert\llbracket C\rrbracket(t)
    =
    \ert\llbracket C\rrbracket(0)
    +
    \wp\llbracket C\rrbracket(t).
\]
Thus \(\ert\llbracket C\rrbracket(0)\) is not an external scalar annotation:
it is itself an expectation, i.e. a predicate as
ordinary quantitative assertions. 

\paragraph*{Finite-dimensional quantum programs.}
The finite-dimensional quantum ERT calculus keeps the same cost-accounting
structure, but separates two closely related views~\cite{liu2025quantum}. There are two views. The scalar
view assigns an expected runtime \(\mathrm{ERT}[S](\rho)\) to an input state \(\rho\).  The predicate-transformer view maps a continuation
observable \(A\) to an observable \(\ert[S](A)\).  The zero-continuation case
\(\ert\) is the runtime observable.  On the almost-surely terminating
inputs considered in that setting, the scalar and predicate views are linked
by trace duality:
\[
    \mathrm{ERT}[S](\rho)
    =
    \tr(\ert\rho),
    \qquad
    \tr(\ert[S](A)\rho)
    =
    \tr(\ert[S]\rho)
    +
    \tr(A\llbracket S\rrbracket(\rho)).
\]
This is the finite-dimensional quantum analogue of the probabilistic fact that
runtime is a predicate: expected runtime can be represented by an observable,
and the continuation transformer is dual to the forward state semantics plus
the accumulated cost.

However, when dealing with infinite-dimensional systems, the bounded predicate is not enough to capture the runtime behavior, since predicates may be unbounded and naturally take the value \(+\infty\). 
By enlarging the predicate domain to
closed positive quadratic forms, we need not impose almost-sure termination as
an external side condition: the cost predicate is defined for arbitrary
inputs, and finiteness of its expected value entails finite expected cost and
hence almost-sure termination. Thus termination becomes a consequence
expressible within the assertion semantics rather than a prerequisite of the
transformer construction, providing the basis for the sound and relatively
complete cost reasoning system developed below.

\subsection{Cost as a Predicate}
\label{subsec:cost_and_correctness}

In this section, all costs, traces, and quadratic-form values range over $[0,\infty]$, with
$a+\infty=\infty$, $a\cdot\infty=\infty$ for $a>0$, and $0\cdot\infty=0$.
Scalar multiplication on predicates is pointwise, i.e., 
$(\lambda\QF)[u]\triangleq\lambda\QF[u]$.  We write $\top$ for the predicate
$$\QF_{\top}[\u]\triangleq
\left\{\begin{aligned}
    &0,\,\u=\mathbf{0};\\
    &+\infty\text{,\,otherwise.}
\end{aligned}\right.~.$$

We assume that the execution of any single command (e.g., a basic unitary, a measurement, or evaluating a loop guard) consumes a state-independent base cost $c \ge 0$.

\begin{definition}[Cost Functional]
\label{def:cost_predicate}
The cost $\cC_c(S,\rho)$ is defined structurally as follows:
\[
\begin{alignedat}{2}
& \cC_c(\mathbf{skip},\rho)
\qquad
& \triangleq \quad
& c\Tr(\rho),
\\
& \cC_c(\mathbf{abort},\rho)
& \triangleq\quad
& c\Tr(\rho)\cdot(+\infty),
\\
& \cC_c(\seq q\coloneqq\ket{0},\rho)
& \triangleq\quad
& c\Tr(\rho),
\\
& \cC_c(\seq q\coloneqq U[\seq q],\rho)
& \triangleq\quad
& c\Tr(\rho),
\\
& \cC_c(S_1;S_2,\rho)
& \triangleq\quad
& \cC_c(S_1,\rho)
  +\cC_c\bigl(S_2,\sem{S_1}(\rho)\bigr),
\\
& \cC_c\bigl(\mathbf{bind}(M[\seq q],x.S),\rho\bigr)
& \triangleq\quad
& c\Tr(\rho)
  +\int_{\Omega_M}
    \cC_c\bigl(S_\nu,M_{\nu}\rho M_{\nu}^\dagger\bigr)
    \,d\mu_M(\nu),
\\
& \cC_c\bigl(
    \mathbf{while}\ M[\seq q]=1\
    \mathbf{do}\ S\ \mathbf{od},\rho
  \bigr)
& \triangleq\quad
& \sum_{n=0}^{\infty}
  \Bigl(
    c\Tr(\rho_n)
    +\cC_c\bigl(S,M_1\rho_nM_1^\dagger\bigr)
  \Bigr).
\end{alignedat}
\]
where in the last clause $\rho_0\triangleq\rho$ and
$\rho_{n+1}\triangleq\sem S(M_1\rho_nM_1^\dagger)$.
\end{definition}

Definition~\ref{def:cost_predicate} preserves the scalar view of \cite{liu2025quantum} reviewed above.  Sequencing evaluates the continuation cost of \(S_2\)
on the state produced by \(S_1\), and loops are still the monotone limits of
finite unrollings.  The only structural replacement is the branching operation:
the finite sum over conditional branches is replaced by the integral over the
outcomes of \(\mathbf{bind}\).

The real change is the representation of the runtime observable. In the finite-dimensional ERT calculus, the scalar runtime is recovered from an observable by the trace pairing,
\[
    \mathrm{ERT}[S](\rho)=\tr(\ert[S]\rho).
\]
Here we first define the cost as an extended scalar function
\(\cC_c(S,\rho)\in[0,\infty]\).  The theorem below says that its values on pure
states assemble into a closed positive form \(C_{c,S}\).  Thus \(C_{c,S}\) plays the
role of the runtime observable \(\ert[S](0)\), but the new predicate
representation can express \(+\infty\) naturally.  Divergence, abort, infinite
expected cost, and hard domain constraints are therefore handled as ordinary
predicate values, rather than as external side conditions.

\begin{theorem}[Well-Definedness of Cost Predicate]
\label{thm:cost_well_defined}
For any quantum program $S$, the expected cost evaluated on pure
states induces a functional
\[
    \QF_{\cC_c(S)}[\u]\triangleq\cC_c(S,\sP_{\u}).
\]
This functional is a closed positive quadratic form on $\cH$.  Thus there exists
a unique cost predicate $C_{c,S}\in\Pred$ that represents the exact expected cost,
satisfying $\QF_{C_{c,S}}=\QF_{\cC_c(S)}$.
\end{theorem}

\begin{proof}[Proof sketch]
The proof is by structural induction, with one important strengthening: the
induction simultaneously proves that the cost functional induces a closed
positive form and that the branchwise forms appearing in bind are measurable.
This measurability is what makes the integral clauses legitimate, not
just formal notation.  For the core rules, sequencing uses the predicate
transformer in \Cref{thm:unbounded_heisenberg_transformer}.  Bind uses quadratic form-level measurement integration.
While uses the supremum of finite-unrolling costs and the \(\omega\)-CPO
structure of \(\Pred\).
\end{proof}

The validity of cost Hoare triple below is the quantum analogue of the quantitative
Hoare triple in probabilistic program \cite{kaminski2018weakest}.  The classical pre- and post-expectations \(g\) and \(f\) are
replaced by predicates paired with states through the extended trace, and the
plain runtime term \(\ert\llbracket C\rrbracket(0)\) is replaced by the scalar
cost \(\cC_c(S,\rho)\).  Thus a single inequality bounds both the expected
postcondition and the cost.

\begin{definition}[Validity]
\label{def:correctness_validity}
A cost Hoare triple is valid, denoted as $\models_c\{P\}S\{Q\}$, if for all
$\rho\in\pardensity{\cH}$,
\[
    \Tr(P\rho)
    \ge
    \Tr(Q\sem{S}(\rho))+\cC_c(S,\rho).
\]
\end{definition}

As in probabilistic expected-runtime reasoning \cite{kaminski2018weakest}, positive cost gives a
total-correctness mode that is stronger than almost-sure termination.  When
\(c=0\), the cost component vanishes and the judgment is zero-cost partial correctness.
When \(c>0\), finite pre-expectation bounds the expected accumulated cost, and
therefore rules out positive-mass nontermination under the positive-step model.

\begin{proposition}[Finite Pre-expectation Entails Almost-Sure Termination]
\label{prop:finite_pre_ast}
If \(c>0\), \(\models_c\{P\}S\{Q\}\), and \(\Tr(P\rho)<\infty\), then
\[
    \Tr(\sem{S}(\rho))=\Tr(\rho),
\]
i.e., \(S\) terminates almost surely from \(\rho\).
\end{proposition}

The converse implication does not hold in general: almost-sure termination need
not imply finite expected cost.  The random-walk case study in the next section
exhibits this separation.

\subsection{Proof System}
\label{subsec:proof_system}

We present in Fig. \ref{fig:proof rule} the proof system of our logic.
\begin{figure}
    \centering
\begin{mathpar}
\inferrule*[right=\textsc{Skip}]
    { }
    { \vdash_{c} \{ Q + c I \}\ \mathbf{skip}\ \{ Q \} }
\\

\inferrule*[right=\textsc{Abort}]
    { }
    { \vdash_{c} \{ c\cdot\top \}\ \mathbf{abort}\ \{ P \} }
\\

\inferrule*[right=\textsc{Init}]
    { }
    { 
      \vdash_{c}
      \left\{
        c I + \sum_{k=0}^{\infty}
        \ket k_q\!\bra 0 P \ket 0_q\!\bra k
      \right\}
      \ q \coloneqq \ket{0}\ 
      \{ P \}
    }
\\

\inferrule*[right=\textsc{Unitary}]
    { }
    { 
      \vdash_{c}
      \{ c I + U^\dagger P U \}
      \ \seq q\coloneqq U \seq q\ 
      \{ P \}
    }
\\

\inferrule*[right=\textsc{Seq}]
    {
      \vdash_{c} \{P\}\ S_1\ \{R\}
      \\
      \vdash_{c} \{R\}\ S_2\ \{Q\}
    }
    {
      \vdash_{c} \{P\}\ S_1 ; S_2\ \{Q\}
    }
\\

\inferrule*[right=\textsc{Bind}]
    {
      \{P_{\nu}\}_{\nu\in\Omega_M}
      \text{ is }M\text{-admissible}
      \\
      \forall \nu \in \Omega_M,\ 
      \vdash_{c} \{ P_\nu \}\ S_\nu\ \{ Q \}
    }
    {
      \vdash_{c}
      \left\{
        c I +
        \int_{\Omega_M}
          M_{\nu}^\dagger P_\nu M_{\nu}
        \, d\mu_M(\nu)
      \right\}
      \ \mathbf{bind}(M[\seq{q}],\, x.S)\ 
      \{ Q \}
    }
\\

\inferrule*[right=\textsc{While}]
    {
      P \sqsupseteq
      M_0^\dagger Q M_0 + c I + M_1^\dagger R M_1
      \\
      \vdash_{c} \{ R \}\ S\ \{ P \}
    }
    {
      \vdash_{c}
      \{ P \}
      \ \mathbf{while}\ M[\seq{q}]=1\
      \mathbf{do}\ S\ \mathbf{od}\ 
      \{ Q \}
    }
\\

\inferrule*[right=\textsc{Conseq}]
    {
      P \sqsupseteq P'
      \\
      \vdash_{c} \{P'\}\ S\ \{Q'\}
      \\
      Q' \sqsupseteq Q
    }
    {
      \vdash_{c} \{P\}\ S\ \{Q\}
    }
\end{mathpar}
    \caption{Proof system for validity}
    \label{fig:proof rule}
    \Description{Proof system for validity.}
\end{figure}
The main ideas are explained as follows.

\emph{Skip}.
  Even \(\mathbf{skip}\) pays the base cost of executing one command.  Since
  \(\Tr(cI\rho)=c\Tr(\rho)\), the precondition \(Q+cI\) exactly combines
  the unchanged postcondition with this one-step cost.  When \(c=0\), the rule
  reduces to the usual partial-correctness axiom.

 \emph{Abort}.
  For positive cost, aborting is assigned infinite cost on every nonzero input
  state, represented by \(c\cdot\top\).  Thus the rule records that abort has no
  finite-cost proof obligation unless the pre-expectation is allowed to be
  infinite.  At \(c=0\), the precondition becomes \(0\), matching ordinary
  zero-cost partial correctness: no terminal state is produced, so no postcondition has to
  be established.

  \emph{Init}.
  Initialization resets \(q\) to \(\ket{0}\) and discards its previous contents.
  The infinite sum is the form-level Heisenberg dual of the postcondition
  through this reset operation: it evaluates \(P\) after replacing the
  initialized register by \(\ket{0}\).  The additional \(cI\) is the base cost
  of performing the initialization.

  \emph{Unitary}.
  The precondition \(U^\dagger P U\) is the unbounded predicate transformer of the
  postcondition through a unitary command, read at the quadratic-form level when
  \(P\) is unbounded.  The term \(cI\) again pays the base command cost.

  \emph{Seq}.
  The intermediate predicate \(R\) is the continuation assertion for \(S_1\):
  after executing \(S_1\), it must be sufficient to run \(S_2\) and establish
  \(Q\).  There is no extra \(cI\) for sequencing itself, because the costs of
  the two subprograms are already included in the two premises; the semicolon is
  only the compositional structure.

  \emph{Bind}.
  This is the rule for continuous measurements.  For each outcome \(\nu\), the premise
  proves that \(P_\nu\) is sufficient for the continuation \(S_\nu\).  The
  precondition of the whole bind is obtained by pulling these branch
  preconditions back through the measurement operators and integrating over
  outcomes.  Since \(P_\nu\) may be unbounded, the expression
  \(M_{\nu}^\dagger P_\nu M_{\nu}\) is read at 
  quadratic-form level (defer to \Cref{subsec:Unbounded_Predicate_Transformer}), and the admissibility assumption ensures that the
  branchwise scalar form values are measurable and define a closed positive
  form after integration.  The leading \(cI\) pays for performing the
  measurement/bind step itself; the costs of the outcome-dependent continuations
  are already contained in the predicates \(P_\nu\).

  \emph{While}.
  The loop rule is the quantum analogue of the invariant/potential rule used in
  probabilistic expected-runtime reasoning~\cite{kaminski2018weakest}.  The
  predicate \(P\) is a potential (quantitative invariant).  One guard evaluation costs
  \(cI\).  On the terminating branch, the postcondition
  \(M_0^\dagger Q M_0\) must be paid.  On the continuing branch, the loop must
  provide \(M_1^\dagger R M_1\), where \(R\) is strong enough to run the body
  and re-establish \(P\).  Thus, the side condition
  \[
      P \sqsupseteq M_0^\dagger Q M_0+cI+M_1^\dagger R M_1
  \]
  says that the current potential pays for one guard step and for either exit
  or one body execution followed by the restored invariant.  Soundness follows
  by induction over finite loop unrollings and then by taking the monotone
  supremum in \(\Pred(\cH)\).

  \emph{Conseq}.
  Consequence uses the quantitative order of predicates.  A larger precondition
  gives more budget, and a smaller postcondition asks for less after termination.
  Hence from \(P\sqsupseteq P'\), a proof of \(\{P'\}S\{Q'\}\), and
  \(Q'\sqsupseteq Q\), one obtains \(\{P\}S\{Q\}\).

\subsection{Soundness and Completeness}
\label{subsec:weakest_precondition}

The proof system introduced above is both sound and relatively complete, as stated below.
\begin{theorem}[Soundness]
\label{thm:soundness}
The proof system displayed in Fig. \ref{fig:proof rule} is sound for reasoning about the cost; that is,
for any quantum program $S$, predicates $P,Q\in\Pred$, and $c\ge0$, we have:
\[
    \vdash_c\{P\}S\{Q\}
    \quad\mbox{implies}\quad
    \models_c\{P\}S\{Q\}.
\]
\end{theorem}

\begin{theorem}[Completeness]
\label{thm:completeness}
The proof system displayed in Fig. \ref{fig:proof rule} is (relatively) complete for reasoning about the cost; that is,
for any quantum program $S$, predicates $P,Q\in\Pred$, and $c\ge0$, we have:
\[
    \models_c\{P\}S\{Q\}
    \quad\mbox{implies}\quad
    \vdash_c\{P\}S\{Q\}.
\]
\end{theorem}

The proof of completeness follows the standard approach by introducing the weakest precondition; the detailed proof can be found in Supplementary Material.

\begin{definition}[Weakest Precondition (\texorpdfstring{$\wp_c$}{wp})]
\label{def:wp}
For any quantum program $S$ and post-condition $Q\in\Pred$, the weakest
precondition $\wp_c(S,Q)$ is the unique predicate in $\Pred$ satisfying:
\begin{enumerate}[(1)]
    \item \textbf{Validity:} $\models_c\{\wp_c(S,Q)\}S\{Q\}$.
    \item \textbf{Minimality:} For any $P\in\Pred$, if
    $\models_c\{P\}S\{Q\}$, then $\wp_c(S,Q)\sqsubseteq P$.
\end{enumerate}
\end{definition}

We can now state the quantum analogue of the probabilistic
weakest-preexpectation and expected-runtime calculi: \(\wp_c\) pulls back the
post-expectation and adds the cost accumulated by the program.

\begin{proposition}[Well-Definedness of $\wp_c$]
\label{thm:wp_well_defined}
For any quantum program $S$, post-expectation $Q\in\Pred$, and $c\ge0$,
$\wp_c(S,Q)$ is uniquely well-defined as $\wp_c(S,Q)=\sem{S}^{\dagger}(Q)+C_{c,S}$.

In addition, $\wp_c$ satisfies the following structure representation:
\[
\begin{alignedat}{2}
& \wp_c(\mathbf{skip},Q)
& = \quad
& Q+cI,
\\
& \wp_c(\mathbf{abort},Q)
& =\quad
& c\cdot\top,
\\
& \wp_c(q\coloneqq\ket{0},Q)
& =\quad
& cI+\sum_{k=0}^{\infty}
  \ket{k}_q\!\bra{0}\,Q\,\ket{0}_q\!\bra{k},
\\
& \wp_c(\seq q\coloneqq U[\seq q],Q)
& =\quad
& cI+U^\dagger Q U,
\\
& \wp_c(S_1;S_2,Q)
& =\quad
& \wp_c\bigl(S_1,\wp_c(S_2,Q)\bigr),
\\
& \wp_c\bigl(\mathbf{bind}(M[\seq q],x.S_x),Q\bigr)
& =\quad
& cI+\int_{\Omega_M}
  M_{\nu}^\dagger
  \wp_c(S_\nu,Q)
  M_{\nu}\,d\mu_M(\nu),
\\
& \wp_c\bigl(
    \mathbf{while}\ M[\seq q]=1\
    \mathbf{do}\ S\ \mathbf{od},Q
  \bigr)
& =\quad
& \sup_{n\ge 0}X_n.
\end{alignedat}
\]
where $X_0\coloneqq0$ and
$X_{n+1}\coloneqq M_0^\dagger Q M_0+cI+M_1^\dagger\wp_c(S,X_n)M_1$.
\end{proposition}

\section{Case Study: Quantum Symmetric Random Walk}
\label{sec:random_walk}

The one-dimensional symmetric random walk with an absorbing boundary usually serves as a
canonical example across probabilistic programming and quantitative verification. 
It is also one of the simplest models where qualitative termination does not
imply a finite quantitative bound: recurrence says that the absorbing state is reached with
probability one, whereas null recurrence says that the first hitting time has
infinite expectation. For probabilistic programs, this example has become a
standard stress test for expectation-transformer calculi: an almost-sure
termination proof is not enough to establish positive almost-sure termination
or any finite expected-runtime bound~\cite{kaminski2018weakest}.

Quantum walks provide a quantum version of this benchmark, and
are an important source of quantum algorithms because their
dynamics can exploit interference~\cite{ABN+01,kempe2003quantum}.  They have also
served as demanding examples for quantum expected-runtime calculi: the work \cite{liu2025quantum} computes expected-runtime observables for finite-cycle quantum walks and solve an open runtime-computation problem for arbitrary coins and initial states.

Our case study stresses a different axis of the same benchmark.  We deliberately
reset the coin at every iteration, so the position marginal is exactly the
classical symmetric random walk on the integer lattice.  Thus the example is
not meant to demonstrate interference or a quantum speedup.  Its role is proof-theoretic: by encoding a null-recurrent process
as an infinite-dimensional quantum program, we obtain a small but sharp
separation between qualitative termination
and quantitative cost inside one weakest-precondition framework. 

This separation is invisible in the finite-dimensional runtime theory. In the
finite-dimensional setting of existing quantum expected-runtime calculi, almost
sure termination and finite expected runtime coincide under the usual
non-singular one-step cost assumption~\cite{liu2025quantum}.  Our integer-lattice walk shows that this coincidence is not preserved by the
general infinite-dimensional case; it is a finite-dimensional phenomenon. 
Thus the case study exercises exactly the feature introduced in this paper: the
same positive-form assertion domain accommodates both the bounded weakest-
precondition calculation used to establish almost-sure termination and the
unbounded cost precondition whose value at the initialized state is \(+\infty\).
The limiting argument is part of the weakest-precondition calculation, rather
than a separate truncation analysis outside the logic.

\paragraph{Program.}
The program has a coin qubit $q$ with Hilbert space $\cH_q=\mathbb C^2$ and a walker register $w$ with Hilbert space $\cH_w\triangleq\ell^2(\mathbb Z)$ and basis $\{\ket{k}:k\in\mathbb Z\}$.  Let $U_{\mathrm{sh}}\ket{k}\triangleq\ket{k+1}$ be the bilateral shift.  The controlled walk unitary is
\[
U_{\mathrm{walk}}\ket{0,k}\triangleq\ket{0,k+1},\qquad
U_{\mathrm{walk}}\ket{1,k}\triangleq\ket{1,k-1}.
\]
The loop guard tests whether the walker has hit the origin:
\[
M_0\triangleq I_q\otimes\ket{0}\!\bra{0},\qquad
M_1\triangleq I_q\otimes(I_w-\ket{0}\!\bra{0}).
\]
Outcome $0$ exits and outcome $1$ continues.  With unit cost $c=1$, the program is
\[
\begin{aligned}
S\equiv w&:=\ket{1};\ W,\qquad\text{where}\\
W\equiv \mathbf{while}\ M[w]&=1\ \mathbf{do}\\
        q&:=\ket{0};\\
        q&:=Hq;\\\
        q,w&:=U_{\mathrm{walk}}[q,w]\ \mathbf{od}.
\end{aligned}
\]
The prefix $w:=\ket{1}$ is trace preserving and has finite cost, so the relevant loop-entry state is
\[
\rho_1=\ket{0}_q\!\bra{0}\otimes\ket{1}_w\!\bra{1}.
\]

\begin{lemma}[Body transformer]
\label{lem:rw_body_transformer}
Let $B\equiv q:=\ket{0};\ q:=Hq;\ q,w:=U_{\mathrm{walk}}[q,w]$.  For every positive predicate of the form $X=I_q\otimes X_w$,
\[
\wp_1(B,X)=I_q\otimes\bigl(3I_w+\mathcal E^*(X_w)\bigr),\,\wp_0(B,X)=I_q\otimes\mathcal E^*(X_w)
\]
where
\(
\mathcal E^*(Y)=\frac12U_{\mathrm{sh}}^\dagger YU_{\mathrm{sh}}+\frac12U_{\mathrm{sh}}YU_{\mathrm{sh}}^\dagger
\).
\end{lemma}

This is a direct $\wp_c$ calculus in \Cref{subsec:weakest_precondition} for the loop body.  We now use the same loop equation
in two different ways.  Its bounded zero-cost instance has the same defining
equation as the bounded quantum weakest-precondition calculus of \cite{Ying11}, and hence can be read through that calculus
to establish almost-sure termination.  Its 1-cost instance, on the other
hand, is interpreted in our positive-form domain and computes the expected
operational cost.

Concretely, for bounded predicates $0\sqleq Q\sqleq I$, the defining equations of
$\wp_0(S,Q)$ are the same as the equations used in bounded
weakest-precondition calculus.  
Thus the equality
\[
    \wp_0(W,I_q\otimes I_w)=I_q\otimes I_w
\]
below should be understood via the bounded weakest-precondition calculus, rather than as a termination consequence of the $c=0$ fragment
of our cost proof system.
The new part of the example is the 1-cost calculation: $\wp_1(W,0)$ is evaluated in
the larger positive-form domain, where it may be unbounded and can witness an infinite expected cost.

\begin{theorem}[Almost-sure termination without finite expected cost]
\label{thm:rw_ast_infinite_cost}
The loop $W$ is almost surely terminating:
\[
    \wp_0(W,I_q\otimes I_w)=I_q\otimes I_w .
\]
However, in the $\wp_1$ calculus,
\[
    \Tr\bigl(\wp_1(W,0)\rho_1\bigr)=+\infty .
\]
Consequently, the full program $S$ terminates almost surely from its initialized
state, but has infinite expected operational cost.
\end{theorem}
\begin{proof}
First consider expected cost.  Let $T\triangleq\wp_1(W,0)$.  By the while rule, $T$ is the supremum of the finite-unrolling approximants
\[
X_0=0,\qquad
X_{n+1}=I+M_1^\dagger\wp_1(B,X_n)M_1 .
\]
The guard and Lemma~\ref{lem:rw_body_transformer} preserve coin-independent diagonal predicates, hence
\[
T=I_q\otimes D,
\qquad
D=\sum_{k\in\mathbb Z}t_k\ket{k}\!\bra{k}
\]
as a closed positive diagonal form, with $t_k\in[0,+\infty]$.  Taking diagonal matrix elements in the fixed-point equation gives
\[
t_0=1,
\qquad
t_k=4+\frac12(t_{k+1}+t_{k-1})\quad(k\ne0).
\]
Here $1$ is the guard cost and $3$ is the body cost.  If $t_1<\infty$, then the recurrence forces every $t_k$ with $k\ge0$ to be finite.  Writing $\Delta_k=t_k-t_{k-1}$, the equation becomes
\[
\Delta_{k+1}=\Delta_k-8,
\qquad
\text{so}\qquad
 t_n=t_0+n\Delta_1-4n(n-1).
\]
The right-hand side is eventually negative, contradicting $t_n\ge0$.  Therefore $t_1=+\infty$, and
\[
\Tr(T\rho_1)=t_1=+\infty .
\]
The initialization before the loop contributes only finite cost, so it cannot remove this divergence.

It remains to show that the same loop is nevertheless almost surely terminating.  Let $H\triangleq\wp_0(W,I_q\otimes I_w)$ be the termination-probability predicate.  The zero-cost approximants again have coin-independent diagonal form
\[
H=I_q\otimes\sum_{k\in\mathbb Z}h_k\ket{k}\!\bra{k},
\qquad 0\le h_k\le1 .
\]
The zero-cost while equation and Lemma~\ref{lem:rw_body_transformer} yield the bounded harmonic system
\[
h_0=1,
\qquad
h_k=\frac12(h_{k+1}+h_{k-1})\quad(k\ne0).
\]
On the nonnegative half-line, $h_k=1+k(h_1-h_0)$; boundedness in $[0,1]$ forces $h_1=h_0=1$.  The same argument on the negative half-line gives $h_k=1$ for all $k\le0$.  Thus $h_k=1$ for every $k$, so $H=I_q\otimes I_w$ and $W$ is AST.  Since $w:=\ket{1}$ is trace preserving, $S$ is also AST from its initialized state.
\end{proof}

Therefore, this example shows that our $\wp_c$ calculus can expose both
the bounded termination argument and the unbounded cost divergence, without
treating the latter as an external probabilistic calculation.

\section{Case Study: One-Round GKP Error Correction}
\label{sec:gkp_verification}


The Gottesman--Kitaev--Preskill (GKP) code is one of the central
bosonic encodings for fault-tolerant quantum information processing.  It
stores a logical qubit in the phase space of a single harmonic oscillator, so
that small displacement errors in the position and momentum quadratures can be
detected and corrected before they become logical Pauli errors
\cite{gottesman2001encoding,terhal2020towards}.  This hardware-efficient use of
an oscillator is a major reason why GKP states have become a leading candidate
for bosonic quantum error correction: GKP encodings have been realized in
trapped-ion motion and superconducting microwave cavities, and they are used as
basic resources in proposed photonic fault-tolerant architectures
\cite{fluhmann2019encoding,campagne2020quantum,Bourassa2021blueprintscalable}.
More broadly, GKP codes sit at the interface between discrete-variable fault
tolerance and continuous-variable hardware, where Gaussian couplings, homodyne
measurements, and displacement feedback are native operations
\cite{weedbrook2012gaussian,terhal2020towards}.

At the level of program structure, one round of GKP correction is an ordinary continuous-variable feedback program.  The data oscillator is coupled
to an ancilla, a homodyne measurement returns a real-valued syndrome, and the
program applies a displacement determined by that real outcome
\cite{gottesman2001encoding,terhal2020towards}.  In the local correction
regime, the purpose of this step is to reduce residual displacement noise:
after feedback, the relevant quadrature variance should be controlled by the
finite variance of the ancilla and by the measurement noise.  Thus
the core quantitative assertion is not a bounded yes/no property, but a
second-moment bound.  Under the usual zero-mean local-noise convention, this
second moment is exactly the variance of the residual displacement.

This is precisely the kind of property that bounded-effect QHL cannot express
directly.  The postcondition asks for a bound on the second moment of the data
oscillator's position, represented by a physically meaningful unbounded
observable, and the correction command is a continuous bind whose branches are
indexed by the real-valued homodyne outcome.

A finite-dimensional or bounded-effect treatment can access this specification
only after an additional encoding step: one either imposes an energy or
phase-space cutoff, or replaces the unbounded observable by a bounded
truncation.  Such encodings are useful for numerical and finite-dimensional
analyses, but they change the verification goal into a family of
cutoff-dependent assertions; any statement about the physical oscillator then
requires an additional limiting argument.  They also hide part of the program
structure, since syndrome extraction is naturally a real-outcome feedback bind
rather than a finite branch.  Our logic avoids this detour.  It verifies the
one-round local-regime variance bound directly in the continuous-variable
oscillator model: the real-valued syndrome measurement is represented as an
admissible continuous bind, and the variance postcondition is represented as a
closed positive form.

We use this local one-round statement as a focused benchmark: it isolates the
two verification challenges that this case study is meant to exercise, namely
real-valued feedback through continuous bind and a physically meaningful
unbounded variance predicate.  The full modular decoder, Pauli-frame update,
repeated correction rounds, and threshold analysis belong to the broader GKP
fault-tolerance stack; including them here would obscure rather than clarify
the predicate-transformer principles addressed by the proof system.

\paragraph*{Model.}
There are two modes. The data mode has Hilbert space $\cH_{\mathsf{s}}=L^2(\bR)$ and the ancilla mode has Hilbert space $\cH_{\mathsf{a}}=L^2(\bR)$, so the global space is
\[
    \cH \triangleq \cH_{\mathsf{s}}\otimes \cH_{\mathsf{a}} .
\]
Let $\hat q_{\mathsf{s}}$ and $\hat q_{\mathsf{a}}$ be the position operators on the two modes, extended cylindrically to $\cH$. The postcondition is the closed positive form induced by the system second moment
\[
    Q_{\mathsf{GKP}} \triangleq \hat q_{\mathsf{s}}^{\,2}\otimes I_{\mathsf{a}} .
\]
Under the usual zero-mean and local-noise assumption of GKP state, which is conventional in physics, this second moment is the position variance.

The ancilla is initialized to a normalized finite-energy approximate GKP state $\ket{\phi_{\Delta,\kappa}}$~\cite{matsuura2020equivalence,terhal2020towards}. We only use the following two moment facts:
\[
    \langle \phi_{\Delta,\kappa},\hat q_{\mathsf{a}}\phi_{\Delta,\kappa}\rangle=0,
    \qquad
    \langle \phi_{\Delta,\kappa},\hat q_{\mathsf{a}}^{\,2}\phi_{\Delta,\kappa}\rangle
    = \sigma_{\mathsf{a}}^2 <\infty .
\]
The coupling is the continuous-variable SUM gate
\[
    U_{\mathsf{SUM}} \triangleq\exp(-i\hat q_{\mathsf{s}}\otimes \hat p_{\mathsf{a}}),
\]
which shifts the ancilla position by the system position in the Heisenberg picture. The feedback displacement on the system is
\[
    D_{\mathsf{s}}(v)\triangleq \exp(-iv\hat p_{\mathsf{s}}),
    \qquad
    D_{\mathsf{s}}(v)^\dagger \hat q_{\mathsf{s}}D_{\mathsf{s}}(v)=\hat q_{\mathsf{s}}+vI_{\mathsf{s}} .
\]

The homodyne readout of the ancilla is represented by the Gaussian-noisy position instrument
\[
    M_x \triangleq I_{\mathsf{s}}\otimes g_{\sigma_M}(xI_{\mathsf{a}}-\hat q_{\mathsf{a}})^{1/2},
    \qquad
    g_{\sigma_M}(t)\triangleq\frac{1}{\sqrt{2\pi\sigma_M^2}}\exp({-t^2/(2\sigma_M^2)}) .
\]
By the Borel functional calculus, $x\mapsto M_x$ is strongly measurable and
$\int_{\bR}M_x^\dagger M_x\,dx=I$ in the quadratic-form sense. Hence it is an admissible continuous measurement family for the bind rule.

\paragraph*{Program.}
In the local GKP correction regime, the modular residual agrees with the linear residual: no lattice wrap-around occurs, and the feedback function $r(x)$ used in the full GKP decoder satisfies $r(x)=x$ on the branch being verified. The one-round program is therefore the linearized GKP correction step
\[
\begin{array}{rcl}
S_{\mathsf{GKP}} &\equiv& A_{\mathsf{GKP}};\ U_{\mathsf{SUM}};\ B_{\mathsf{hom}},\\[0.3ex]
\text{where~~~~}A_{\mathsf{GKP}} &\equiv& \mathsf{a}:=\ket{\phi_{\Delta,\kappa}},\\[0.3ex]
B_{\mathsf{hom}} &\equiv&
    \bind(M[\mathsf{a}],x.\,D_{\mathsf{s}}(-x)[\mathsf{s}]).
\end{array}
\]
The command $A_{\mathsf{GKP}}$ is an initialization command with Kraus operators
$E_k=I_{\mathsf{s}}\otimes\ket{\phi_{\Delta,\kappa}}\bra{k}$ for any fixed orthonormal basis $\{\ket{k}\}_k$ of $\cH_{\mathsf{a}}$.

\paragraph*{Verification goal.}
We use the zero-cost fragment, since this case study concerns a physical post-expectation rather than runtime. The desired Hoare triple is
\begin{equation}
\label{eq:gkp-target-triple}
    \vdash_0
    \left\{(\sigma_{\mathsf{a}}^2+\sigma_M^2)I\right\}
    \; S_{\mathsf{GKP}} \;
    \left\{Q_{\mathsf{GKP}}\right\} .
\end{equation}
Equivalently, for every partial density operator $\rho$,
\[
    \Tr(Q_{\mathsf{GKP}}\,\sem{S_{\mathsf{GKP}}}(\rho))
    \le (\sigma_{\mathsf{a}}^2+\sigma_M^2)\Tr(\rho).
\]
The physical meaning is that, for any finite-energy inputs, the output second moment (i.e., variance) is therefore bounded by the finite ancilla variance plus the measurement-noise variance.

The calculation needed for the proof is isolated in the following lemma. All equalities are equalities of closed positive quadratic forms.

\begin{lemma}[One-round GKP form calculations]
\label{lem:gkp-form-calculations}
Let
\[
    P_x \triangleq (\hat q_{\mathsf{s}}-xI_{\mathsf{s}})^2\otimes I_{\mathsf{a}},
    \quad
    R \triangleq (\hat q_{\mathsf{s}}\otimes I_{\mathsf{a}}-I_{\mathsf{s}}\otimes\hat q_{\mathsf{a}})^2+\sigma_M^2 I,\quad
    T \triangleq I_{\mathsf{s}}\otimes\hat q_{\mathsf{a}}^{\,2}+\sigma_M^2 I,
    \quad
    P_0 \triangleq (\sigma_{\mathsf{a}}^2+\sigma_M^2)I .
\]
Then:
\begin{align}
    P_x&=(D_{\mathsf{s}}(-x)\otimes I_{\mathsf{a}})^\dagger
    Q_{\mathsf{GKP}}
    (D_{\mathsf{s}}(-x)\otimes I_{\mathsf{a}})
    , \label{eq:gkp-feedback-calc}\\
    R&=\int_{\bR}M_x^\dagger P_xM_x\,dx
    , \label{eq:gkp-bind-calc}\\
    T&=U_{\mathsf{SUM}}^\dagger R U_{\mathsf{SUM}}
    , \label{eq:gkp-sum-calc}\\
    P_0&=\sum_k E_k^\dagger T E_k.
    \label{eq:gkp-init-calc}
\end{align}
\end{lemma}
\begin{proof}
Deferred to the supplementary material. The proof uses the Weyl translation relation, the joint spectral calculus for $\hat q_{\mathsf{s}}\otimes I$ and $I\otimes\hat q_{\mathsf{a}}$, and the first two Gaussian moments of the noisy homodyne kernel.
\end{proof}

\paragraph*{Hoare derivation.}
We now derive \eqref{eq:gkp-target-triple} using the rules of \Cref{subsec:proof_system}. First, by the unitary rule and \eqref{eq:gkp-feedback-calc}, for every $x\in\bR$,
\[
    \vdash_0\{P_x\}\ D_{\mathsf{s}}(-x)[\mathsf{s}]\ \{Q_{\mathsf{GKP}}\} .
\]
The family $(P_x)_x$ is $M$-admissible by \Cref{lem:gkp-form-calculations}. Applying the bind rule and then \eqref{eq:gkp-bind-calc} gives
\[
    \vdash_0\{R\}\ B_{\mathsf{hom}}\ \{Q_{\mathsf{GKP}}\} .
\]
Next, the unitary rule for $U_{\mathsf{SUM}}$ and \eqref{eq:gkp-sum-calc} give
\[
    \vdash_0\{T\}\ U_{\mathsf{SUM}}\ \{R\} .
\]
By the sequence rule,
\[
    \vdash_0\{T\}\ U_{\mathsf{SUM}};B_{\mathsf{hom}}\ \{Q_{\mathsf{GKP}}\} .
\]
Finally, the initialization rule for $A_{\mathsf{GKP}}$ and \eqref{eq:gkp-init-calc} give
\[
    \vdash_0\{P_0\}\ A_{\mathsf{GKP}}\ \{T\} .
\]
A second application of the sequence rule yields
\[
    \vdash_0\{P_0\}\ A_{\mathsf{GKP}};U_{\mathsf{SUM}};B_{\mathsf{hom}}\ \{Q_{\mathsf{GKP}}\},
\]
which is exactly \eqref{eq:gkp-target-triple}.

\paragraph*{Summary.}
The Hoare derivation above is the endpoint of the case study: it proves
\eqref{eq:gkp-target-triple} directly, without any finite dimensional truncation / bounded approximation.  The
main point is that \(Q_{\mathsf{GKP}}\) and the intermediate predicates are
well-formed assertions of the logic, so the proof can move them backward through
the program by syntax-directed rules.  More importantly, the derivation illustrates the role of
the proof system as a compositional interface between continuous-variable
semantics and Hoare-style verification.  The program syntax generates
local proof obligations, while the required analytic facts are used only as
lemmas for discharging those obligations.  Thus the case study demonstrates that
unbounded continuous-variable reasoning can be presented as a structured Hoare
proof rather than as a single external calculation.

\section{Related Work}
\label{sec:related-work}



\paragraph{Motivating works}
Four lines of work are most relevant.  First, a measure-theoretic denotational
semantics for probabilistic programs was given by Kozen
\cite{kozen1979semantics}, while finite-dimensional quantum programs were
interpreted by Selinger using CPTN maps \cite{selinger2004towards}.  We combine these
perspectives to treat continuous measurements over
infinite-dimensional Hilbert spaces.

Second, probabilistic weakest-preexpectation transformers were introduced in
\cite{MMS96}; quantum weakest preconditions were then defined in finite
dimensions \cite{dhondt2006quantum} and extended to general Hilbert spaces with
bounded predicates \cite{YingDuanFengJi2010PredicateTransformer}.  Separately,
probabilistic ERT was developed in \cite{kaminski2018weakest} and extended to
finite-dimensional quantum programs \cite{liu2025quantum}.  Both transformer
lines are extended here to unbounded assertions in
infinite dimensions.

Third, a systematic proof framework for probabilistic programs was developed
in \cite{McIverM05}, while a sound and relatively complete QHL was established
in \cite{Ying11}.  Our system adds unbounded assertions and potential-based cost
reasoning while retaining soundness and relative completeness.

Finally, infinite-valued predicates were introduced in relational quantum logics
\cite{qotl} and generalized to positive closed forms in infinite dimensions
\cite{LICS2026LinearRelation}.  Whereas these logics compare two programs, their assertion techniques are used here for reasoning about properties and costs of one single program.

\paragraph{Verification of infinite-dimensional quantum programs.}
Predicate-transformer semantics for quantum programs were developed over general
Hilbert-space models \cite{YingDuanFengJi2010PredicateTransformer}, and
Floyd--Hoare logic for quantum while programs was formulated for finite- and
countably infinite-dimensional state spaces \cite{Ying11}.  This line was further
connected to classical verification methods through quantum Markov-chain based
program verification over separable Hilbert spaces
\cite{YingYuFengDuan2013Verification}.  Later work studied proof principles and
automation for the same broad setting, including invariant generation
\cite{YingYingWu2017Invariants}, termination analysis
\cite{LiYing2018Termination}, and automatic verification
\cite{Ying2019AutomaticVerification}.

A complementary line refined the predicate and proof methodology:
applied quantum Hoare logic uses subspace/projection assertions
\cite{ZhouYY19}, ghost variables support reasoning about auxiliary quantum data
\cite{Unr19b}, quantum relational Hoare logic treats relational assertions as
subspaces \cite{Unr19a}, expectation-based relational reasoning generalizes this
to quantitative assertions \cite{LU19}, and disjoint parallel quantum programs
are handled by a dedicated proof system \cite{YZL18}.  More recent work clarifies
logical foundations for infinite-dimensional reasoning, including
Birkhoff--von Neumann quantum logic as an assertion language
\cite{Ying2022BirkhoffAssertion}, structured theorems for quantum programs
\cite{Yu2023StructuredTheorem}, quantum/classical registers
\cite{Unruh2021QuantumClassicalRegisters}, algebraic laws of quantum programming
\cite{YingZhouBarthe2025Laws}, Bayesian conditioning in infinite-dimensional
quantum programs \cite{GehnenUnruhKatoen2025Bayesian}, and complete relational
logics with unbounded assertions \cite{LICS2026LinearRelation}.

Our work is complementary to these logics rather than a direct extension of any
single one of them.  We focus on continuous-variable programs with
continuous-outcome measurements and outcome binding, and use closed positive
quadratic forms to represent unbounded physical observables.  This yields
domain-sensitive weakest-precondition and cost reasoning tailored to
continuous-variable examples such as symmetric quantum walks and GKP error
correction, rather than to subspace, expectation, or relational assertions alone.


\paragraph*{Verification of quantum error correction (QEC) codes}
There are already prior works on formal verification of QEC codes, although they focus on a discrete setting---based on a qubit system. 
In the program logic aspect,
\citet{Rand2021gottesman,sundaram2026hoaremeetsheisenberglightweight} first introduced stabilizer formalism as assertions, upon which \citet{sundaram2026hoaremeetsheisenberglightweight} further established a Hoare-like logic that is capable of reasoning about QEC programs and has been formalized in Rocq proof assistant~\cite{Cho2026}.
\citet{wu2021qecv} designed a programming language with a stabilizer constructor in the syntax, with assertions allowing sums of stabilizers, and can effectively prove the correctness of a given QEC program.
\citet{VeriQEC} revisit these two approaches, but using subspace as assertion semantics. By designing a substitution-based quantum Hoare logic they can automatically verify medium-sized (~300 qubits) QEC programs.
Symbolic execution is another appealing way, as shown by \cite{Fang2024symbolic,Chen2025}, capable of verifying QEC code, and appears most efficient in finding bugs in a QEC implementation.
More recently, a mechanized framework, Lean-QEC~\cite{ehatamm2026endtoendformalizationquantumerror} in Lean 4, has been announced, which can provide end-to-end, machine-checked distance certificates of QEC code at industrial code sizes.

\paragraph{Continuous-variable quantum system}
Continuous-variable (CV) quantum computation and its broader quantum information framework were established and systematically developed in \cite{lloyd1999quantum,braunstein2005quantum}.
Recently it has been studied as a hardware-efficient route to fault-tolerant quantum computation \cite{leghtas2013hardware, Bourassa2021blueprintscalable,madsen2022quantum}. Its main platforms include photonic modes \cite{madsen2022quantum}, microwave cavities \cite{campagne2020quantum}, and trapped-ion mechanical oscillators \cite{fluhmann2019encoding}. Experimental and architectural work has demonstrated bosonic error-correcting codes, including the Gottesman--Kitaev--Preskill (GKP) code \cite{gottesman2001encoding, campagne2020quantum}; large-scale cluster states \cite{asavanant2019generation}; programmable photonic processors \cite{madsen2022quantum}; and quantum squeezing for precision metrology \cite{vahlbruch2016detection,tse2019quantum}.

\section{Conclusions and Future Work}
\label{sec:conclusion}
Our work establishes semantic and logical foundations for reasoning about continuous-variable quantum systems. Building on these foundations, we intend to develop formal models and verification methods for the continuous-variable quantum programs and protocols in the future. Because these systems are inherently complex, it is essential that these verifications are machine-checked. Therefore, our next step is to formalize the results of this paper in a proof assistant, which would provide a basis for future mechanized verification of continuous-variable quantum systems.

\begin{acks}
	Tianshi Yu, Minbo Gao and Li Zhou are supported by the Strategic Priority Research Program of the Chinese Academy of Sciences under Grant XDA0480501.
\end{acks}
\bibliographystyle{ACM-Reference-Format}
\bibliography{ref}
\appendix
\section*{The Guide of Appendices}
The appendices constitute the supplementary material for this paper. It is designed to be read independently and is fully self-contained. In particular, it includes all mathematical preliminaries, definitions, semantic constructions, logical judgments, and proof rules required for the development. It also provides complete proofs and detailed case studies for results that are presented more concisely in the main body of the paper. Thus, the technical results established in this supplementary material do not rely on definitions or arguments external to the appendices.

The remainder of this supplementary material is organized as follows. \Cref{sec:math_foundations} develops the mathematical foundations of the work. It introduces positive self-adjoint linear relations and their associated closed positive quadratic forms, defines the extended L"owner order and extended trace, and establishes the monotone-convergence and order-completeness properties used throughout the subsequent development.

\Cref{SEC:syntax_semantics} presents the syntax and denotational semantics of the quantum programming language. Quantum states are modeled by partial density operators on separable Hilbert spaces, while programs are interpreted as completely positive trace-nonincreasing maps on the trace class. Particular attention is given to parameterized quantum operations, continuous-outcome measurement binding, and measurement-guarded loops. The section also establishes the well-definedness, physical validity, and required measurability properties of the semantics.

\Cref{sec:logic} develops the quantitative predicate-transformer semantics. It defines quantum predicates using positive self-adjoint linear relations, constructs weakest liberal preconditions through pullbacks of quadratic forms, and proves their structural and duality properties. It then introduces the operational cost predicate and uses it to define weakest preconditions for finite-expected-cost total correctness.

\Cref{sec:proof_system} and \Cref{sec:proof_system_total} present the Hoare-style proof systems for partial and total correctness, respectively. For each system, we give the inference rules and establish soundness and semantic relative completeness. In the total-correctness system, the rule for while loops uses a quantitative predicate simultaneously as an invariant and as a potential function that accounts for the cost of further iterations.

\Cref{Sec:random_walk} and \Cref{Sec:gkp_verification} illustrate the framework through two case studies. Section~7 formalizes a symmetric random walk as an infinite-dimensional quantum program and demonstrates the distinction between almost-sure termination and finite expected runtime. Section~8 studies a continuous-variable GKP error-correction procedure and derives a bound on the resulting second moment through the continuous-measurement predicate-transformer calculus.

Finally, \Cref{app:further_details} provides further technique details. \Cref{app:spectral_approximation} supplies the detailed construction and properties of the extended trace, while \Cref{app:trace_proof} provides the spectral-approximation result used to approximate unbounded predicates monotonically by bounded ones.

\section{Mathematical Foundations}
\label{sec:math_foundations}

\subsection{Linear Relations and Associated Quadratic Forms}

Throughout this paper, $\cH$ denotes a separable complex Hilbert space. We use the following convention consistently:
\[
    \langle \cdot,\cdot\rangle_{\cH} \text{ is conjugate-linear in the first argument and linear in the second argument.}
\]
All sesquilinear forms in this paper follow the same convention. 

Unlike finite-dimensional cases, observables in infinite-dimensional and continuous-variable systems are often unbounded. We model such observables by \emph{linear relations}. A linear relation $T$ on $\cH$ is a linear subspace of $\cH\oplus\cH$. An operator $A$ is identified with its graph
\[
    \mathcal{G}(A)=\{(\x,A\x)\mid \x\in\Dom(A)\}.
\]
Linear relations generalize operators by allowing multivalued behavior.

\begin{definition}[Basic notions]
    Let $T\subseteq \cH\oplus\cH$ be a linear relation.
    \begin{itemize}
        \item \textbf{Domain and range:}
        \[
            \Dom(T)=\{\x\in\cH\mid \exists\,\y\in\cH,\ (\x,\y)\in T\},\qquad
            \Ran(T)=\{\y\in\cH\mid \exists\,\x\in\cH,\ (\x,\y)\in T\}.
        \]
        \item \textbf{Kernel and multivalued part:}
        \[
            \ker(T)=\{\x\in\cH\mid (\x,0)\in T\},\qquad
            \mul(T)=\{\y\in\cH\mid (0,\y)\in T\}.
        \]
        \item \textbf{Inverse:}
        \[
            T^{-1}=\{(\y,\x)\mid (\x,\y)\in T\}.
        \]
        In particular, $\ker(T)=\mul(T^{-1})$.
    \end{itemize}
    The relation $T$ is the graph of an operator if and only if $\mul(T)=\{0\}$.
\end{definition}

\begin{definition}[Closed and densely defined relations]
    Let $T\subseteq\cH\oplus\cH$ be a linear relation.
    \begin{itemize}
        \item $T$ is \emph{closed} if it is a closed subspace of $\cH\oplus\cH$ with respect to the product Hilbert-space topology.
        \item $T$ is \emph{densely defined} if $\Dom(T)$ is dense in $\cH$.
    \end{itemize}
\end{definition}

\begin{definition}[Adjoint and self-adjointness]
    The adjoint relation $T^*$ is defined by
    \[
        T^*\coloneqq \{(\u,\v)\in\cH\oplus\cH \mid
        \langle \v,\x\rangle_{\cH}=\langle \u,\y\rangle_{\cH}
        \text{ for all }(\x,\y)\in T\}.
    \]
    The relation $T$ is called \emph{symmetric} if $T\subseteq T^*$, and \emph{self-adjoint} if $T=T^*$.
\end{definition}

A crucial structural property of self-adjoint linear relations is the orthogonal decomposition of the Hilbert space. Since $\mul(T^*)=(\Dom(T))^\perp$ holds for every linear relation $T$, self-adjointness implies the following decomposition.

\begin{proposition}[Canonical decomposition, \cite{cross1998multivalued,behrndt2020boundary}]
\label{prop:canonical_decomposition}
    Let $T$ be a self-adjoint linear relation on $\cH$. Then
    \[
        \mul(T)=(\Dom(T))^\perp,
    \]
    and hence
    \[
        \cH=\overline{\Dom(T)}\oplus\mul(T).
    \]
    With $\cH_0\coloneqq\overline{\Dom(T)}$ and $\cH_\infty\coloneqq\mul(T)$, the relation decomposes as
    \[
        T=T_{\operatorname{op}}\oplus(\{0\}\times\cH_\infty),
        \qquad
        T_{\operatorname{op}}\coloneqq T\cap(\cH_0\times\cH_0),
    \]
    where $T_{\operatorname{op}}$ is a densely defined self-adjoint operator on $\cH_0$.
\end{proposition}

\begin{proof}
    By the definition of $T^*$,
    \[
        (0,\y)\in T^*
        \iff
        \langle \y,\x\rangle_{\cH}=0 \text{ for all }(\x,\z)\in T
        \iff
        \y\in(\Dom(T))^\perp.
    \]
    Thus $\mul(T^*)=(\Dom(T))^\perp$. If $T=T^*$, then $\mul(T)=(\Dom(T))^\perp$, and the Hilbert-space decomposition follows.

    Let $(\x,\y)\in T$ and write $\y=\y_0+\y_\infty$ according to $\cH=\cH_0\oplus\cH_\infty$. Since $\y_\infty\in\mul(T)$, we have $(0,\y_\infty)\in T$, and therefore $(\x,\y_0)=(\x,\y)-(0,\y_\infty)\in T$. Hence every element of $T$ decomposes into an operator part in $\cH_0\times\cH_0$ plus a purely multivalued part in $\{0\}\times\cH_\infty$.
    Moreover, $\mul(T_{\operatorname{op}})=\mul(T)\cap\cH_0=\{0\}$, so $T_{\operatorname{op}}$ is single-valued. Its domain is dense in $\cH_0$ by definition of $\cH_0$. The self-adjointness of $T_{\operatorname{op}}$ follows by taking the adjoint inside the reducing subspace $\cH_0$.
\end{proof}

In quantitative verification, observables such as execution time, energy, and variance are typically non-negative or at least bounded below. We use the following relation-level definition.

\begin{definition}[Bounded-below linear relations]
    A self-adjoint linear relation $T$ is \emph{bounded below} if there exists $\gamma\in\bR$ such that
    \[
        \langle \x,\y\rangle_{\cH}\ge \gamma\|\x\|^2,
        \qquad \forall(\x,\y)\in T.
    \]
    The quantity $\langle\x,\y\rangle_{\cH}$ is real and independent of the representative $\y$ in the affine fiber $T\x$, because $T$ is self-adjoint and $\mul(T)\perp\Dom(T)$. If $\gamma=0$, then $T$ is called \emph{positive}.
\end{definition}

Linear relations provide the structural framework for unbounded observables. To define the logical order and expected values, we pass to scalar-valued quadratic forms.

\begin{definition}[Minimal quadratic form and form norm]
\label{def:minimal-quadratic-form}
    Let $T$ be a positive self-adjoint linear relation. Its \emph{minimal sesquilinear form} $\QF_0$ on $\Dom(T)$ is defined by
    \[
        \QF_0[\u,\v]\coloneqq \langle \w,\v\rangle_{\cH},
        \qquad (\u,\w)\in T,
        \quad \v\in\Dom(T).
    \]
    The Symmetry of $\QF_0$ follows from the self-adjointness of $T$.
    This definition is independent of the choice of $\w$: if $(\u,\w_1),(\u,\w_2)\in T$, then $\w_1-\w_2\in\mul(T)=\Dom(T)^\perp$. Hence $\langle \w_1-\w_2,\v\rangle_{\cH}=0$ for all $\v\in\Dom(T)$.

    The associated quadratic form is $\QF_0[\u]\coloneqq\QF_0[\u,\u]$. It induces the form norm
    \[
        \|\u\|_{\QF_0}\coloneqq\sqrt{\QF_0[\u]+\|\u\|^2}.
    \]
\end{definition}

The space $\Dom(T)$ equipped with the inner product
\[
    \langle \u,\v\rangle_{\QF_0}\coloneqq\QF_0[\u,\v]+\langle \u,\v\rangle_{\cH}
\]
is a pre-Hilbert space. It is generally not complete. A prerequisite for completing it as a form domain embedded in $\cH$ is closability.

\begin{proposition}[Closability]
\label{prop:closability}
    The minimal form $\QF_0$ is closable. Equivalently, for every sequence $\{\u_n\}\subseteq\Dom(T)$,
    \[
        \|\u_n\|\to0 \text{ in }\cH
        \quad\text{and}\quad
        \{\u_n\} \text{ is a Cauchy sequence under the form norm }\|\cdot\|_{\QF_0}
        \quad\Longrightarrow\quad
        \QF_0[\u_n]\to0.
    \]
\end{proposition}

\begin{proof}[Proof Sketch]
    Let $\{\u_n\}\subseteq\Dom(T)$ satisfy the two assumptions. Let $\cH_{\QF}$ be the abstract completion of $\Dom(T)$ under $\|\cdot\|_{\QF_0}$. Then $\{\u_n\}$ converges in $\cH_{\QF}$ to some $\tilde\u$.

    Fix $\v\in\Dom(T)$ and choose $(\v,\w)\in T$. For each $n$, choose $(\u_n,\w_n)\in T$. Since $T$ is symmetric,
    \[
        \QF_0[\u_n,\v]=\langle \w_n,\v\rangle_{\cH}=\langle \u_n,\w\rangle_{\cH}.
    \]
    Therefore
    \[
        \langle \u_n,\v\rangle_{\QF_0}
        =\QF_0[\u_n,\v]+\langle\u_n,\v\rangle_{\cH}
        =\langle\u_n,\w\rangle_{\cH}+\langle\u_n,\v\rangle_{\cH}\to0.
    \]
    Hence $\langle\tilde\u,\v\rangle_{\cH_{\QF}}=0$ for all $\v\in\Dom(T)$. Since $\Dom(T)$ is dense in its completion $\cH_{\QF}$, it follows that $\tilde\u=0$. Thus $\|\u_n\|_{\QF_0}\to0$, and in particular $\QF_0[\u_n]\to0$.
\end{proof}

\begin{definition}[Associated closed quadratic form]
\label{def:associated_quadratic_form}
    Let $T$ be a positive self-adjoint linear relation and let $\QF_0$ be its minimal form. The associated closed form $\QF_T$ is the closure of $\QF_0$.
    \begin{enumerate}[(1)]
        \item \textbf{Domain:} $\Dom(\QF_T)$ consists of all $\u\in\cH$ for which there exists a sequence $\{\u_n\}\subseteq\Dom(T)$ such that $\u_n\to\u$ in $\cH$ and $\{\u_n\}$ is Cauchy with respect to $\|\cdot\|_{\QF_0}$.
        \item \textbf{Form value:} for such $\u$, define
        \[
            \QF_T[\u]\coloneqq\lim_{n\to\infty}\QF_0[\u_n].
        \]
        More generally, if $\u_n\to\u$ and $\v_n\to\v$ are approximating sequences that are Cauchy in the form norm, then
        \[
            \QF_T[\u,\v]\coloneqq\lim_{n\to\infty}\QF_0[\u_n,\v_n].
        \]
    \end{enumerate}
    By \Cref{prop:closability}, these limits are independent of the approximating sequences\footnote{Indeed, let $\{\u_n\},\{\v_n\}$ be two $\QF_0$-Cauchy sequences that are both convergent to $\u$, and $\w_n \coloneqq \u_n - \v_n$. By \Cref{prop:closability}, $\QF_0[\w_n] \to 0$. Applying the Minkowski inequality for the semi-norm $\sqrt{\QF_0[\cdot]}$, we obtain $|\sqrt{\QF_0[u_n]} - \sqrt{\QF_0[v_n]}| \le \sqrt{\QF_0[w_n]} \to 0$. This ensures that the limits of the quadratic forms coincide.}. Therefore, the space $(\Dom(\QF_T),\|\cdot\|_{\QF_T})$ is a Hilbert space, where
    \[
        \|\u\|_{\QF_T}^2=\QF_T[\u]+\|\u\|^2.
    \]
\end{definition}

We now state the form-theoretic definitions used throughout the semantic construction.

\begin{definition}[Positive quadratic form]
\label{def:Positive_QF}
    A \emph{positive quadratic form} on $\cH$ is an extended map $\QF:\cH\to[0,\infty]$ whose finite domain
    \[
        \Dom(\QF)\coloneqq\{\u\in\cH\mid \QF[\u]<\infty\}
    \]
    is a linear subspace of $\cH$, and for which there exists a positive sesquilinear form
    \[
        \mathfrak{s}:\Dom(\QF)\times\Dom(\QF)\to\bC
    \]
    such that $\QF[\u]=\mathfrak{s}[\u,\u]$ for all $\u\in\Dom(\QF)$. We write $\QF[\u,\v]$ for $\mathfrak{s}[\u,\v]$ on $\Dom(\QF)$, and set $\QF[\u]=+\infty$ when $\u\notin\Dom(\QF)$.

    The form induces an inner product on $\Dom(\QF)$ by
    \[
        \langle \u,\v\rangle_{\QF}\coloneqq\QF[\u,\v]+\langle\u,\v\rangle_{\cH}.
    \]
\end{definition}

\begin{definition}[Closed quadratic form]
\label{def:Positive_Closed_QF}
    A positive quadratic form $\QF$ is \emph{closed} if $(\Dom(\QF),\langle\cdot,\cdot\rangle_{\QF})$ is a Hilbert space.
\end{definition}

In practice, checking completeness via the form norm can be cumbersome. The following standard characterization connects closedness to lower semicontinuity. An extended functional $f:\cH\to\bR\cup\{+\infty\}$ is \emph{lower semicontinuous} (LSC for short) if
\[
    \u_n\to\u \text{ in }\cH
    \quad\Longrightarrow\quad
    f(\u)\leq\liminf_{n\to\infty}f(\u_n).
\]

\begin{theorem}[Closedness $\iff$ lower semicontinuity]
\label{thm:closed_iff_lsc}
    A positive quadratic form $\QF$ is closed if and only if the extended functional $\u\mapsto\QF[\u]$ is lower semicontinuous on $\cH$.
\end{theorem}

\begin{proof}
    ($\implies$) Suppose first that $\QF$ is closed and let $\u_n\to\u$ in $\cH$. If $\liminf_n\QF[\u_n]=+\infty$, there is nothing to prove. Otherwise choose a subsequence $\u_{n_k}$ such that
    \[
        \QF[\u_{n_k}]\to L\coloneqq\liminf_{n\to\infty}\QF[\u_n]<\infty.
    \]
    Discarding finitely many terms if necessary, we may assume $\u_{n_k}\in\Dom(\QF)$ for every $k$. Since $\u_{n_k}\to\u$ in $\cH$ and $\QF[\u_{n_k}]$ is bounded, the sequence $\{\u_{n_k}\}$ is bounded in the form Hilbert space $(\Dom(\QF),\|\cdot\|_{\QF})$. By Banach-Alaoglu theorem and the reflexivity of Hilbert spaces, after passing to a further subsequence, not relabelled, there is $\u'\in\Dom(\QF)$ such that
    \[
        \u_{n_k}\rightharpoonup \u'
        \quad\text{weakly in }(\Dom(\QF),\langle\cdot,\cdot\rangle_{\QF}).
    \]
    Consider the inclusion map $\iota:(\Dom(\QF),\|\cdot\|_{\QF})\hookrightarrow(\cH,\|\cdot\|_{\cH})$. Since $\|\v\|_{\cH}\le\|\v\|_{\QF}$ for all $\v\in\Dom(\QF)$, $\iota$ is continuous and linear. Therefore $\u_{n_k}\rightharpoonup\u'$ weakly in $\cH$. On the other hand, the same subsequence converges to $\u$ in the sense of $\cH$-norm, hence also weakly to $\u$ in $\cH$. Uniqueness of weak limits gives $\u'=\u$, and therefore $\u\in\Dom(\QF)$.

    We now prove the weak lower semicontinuity of the form norm. With our convention that the inner product is linear in the second variable, weak convergence in the form Hilbert space gives
    \[
        \lim_{k\to\infty}\langle \u,\u_{n_k}\rangle_{\QF}
        =\langle \u,\u\rangle_{\QF}
        =\|\u\|_{\QF}^2.
    \]
    By Cauchy-Schwarz inequality,
    \[
        |\langle \u,\u_{n_k}\rangle_{\QF}|
        \le \|\u\|_{\QF}\,\|\u_{n_k}\|_{\QF}.
    \]
    If $\|\u\|_{\QF}=0$, the desired inequality is trivial. Otherwise, taking the limit inferior in the above estimate and dividing by $\|\u\|_{\QF}$ yields
    \[
        \|\u\|_{\QF}\le\liminf_{k\to\infty}\|\u_{n_k}\|_{\QF}.
    \]
    Hence
    \[
        \|\u\|_{\QF}^2
        \le \liminf_{k\to\infty}\|\u_{n_k}\|_{\QF}^2
        =\lim_{k\to\infty}\bigl(\QF[\u_{n_k}]+\|\u_{n_k}\|_{\cH}^2\bigr)
        =L+\|\u\|_{\cH}^2.
    \]
    Since $\|\u\|_{\QF}^2=\QF[\u]+\|\u\|_{\cH}^2$, we obtain
    \[
        \QF[\u]\le L=\liminf_{n\to\infty}\QF[\u_n].
    \]

    ($\impliedby$) Conversely, assume that $\QF$ is lower semicontinuous. Let $\{\u_n\}$ be Cauchy in $\|\cdot\|_{\QF}$. Since $\|\v\|\le\|\v\|_{\QF}$, the sequence is Cauchy in $\cH$, so $\u_n\to\u$ in $\cH$ for some $\u\in\cH$. For every $\epsilon>0$, choose $N$ such that $\|\u_n-\u_m\|_{\QF}^2<\epsilon$ for $n,m\ge N$. Fixing $n\ge N$ and letting $m\to\infty$, lower semicontinuity yields
    \[
        \QF[\u_n-\u]\le\liminf_{m\to\infty}\QF[\u_n-\u_m]\le\epsilon.
    \]
    Thus $\u_n-\u\in\Dom(\QF)$ and hence $\u\in\Dom(\QF)$. Moreover, $\QF[\u_n-\u]\to0$ and $\|\u_n-\u\|\to0$, so $\u_n\to\u$ in the form norm. Therefore $\Dom(\QF)$ is complete.
\end{proof}

\begin{theorem}[Kato's First representation theorem]
\label{thm:kato_first}
    There is a one-to-one correspondence between positive self-adjoint linear relations $T$ on $\cH$ and closed positive quadratic forms $\QF$ on $\cH$:
    \[
        T\longleftrightarrow\QF.
    \]
    Under our inner-product convention, this correspondence is characterized by
    \begin{equation}\label{eq:variational}
        (\u,\w)\in T
        \iff
        \u\in\Dom(\QF)\text{ and }
        \QF[\u,\v]=\langle\w,\v\rangle_{\cH}
        \text{ for all }\v\in\Dom(\QF).
    \end{equation}
\end{theorem}

\begin{remark}[Domain distinction]
    The condition $\u\in\Dom(\QF)$ is necessary for $\QF[\u,\v]$ to be defined, but it is not sufficient for $\u\in\Dom(T)$. In general,
    \[
        \Dom(T)\subsetneq\Dom(\QF).
    \]
    The variational equation \eqref{eq:variational} is a regularity constraint: it selects precisely those $\u\in\Dom(\QF)$ for which the functional $\v\mapsto\QF[\u,\v]$ is bounded with respect to the ambient $\cH$-norm, and hence can be represented by some $\w\in\cH$. The representing vector is unique modulo $(\Dom(\QF))^\perp$, which is exactly the multivalued part of the corresponding relation.
\end{remark}

To prove this theorem, we firstly state and prove the following lemma:
\begin{lemma}[Friedrichs resolvent construction]
\label{lem:friedrichs_resolvent}
    Let $\QF$ be a closed positive quadratic form with domain $D\coloneqq\Dom(\QF)$, and put
    \[
        \cH_0\coloneqq\overline D,\qquad
        \mathcal V\coloneqq(D,\langle\cdot,\cdot\rangle_{\QF}).
    \]
    For each $\x\in\cH_0$, there is a unique $B\x\in D$ such that
    \begin{equation}\label{eq:friedrichs_resolvent}
        \QF[B\x,\v]+\langle B\x,\v\rangle_{\cH}
        =
        \langle \x,\v\rangle_{\cH},
        \qquad \forall \v\in D .
    \end{equation}
    The operator $B$ is an injective positive self-adjoint contraction on $\cH_0$, and $\Ran(B)$ is dense in $\mathcal V$. Furthermore, define
    \[
        A\coloneqq B^{-1}-I,
        \qquad
        \Dom(A)=\Ran(B),
    \]
    then $A$ is a positive self-adjoint operator (maybe unbounded) on $\cH_0$. Moreover,
    \begin{equation}\label{eq:friedrichs_variational}
        \u\in\Dom(A),\ A\u=\mathbf{f}
        \iff
        \u\in D
        \text{ and }
        \QF[\u,\v]=\langle\mathbf{f},\v\rangle_{\cH}
        \text{ for all }\v\in D .
    \end{equation}
\end{lemma}

\begin{proof}
    The functional $\v\mapsto\langle\x,\v\rangle_{\cH}$ is bounded on $\mathcal V$ since
    \[
        |\langle\x,\v\rangle_{\cH}|
        \le \|\x\|_{\cH}\|\v\|_{\cH}
        \le \|\x\|_{\cH}\|\v\|_{\QF}.
    \]
    Hence \eqref{eq:friedrichs_resolvent} follows from Riesz representation, equivalently Lax--Milgram, applied to the Hilbert space $\mathcal V$.

    Taking $\v=B\x$ in \eqref{eq:friedrichs_resolvent} gives
    \[
        \|B\x\|_{\QF}^2
        =
        \langle\x,B\x\rangle_{\cH}
        \le
        \|\x\|_{\cH}\|B\x\|_{\cH}
        \le
        \|\x\|_{\cH}\|B\x\|_{\QF},
    \]
    so $B$ is bounded $\cH_0\to\mathcal V$ and a contraction on $\cH_0$. The same identity gives
    $\langle\x,B\x\rangle_{\cH}\ge0$, so $B$ is positive. For $\x,\y\in\cH_0$, applying \eqref{eq:friedrichs_resolvent} with $\v=B\y$ in the representation of $\x$, and with $\v=B\x$ in the representation of $\y$, gives
    \[
        \langle\x,B\y\rangle_{\cH}
        =
        \langle B\x,B\y\rangle_{\QF}
        =
        \langle B\x,\y\rangle_{\cH},
    \]
    where the second equality uses conjugate symmetry. Hence $B$ is self-adjoint. If $B\x=0$, then \Cref{eq:friedrichs_resolvent} gives that $\langle\x,\v\rangle_{\cH}=0$ for all $\v\in D$, and since $\x\in\overline D$, we get $\x=0$, which implies $B$ is injective.

    If $\eta\in\mathcal{V}$ is orthogonal to $\Ran(B)$ in $\mathcal V$, then
    \[
        0=\langle B\x,\eta\rangle_{\QF}
        =
        \langle\x,\eta\rangle_{\cH},
        \qquad \forall\x\in\cH_0.
    \]
    Taking $\x=\eta$ gives $\eta=0$, so $\Ran(B)$ is dense in $\mathcal V$. Note that the density needed here is form-norm density, not merely $\cH_0$-density.

    Define $A=B^{-1}-I$ on $\Ran(B)$. For $\u=B\x$, \eqref{eq:friedrichs_resolvent} gives
    \[
        \QF[\u,\v]
        =
        \langle \x-B\x,\v\rangle_{\cH}
        =
        \langle A\u,\v\rangle_{\cH},
        \qquad \forall\v\in D.
    \]
    Conversely, if $\u\in D$ and $\QF[\u,\v]=\langle\mathbf{f},\v\rangle_{\cH}$ for all $\v\in D$, then
    \[
        \langle\u+\mathbf{f},\v\rangle_{\cH}
        =
        \QF[\u,\v]+\langle\u,\v\rangle_{\cH}
        =
        \langle\u,\v\rangle_{\QF},
    \]
    so uniqueness in \eqref{eq:friedrichs_resolvent} gives $B(\u+\mathbf{f})=\u$, hence $\u\in\Ran(B)=\Dom(A)$ and $A\u=\mathbf{f}$. This proves \eqref{eq:friedrichs_variational}. In particular, $A$ is positive and symmetric. Since $(A+I)B\x=\x$ for every $\x\in\cH_0$, we have $\Ran(A+I)=\cH_0$, by the elementary range criterion for positive symmetric operators (see such as \cite[Proposition 3.11]{schmudgen2012unbounded}), $A$ is self-adjoint.
\end{proof}

\begin{proof}[Proof of \Cref{thm:kato_first}]
    \textbf{1. From Form to Relation ($\QF \mapsto T_{\QF}$):}
    Let $\QF$ be a closed positive form and put
    \[
        \cH_0\coloneqq\overline{\Dom(\QF)},
        \qquad
        \cK\coloneqq(\Dom(\QF))^\perp .
    \]
    Apply \Cref{lem:friedrichs_resolvent} to $\QF$ as a densely defined closed form on $\cH_0$, and let $A$ be the resulting positive self-adjoint operator on $\cH_0$. Define
    \[
        T_{\QF}
        \coloneqq
        \mathcal G(A)\oplus(\{0\}\times\cK).
    \]
    Then $T_{\QF}$ is a positive self-adjoint linear relation on $\cH$ and
    \[
        \mul(T_{\QF})=\cK=(\Dom(\QF))^\perp.
    \]
    If $(\u,\w)\in T_{\QF}$, then $\w=A\u+\kappa$ for some $\kappa\in\cK$, and therefore, for every $\v\in\Dom(\QF)$,
    \[
        \langle\w,\v\rangle_{\cH}
        =
        \langle A\u,\v\rangle_{\cH}
        =
        \QF[\u,\v].
    \]
    Conversely, suppose $\u\in\Dom(\QF)$ and satisfies
    \[
        \QF[\u,\v]=\langle\w,\v\rangle_{\cH},
        \qquad \forall\v\in\Dom(\QF).
    \]
    Write $\w=\w_0+\kappa$ with $\w_0\in\cH_0$ and $\kappa\in\cK$. Since $\kappa\perp\Dom(\QF)$, we have
    \[
        \QF[\u,\v]=\langle\w_0,\v\rangle_{\cH},
        \qquad \forall\v\in\Dom(\QF).
    \]
    By \Cref{eq:friedrichs_variational}, $\u\in\Dom(A)$ and $A\u=\w_0$. Hence $(\u,\w)\in T_{\QF}$. This proves \eqref{eq:variational} for the form-to-relation construction.

    \textbf{2. From Relation to Form ($T \mapsto \QF_T$):}
    Given a positive self-adjoint linear relation $T$, we define $\QF_T$ as the closure of the minimal form $\QF_0[\u] = \langle \w, \u \rangle$ for any $(\u,\w)\in T$ and $\v\in\Dom(T)$, as constructed in \Cref{def:minimal-quadratic-form} and \Cref{def:associated_quadratic_form}. By the definition of closure, the domain $\Dom(T)$ is dense in $\Dom(\QF_T)$ with respect to the form norm.

    \textbf{3. Bijectivity (Inverse Property):}
    Starting from $\QF$ and then forming the minimal form of $T_{\QF}$ gives the restriction of $\QF$ to $\Dom(A)=\Ran(B)$. By \Cref{lem:friedrichs_resolvent}, $\Ran(B)$ is dense in $\Dom(\QF)$ in the form norm; hence the closure is exactly $\QF$.

    Starting from $T$, take $(\u,\w)\in T$. For $\v\in\Dom(\QF_T)$, choose $\v_n\in\Dom(T)$ with $\v_n\to\v$ in the form norm. Then
    \[
        \QF_T[\u,\v]
        =
        \lim_n \QF_0[\u,\v_n]
        =
        \lim_n \langle\w,\v_n\rangle_{\cH}
        =
        \langle\w,\v\rangle_{\cH}.
    \]
    Thus $T\subseteq T_{\QF_T}$. Since both relations are self-adjoint,
    \[
        T_{\QF_T}=T_{\QF_T}^*\subseteq T^*=T.
    \]
    Hence $T_{\QF_T}=T$. Therefore the two constructions are inverse to each other.
\end{proof}

\subsection{The L\"owner Order via Quadratic Forms}
\label{subsec:loewner-order-and-extended-trace}
To define a logic, we need an order structure on predicates. The standard operator order is inadequate for unbounded operators because their domains need not match. We therefore use the extended L\"owner order induced by quadratic forms.

\begin{definition}[Extended L\"owner order]
\label{def:extended-lowner-order}
    Let $A$ and $B$ be positive self-adjoint linear relations with associated closed forms $\QF_A$ and $\QF_B$. We write $A\sqleq B$ if and only if
    \begin{enumerate}[(1)]
        \item $\Dom(\QF_B)\subseteq\Dom(\QF_A)$, and
        \item $\QF_A[\u]\le\QF_B[\u]$ for every $\u\in\Dom(\QF_B)$.
    \end{enumerate}
    Equivalently, after extending each form by $+\infty$ outside its domain, $A\sqleq B$ iff $\QF_A[\u]\le\QF_B[\u]$ for all $\u\in\cH$.
\end{definition}

\begin{remark}
    The domain inclusion is reversed. This is the usual behavior of form order: a larger unbounded predicate may have a smaller finite-value domain.
\end{remark}

The extended L\"owner order connects directly to expectation values. Since an unbounded relation $T$ need not define a bounded operator product $T\rho$, the usual trace is generally not available. We therefore define the trace through the associated closed form.

For any $\u\in\cH$, let $\sP_{\u}$ be the rank-one positive operator
\[
    \sP_{\u}(\v)=\langle\u,\v\rangle_{\cH}\u,
    \qquad \v\in\cH.
\]
If $\|\u\|=1$, then $\sP_{\u}$ is the orthogonal projection onto $\operatorname{span}\{\u\}$.

\begin{definition}[Extended trace]
\label{def:extended_trace}
    Let $T$ be a positive self-adjoint linear relation with associated closed form $\QF_T$.
    \begin{enumerate}[(1)]
        \item For a pure state $\rho=\sP_{\u}$ with $\|\u\|=1$, define
        \[
            \Tr(T\rho)\coloneqq\QF_T[\u],
        \]
        where the value is $+\infty$ if $\u\notin\Dom(\QF_T)$.
        \item For a partial density operator $\rho$ with spectral decomposition
        \[
            \rho=\sum_i p_i\sP_{\u_i},
            \qquad p_i>0,
            \quad \sum_i p_i\le1,
        \]
        where $\{\u_i\}$ is an finite or countable orthonormal family, define
        \[
            \Tr(T\rho)\coloneqq\sum_i p_i\,\QF_T[\u_i]\in[0,\infty].
        \]
    \end{enumerate}
\end{definition}

\begin{remark}[Well-definedness and additivity]
    The above definition of the extended trace has a direct physical
    interpretation, but its independence of the chosen decomposition
    $\rho=\sum_i p_i\sP_{\u_i}$ is not immediate. We justify it by showing
    that our definition agrees with the spectral-measure definition of
    expectations for positive unbounded observables used in mathematical
    physics; see, e.g.,
    \cite{davies1976quantum,holevo2011probabilistic}. This equivalence
    also gives additivity and positive homogeneity on the positive
    trace-class cone. The details are given in
    \hyperref[app:trace_proof]{\Cref*{app:trace_proof}}.
\end{remark}

This expectation characterizes the logical order.

\begin{proposition}[Order characterization]
\label{prop:order_characterization}
    Let $A$ and $B$ be positive self-adjoint linear relations. Then $A\sqleq B$ if and only if
    \[
        \Tr(A\rho)\le\Tr(B\rho),
        \qquad \forall\rho\in\pardensity{\cH}.
    \]
\end{proposition}

\begin{proof}[Proof Sketch]
    Suppose $A\sqleq B$ and let $\rho=\sum_i p_i\sP_{\u_i}$ be a partial density operator. For every $i$, the extended-form inequality gives $\QF_A[\u_i]\le\QF_B[\u_i]$. Multiplying by $p_i$ and summing yields $\Tr(A\rho)\le\Tr(B\rho)$.

    Conversely, assume the trace inequality holds for all partial density operators. Choosing the pure state $\rho=\sP_{\v}$ for every unit vector $\v$ gives $\QF_A[\v]\le\QF_B[\v]$. By homogeneity of quadratic forms, the same inequality holds for every $\u\in\cH$ after writing $\v=\u/\|\u\|$ when $\u\ne0$; the case $\u=0$ is trivial. Therefore $\QF_A[\u]\le\QF_B[\u]$ as extended functions on $\cH$, which is equivalent to $A\sqleq B$.
\end{proof}

\subsection{Monotone Convergence and CPO Structure}

The semantics of loops relies on least fixed points, which require an order-complete predicate domain. In the unbounded setting, weak and strong operator convergence are not the right primitives; monotone convergence of closed forms is.

We use the standard monotone convergence theorem for positive closed forms, including singular forms \cite{Simon_1978,Behrndt_Hassi_de_Snoo_Wietsma_2010,behrndt2020boundary}.

\begin{theorem}[Monotone convergence for forms]
\label{thm:monotone_convergence}
    Let $\{\QF_n\}_{n=1}^{\infty}$ be a non-decreasing sequence of positive closed quadratic forms, i.e.,
    \[
        \Dom(\QF_{n+1})\subseteq\Dom(\QF_n),
        \qquad
        \QF_n[\u]\le\QF_{n+1}[\u]
        \quad \forall\u\in\Dom(\QF_{n+1}).
    \]
    Define
    \[
        \Dom(\QF)\coloneqq
        \left\{\u\in\bigcap_{n=1}^{\infty}\Dom(\QF_n)
        \;\middle|\;
        \sup_n\QF_n[\u]<\infty\right\},
    \]
    and
    \[
        \QF[\u]\coloneqq\sup_n\QF_n[\u]
        =\lim_{n\to\infty}\QF_n[\u],
        \qquad \u\in\Dom(\QF),
    \]
    with $\QF[\u]=+\infty$ outside $\Dom(\QF)$. Then $\QF$ is a positive closed quadratic form.
\end{theorem}

\begin{proof}
    We show completeness of $(\Dom(\QF),\|\cdot\|_{\QF})$. Let $\{\u_k\}_{k=1}^{\infty}$ be Cauchy in $\|\cdot\|_{\QF}$. Since $\|\v\|_{\cH}\le\|\v\|_{\QF}$, there exists $\u\in\cH$ such that
    \begin{equation}\label{eq:H_limit}
        \lim_{k\to\infty}\|\u_k-\u\|_{\cH}=0.
    \end{equation}

    Fix $n$. Since $\QF_n\le\QF$ on $\Dom(\QF)$, the sequence $\{\u_k\}$ is Cauchy in $\|\cdot\|_{\QF_n}$. Closedness of $\QF_n$ gives a limit $\u^{(n)}\in\Dom(\QF_n)$ in the $\QF_n$-norm. This convergence implies convergence in $\cH$, so uniqueness of the $\cH$-limit gives $\u^{(n)}=\u$. Hence
    \begin{equation}\label{eq:norm_conv_n}
        \u\in\Dom(\QF_n),
        \qquad
        \lim_{k\to\infty}\|\u_k-\u\|_{\QF_n}=0
        \quad\text{for every }n.
    \end{equation}
    In particular, $\QF_n[\u]=\lim_{k\to\infty}\QF_n[\u_k]$.

    Since $\{\u_k\}$ is Cauchy in $\|\cdot\|_{\QF}$, it is bounded in that norm. Hence there exists $M$ such that $\QF[\u_k]\le M$ for all $k$. For every fixed $n$,
    \[
        \QF_n[\u]=\lim_{k\to\infty}\QF_n[\u_k]
        \le\limsup_{k\to\infty}\QF[\u_k]
        \le M.
    \]
    Taking the supremum over $n$ gives $\sup_n\QF_n[\u]<\infty$, so $\u\in\Dom(\QF)$.

    It remains to show form-norm convergence. Let $\epsilon>0$. Choose $N$ such that
    \[
        \|\u_k-\u_m\|_{\QF}^2<\epsilon,
        \qquad k,m\ge N.
    \]
    Then $\QF_n[\u_k-\u_m]\le\QF[\u_k-\u_m]<\epsilon$ for all $n$ and all $k,m\ge N$. Fixing $k\ge N$ and letting $m\to\infty$, \eqref{eq:norm_conv_n} gives
    \[
        \QF_n[\u_k-\u]\le\epsilon,
        \qquad\forall n.
    \]
    Therefore
    \[
        \QF[\u_k-\u]=\sup_n\QF_n[\u_k-\u]\le\epsilon.
    \]
    Combining this with \eqref{eq:H_limit} yields $\|\u_k-\u\|_{\QF}\to0$. Thus $\QF$ is closed.
\end{proof}

This theorem is the cornerstone of our semantic domain.

\begin{corollary}[$\omega$-CPO property]
\label{cor:omega-CPO-for-LRs}
    The set of all positive self-adjoint linear relations, equipped with the extended L\"owner order $\sqleq$, forms an $\omega$-complete partial order ($\omega$-CPO). Every countable ascending chain has a least upper bound, obtained by taking the pointwise supremum of the associated quadratic forms. The bottom element is the zero relation
    \[
        \mathbf{0}=\cH\oplus\{0\},
    \]
    represented by the zero form.
\end{corollary}

\begin{proof}
    Let $T_1\sqleq T_2\sqleq\cdots$ be an ascending chain of positive self-adjoint linear relations, and let $\QF_n$ be the closed form associated with $T_n$. By \Cref{def:extended-lowner-order}, the forms are non-decreasing in the sense of \Cref{thm:monotone_convergence}. Let $\QF$ be the closed limit form given there, and let $T_\infty$ be the corresponding relation from \Cref{thm:kato_first}.

    For every $n$, we have $\Dom(\QF)\subseteq\Dom(\QF_n)$ and $\QF_n[\u]\le\QF[\u]$ on $\Dom(\QF)$, hence $T_n\sqleq T_\infty$. Thus $T_\infty$ is an upper bound.

    If $K$ is another upper bound, then $T_n\sqleq K$ for all $n$. Equivalently,
    \[
        \Dom(\QF_K)\subseteq\Dom(\QF_n),
        \qquad
        \QF_n[\u]\le\QF_K[\u]
        \quad\forall\u\in\Dom(\QF_K),\ \forall n.
    \]
    Taking the supremum over $n$ gives $\QF[\u]\le\QF_K[\u]$ for all $\u\in\Dom(\QF_K)$, and also $\Dom(\QF_K)\subseteq\Dom(\QF)$. Therefore $T_\infty\sqleq K$. Hence $T_\infty$ is the least upper bound.
\end{proof}

\begin{remark}[Relation to resolvent convergence]
    For a non-decreasing sequence of closed forms, the form limit above is equivalent to strong resolvent convergence of the associated positive self-adjoint relations in the generalized sense of Kato \cite{kato2013perturbation}. Thus logical suprema agree with the standard spectral limiting process.
\end{remark}

\begin{corollary}[Monotone convergence of traces over linear relations]
\label{cor:mct_traces}
    Let $\{P_n\}_{n\in\mathbb{N}}$ be a monotonically increasing sequence of positive self-adjoint linear relations with supremum $P\coloneqq\sup_n P_n$. For every partial density operator $\rho$,
    \[
        \sup_n\Tr(P_n\rho)=\Tr(P\rho).
    \]
\end{corollary}

\begin{proof}
    Let $\rho=\sum_{k=1}^{\infty}\lambda_k\sP_{\u_k}$ be a spectral decomposition over the positive eigenvalues of $\rho$, where $\lambda_k>0$, $\sum_k\lambda_k\le1$, and $\{\u_k\}$ is an orthonormal family. By \Cref{def:extended_trace},
    \[
        \Tr(P_n\rho)=\sum_{k=1}^{\infty}\lambda_k\QF_{P_n}[\u_k].
    \]
    Since $P_n\sqleq P_{n+1}$, the sequence $\QF_{P_n}[\u_k]$ is non-decreasing for each fixed $k$. Beppo Levi's monotone convergence theorem for series gives
    \[
        \sup_n\Tr(P_n\rho)
        =\sup_n\sum_{k=1}^{\infty}\lambda_k\QF_{P_n}[\u_k]
        =\sum_{k=1}^{\infty}\lambda_k\sup_n\QF_{P_n}[\u_k].
    \]
    By \Cref{thm:monotone_convergence}, $\sup_n\QF_{P_n}[\u_k]=\QF_P[\u_k]$. Therefore
    \[
        \sup_n\Tr(P_n\rho)
        =\sum_{k=1}^{\infty}\lambda_k\QF_P[\u_k]
        =\Tr(P\rho),
    \]
    which completes the proof.
\end{proof}

Finally, to bridge the gap between standard quantum information theory, which mostly deals with bounded operators, and our domain of linear relations, we introduce bounded predicates and their spectral approximation.

\begin{definition}[Bounded positive linear relations]
    A positive self-adjoint linear relation $T$ is \emph{bounded} if it is the graph of a bounded positive self-adjoint operator on $\cH$. Equivalently, $\Dom(T)=\cH$, $\mul(T)=\{0\}$, and there exists $A\in\cB(\cH)$ such that
    \[
        T=\{(\u,A\u)\mid \u\in\cH\}.
    \]
\end{definition}

The following proposition ensures that every positive linear relation is generated as the supremum of bounded predicates, allowing properties proved in the bounded case to be lifted to the unbounded case via the $\omega$-CPO structure.

\begin{proposition}[Spectral approximation]
\label{prop:spectral_approximation}
    For every positive self-adjoint linear relation $T$, there exists a countable sequence of bounded positive linear relations $\{T_n\}_{n\in\mathbb{N}}$ such that
    \begin{enumerate}[(1)]
        \item $T_1\sqleq T_2\sqleq\cdots$, and
        \item $\QF_T[\u]=\sup_n\QF_{T_n}[\u]$ for every $\u\in\cH$.
    \end{enumerate}
    One may choose $T_n$ to be the bounded spectral saturation of $T$ at height $n$.
\end{proposition}

\begin{proof}
    By the spectral theorem for positive self-adjoint operators, applied
    to the operator part in the canonical decomposition of a positive
    self-adjoint linear relation, one obtains an extended spectral
    measure $E_T$ of $T$ on $[0,\infty]$; see
    \cite{kato2013perturbation,schmudgen2012unbounded,behrndt2020boundary}.
    We recall the precise relation-valued version needed here in
    \Cref{app:spectral_approximation}. For $n\in\mathbb N$, set
    \[
        f_n(\lambda)\coloneqq\min(\lambda,n),
        \qquad \lambda\in[0,\infty],
    \]
    with $f_n(\infty)=n$, and define
    \[
        T_n\coloneqq \int_{[0,\infty]} f_n(\lambda)\,dE_T(\lambda).
    \]
    Then $T_n$ is a bounded positive self-adjoint operator,
    $0\le T_n\le nI$, and $T_n\sqleq T_{n+1}$. Since
    $f_n(\lambda)\nearrow\lambda$ on the extended half-line, the
    spectral monotone convergence theorem gives
    \[
        \QF_T[\u]=\sup_n \QF_{T_n}[\u],
        \qquad \forall\u\in\cH.
    \]
    The only point not covered by the ordinary operator-valued spectral
    theorem is the mass at $\infty$, which represents the multivalued
    part of $T$; this is treated in
    \Cref{app:spectral_approximation}.
\end{proof}

\section{Syntax and Semantics of Quantum Programs}
\label{SEC:syntax_semantics}

\subsection{Syntax and State Space}
\label{subsec:syntax}

We extend the standard quantum programming syntax to explicitly support continuous-variable (CV) systems and continuous control flow. This involves introducing classical arithmetic expressions to parameterize unitary transformations and a generic binding mechanism for measurement outcomes.

\subsubsection*{Basic Setup and State Space}

Let $\mathit{QVar}$ be a \emph{finite} set of quantum variables representing the quantum registers available in the system (e.g., modes of an optical field). Each variable $q \in \mathit{QVar}$ is associated with a separable (not necessarily finite-dimensional) Hilbert space $\cH_q$ (e.g., $L^2(\mathbb{R})$).
In every $\cH_q$ we fix a distinguished unit vector \(\ket{0}_q\in\cH_q\), and an orthonormal basis
\[
    \{\ket{k}_q\}_{k\in J_q}
\]
containing \(\ket{0}_q\), where \(J_q\) is finite or countably infinite.
The vector \(\ket{0}_q\) is the reset state of \(q\).

The global system resides in the tensor product space:
\[
    \cH = \bigotimes_{q \in \mathit{QVar}} \cH_q.
\]
Since $\mathit{QVar}$ is finite and each $\cH_q$ is separable, the global Hilbert space $\cH$ remains \emph{separable}. This ensures that the mathematical foundations laid out in Sec.~\ref{sec:math_foundations} are strictly applicable.

\begin{remark}[Finiteness and Separability]
    We restrict $\mathit{QVar}$ to be finite for two reasons:
    \begin{enumerate}
        \item \textbf{Sufficiency:} The assumption of finite variables involves no loss of generality for program verification. Any concrete program $S$ written in our syntax is a finite string of symbols and thus can only reference a finite number of quantum variables.
        \item \textbf{Mathematical Stability:} The extension of this logic to systems with infinitely many degrees of freedom (potentially involving non-separable Hilbert spaces, as in general Quantum Field Theory) presents distinct topological challenges, particularly regarding the definitions of the trace functional and operator decompositions. We leave the exploration of such non-separable domains as an open problem for future mathematical investigation.
    \end{enumerate}
\end{remark}

The state of the system is modeled by a partial density operator.
\begin{definition}[State Space]
    The state space is the set of partial density operators on the global Hilbert space $\cH$:
    \[
        \pardensity{\cH} \coloneqq \{ \rho \in \cT(\cH) \mid \rho \sqgeq 0, \tr(\rho) \le 1 \}.
    \]
    We equip $\pardensity{\cH}$ with the standard L\"owner order: $\rho \sqsubseteq \rho' \iff \rho' - \rho \sqgeq 0$. 
\end{definition}

To seamlessly bridge the algebraic fixed-point semantics of loops with the measure-theoretic requirements of continuous environments, we establish the exact topological nature of this order.

\begin{proposition}[$\omega$-CPO and Trace-Norm Convergence]\label{prop:state_cpo_topology}
    The poset $(\pardensity{\cH}, \sqsubseteq)$ is an $\omega$-Complete Partial Order ($\omega$-CPO) with a least element $0$. 
    
    Moreover, for any increasing chain $\rho_0 \sqsubseteq \rho_1 \sqsubseteq \dots$ in $\pardensity{\cH}$, its supremum $\rho_\infty = \sup_n \rho_n$ exists in $\pardensity{\cH}$, and:
    \[
        \lim_{n \to \infty} \|\rho_\infty - \rho_n\|_1 = 0.
    \]
\end{proposition}

\begin{proof}
    Let \((\rho_n)_{n\in\mathbb N}\) be an increasing sequence in
    \(\pardensity{\cH}\).  Since \(0\sqleq\rho_n\) and
    \(\tr(\rho_n)\leq 1\), we have \(\rho_n\sqleq I\).  By the monotone
    convergence theorem for bounded positive self-adjoint operators
    \cite[Lemma 5.1.4]{kadison1997fundamentals}, there exists a
    positive operator \(\rho_\infty\in\cB(\cH)\) such that
    \[
        \rho_n\to\rho_\infty
        \,\,\text{strongly, hence weakly,}
    \]
    and \(\rho_\infty=\sup_n\rho_n\) in the L\"owner order.

    To verify that $\rho_\infty \in \pardensity{\cH}$, we evaluate its trace. Let $\{e_k\}_{k=1}^\infty$ be an arbitrary orthonormal basis of $\cH$. The trace is given by $\tr(\rho_n) = \sum_{k=1}^\infty \langle e_k, \rho_n e_k \rangle$. By the Monotone Convergence Theorem for series, we can exchange the supremum over $n$ and the infinite sum over $k$:
    \[
        \tr(\rho_\infty) = \sum_{k=1}^\infty \sup_n \langle e_k, \rho_n e_k \rangle = \sup_n \sum_{k=1}^\infty \langle e_k, \rho_n e_k \rangle = \sup_n \tr(\rho_n).
    \]
    Since $\tr(\rho_n) \le 1$ for all $n$, we have $\tr(\rho_\infty) \le 1$. Thus, $\rho_\infty \in \pardensity{\cH}$, which establishes the $\omega$-CPO property.
    
    Finally, we establish the trace-norm convergence. Since $\rho_\infty \ge \rho_n \ge 0$, the operator $\rho_\infty - \rho_n$ is positive semidefinite. For positive operators, the trace norm equals the trace itself. Therefore:
    \[
        \|\rho_\infty - \rho_n\|_1 = \tr(\rho_\infty - \rho_n) = \tr(\rho_\infty) - \tr(\rho_n).
    \]
    As $n \to \infty$, $\tr(\rho_n) \to \tr(\rho_\infty)$ by our earlier derivation, yielding $\|\rho_\infty - \rho_n\|_1 \to 0$. 
\end{proof}

\subsubsection*{Classical Expressions and Quantum Syntax}

Let $\mathit{CVar}$ be a finite set of classical parameter variables, denoted
by $x,y,\dots$.  These variables are meta-level variables used to
parameterize quantum primitives and program phrases.  They are not mutable
program variables, and the language contains no classical assignment command.

Each classical parameter variable $x\in\mathit{CVar}$ is assigned a
$\sigma$-finite measure space
\[
    (\Omega_x,\Sigma_x,\mu_x).
\]
A classical context $\Gamma\subseteq\mathit{CVar}$ records the classical
parameter variables currently in scope.

For a finite context $\Gamma$, we define the corresponding meta-parameter
space by
\[
    \Omega_\Gamma
    \coloneqq
    \prod_{x\in\Gamma}\Omega_x,
    \qquad
    \Sigma_\Gamma
    \coloneqq
    \bigotimes_{x\in\Gamma}\Sigma_x,
    \qquad
    \mu_\Gamma
    \coloneqq
    \bigotimes_{x\in\Gamma}\mu_x .
\]
No completion of the product $\sigma$-algebra is taken.  Since $\Gamma$ is
finite and each $(\Omega_x,\Sigma_x,\mu_x)$ is $\sigma$-finite,
$(\Omega_\Gamma,\Sigma_\Gamma,\mu_\Gamma)$ is again $\sigma$-finite.

Elements of $\Omega_\Gamma$ are denoted by $\omega$.  When
$x\in\Gamma$, the $x$-component of $\omega$ is written
\[
    \omega(x)\in\Omega_x .
\]
For the empty context, we use the one-point measure space
\[
    \Omega_\emptyset=\{\ast\},
    \qquad
    \Sigma_\emptyset=\{\emptyset,\{\ast\}\},
    \qquad
    \mu_\emptyset(\{\ast\})=1.
\]

\begin{remark}[$\sigma$-finite reference spaces]
    Throughout the syntax, all classical parameter spaces and measurement
    outcome spaces are assumed to carry $\sigma$-finite reference measures.
    This assumption covers the standard spaces used in continuous-variable
    quantum programming, such as finite products of Euclidean spaces with
    Borel measure and countable discrete spaces with counting measure.  It
    is also technically convenient for the Bochner-integral semantics, since it
    allows product-measurability and parameterized integration arguments to be
    applied uniformly in the induction on program structure.
\end{remark}

Classical parameter expressions over $\Gamma$ are defined by the grammar
\begin{equation}
    e ::= c \mid x \mid f(e_1,\dots,e_n)
    \qquad (\text{Parameter expressions}),
\end{equation}
where $c$ is a constant, $x\in\Gamma$ is a classical parameter variable
currently in scope, and $f$ ranges over measurable functions of the appropriate
arity.

Each well-typed expression $e$ has a $\sigma$-finite value space
\[
    (\Omega_e,\Sigma_e,\mu_e)
\]
and denotes a measurable function
\[
    \sem{e}_\Gamma:
    (\Omega_\Gamma,\Sigma_\Gamma)
    \longrightarrow
    (\Omega_e,\Sigma_e).
\]
For a meta-parameter instance $\omega\in\Omega_\Gamma$, the value of $e$ at
$\omega$ is written
\[
    \sem{e}_\Gamma(\omega)\in\Omega_e .
\]
This interpretation is only a meta-level interpretation of syntactic
parameters.  It is not a program command, does not modify a classical store,
and does not enter the quantum state space.

For an expression tuple $\seq e=(e_1,\ldots,e_k)$, we write
\[
    \Omega_{\seq e}
    \coloneqq
    \Omega_{e_1}\times\cdots\times\Omega_{e_k},
    \qquad
    \Sigma_{\seq e}
    \coloneqq
    \Sigma_{e_1}\otimes\cdots\otimes\Sigma_{e_k},
    \qquad
    \mu_{\seq e}
    \coloneqq
    \mu_{e_1}\otimes\cdots\otimes\mu_{e_k},
\]
and
\[
    \sem{\seq e}_\Gamma
    \coloneqq
    \bigl(\sem{e_1}_\Gamma,\ldots,\sem{e_k}_\Gamma\bigr)
    :
    \Omega_\Gamma
    \longrightarrow
    \Omega_{\seq e}.
\]

For the empty expression tuple \(()\), we use the one-point space
\[
    \Omega_{()}=\{\ast\},\qquad
    \Sigma_{()}=\{\emptyset,\{\ast\}\},\qquad
    \mu_{()}(\{\ast\})=1,
\]
and \(\sem{()}_\Gamma(\omega)=\ast\).

We use $\omega$ as a generic symbol for elements of classical parameter
spaces.  Thus $\omega$ may denote an element of $\Omega_\Gamma$ or, locally,
an element of an expression value space such as $\Omega_e$ or
$\Omega_{\seq e}$; its type is determined by the surrounding domain.  In
primitive clauses, after evaluating a parameter tuple $\seq e$ at an outer
meta-parameter instance, we may locally write
\[
    \omega \coloneqq \sem{\seq e}_\Gamma(\omega_{\mathrm{out}})
    \in\Omega_{\seq e},
\]
where $\omega_{\mathrm{out}}\in\Omega_\Gamma$ denotes the original outer
meta-parameter instance. Later, when a binding outcome is displayed separately as $\nu$, the symbol $\omega$
will denote the surrounding valuation, while $(\omega,\nu)$ denotes the
corresponding valuation in the extended binding context.

The syntax of quantum programs is defined relative to a finite classical
context $\Gamma$.  If all classical parameter expressions occurring in a
program are well-typed over $\Gamma$, except for parameters locally bound by
binding constructs, we write $\Gamma\vdash S$.  The grammar is
\begin{equation}
    S ::= \mathbf{skip}
    \mid \mathbf{abort}
    \mid q := 0
    \mid U(\seq{e})[\seq{q}]
    \mid S_1; S_2
    \mid \bind(M(\seq{e}), x. S(x))
    \mid
    \mathbf{while}\ M(\seq e)[\seq{q}] = 1\ \mathbf{do}\ S\ \mathbf{od}.
\end{equation}
The binding construct and its associated measurability assumptions are
specified below in the explanation of syntax.

\subsubsection*{Explanation of Syntax and Measurability Assumptions}

\paragraph*{Primitive-signature convention.}
The language is used schematically, as is customary for quantum while
languages: primitive operation symbols are not generated from a fixed finite
universal gate set.  A primitive unitary symbol may denote any unitary family
satisfying the measurability assumption below, and a primitive measurement
symbol may denote any Kraus-density family satisfying the joint measurability
and normalization assumptions below.  Such a family is treated as one syntactic
primitive, even when its semantic specification contains infinitely or
continuously many operators.

To ensure physical and mathematical validity, we specify the behavior of each
primitive.  Central to our continuous-parameter operations is the following
measurability concept.

\begin{definition}[Strong operator measurability]
    Let $(X,\Sigma_X)$ be a measurable space.  A family
    $\{A_x\}_{x\in X}\subseteq\cB(\cH)$ is \emph{strongly operator measurable}
    if, for every $\u\in\cH$, the map
    \[
        x\longmapsto A_x\u
    \]
    is a measurable function from $(X,\Sigma_X)$ to $\cH$, where $\cH$ carries
    its norm Borel $\sigma$-algebra.  Equivalently, for every $\u\in\cH$ and
    every norm-open set $O\subseteq\cH$,
    \[
        \{x\in X\mid A_x\u\in O\}\in\Sigma_X .
    \]
\end{definition}

With this definition, we impose the following assumptions on the valid
primitives in our syntax.

\begin{itemize}
    \item $\mathbf{skip}$: Performs no operation and leaves the global quantum
    state unchanged.

    \item $\mathbf{abort}$: Represents anomalous termination or
    non-termination, mapping every input state to the zero operator.

    \item $q := 0$: Resets the quantum system $q$ to the ground state
    $|0\rangle$.

    \item $U(\seq{e})[\seq{q}]$: Applies a unitary transformation on the
    quantum variables $\seq q$, parameterized by the classical evaluation of
    the expression tuple $\seq e$.  This allows quantum operations to depend on
    classical parameters, including prior measurement outcomes that have been
    bound by a preceding binding construct.
    
    \textbf{Assumption:}
    Let
    \[
        (\Omega_{\seq e},\Sigma_{\seq e},\mu_{\seq e})
    \]
    be the value space of the expression tuple $\seq e$.  A well-formed
    parameterized unitary primitive is specified by a family
    \[
        \{U_{\omega}\}_{\omega\in\Omega_{\seq e}}\subseteq\cB(\cH).
    \]
    For each parameter value $\omega\in\Omega_{\seq e}$, the operator
    $U_{\omega}$ is unitary.  If the gate is specified only on the subsystem
    $\seq q$, we use the same notation $U_{\omega}$ for its lift to the global
    Hilbert space $\cH$, obtained by tensoring with the identity on all variables
    outside $\seq q$.
    
    The family
    \[
        \omega\longmapsto U_{\omega}
    \]
    is required to be strongly operator measurable over
    $(\Omega_{\seq e},\Sigma_{\seq e})$.  Equivalently, for every $\u\in\cH$,
    the map
    \(
        \omega\longmapsto U_{\omega}\u
    \)
    is $\Sigma_{\seq e}$-measurable as an $\cH$-valued function.

    \item $S_1; S_2$: Sequential composition.

    \item $\bind(M(\seq{e}),x.S(x))$: Performs a measurement using the
    instrument selected by the value of the classical parameter expression tuple
    $\seq e$.  The measurement outcome is bound to the classical parameter
    variable $x$, and the continuation is then executed with access to this
    outcome value. The notation $S(x)$ denotes the outcome-dependent
    continuation of the binding construct.
    
    \textbf{Assumption:}
    The variable $x$ is an already declared classical parameter variable, but it
    is not in the current context $\Gamma$. Thus \(x\) is fresh relative to the current context.
    The binding construct temporarily extends the context to
    \[
        \Gamma,x \coloneqq \Gamma\cup\{x\}.
    \]
    The continuation $S(x)$ is well-formed over this extended context.
    
    Let
    \[
        (\Omega_M,\Sigma_M,\mu_M)
    \]
    be the $\sigma$-finite measure space of measurement outcomes.  Well-formedness
    requires the value space of the bound variable $x$ to agree with the
    measurement outcome space:
    \[
        (\Omega_x,\Sigma_x,\mu_x)
        =
        (\Omega_M,\Sigma_M,\mu_M).
    \]
    Accordingly, we identify
    \[
        \Omega_{\Gamma,x}
        =
        \Omega_\Gamma\times\Omega_M,
        \qquad
        \Sigma_{\Gamma,x}
        =
        \Sigma_\Gamma\otimes\Sigma_M,
        \qquad
        \mu_{\Gamma,x}
        =
        \mu_\Gamma\otimes\mu_M .
    \]
    No completion of the product $\sigma$-algebra is taken.
    
    A meta-parameter instance for the continuation is therefore a pair
    \[
        (\omega,\nu)\in\Omega_\Gamma\times\Omega_M
        =
        \Omega_{\Gamma,x},
    \]
    where $\omega$ is the old outer meta-parameter instance and
    $\nu\in\Omega_M$ is the concrete measurement outcome assigned to $x$.
    In the semantic clauses below, the continuation under this extended
    meta-parameter instance will be indexed by $(\omega,\nu)$.
    Moreover, when the continuation phrase is denoted by $S(x)$, we shall write
    \[
        S_{\omega,\nu}
    \]
    for this continuation interpreted under the extended valuation
    \[
        (\omega,\nu)\in\Omega_{\Gamma,x}
        =
        \Omega_\Gamma\times\Omega_M.
    \]
    Equivalently, $S_{\omega,\nu}$ means $S(x)_{\omega[x:=\nu]}$.  Here
    $\omega$ records the valuation of all classical meta-variables from the
    surrounding context $\Gamma$, while $\nu$ records the value assigned to the
    freshly bound outcome variable $x$.
    
    The result of the binding construct is again a program phrase over the
    original context $\Gamma$; the variable $x$ is local to the continuation.
    Thus the construct does not create a new classical variable and does not
    introduce a mutable classical store.
    
    The notation $S(x)$ should not be understood as a completely arbitrary
    family of programs indexed by the outcome value.  It denotes an
    outcome-indexed continuation subject to a finiteness condition: there exist
    finitely many pairwise disjoint measurable subsets
    \[
        A_1,\ldots,A_m\in\Sigma_M,
        \qquad
        \Omega_M=\bigsqcup_{i=1}^{m}A_i,
    \]
    and finitely many program fragments
    \[
        S_1(x),\ldots,S_m(x),
    \]
    each well-formed over the extended context $\Gamma,x$, such that for every
    concrete outcome value $\nu\in A_i$, the continuation selected by $S(x)$ is
    the fragment $S_i(x)$ interpreted with the bound parameter $x$ assigned the
    value $\nu$.  Thus the syntax remains finite, although each measurable set
    $A_i$ may contain infinitely or uncountably many outcome values.
    
    Let
    \[
        (\Omega_{\seq e},\Sigma_{\seq e},\mu_{\seq e})
    \]
    be the value space of the parameter tuple $\seq e$.  A well-formed
    parameterized instrument is specified by a Kraus-density family
    \[
        \{M_{\omega,\nu}\}_{(\omega,\nu)\in
        \Omega_{\seq e}\times\Omega_M}
        \subseteq \cB(\cH)
    \]
    over the product measurable space
    \[
        (\Omega_{\seq e}\times\Omega_M,
        \Sigma_{\seq e}\otimes\Sigma_M).
    \]
    This family is required to be jointly strongly operator measurable in the
    following pointwise sense: for every $\u\in\cH$, the map
    \[
        (\omega,\nu)
        \longmapsto
        M_{\omega,\nu}\u
    \]
    is measurable from
    \[
        (\Omega_{\seq e}\times\Omega_M,
        \Sigma_{\seq e}\otimes\Sigma_M)
    \]
    to $\cH$.  Equivalently, for every $\u\in\cH$ and every norm-open set
    $O\subseteq\cH$,
    \[
        \{(\omega,\nu)\in\Omega_{\seq e}\times\Omega_M
          \mid M_{\omega,\nu}\u\in O\}
        \in
        \Sigma_{\seq e}\otimes\Sigma_M .
    \]
    This is a pointwise measurability assumption on the displayed raw product
    $\sigma$-algebra, not an almost-everywhere statement on the completion of the
    product measure.
    
    Furthermore, for every fixed parameter value
    $\omega\in\Omega_{\seq e}$, the instrument satisfies the normalization
    condition
    \[
        \int_{\Omega_M}
        \langle \u,M_{\omega,\nu}^{\dagger}M_{\omega,\nu}\u\rangle
        \,d\mu_M(\nu)
        =
        \langle \u,\u\rangle,
        \qquad
        \forall \u\in\cH .
    \]
    
    When the instrument is not parameterized by classical expressions, the
    parameter space $\Omega_{\seq e}$ may be taken to be the one-point space, and
    we write the Kraus family simply as
    \[
        \{M_{\nu}\}_{\nu\in\Omega_M}.
    \]

    \item
    $\mathbf{while}\ M(\seq e)[\seq{q}] = 1\ \mathbf{do}\ S\ \mathbf{od}$:
    Executes the loop body $S$ repeatedly as long as the binary measurement
    $M(\seq e)$ yields the outcome $1$.
    Here $M(\seq e)$ denotes a parameterized binary instrument: the
    expressions $\seq e$ are evaluated under the meta-valuation and
    remain fixed throughout the loop; only the outcome $b\in\{0,1\}$ is
    produced at runtime.

    \textbf{Assumption:}
    For each valuation $\omega$, write
    \[
        M_\omega
        =
        M_{\sem{\seq e}_\Gamma(\omega)}
        =
        \{M_{\omega,0},M_{\omega,1}\}\subseteq\cB(\cH).
    \]
    The loop guard is a discrete binary instrument satisfying
    \[
        M_{\omega,0}^\dagger M_{\omega,0}
        +
        M_{\omega,1}^\dagger M_{\omega,1}
        =
        I
        \qquad
        \text{for every }\omega\in\Omega_\Gamma .
    \]
    Moreover, for $b\in\{0,1\}$, the map
    \[
        \omega\longmapsto M_{\omega,b}
    \]
    is strongly operator measurable.  The while loop does not modify the
    classical context or its valuation; hence, for each fixed valuation
    $\omega$, all iterations use the same binary guard.  When $\omega$ is fixed
    and no ambiguity arises, we write the guard simply as
    $M=\{M_0,M_1\}$.
    The outcome $1$ dictates continuation of the loop, while outcome $0$
    dictates termination.
\end{itemize}

\begin{remark}[Boolean guards as coarse-grained continuous guards]
    The restriction of the core $\mathbf{while}$ syntax to Boolean guards
    should not be read as excluding tests obtained from continuous
    measurements.  It abstracts the situation where the measurement outcome is
    used only through a measurable predicate, such as ``continue'' versus
    ``terminate'', and the concrete outcome is not exposed to the loop body.

    Suppose that a continuous measurement is specified by a Kraus-density
    family
    \(
        \{M_\nu\}_{\nu\in\Omega}
    \)
    over a $\sigma$-finite outcome space $(\Omega,\Sigma,\mu)$, and that the
    loop continues exactly when $\nu\in\Delta$ for some $\Delta\in\Sigma$.
    If the loop body does not depend on the concrete value of $\nu$, then the
    guard may be coarse-grained, at the system level, into the binary CP
    instrument
    \[
        \mathcal G_1(\rho)
        =
        \int_{\Delta} M_\nu\rho M_\nu^\dagger\,d\mu(\nu),
        \qquad
        \mathcal G_0(\rho)
        =
        \int_{\Omega\setminus\Delta}
        M_\nu\rho M_\nu^\dagger\,d\mu(\nu).
    \]
    If the original measurement is normalized, then
    $\mathcal G_0+\mathcal G_1$ is trace-preserving.

    In general, these two branch maps need not be single-Kraus maps on the
    original Hilbert space, and hence should not be identified with two
    operators $M_0,M_1\in\cB(\cH)$.  Thus, under the single-Kraus guard syntax
    used in the core language, such a coarse-grained continuous guard is a
    semantic abstraction rather than necessarily a literal syntactic instance.
    By Ozawa's dilation theorem for quantum instruments
    \cite{ozawa1984quantum}, the resulting two-outcome instrument
    \((\mathcal G_0,\mathcal G_1)\) can be implemented, outside the core syntax, in an enlarged operational model by adjoining a fresh pointer system and measuring a two-outcome pointer observable at each guard evaluation.

    Guards whose loop body depends on the actual value of $\nu$ require an
    outcome-binding recursive construct and are outside the single-Kraus
    Boolean guard fragment considered here.
\end{remark}

\subsection{Semantics via Superoperators}
\label{subsec:semantics}

\paragraph*{Remark on Quantum Variables and Global State}
Although the syntax $S$ involves quantum variables $\bar{q}$ referencing specific subsystems (implying a tensor product structure $\cH = \cH_{\bar{q}} \otimes \cH_{\text{env}}$), our semantic model operates exclusively on the \emph{global} Hilbert space $\cH$.
To rigorously handle unbounded observables, we adopt a \emph{holistic perspective}:
\begin{enumerate}
    \item All semantic maps $\sem{S}$ are defined as superoperators on the global state space $\pardensity{\cH}$.
    \item Any local operation (e.g., a unitary $U$ acting on $\bar{q}$) is interpreted as its \emph{cylindrical extension} to the entire space (i.e., $U \otimes I_{\text{env}}$).
    \item We avoid explicit tensor product decompositions of linear relations in the subsequent logic, as they pose significant domain-theoretic challenges. Treating operations as global bounded maps ensures that the weakest preconditions (pullbacks) are well-defined closed forms.
\end{enumerate}

For each fixed meta-parameter instance \(\omega\in\Omega_\Gamma\), the state
semantics of a well-formed phrase \(\Gamma\vdash S\) is first defined on
\(\cT(\cH)_{+}\). The well-definedness theorem below shows that it extends
uniquely to a CPTNI superoperator
\[
    \sem{S}_\omega:\cT(\cH)\to\cT(\cH).
\]
For classically closed programs, where \(\Omega_\emptyset=\{\ast\}\), we write
\(\sem{S}\) for \(\sem{S}_\ast\).

\subsubsection{CPTNI Maps and Dual Normality}

A fundamental requirement in quantum theory is that valid physical evolutions are
Completely Positive and Trace-Non-Increasing (CPTNI).  In the present paper the
Schrödinger picture is primary: program denotations act on trace-class operators,
whereas normality is assigned to the corresponding Heisenberg-picture dual map.

\begin{definition}[Superoperator and CPTNI Map]
\label{def:normal_superoperator}
    A \emph{superoperator} on $\cH$ is a bounded linear map
    \[
        \cE:\cT(\cH)\longrightarrow \cT(\cH).
    \]
    It is called:
    \begin{enumerate}[(1)]
        \item \textbf{Completely Positive (CP)} if, for every $n\in\mathbb{N}$,
        the amplification
        \[
            \cE\otimes \mathcal{I}_n:
            \cT(\cH\otimes \mathbb{C}^n)\longrightarrow
            \cT(\cH\otimes \mathbb{C}^n)
        \]
        maps the positive cone of $\cT(\cH\otimes \mathbb{C}^n)$ into itself.

        \item \textbf{Trace-Non-Increasing (TNI)} if $\cE$ is positive and
        \[
            \tr(\cE(A))\leq \tr(A),
            \qquad \forall A\in \cT(\cH)_+ .
        \]
    \end{enumerate}
    A \emph{CPTNI map} is a superoperator that is both CP and TNI.  When a
    CPTNI map is applied to partial density operators, we always mean its
    restriction to
    \[
        \pardensity{\cH}=\{\rho\in\cT(\cH)\mid \rho\sqgeq 0,\ \tr(\rho)\leq 1\}.
    \]
\end{definition}

\begin{remark}
    The condition $\cE(\pardensity{\cH})\subseteq \pardensity{\cH}$ is a
    consequence of CP and TNI: complete positivity implies positivity, and TNI
    bounds the trace.  Thus it should not be built into the word
    ``superoperator'' itself.
\end{remark}

\begin{proposition}[Automatic \(\omega\)-order continuity on the trace class]
\label{prop:cptni_trace_class_normality}
    Let $\cE:(\cT(\cH),\|\cdot\|_1)\to(\cT(\cH),\|\cdot\|_1)$ be a bounded positive linear map.  If
    $\{\rho_n\}_{n\in\bN}\subseteq \cT(\cH)_+$ is an increasing sequence, i.e., $\rho_n\sqleq \rho_{n+1}$,
    with supremum $\rho=\sup_n\rho_n\in\cT(\cH)_+$ in the L\"owner order, then
    \[
        \cE(\rho)=\sup_{n}\cE(\rho_n).
    \]
    Moreover,
    \[
        \lim_{n\to\infty}\|\cE(\rho)-\cE(\rho_n)\|_1=0 .
    \]
    In particular, every CPTNI map is \(\omega\)-order continuous on $\cT(\cH)_+$ and hence on $\pardensity{\cH}$.
\end{proposition}

\begin{proof}
    If $\rho=0$, then $\rho_n=0$ for all $n$, and the claim is trivial.  Assume $\tr(\rho)>0$.      
    Since scalar multiplication by the positive number $\tr(\rho)$ is an
    order isomorphism, the sequence
    $\{\rho_n/\tr(\rho)\}_n$ is increasing in $\pardensity{\cH}$ and has
    supremum $\rho/\tr(\rho)$.  By
    Proposition~\ref{prop:state_cpo_topology},
    \[
        \lim_{n\to\infty}
        \left\|
            \frac{\rho}{\tr(\rho)}
            -
            \frac{\rho_n}{\tr(\rho)}
        \right\|_1
        =0.
    \]
    Multiplying by $\tr(\rho)$ gives
    \[
        \lim_{n\to\infty}\|\rho-\rho_n\|_1=0.
    \]
    Since $\cE$ is bounded,
    \[
        \|\cE(\rho)-\cE(\rho_n)\|_1
        \leq \|\cE\|\,\|\rho-\rho_n\|_1
        \xrightarrow[n\to\infty]{}0 .
    \]
    Positivity of $\cE$ gives
    \[
        0\sqleq \cE(\rho_n)\sqleq \cE(\rho_{n+1})\sqleq \cE(\rho),\,\forall n\in\bN.
    \]
    Thus $\{\cE(\rho_n)\}_n$ is an increasing sequence bounded above by
    $\cE(\rho)$.  If $\sigma\in\cT(\cH)_+$ is any other upper bound, then
    \[
        \sigma-\cE(\rho_n)\sqgeq 0,
        \qquad \forall n\in\bN .
    \]
    Since $\cE(\rho_n)\to\cE(\rho)$ in trace norm, we have
    \[
        \sigma-\cE(\rho_n)
        \xrightarrow{\|\cdot\|_1}
        \sigma-\cE(\rho).
    \]
    The positive cone $\cT(\cH)_+$ is closed in trace norm: indeed, if
    $X_k\in\cT(\cH)_+$ and $X_k\to X$ in trace norm, then
    $\|X_k-X\|_\infty\leq \|X_k-X\|_1$, so $X_k\to X$ in operator norm.
    Hence, for every $\u\in\cH$,
    \[
        \langle\u,X\u\rangle
        =
        \lim_{k\to\infty}\langle\u,X_k\u\rangle
        \geq 0,
    \]
    and therefore $X\in\cT(\cH)_+$.  Applying this to
    $X_k=\sigma-\cE(\rho_k)$ yields
    \[
        \sigma-\cE(\rho)\sqgeq 0.
    \]
    Hence $\cE(\rho)\sqleq\sigma$.  Since $\sigma$ was arbitrary,
    $\cE(\rho)$ is the least upper bound of $\{\cE(\rho_n)\}_n$.
\end{proof}

\begin{definition}[Normality in the Heisenberg picture]
\label{def:normal_heisenberg_map}
    Let $\Phi:\cB(\cH)\to\cB(\cH)$ be a positive linear map.  We say that
    $\Phi$ is \emph{normal} if for every bounded increasing net
    $\{A_\alpha\}_\alpha\subseteq\cB(\cH)_+$ with supremum
    $A=\sup_\alpha A_\alpha$ in $\cB(\cH)_+$, one has
    \[
        \Phi(A)=\sup_\alpha \Phi(A_\alpha).
    \]
\end{definition}

\begin{remark}
    Equivalently, a positive map between von Neumann algebras is normal iff it
    is ultraweakly continuous, or iff it admits a predual map.  In the special
    case of $\cB(\cH)$, the ultraweak topology is the weak-* topology induced
    by the predual $\cT(\cH)$. See \cite[Chapter 7]{kadison1997fundamentals} for more details.
\end{remark}

\begin{proposition}[Schr\"odinger--Heisenberg duality]
\label{prop:schrodinger_heisenberg_dual}
    Let $\cE:\cT(\cH)\to\cT(\cH)$ be a bounded linear map.  There exists a
    unique bounded linear map
    \[
        \cE^\dagger:\cB(\cH)\longrightarrow\cB(\cH)
    \]
    such that
    \[
        \tr(\cE(\rho)A)=\tr(\rho\,\cE^\dagger(A)),
        \qquad
        \forall \rho\in\cT(\cH),\ \forall A\in\cB(\cH).
    \]
    Moreover, $\cE$ is positive (respectively, completely positive) iff
    $\cE^\dagger$ is positive (respectively, completely positive).  
    If $\cE$ is positive, then
    \[
        \cE \text{ is TNI}
        \quad\Longleftrightarrow\quad
        \cE^\dagger(I)\sqleq I\, (\text{subunital}).
    \]
    In particular, if $\cE$ is CPTNI, then $\cE^\dagger$ is a \textbf{normal} CP subunital map on $\cB(\cH)$.
\end{proposition}

\begin{proof}
    The existence and uniqueness of $\cE^\dagger$, as well as the equivalence
    of positivity and complete positivity for $\cE$ and $\cE^\dagger$, follow
    from the trace-class duality $\cT(\cH)^*\cong\cB(\cH)$; see
    \cite[Thm.~1.55]{WolfQIPNotes}.

    It remains to record the trace-nonincreasing condition.  For
    $\rho\in\cT(\cH)_+$, the trace-pairing identity gives
    \[
        \tr(\rho)-\tr(\cE(\rho))
        =
        \tr\bigl(\rho(I-\cE^\dagger(I))\bigr).
    \]
    Hence $\tr(\cE(\rho))\leq\tr(\rho)$ for all
    $\rho\in\cT(\cH)_+$ iff
    \[
        \tr\bigl(\rho(I-\cE^\dagger(I))\bigr)\geq 0,
        \qquad
        \forall \rho\in\cT(\cH)_+ .
    \]
    Since positivity of a bounded self-adjoint operator can be checked on
    rank-one projections,
    \[
    X\sqgeq 0
    \quad\Longleftrightarrow\quad
    \tr(\sP_{\u}X)=\langle\u,X\u\rangle\geq 0
    \quad \forall\u\in\cH,
    \]
    this is equivalent to $I-\cE^\dagger(I)\sqgeq 0$, i.e., to
    $\cE^\dagger(I)\sqleq I$.

    Finally, assume that $\cE$ is CPTNI.  We prove normality of
    $\cE^\dagger$ in the sense of Definition~\ref{def:normal_heisenberg_map}.
    Let $\{A_\alpha\}_\alpha\subseteq\cB(\cH)_+$ be a bounded increasing net with supremum $A$.  
    Then $\{\cE^\dagger(A_\alpha)\}_\alpha$ is an increasing net in
    $\cB(\cH)_+$.  Moreover, since $A_\alpha\sqleq A$ and $\cE^\dagger$ is positive,
    \[
        0\sqleq \cE^\dagger(A_\alpha)\sqleq \cE^\dagger(A)
        \sqleq \|\cE^\dagger(A)\|\,I .
    \]
    By the monotone convergence theorem for bounded self-adjoint operators
    \cite[Lemma 5.1.4]{kadison1997fundamentals}, the supremum
    \[
        B\coloneqq \sup_\alpha \cE^\dagger(A_\alpha)
    \]
    exists in $\cB(\cH)_+$.
    We now show that $B=\cE^\dagger(A)$.  Fix
    $\rho\in\cT(\cH)_+$.  We first use the following elementary fact: if
    $\tau\in\cT(\cH)_+$ and $\{C_\alpha\}_\alpha\subseteq\cB(\cH)_+$ is a
    bounded increasing net with supremum $C$, then
    \[
        \tr(\tau C)=\sup_\alpha \tr(\tau C_\alpha).
    \]
    Indeed, writing the spectral decomposition
    \[
        \tau=\sum_{k=1}^{\infty}\lambda_k\sP_{\u_k},
        \qquad
        \lambda_k\geq 0,\quad \sum_k\lambda_k=\tr(\tau)<\infty,
    \]
    we have, for each $k$,
    \[
        \langle \u_k,C_\alpha \u_k\rangle
        \uparrow
        \langle \u_k,C \u_k\rangle .
    \]
    Hence, by monotone convergence for non-negative scalar nets and series,
    \[
    \begin{aligned}
        \tr(\tau C)
        &=
        \sum_{k=1}^{\infty}
        \lambda_k\langle \u_k,C \u_k\rangle        \\
        &=
        \sup_\alpha
        \sum_{k=1}^{\infty}
        \lambda_k\langle \u_k,C_\alpha \u_k\rangle \\
        &=
        \sup_\alpha \tr(\tau C_\alpha).
    \end{aligned}
    \]

    Applying this fact first to
    $\tau=\rho$, $C_\alpha=\cE^\dagger(A_\alpha)$ and $C=B$, and then to
    $\tau=\cE(\rho)$, $C_\alpha=A_\alpha$ and $C=A$, we obtain
    \[
    \begin{aligned}
        \tr(\rho B)
        &=
        \sup_\alpha
        \tr\bigl(\rho\,\cE^\dagger(A_\alpha)\bigr)  \\
        &=
        \sup_\alpha
        \tr\bigl(\cE(\rho)A_\alpha\bigr) \\
        &=
        \tr\bigl(\cE(\rho)A\bigr) \\
        &=
        \tr\bigl(\rho\,\cE^\dagger(A)\bigr).
    \end{aligned}
    \]
    Taking $\rho=\sP_{\u}$ gives
    \[
        \langle \u,B\u\rangle
        =
        \langle \u,\cE^\dagger(A)\u\rangle,
        \qquad \forall \u\in\cH .
    \]
    Since both $B$ and $\cE^\dagger(A)$ are self-adjoint, the polarization
    identity implies
    \[
        B=\cE^\dagger(A).
    \]
    Therefore
    \[
        \cE^\dagger(A)
        =
        \sup_\alpha \cE^\dagger(A_\alpha),
    \]
    so $\cE^\dagger$ is normal.
\end{proof}

We denote by $\cQ(\cH)$ the set of all CPTNI superoperators
$\cE:\cT(\cH)\to\cT(\cH)$.  By
Proposition~\ref{prop:cptni_trace_class_normality}, every element of
$\cQ(\cH)$ is automatically order-normal on the trace-class positive cone.
By Proposition~\ref{prop:schrodinger_heisenberg_dual}, its Heisenberg dual
$\cE^\dagger:\cB(\cH)\to\cB(\cH)$ is a normal CP subunital map. 
Now we establish two key structural properties of $\cQ(\cH)$: its representation via operators and its domain-theoretic completeness.

\begin{theorem}[Countable Kraus representation, \cite{kraus1971general,holevo2001statistical,WolfQIPNotes}]
\label{thm:kraus_rep}
    If $\cE:\cT(\cH)\to\cT(\cH)$ is CPTNI, then there are
    $\{K_j\}_{j\in J}\subseteq\cB(\cH)$, with $J$ \textbf{finite or countably infinite},
    such that
    \[
        \sum_{j\in J}K_j^\dagger K_j \sqleq I
        \quad\text{in the strong operator topology,}
    \]
    and
    \[
        \cE(\rho)=\sum_{j\in J}K_j\rho K_j^\dagger
        \quad(\rho\in\cT(\cH)),
        \qquad
        \cE^\dagger(A)=\sum_{j\in J}K_j^\dagger A K_j
        \quad(A\in\cB(\cH)).
    \]
    The first series converges in trace norm, and the second converges
    ultraweakly.\footnote{That is, after enumerating $J$ as a finite set or as
    $\mathbb{N}$, for every $A\in\cB(\cH)$ and every $\rho\in\cT(\cH)$,
    \[
        \tr\!\left(
            \rho\sum_{j=1}^{N}K_j^\dagger A K_j
        \right)
        \quad\text{converges as }N\to\infty,
    \]
    with limit $\tr\!\left(\rho\,\cE^\dagger(A)\right)$.}
\end{theorem}

\begin{proof}[Proof sketch]
    By Proposition~\ref{prop:schrodinger_heisenberg_dual}, the Heisenberg dual
    $\cE^\dagger:\cB(\cH)\to\cB(\cH)$ is a normal CP map satisfying
    $\cE^\dagger(I)\sqleq I$.  
    The Kraus representation theorem (see for example \cite{kraus1971general,holevo2001statistical}) for normal CP maps on $\cB(\cH)$ gives a finite or countable family
    $\{K_j\}_{j\in J}\subseteq\cB(\cH)$ such that
    \[
        \cE^\dagger(A)=\sum_{j\in J}K_j^\dagger A K_j
    \]
    ultraweakly for every $A\in\cB(\cH)$.  Taking $A=I$ gives
    \[
        \sum_{j\in J}K_j^\dagger K_j
        =
        \cE^\dagger(I)
        \sqleq I
    \]
    in the strong operator topology.

    It remains to identify the predual action. 
    For any $\rho\in\cT(\cH)_+$, let $S_N(\rho)=\sum_{j=1}^{N}K_j\rho K_j^\dagger$.  
    Since $S_N(\rho)$ is increasing and
    \[
        \tr(S_N(\rho))
        =
        \tr\!\left(\rho\sum_{j=1}^{N}K_j^\dagger K_j\right)
        \leq \tr(\rho),
    \]
    it is trace-norm Cauchy, hence converges in $\cT(\cH)_+$ to some
    $S(\rho)$.  For every $A\in\cB(\cH)$,
    \[
        \tr(S(\rho)A)
        =
        \lim_N
        \tr\!\left(\rho\sum_{j=1}^{N}K_j^\dagger A K_j\right)
        =
        \tr(\rho\,\cE^\dagger(A))
        =
        \tr(\cE(\rho)A).
    \]
    Thus $S(\rho)=\cE(\rho)$ by trace duality.  The general case follows by
    linearity from the decomposition of an arbitrary trace-class operator into
    four positive ones.
\end{proof}

\begin{remark}
    Conversely, any family $\{K_j\}_{j\in J}$ satisfying
    $\sum_j K_j^\dagger K_j\sqleq I$ and for which
    $\sum_j K_j\rho K_j^\dagger$ converges in trace norm defines a CPTNI map.
    This converse direction is standard and will not be used in the sequel.
\end{remark}

For loop semantics, we need closure of CPTNI maps under increasing countable suprema.
\begin{lemma}[Closure of CPTNI Maps under Increasing Countable Suprema]
\label{lem:cpo_domain}
    We equip $\cQ(\cH)$ with the \emph{pointwise L\"owner order} $\sqsubseteq$:
    \[
        \cE \sqsubseteq \mathcal{F} \iff \cE(\rho) \sqsubseteq \mathcal{F}(\rho), \quad \forall \rho \in \pardensity{\cH}.
    \]
    The poset $(\cQ(\cH), \sqsubseteq)$ forms an $\omega$-Complete Partial Order ($\omega$-CPO) with a least element $\bot$ (the zero map). Specifically, for any increasing sequence $\cE_0 \sqsubseteq \cE_1 \sqsubseteq \dots$, the pointwise supremum $\cE = \sup_n \cE_n$ is a well-defined CPTNI map.
    Moreover, for any state $\rho \in \pardensity{\cH}$, the sequence $\cE_n(\rho)$ converges to $\cE(\rho)$ in the trace norm topology.
\end{lemma}

\begin{proof}
    The zero map $\bot(A)=0$ is CPTNI and is clearly the least element.

    We first note that $\sqsubseteq$ is indeed a partial order on $\cQ(\cH)$.
    If $\cE\sqsubseteq\mathcal F$ and $\mathcal F\sqsubseteq\cE$, then
    $\cE(\rho)=\mathcal F(\rho)$ for all $\rho\in\pardensity{\cH}$.  For every
    $A\in\cT(\cH)_+$ with $A\neq 0$,
    \[
        \frac{A}{\tr(A)}\in\pardensity{\cH},
    \]
    hence $\cE(A)=\mathcal F(A)$ by linearity.  Since every trace-class
    operator is a linear combination of four positive trace-class operators, we
    obtain $\cE=\mathcal F$ on $\cT(\cH)$.

    Let
    \[
        \cE_0\sqsubseteq \cE_1\sqsubseteq \cdots
    \]
    be an increasing sequence in $\cQ(\cH)$.  For each \(\rho\in\pardensity{\cH}\), the chain assumption
    \(\cE_n\sqsubseteq \cE_{n+1}\) gives
    $\cE_n(\rho)\sqleq \cE_{n+1}(\rho)$.
    Moreover, since each \(\cE_n\) is positive and TNI, each
    \(\cE_n(\rho)\) belongs to \(\pardensity{\cH}\).  By
    Proposition~\ref{prop:state_cpo_topology}, the supremum exists in
    $\pardensity{\cH}$ and
    \[
        \cE_n(\rho)\xrightarrow{\|\cdot\|_1}
        \sup_n \cE_n(\rho).
    \]
    Define
    \[
        \cE(\rho)\coloneqq \sup_n\cE_n(\rho),
        \qquad \rho\in\pardensity{\cH}.
    \]

    We extend $\cE$ to $\cT(\cH)$.  For $A\in\cT(\cH)_+$, set
    $\cE(0)=0$, and if $A\neq 0$, define
    \[
        \cE(A)
        \coloneqq
        \tr(A)\,
        \cE\!\left(\frac{A}{\tr(A)}\right).
    \]
    Then $\cE_n(A)\to\cE(A)$ in trace norm.

    For an arbitrary $X\in\cT(\cH)$, choose a decomposition
    \[
        X=X_1-X_2+iX_3-iX_4,
        \qquad X_j\in\cT(\cH)_+ .
    \]
    Since each sequence $\{\cE_n(X_j)\}_n$ converges in trace norm, the sequence
    $\{\cE_n(X)\}_n$ also converges in trace norm.  We define
    \[
        \cE(X)\coloneqq \lim_{n\to\infty}\cE_n(X).
    \]
    This definition is independent of the chosen decomposition.  Linearity of
    $\cE$ follows by passing to the trace-norm limit in
    \[
        \cE_n(\alpha X+\beta Y)
        =
        \alpha\cE_n(X)+\beta\cE_n(Y).
    \]

        We next prove complete positivity.  Fix a finite-dimensional Hilbert space
    $\cK$ with $\dim\cK=k$, and let
    \[
        X\in\cT(\cH\otimes\cK)_+ .
    \]
    Choose an orthonormal basis $\{\ket{1},\ldots,\ket{k}\}$ of $\cK$.  For
    $1\leq i,j\leq k$, set
    \[
        L_i\coloneqq I_{\cH}\otimes\bra{i}:
        \cH\otimes\cK\to\cH,
        \qquad
        R_j\coloneqq I_{\cH}\otimes\ket{j}:
        \cH\to\cH\otimes\cK,
    \]
    and define the $(i,j)$-block of $X$ by
    \[
        X_{ij}\coloneqq L_iXR_j .
    \]
    Since $L_i$ and $R_j$ are bounded and trace-class operators form a two-sided
    ideal in bounded operators, each $X_{ij}$ belongs to $\cT(\cH)$.  Moreover,
    with respect to the chosen basis of $\cK$,
    \[
        X=\sum_{i,j=1}^{k}X_{ij}\otimes\ket{i}\bra{j}.
    \]

    By the construction of $\cE$ above, for every trace-class operator
    $A\in\cT(\cH)$ we have
    \[
        \cE_n(A)\xrightarrow{\|\cdot\|_1}\cE(A).
    \]
    Applying this to each block $X_{ij}$ gives
    \[
        \cE_n(X_{ij})
        \xrightarrow{\|\cdot\|_1}
        \cE(X_{ij}),
        \qquad 1\leq i,j\leq k .
    \]
    Now put
    \[
        Y_n\coloneqq(\cE_n\otimes\mathcal I_{\cK})(X),
        \qquad
        Y\coloneqq(\cE\otimes\mathcal I_{\cK})(X).
    \]
    Then
    \[
        Y_n-Y
        =
        \sum_{i,j=1}^{k}
        \bigl(\cE_n(X_{ij})-\cE(X_{ij})\bigr)
        \otimes\ket{i}\bra{j}.
    \]
    Hence, using the triangle inequality and
    $\|A\otimes B\|_1=\|A\|_1\|B\|_1$ for trace-class operators,
    \[
    \begin{aligned}
        \|Y_n-Y\|_1
        &\leq
        \sum_{i,j=1}^{k}
        \left\|
            \bigl(\cE_n(X_{ij})-\cE(X_{ij})\bigr)
            \otimes\ket{i}\bra{j}
        \right\|_1                                      \\
        &=
        \sum_{i,j=1}^{k}
        \|\cE_n(X_{ij})-\cE(X_{ij})\|_1\,
        \|\ket{i}\bra{j}\|_1                            \\
        &=
        \sum_{i,j=1}^{k}
        \|\cE_n(X_{ij})-\cE(X_{ij})\|_1
        \xrightarrow[n\to\infty]{}0 .
    \end{aligned}
    \]
    Therefore
    \[
        (\cE_n\otimes\mathcal I_{\cK})(X)
        \xrightarrow{\|\cdot\|_1}
        (\cE\otimes\mathcal I_{\cK})(X).
    \]

    Since each $\cE_n$ is completely positive,
    \[
        (\cE_n\otimes\mathcal I_{\cK})(X)\sqgeq 0
        \qquad\forall n.
    \]
    The positive cone of $\cT(\cH\otimes\cK)$ is closed in trace norm: indeed,
    trace-norm convergence implies operator-norm convergence, and positivity is
    closed in operator norm.  Hence the trace-norm limit is positive:
    \[
        (\cE\otimes\mathcal I_{\cK})(X)\sqgeq 0.
    \]
    Since $\cK$ was arbitrary finite-dimensional, $\cE$ is completely positive.

    We now prove trace non-increase.  Let $A\in\cT(\cH)_+$.  Since
    $\cE_n(A)\to\cE(A)$ in trace norm and the trace is trace-norm continuous,
    \[
        \tr(\cE(A))
        =
        \lim_{n\to\infty}\tr(\cE_n(A))
        \leq
        \tr(A).
    \]
    Hence $\cE$ is TNI.

    Finally, $\cE$ is bounded.  If $X=X^\dagger$ and
    $X=X_+-X_-$ is its Jordan decomposition, then CP implies positivity, and
    TNI gives
    \[
        \|\cE(X)\|_1
        \leq
        \|\cE(X_+)\|_1+\|\cE(X_-)\|_1
        =
        \tr(\cE(X_+))+\tr(\cE(X_-))
        \leq
        \tr(X_+)+\tr(X_-)
        =
        \|X\|_1 .
    \]
    For general $X$, write
    \[
        X=\operatorname{Re}X+i\,\operatorname{Im}X.
    \]
    Since
    \[
        \|\operatorname{Re}X\|_1\leq\|X\|_1,
        \qquad
        \|\operatorname{Im}X\|_1\leq\|X\|_1,
    \]
    we obtain
    \[
        \|\cE(X)\|_1\leq 2\|X\|_1.
    \]
    Thus $\cE:\cT(\cH)\to\cT(\cH)$ is a bounded linear map.

    Therefore $\cE\in\cQ(\cH)$.  Moreover, for every
    $\rho\in\pardensity{\cH}$,
    \[
        \cE_n(\rho)\sqleq \cE(\rho),
    \]
    so $\cE$ is an upper bound of the chain.  If $\mathcal F\in\cQ(\cH)$ is any
    other upper bound, then for every $\rho\in\pardensity{\cH}$,
    \[
        \cE_n(\rho)\sqleq \mathcal F(\rho)
        \quad\forall n,
    \]
    hence
    \[
        \cE(\rho)=\sup_n\cE_n(\rho)\sqleq \mathcal F(\rho).
    \]
    Thus $\cE=\sup_n\cE_n$ in $(\cQ(\cH),\sqsubseteq)$.  The trace-norm
    convergence on $\pardensity{\cH}$ was established above, so
    $(\cQ(\cH),\sqsubseteq)$ is an $\omega$-CPO.
\end{proof}

\begin{remark}[\(\omega\)-CPO Structure for the Complete-Positivity Order]
\label{rem:cp_order_cpo}
    The semantic domain $\cQ(\cH)$ may also be equipped with the stronger
    complete-positivity order
    \[
        \cE\sqleq_{\mathrm{cp}}\mathcal{F}
        \iff
        \mathcal{F}-\cE \text{ is completely positive}.
    \]
    Equivalently,
    \[
        (\cE\otimes\mathcal{I}_k)(X)
        \sqleq
        (\mathcal{F}\otimes\mathcal{I}_k)(X)
    \]
    for every $k\in\mathbb{N}$ and every
    $X\in\cT(\cH\otimes\mathbb{C}^k)_+$.  The proof of
    Lemma~\ref{lem:cpo_domain} shows that trace-norm limits preserve complete positivity through finite block amplifications.  
    Hence every $\sqleq_{\mathrm{cp}}$-increasing sequence has the same pointwise supremum
    as in the pointwise L\"owner-order construction, and this supremum again belongs to $\cQ(\cH)$. 
    Consequently, $(\cQ(\cH), \sqsubseteq_{\mathrm{cp}})$ is also a valid $\omega$-CPO sharing the exact same supremum structure and least element $\bot$.
    This stronger order is recorded only for orientation; the loop semantics below uses the pointwise L\"owner order.
\end{remark}

\subsubsection{Structural Semantics}

Since continuous measurements require integration of trace-class-valued branch
states over measurement outcome spaces, we use the Bochner integral framework.
We distinguish $\Sigma$-measurability from the almost-everywhere
measurability used in Bochner integration.
Under our standing assumption that $\cH$ is separable, the trace-class space
$\cT(\cH)$ is a separable Banach space under the trace norm.  Indeed,
finite-rank operators whose matrix coefficients, relative to a fixed countable orthonormal basis, lie in $\bQ+\ii\bQ$ form a countable
trace-norm dense subset.

\begin{definition}[$\Sigma$-simple and $\mu$-simple functions]
\label{def:simple_functions}
Let $(\Omega,\Sigma,\mu)$ be a $\sigma$-finite measure space.  A function
$s:\Omega\to\cT(\cH)$ is called $\Sigma$-simple if it has the form
\[
    s(\omega)=\sum_{i=1}^k \rho_i\chi_{E_i}(\omega),
\]
where $k\in\mathbb N$, $\rho_i\in\cT(\cH)$, and
$E_i\in\Sigma$ are measurable.

It is called $\mu$-simple if it admits such a representation with
\(\mu(E_i)<\infty\) for every $i$ .
\end{definition}

\begin{definition}[Bochner measurability]
\label{def:bochner_measurability}
Let $(\Omega,\Sigma,\mu)$ be a $\sigma$-finite measure space.  A function
\[
    f:\Omega\to\cT(\cH)
\]
is \emph{Bochner measurable}, or \emph{strongly $\mu$-measurable}, if there
exists a sequence of $\mu$-simple functions $f_n$ such that
\[
    \|f_n(\omega)-f(\omega)\|_1\to 0
\]
for $\mu$-almost every $\omega\in\Omega$.
\end{definition}

We recall a standard completion-based characterization of Bochner
measurability. We denote by $\Sigma^*$ the \emph{Lebesgue extension} (or \emph{completion}) of $\Sigma$, which is the smallest $\sigma$-algebra containing $\Sigma$ and all subsets of $\mu$-null sets.

\begin{proposition}[Completion criterion for Bochner measurability, {\cite[Corollary 1.1.10, Proposition 1.1.16]{hytonen2016analysis};
see also \cite{dunford1958linear}}]
\label{prop:bochner_equivalent_def}
Let $(\Omega,\Sigma,\mu)$ be a
$\sigma$-finite measure space.  A function
\(
    f:\Omega\to\cT(\cH)
\)
is Bochner measurable if and only if it is
\((\Sigma^*,\operatorname{Borel}(\cT(\cH)))\)-measurable\footnote{i.e., the preimage of every open set $U \subseteq \mathcal{T}(\mathcal{H})$ (under the trace-norm topology) is $\Sigma^*$-measurable.}, where $\cT(\cH)$ carries the trace-norm topology.

In particular, since $\Sigma\subseteq\Sigma^*$, every
\((\Sigma,\operatorname{Borel}(\cT(\cH)))\)-measurable function is Bochner measurable.
\end{proposition}

Direct verification of trace-norm Borel measurability is sometimes
inconvenient.  For trace-class-valued maps, the following scalar criterion is
often easier to apply.

\begin{definition}[Weak $\Sigma$-measurability]
\label{def:weak_measurability}
Let $(\Omega,\Sigma,\mu)$ be a $\sigma$-finite measure space.  By the canonical duality
\(\cT(\cH)^*\cong\cB(\cH)\),
a map
\(
    f:\Omega\to\cT(\cH)
\)
is called \emph{weakly $\Sigma$-measurable} if, for every bounded operator
$B\in\cB(\cH)$, the scalar function
\[
    \omega\longmapsto \tr(Bf(\omega))
\]
is $(\Sigma,\operatorname{Borel}(\bC))$-measurable.
\end{definition}

\begin{theorem}[Weak-to-strong $\Sigma$-measurability for trace-class maps, {\cite[Corollary 1.1.2, Corollary 1.1.10]{hytonen2016analysis}}]
\label{thm:pettis}
Let $(\Omega,\Sigma)$ be a measurable space.
Since $\cT(\cH)$ is separable in the trace norm,
if
\(
    f:\Omega\to\cT(\cH)
\)
is weakly $\Sigma$-measurable, then \(f\) is
\((\Sigma,\operatorname{Borel}(\cT(\cH)))\)-measurable, where \(\cT(\cH)\) carries the trace-norm topology. In addition, if \((\Omega,\Sigma,\mu)\) is a a $\sigma$-finite measure space, by \Cref{prop:bochner_equivalent_def} \(f\) is Bochner measurable.
\end{theorem}

\begin{proposition}[Bochner integrability criterion, {\cite[Proposition 1.2.2]{hytonen2016analysis}}]
\label{prop:bochner_integrability_criterion}
Let $(\Omega,\Sigma,\mu)$ be a $\sigma$-finite measure space.  If
\(f:\Omega\to\cT(\cH)\)
is Bochner measurable and
\[
    \int_\Omega \|f(\omega)\|_1\,d\mu(\omega)<\infty,
\]
then \(f\) is Bochner integrable.
\end{proposition}

Equipped with these measurability preliminaries, we now state the
well-definedness and structural properties of the denotational semantics.

The structural clauses are first written on the positive trace-class cone
$\cT(\cH)_+$.  This is intentional: measurement branches and loop approximants
are positive, and the loop semantics is constructed using increasing suprema
of positive operators.  The theorem below asserts that these positive-cone
clauses are not merely state transformers; for each fixed meta-parameter
instance, they extend uniquely to CPTNI superoperators on the full trace-class
space.

We shall use the following elementary reduction from positive-cone
measurability to full trace-class measurability.

\begin{lemma}[Positive-cone extension and measurability reduction]
\label{lem:positive_cone_reduction}
Let $(\Omega,\Sigma)$ be a measurable space.  Suppose that, for every
$\omega\in\Omega$, a map
\[
    \Phi^+_\omega:\cT(\cH)_+\longrightarrow\cT(\cH)_+
\]
is cone-linear, i.e.,
\[
    \Phi^+_\omega(\rho+\sigma)
    =
    \Phi^+_\omega(\rho)+\Phi^+_\omega(\sigma),
    \qquad
    \Phi^+_\omega(a\rho)=a\Phi^+_\omega(\rho)
\]
for all $\rho,\sigma\in\cT(\cH)_+$ and all $a\geq 0$.  Assume also that it is
trace-non-increasing on the positive cone:
\[
    \tr(\Phi^+_\omega(\rho))\leq \tr(\rho),
    \qquad
    \forall \rho\in\cT(\cH)_+ .
\]
Then $\Phi^+_\omega$ extends uniquely to a bounded positive linear map
\[
    \Phi_\omega:\cT(\cH)\longrightarrow\cT(\cH)
\]
whose restriction to $\cT(\cH)_+$ is $\Phi^+_\omega$, and this extension is
trace-non-increasing.

If, moreover, for every $k\in\mathbb N$ the amplification
\[
    \Phi_\omega\otimes\mathcal I_k:
    \cT(\cH\otimes\mathbb C^k)
    \longrightarrow
    \cT(\cH\otimes\mathbb C^k)
\]
maps $\cT(\cH\otimes\mathbb C^k)_+$ into itself, then
\[
    \Phi_\omega\in\cQ(\cH).
\]

Finally, if for every $\rho\in\cT(\cH)_+$ the map
$\omega\mapsto\Phi^+_\omega(\rho)$ is
$(\Sigma,\operatorname{Borel}(\cT(\cH)_+))$-measurable,
then for every $X\in\cT(\cH)$ the map
$\omega\mapsto\Phi_\omega(X)$ is
$(\Sigma,\operatorname{Borel}(\cT(\cH)))$-measurable.
Here $\cT(\cH)$ carries the trace-norm topology, and $\cT(\cH)_+$ carries the corresponding subspace topology.
\end{lemma}
\begin{proof}
Fix $\omega\in\Omega$.  We first construct the extension.  For a self-adjoint
$A\in\cT(\cH)$, choose its Jordan decomposition
\[
    A=A_+-A_-,
    \qquad
    A_+,A_-\in\cT(\cH)_+,
    \qquad
    A_+A_-=0,
\]
and set
\[
    \Phi_\omega(A)
    \coloneqq
    \Phi^+_\omega(A_+)-\Phi^+_\omega(A_-).
\]
This definition is independent of the chosen positive decomposition.  Indeed,
if \(A=P-Q=P'-Q'\) with \(P,Q,P',Q'\in\cT(\cH)_+\), then
\(P+Q'=P'+Q\).  By cone-linearity,
\[
    \Phi^+_\omega(P)+\Phi^+_\omega(Q')
    =
    \Phi^+_\omega(P')+\Phi^+_\omega(Q),
\]
hence
\[
    \Phi^+_\omega(P)-\Phi^+_\omega(Q)
    =
    \Phi^+_\omega(P')-\Phi^+_\omega(Q').
\]
Thus \(\Phi_\omega\) is a well-defined real-linear map on the self-adjoint
trace-class operators.

For a general \(X\in\cT(\cH)\), define
\[
    \Phi_\omega(X)
    \coloneqq
    \Phi_\omega(\operatorname{Re}X)
    +
    i\,\Phi_\omega(\operatorname{Im}X).
\]
This gives a complex-linear extension of \(\Phi^+_\omega\) to all of
\(\cT(\cH)\).  Uniqueness follows because every trace-class operator is a linear
combination of four positive trace-class operators.

The extension is positive by construction.  It is trace-non-increasing because,
for \(\rho\in\cT(\cH)_+\),
\[
    \tr(\Phi_\omega(\rho))
    =
    \tr(\Phi^+_\omega(\rho))
    \leq
    \tr(\rho).
\]
It is bounded.  If \(A=A^\dagger\) and \(A=A_+-A_-\) is its Jordan
decomposition, then positivity and trace non-increase give
\[
\begin{aligned}
    \|\Phi_\omega(A)\|_1
    &\leq
    \|\Phi^+_\omega(A_+)\|_1
    +
    \|\Phi^+_\omega(A_-)\|_1        \\
    &=
    \tr(\Phi^+_\omega(A_+))
    +
    \tr(\Phi^+_\omega(A_-))         \\
    &\leq
    \tr(A_+)+\tr(A_-)
    =
    \|A\|_1 .
\end{aligned}
\]
For general \(X\in\cT(\cH)\),
\[
    X=\operatorname{Re}X+i\,\operatorname{Im}X,
\]
and
\[
    \|\operatorname{Re}X\|_1\leq\|X\|_1,
    \qquad
    \|\operatorname{Im}X\|_1\leq\|X\|_1,
\]
so
\[
    \|\Phi_\omega(X)\|_1\leq 2\|X\|_1.
\]
Thus \(\Phi_\omega:\cT(\cH)\to\cT(\cH)\) is bounded.

If the finite amplifications of \(\Phi_\omega\) preserve positivity, then
\(\Phi_\omega\) is completely positive.  Together with trace non-increase and
bounded linearity, this means \(\Phi_\omega\in\cQ(\cH)\).

It remains to prove the measurability reduction.  Assume that
\(\omega\mapsto\Phi^+_\omega(\rho)\) is
\[
    (\Sigma,\operatorname{Borel}(\cT(\cH)_+))\text{-measurable}
\]
for every \(\rho\in\cT(\cH)_+\).  Fix \(X\in\cT(\cH)\) and choose a fixed
decomposition
\[
    X=X_1-X_2+iX_3-iX_4,
    \qquad
    X_j\in\cT(\cH)_+ .
\]
Then
\[
    \Phi_\omega(X)
    =
    \Phi^+_\omega(X_1)-\Phi^+_\omega(X_2)
    +i\Phi^+_\omega(X_3)-i\Phi^+_\omega(X_4).
\]
Each map \(\omega\mapsto\Phi^+_\omega(X_j)\) is measurable into
\(\cT(\cH)_+\), hence also measurable into \(\cT(\cH)\), because \(\cT(\cH)_+\) carries the
subspace Borel structure. 
The right-hand side of the above equality is a finite linear combination of
$(\Sigma,\operatorname{Borel}(\cT(\cH)))$\text{-measurable}
maps. Therefore \(\omega\mapsto\Phi_\omega(X)\) is
$(\Sigma,\operatorname{Borel}(\cT(\cH)))$-measurable.
\end{proof}

\begin{theorem}[Well-definedness and Structural Properties of Semantics]
\label{thm:semantics_well_defined}
Let $\Gamma$ be a finite classical context, and let $\Gamma\vdash S$ be a
well-formed quantum program phrase.  The structural semantic clauses define,
for each $\omega\in\Omega_\Gamma$, a positive-cone transformer
\[
    \sem{S}^{+}_\omega:
    \cT(\cH)_+\longrightarrow \cT(\cH)_+ .
\]
These transformers satisfy the following two properties.

\begin{enumerate}
    \item \textbf{Physical validity.}
    For every fixed $\omega\in\Omega_\Gamma$, the positive-cone transformer
    $\sem{S}^{+}_\omega$ extends uniquely to a CPTNI superoperator
    \(\sem{S}_\omega:\cT(\cH)\longrightarrow\cT(\cH)\).
    Equivalently, $\sem{S}_\omega\in\cQ(\cH)$ and
    \(\sem{S}^{+}_\omega=\sem{S}_\omega\big|_{\cT(\cH)_+}\).

    \item \textbf{Measurability.}
    For every fixed positive trace-class operator
    $\rho\in\cT(\cH)_+$, the map
    \(\omega\longmapsto \sem{S}^{+}_\omega(\rho)\)
    is
    \((\Sigma_\Gamma,\operatorname{Borel}(\cT(\cH)_+))\text{-measurable}\).
\end{enumerate}

Consequently, by Lemma~\ref{lem:positive_cone_reduction}, for every fixed trace-class operator $X\in\cT(\cH)$, the map
\(\omega\longmapsto \sem{S}_\omega(X)\)
is \((\Sigma_\Gamma,\operatorname{Borel}(\cT(\cH)))\)-measurable.

After this theorem, we write $\sem{S}_\omega$ also for its restriction to
positive trace-class inputs whenever no confusion can arise.  The superscript
``$+$'' is used only inside the well-definedness argument to distinguish the
positive-cone clauses from their CPTNI extensions.
\end{theorem}

Before proving the theorem, we record two elementary measurability facts used
to connect primitive operator families with trace-class-valued semantic clauses.
The first one handles pullback along classical expression evaluation; the second one turns strong operator measurability of an operator family into trace-class-valued measurability of the corresponding sandwich map.

In the primitive assumptions, the parameter $\omega$ in $U_\omega$ or
$M_{\omega,\nu}$ ranges over the value space $\Omega_{\seq e}$ of the displayed
expression tuple.  In the theorem above, the parameter $\omega$ in
$\sem{S}_\omega$ ranges over the context space $\Omega_\Gamma$.  The following
lemma is the general mechanism for transferring measurability from the former
level to the latter.

\begin{lemma}[Pullback along expression evaluation]
\label{lem:pullback_primitive_measurability}
Let $\seq e$ be a classical expression tuple well-typed over $\Gamma$, so that
\[
    \sem{\seq e}_\Gamma:
    (\Omega_\Gamma,\Sigma_\Gamma)
    \longrightarrow
    (\Omega_{\seq e},\Sigma_{\seq e})
\]
is measurable.

\begin{enumerate}
    \item Let $(Y,\Sigma_Y)$ be a measurable space.  If
    \(F:\Omega_{\seq e}\to Y\) is
    \((\Sigma_{\seq e},\Sigma_Y)\)-measurable, then
    \[
        \omega_{\mathrm{out}}
        \longmapsto
        F\bigl(\sem{\seq e}_\Gamma(\omega_{\mathrm{out}})\bigr)
    \]
    is \((\Sigma_\Gamma,\Sigma_Y)\)-measurable.

    \item Let $(Y,\Sigma_Y)$ be a measurable space and let
    \((\Omega_M,\Sigma_M)\) be a measurement outcome space.  If
    \(G:\Omega_{\seq e}\times\Omega_M\to Y\) is
    \((\Sigma_{\seq e}\otimes\Sigma_M,\Sigma_Y)\)-measurable, then
    \[
        (\omega_{\mathrm{out}},\nu)
        \longmapsto
        G\bigl(\sem{\seq e}_\Gamma(\omega_{\mathrm{out}}),\nu\bigr)
    \]
    is \((\Sigma_\Gamma\otimes\Sigma_M,\Sigma_Y)\)-measurable.

    \item In particular, if
    \(\{U_\omega\}_{\omega\in\Omega_{\seq e}}\) is strongly operator measurable
    over \((\Omega_{\seq e},\Sigma_{\seq e})\), then for every \(\u\in\cH\) the
    map
    \[
        \omega_{\mathrm{out}}
        \longmapsto
        U_{\sem{\seq e}_\Gamma(\omega_{\mathrm{out}})}\u
    \]
    is \((\Sigma_\Gamma,\operatorname{Borel}(\cH))\)-measurable.

    Similarly, if
    \[
        \{M_{\omega,\nu}\}_{(\omega,\nu)\in
        \Omega_{\seq e}\times\Omega_M}
    \]
    is jointly strongly operator measurable over
    \((\Omega_{\seq e}\times\Omega_M,\Sigma_{\seq e}\otimes\Sigma_M)\), then for
    every \(\u\in\cH\) the map
    \[
        (\omega_{\mathrm{out}},\nu)
        \longmapsto
        M_{\sem{\seq e}_\Gamma(\omega_{\mathrm{out}}),\nu}\u
    \]
    is \((\Sigma_\Gamma\otimes\Sigma_M,\operatorname{Borel}(\cH))\)-measurable.

    Moreover, if the Kraus-density family satisfies
    \[
        \int_{\Omega_M}
        \langle \u,M_{\omega,\nu}^{\dagger}M_{\omega,\nu}\u\rangle
        \,d\mu_M(\nu)
        =
        \langle \u,\u\rangle,
        \qquad
        \forall \omega\in\Omega_{\seq e},\ \forall \u\in\cH,
    \]
    then, for every \(\omega_{\mathrm{out}}\in\Omega_\Gamma\) and every
    \(\u\in\cH\),
    \[
        \int_{\Omega_M}
        \left\langle
            \u,
            M_{\sem{\seq e}_\Gamma(\omega_{\mathrm{out}}),\nu}^{\dagger}
            M_{\sem{\seq e}_\Gamma(\omega_{\mathrm{out}}),\nu}
            \u
        \right\rangle
        \,d\mu_M(\nu)
        =
        \langle \u,\u\rangle .
    \]
\end{enumerate}
\end{lemma}

\begin{proof}
The first assertion can be directly verified by pulling back measurable sets along the composition of functions.

For the second assertion, consider the map
\[
    \Theta:\Omega_\Gamma\times\Omega_M
    \longrightarrow
    \Omega_{\seq e}\times\Omega_M,
    \qquad
    \Theta(\omega_{\mathrm{out}},\nu)
    \coloneqq
    \bigl(\sem{\seq e}_\Gamma(\omega_{\mathrm{out}}),\nu\bigr).
\]
We verify its measurability on measurable rectangles.  For
\(A\in\Sigma_{\seq e}\) and \(C\in\Sigma_M\),
\[
\begin{aligned}
    \Theta^{-1}(A\times C)
    &=
    \{(\omega_{\mathrm{out}},\nu)
      \mid
      \sem{\seq e}_\Gamma(\omega_{\mathrm{out}})\in A,\ \nu\in C\}  \\
    &=
    \sem{\seq e}_\Gamma{}^{-1}(A)\times C .
\end{aligned}
\]
Since \(\sem{\seq e}_\Gamma\) is
\((\Sigma_\Gamma,\Sigma_{\seq e})\)-measurable, we have
\(\sem{\seq e}_\Gamma{}^{-1}(A)\in\Sigma_\Gamma\), and hence
\[
    \sem{\seq e}_\Gamma{}^{-1}(A)\times C
    \in
    \Sigma_\Gamma\otimes\Sigma_M .
\]
Because the measurable rectangles generate
\(\Sigma_{\seq e}\otimes\Sigma_M\), it follows that \(\Theta\) is
\((\Sigma_\Gamma\otimes\Sigma_M,
  \Sigma_{\seq e}\otimes\Sigma_M)\)-measurable
(cf.~\cite[Lemma~4.49]{aliprantis2006infinite}).
Composing \(G\) with \(\Theta\) gives the desired measurability.

The strong-operator statements follow by applying the first two assertions to
the \(\cH\)-valued maps
\[
    \omega\longmapsto U_\omega\u
    \qquad\text{and}\qquad
    (\omega,\nu)\longmapsto M_{\omega,\nu}\u .
\]
The normalization statement is pointwise substitution of
\(\omega=\sem{\seq e}_\Gamma(\omega_{\mathrm{out}})\).
\end{proof}

\paragraph*{Primitive-level indices and contextual shorthand.}
In the proof of Theorem~\ref{thm:semantics_well_defined}, primitive symbols are
kept at their primitive level.  Thus a unitary primitive is indexed as
\[
    U_\theta,
    \qquad
    \theta\in\Omega_{\seq e},
\]
and a measurement primitive is indexed as
\[
    M_{\theta,\nu},
    \qquad
    (\theta,\nu)\in\Omega_{\seq e}\times\Omega_M.
\]
When such a primitive occurs inside a program over a surrounding context
$\Gamma$, its actual primitive parameter is obtained by expression evaluation:
\[
    \theta=\sem{\seq e}_\Gamma(\omega),
    \qquad
    \omega\in\Omega_\Gamma.
\]

Outside passages where this primitive/context distinction is being verified
explicitly, we suppress the pullback along expression evaluation.  Thus, for a
primitive occurrence over $\Gamma$, we may write
\[
    U_\omega
    \quad\text{for}\quad
    U_{\sem{\seq e}_\Gamma(\omega)},
\]
and, in a binding construct,
\[
    M_{\omega,\nu}
    \quad\text{for}\quad
    M_{\sem{\seq e}_\Gamma(\omega),\nu}.
\]
This is purely a notational abbreviation.  The primitive-level families remain
indexed by the value space $\Omega_{\seq e}$ of the displayed expression tuple,
and all context-level measurability assertions follow from
\Cref{lem:pullback_primitive_measurability}.

Together with the convention for continuations in a binding construct,
\[
    S_{\omega,\nu}
    =
    S(x)_{\omega[x:=\nu]},
\]
this allows later sections to write expressions such as
\[
    M_{\omega,\nu}\rho M_{\omega,\nu}^{\dagger}
    \qquad\text{and}\qquad
    \mathcal C(S_{\omega,\nu},
    M_{\omega,\nu}\rho M_{\omega,\nu}^{\dagger})
\]
without repeatedly displaying the pullback
$M_{\sem{\seq e}_\Gamma(\omega),\nu}$.

The following lemma deals with a special measurability, which will be used for the cases of parameterized unitary and binding.

\begin{lemma}[Measurability of operator sandwich maps]
\label{lem:operator_sandwich_measurability}
Let $(X,\Sigma_X)$ be a measurable space, and let
\[
    \{A_x\}_{x\in X}\subseteq\cB(\cH)
\]
be strongly operator measurable.  Then for every
$\rho\in\cT(\cH)_+$, the map
\[
    x\longmapsto A_x\rho A_x^\dagger
\]
is $(\Sigma_X,\operatorname{Borel}(\cT(\cH)_+))$-measurable.
\end{lemma}

\begin{proof}
By \Cref{thm:pettis}, it suffices to prove weak $\Sigma_X$-measurability as a
$\cT(\cH)$-valued map.  Let $B\in\cB(\cH)$ be arbitrary.  Choose a spectral
decomposition
\[
    \rho=\sum_i\lambda_i\sP_{\u_i},
    \qquad
    \lambda_i\geq 0,\quad
    \sum_i\lambda_i=\tr(\rho)<\infty,
\]
where the vectors $\u_i$ are unit vectors.  For each fixed $x$, trace-norm
convergence of the spectral decomposition gives
\[
    \tr(BA_x\rho A_x^\dagger)
    =
    \sum_i\lambda_i
    \langle A_x\u_i,BA_x\u_i\rangle .
\]
For each $i$, the map $x\mapsto A_x\u_i$ is measurable by strong operator
measurability.  Since $B$ is bounded and the inner product is continuous, it follows that
\(
    x\longmapsto \langle A_x\u_i,BA_x\u_i\rangle
\)
is $(\Sigma_X,\operatorname{Borel}(\bC))$-measurable.

For each fixed $x$, the series is absolutely convergent:
\[
    \sum_i\lambda_i
    |\langle A_x\u_i,BA_x\u_i\rangle|
    \leq
    \|B\|\,\|A_x\|^2\,\tr(\rho)
    <\infty .
\]
Hence $x\mapsto \tr(BA_x\rho A_x^\dagger)$ is the pointwise limit of
measurable partial sums, and is therefore
$(\Sigma_X,\operatorname{Borel}(\bC))$-measurable by
\cite[Lemma~4.29]{aliprantis2006infinite}.  Since $B$ was arbitrary,
$x\mapsto A_x\rho A_x^\dagger$ is weakly $\Sigma_X$-measurable.  By
\Cref{thm:pettis}, it is
$(\Sigma_X,\operatorname{Borel}(\cT(\cH)))$-measurable.  Since the map takes
values in $\cT(\cH)_+$ and $\operatorname{Borel}(\cT(\cH)_+)$ is the subspace
Borel $\sigma$-algebra, it is
$(\Sigma_X,\operatorname{Borel}(\cT(\cH)_+))$-measurable.
\end{proof}

\begin{proof}[Proof of \Cref{thm:semantics_well_defined}]
We prove the theorem by mutual structural induction on the syntax of $S$.
For each well-formed phrase $\Gamma\vdash S$, the induction proves
simultaneously:

\begin{enumerate}
    \item (Physical Validity) for every $\omega\in\Omega_\Gamma$, the positive-cone map
    $\sem{S}^{+}_\omega$ is cone-linear, trace-non-increasing, and its finite
    amplifications preserve positivity; hence it admits a unique CPTNI
    extension by \Cref{lem:positive_cone_reduction};

    \item (Measurability) for every $\rho\in\cT(\cH)_+$, the map
    $\omega\mapsto\sem{S}^{+}_\omega(\rho)$ is
    $(\Sigma_\Gamma,\operatorname{Borel}(\cT(\cH)_+))$-measurable.
\end{enumerate}

Uniqueness of the CPTNI extension and measurability on the full trace-class
space then follow uniformly from
Lemma~\ref{lem:positive_cone_reduction}.

\paragraph{$\bullet$ \textbf{Skip ($\mathbf{skip}$):}}
For every $\omega\in\Omega_\Gamma$ and every $\rho\in\cT(\cH)_+$, define
\[
    \sem{\mathbf{skip}}^{+}_\omega(\rho)\coloneqq \rho .
\]

\emph{Physical Validity (CPTNI):}
The positive-cone map
\[
    \rho\longmapsto \sem{\mathbf{skip}}^{+}_\omega(\rho)=\rho
\]
is cone-linear on \(\cT(\cH)_+\), trace-preserving, and its finite
amplifications preserve positivity.  Hence, by
\Cref{lem:positive_cone_reduction}, it extends uniquely to a CPTNI
superoperator on \(\cT(\cH)\).  This extension is the identity superoperator
\[
    \operatorname{id}_{\cT(\cH)}:\cT(\cH)\to\cT(\cH),
    \qquad
    X\longmapsto X .
\]
Thus \(\sem{\mathbf{skip}}_\omega=\operatorname{id}_{\cT(\cH)}\in\cQ(\cH)\),
and its restriction to \(\cT(\cH)_+\) is exactly
\(\sem{\mathbf{skip}}^{+}_\omega\).

\emph{Measurability:} For fixed $\rho\in\cT(\cH)_+$, the map
$\omega\mapsto\sem{\mathbf{skip}}^{+}_\omega(\rho)=\rho$ is constant as a map
from $\Omega_\Gamma$ to $\cT(\cH)_+$.  Therefore it is trivially 
\((\Sigma_\Gamma,\operatorname{Borel}(\cT(\cH)_+))\)-measurable.

\paragraph{$\bullet$ \textbf{Abort ($\mathbf{abort}$):}}
For every $\omega\in\Omega_\Gamma$ and every $\rho\in\cT(\cH)_+$, define
\[
    \sem{\mathbf{abort}}^{+}_\omega(\rho)\coloneqq 0 .
\]

\emph{Physical Validity (CPTNI):}
The positive-cone map
\[
    \rho\longmapsto \sem{\mathbf{abort}}^{+}_\omega(\rho)=0
\]
is cone-linear on \(\cT(\cH)_+\), trace-non-increasing, and its finite
amplifications preserve positivity.  Hence, by
\Cref{lem:positive_cone_reduction}, it extends uniquely to a CPTNI
superoperator on \(\cT(\cH)\).  This extension is the zero superoperator
\[
    0_{\cT(\cH)}:\cT(\cH)\to\cT(\cH),
    \qquad
    X\longmapsto 0 .
\]
Thus \(\sem{\mathbf{abort}}_\omega=0_{\cT(\cH)}\in\cQ(\cH)\), and its
restriction to \(\cT(\cH)_+\) is exactly
\(\sem{\mathbf{abort}}^{+}_\omega\).

\emph{Measurability:} For fixed $\rho\in\cT(\cH)_+$, the map
$\omega\mapsto\sem{\mathbf{abort}}^{+}_\omega(\rho)=0$ is constant as a map
from $\Omega_\Gamma$ to $\cT(\cH)_+$.  Therefore it is trivially
\((\Sigma_\Gamma,\operatorname{Borel}(\cT(\cH)_+))\)-measurable.

\paragraph{$\bullet$ \textbf{Initialization ($q:=0$):}}
Let $\{|k\rangle_q\}_{k\in J}$ be the fixed orthonormal basis of the Hilbert space
associated with the subsystem $q$, where the index set $J$ is finite or countable.  Set
\[
    E_k\coloneqq |0\rangle_q\langle k|_q\otimes I_{\mathrm{env}} .
\]
For every $\omega\in\Omega_\Gamma$ and every $\rho\in\cT(\cH)_+$, define
\[
    \sem{q:=0}^{+}_\omega(\rho)
    \coloneqq
    \sum_{k\in J} E_k\rho E_k^\dagger .
\]
The series is understood as the trace-norm limit of its finite partial sums.
Indeed, after enumerating \(J\), the partial sums
\[
    S_N(\rho)\coloneqq\sum_{k=1}^{N}E_k\rho E_k^\dagger
\]
form an increasing sequence in \(\cT(\cH)_+\), and
\[
    \tr(S_N(\rho))
    =
    \tr\!\left(\rho\sum_{k=1}^{N}E_k^\dagger E_k\right)
    \leq
    \tr(\rho).
\]
Thus Proposition~\ref{prop:state_cpo_topology}, after scaling by
\(\tr(\rho)\) when \(\rho\neq0\), gives trace-norm convergence.
Operationally, this map discards the
previous content of subsystem $q$, replaces it by
$|0\rangle_q\langle 0|$, and leaves the environment unchanged.

\emph{Physical Validity (CPTNI):}
The positive-cone map
\[
    \rho
    \longmapsto
    \sem{q:=0}^{+}_\omega(\rho)
    =
    \sum_{k\in J}E_k\rho E_k^\dagger
\]
is cone-linear on \(\cT(\cH)_+\).  The Kraus operators satisfy
$\sum_{k\in J}E_k^\dagger E_k=I$ 
in the strong operator topology, so the map is trace-preserving.  Its finite
amplifications are given by the Kraus family
$\{E_k\otimes I_n\}_{k\in J}$,
and therefore preserve positivity on
\(\cT(\cH\otimes\mathbb C^n)_+\) for every \(n\in\mathbb N\).  Hence, by
\Cref{lem:positive_cone_reduction}, the positive-cone reset semantics extends
uniquely to a CPTNI superoperator on \(\cT(\cH)\).  This extension is the reset
channel
\[
    \mathcal R_q:\cT(\cH)\to\cT(\cH),
    \qquad
    \mathcal R_q(X)
    \coloneqq
    \sum_{k\in J}E_kXE_k^\dagger .
\]
Thus \(\sem{q:=0}_\omega=\mathcal R_q\in\cQ(\cH)\), and its restriction to
\(\cT(\cH)_+\) is exactly \(\sem{q:=0}^{+}_\omega\).

\emph{Measurability:} For fixed $\rho\in\cT(\cH)_+$, the value
$\sem{q:=0}^{+}_\omega(\rho)=\mathcal R_q(\rho)$ is independent of
$\omega$.  Hence the map
\(\omega\longmapsto \sem{q:=0}^{+}_\omega(\rho)\)
is constant as a map from $\Omega_\Gamma$ to $\cT(\cH)_+$.  Therefore it is
trivially \((\Sigma_\Gamma,\operatorname{Borel}(\cT(\cH)_+))\)-measurable.

\paragraph*{$\bullet$ \textbf{Parameterized Unitary
($U(\seq e)[\seq q]$):}}
Let $\Gamma\vdash U(\seq e)[\seq q]$.  The primitive unitary family is written
as
\[
    \{U_\omega\}_{\omega\in\Omega_{\seq e}},
\]
where $\omega$ ranges over the value space of the expression tuple $\seq e$.
For an outer meta-parameter instance
$\omega_{\mathrm{out}}\in\Omega_\Gamma$, the expression tuple selects the
primitive parameter value
\[
    \omega=\sem{\seq e}_\Gamma(\omega_{\mathrm{out}})
    \in\Omega_{\seq e}.
\]
Substituting this value into the primitive family, the positive-cone semantic
clause is
\[
    \sem{U(\seq e)[\seq q]}^{+}_{\omega_{\mathrm{out}}}(\rho)
    =
    U_{\omega}\rho U_{\omega}^{\dagger},
    \qquad
    \rho\in\cT(\cH)_+ .
\]

\emph{Physical Validity (CPTNI):}
For fixed $\omega_{\mathrm{out}}\in\Omega_\Gamma$, the positive-cone map
\[
    \rho
    \longmapsto
    \sem{U(\seq e)[\seq q]}^{+}_{\omega_{\mathrm{out}}}(\rho)
    =
    U_{\sem{\seq e}_\Gamma(\omega_{\mathrm{out}})}
    \rho
    U_{\sem{\seq e}_\Gamma(\omega_{\mathrm{out}})}^\dagger
\]
is cone-linear on $\cT(\cH)_+$.  It is trace-preserving, since it is conjugation
by a unitary operator.  Its finite amplifications are conjugations by the
unitaries
\[
    U_{\sem{\seq e}_\Gamma(\omega_{\mathrm{out}})}\otimes I_n,
    \qquad n\in\mathbb N,
\]
and therefore preserve positivity on
$\cT(\cH\otimes\mathbb C^n)_+$.

Hence, by \Cref{lem:positive_cone_reduction}, this positive-cone map extends
uniquely to a CPTNI superoperator on $\cT(\cH)$.  The extension is the unitary
conjugation map
\[
    X
    \longmapsto
    U_{\sem{\seq e}_\Gamma(\omega_{\mathrm{out}})}
    X
    U_{\sem{\seq e}_\Gamma(\omega_{\mathrm{out}})}^\dagger,
    \qquad
    X\in\cT(\cH).
\]
Thus
$\sem{U(\seq e)[\seq q]}_{\omega_{\mathrm{out}}}\in \cQ(\cH)$,
and its restriction to $\cT(\cH)_+$ is exactly
$\sem{U(\seq e)[\seq q]}^{+}_{\omega_{\mathrm{out}}}$.

\emph{Measurability:}
At the primitive level, the primitive assumption guarantees that
$\{U_\omega\}_{\omega\in\Omega_{\seq e}}$ is strongly operator measurable.
By \Cref{lem:operator_sandwich_measurability}, for every
$\rho\in\cT(\cH)_+$ the map
\[
    \omega\longmapsto U_\omega\rho U_\omega^\dagger
\]
is $(\Sigma_{\seq e},\operatorname{Borel}(\cT(\cH)_+))$-measurable.
Applying \Cref{lem:pullback_primitive_measurability} along
$\sem{\seq e}_\Gamma$ gives the required
$(\Sigma_\Gamma,\operatorname{Borel}(\cT(\cH)_+))$-measurability.

\paragraph{$\bullet$ \textbf{Sequential Composition ($S_1;S_2$):}}
Assume that $\Gamma\vdash S_1$ and $\Gamma\vdash S_2$.  For every
$\omega\in\Omega_\Gamma$ and every $\rho\in\cT(\cH)_+$, define
\[
    \sem{S_1;S_2}^{+}_\omega(\rho)
    \coloneqq
    \sem{S_2}^{+}_\omega
    \bigl(
        \sem{S_1}^{+}_\omega(\rho)
    \bigr).
\]

\emph{Physical Validity (CPTNI):}
By the structural induction hypothesis, for each fixed
$\omega\in\Omega_\Gamma$, the positive-cone maps
$\sem{S_1}^{+}_\omega$ and $\sem{S_2}^{+}_\omega$ admit unique CPTNI extensions
\[
    \sem{S_1}_\omega,\sem{S_2}_\omega:\cT(\cH)\to\cT(\cH).
\]
The positive-cone semantics of the sequence is
\[
    \sem{S_1;S_2}^{+}_\omega(\rho)
    =
    \sem{S_2}^{+}_\omega
    \bigl(
        \sem{S_1}^{+}_\omega(\rho)
    \bigr),
    \qquad
    \rho\in\cT(\cH)_+ .
\]
It is cone-linear because both component positive-cone maps are cone-linear.
It is trace-non-increasing because
\[
    \tr\!\left(
        \sem{S_2}^{+}_\omega
        \bigl(
            \sem{S_1}^{+}_\omega(\rho)
        \bigr)
    \right)
    \leq
    \tr\!\left(\sem{S_1}^{+}_\omega(\rho)\right)
    \leq
    \tr(\rho).
\]
Its finite amplifications preserve positivity, since they are the restrictions of
\[
    (\sem{S_2}_\omega\otimes\mathcal I_n)
    \circ
    (\sem{S_1}_\omega\otimes\mathcal I_n),
    \qquad n\in\mathbb N,
\]
and both factors are positive on
$\cT(\cH\otimes\mathbb C^n)_+$.

Hence, by \Cref{lem:positive_cone_reduction},
$\sem{S_1;S_2}^{+}_\omega$ extends uniquely to a CPTNI superoperator on
$\cT(\cH)$.  This extension is the composition
\[
    \sem{S_1;S_2}_\omega
    \coloneqq
    \sem{S_2}_\omega\circ \sem{S_1}_\omega .
\]
Therefore
$\sem{S_1;S_2}_\omega\in\cQ(\cH)$,
and its restriction to $\cT(\cH)_+$ is exactly
$\sem{S_1;S_2}^{+}_\omega$.

\emph{Measurability:}
Fix $\rho\in\cT(\cH)_+$.  Set
\[
    \rho'_\omega
    \coloneqq
    \sem{S_1}^{+}_\omega(\rho).
\]
By the structural induction hypothesis for $S_1$, the map
$\omega\mapsto\rho'_\omega$ is
$(\Sigma_\Gamma,\operatorname{Borel}(\cT(\cH)_+))$-measurable.

Define
\[
    h:\Omega_\Gamma\to\Omega_\Gamma\times\cT(\cH)_+,
    \qquad
    h(\omega)\coloneqq(\omega,\rho'_\omega).
\]
The first component $\omega\mapsto\omega$ is
$(\Sigma_\Gamma,\Sigma_\Gamma)$-measurable, and the second component
$\omega\mapsto\rho'_\omega$ is
$(\Sigma_\Gamma,\operatorname{Borel}(\cT(\cH)_+))$-measurable.  Hence, by
\cite[Lemma~4.49]{aliprantis2006infinite}, the pair map $h$ is
\(
    \bigl(
        \Sigma_\Gamma,
        \Sigma_\Gamma\otimes\operatorname{Borel}(\cT(\cH)_+)
    \bigr)
\)-measurable.

Next define
\[
    \Phi:\Omega_\Gamma\times\cT(\cH)_+\to\cT(\cH)_+,
    \qquad
    \Phi(\omega,\sigma)
    \coloneqq
    \sem{S_2}^{+}_\omega(\sigma).
\]
We verify the Carath\'eodory hypotheses:

\begin{itemize}
    \item By the structural induction hypothesis for $S_2$, for each fixed $\sigma\in\cT(\cH)_+$, the map
    \[
        \omega\longmapsto \Phi(\omega,\sigma)
        =
        \sem{S_2}^{+}_\omega(\sigma)
    \]
    is
    $(\Sigma_\Gamma,\operatorname{Borel}(\cT(\cH)_+))$-measurable.

    \item For each fixed $\omega\in\Omega_\Gamma$, the map
    \[
        \sigma\longmapsto \Phi(\omega,\sigma)
        =
        \sem{S_2}^{+}_\omega(\sigma)
    \]
    is trace-norm continuous on $\cT(\cH)_+$.  Indeed,
    $\sem{S_2}^{+}_\omega$ is the restriction of the bounded linear CPTNI map $\sem{S_2}_\omega:\cT(\cH)\to\cT(\cH)$.
\end{itemize}

Since $\cT(\cH)_+$ is separable and metrizable in the trace-norm subspace
topology, the Carath\'eodory joint measurability theorem
\cite[Lemma~4.51]{aliprantis2006infinite} implies that $\Phi$ is
\(
    \bigl(
        \Sigma_\Gamma\otimes\operatorname{Borel}(\cT(\cH)_+),
        \operatorname{Borel}(\cT(\cH)_+)
    \bigr)
\)-measurable.

We now pull back this joint measurability along $h$.  Let
$W\subseteq\cT(\cH)_+$ be relatively trace-norm open.  Since $\Phi$ is jointly
measurable,
\[
    \Phi^{-1}(W)
    \in
    \Sigma_\Gamma\otimes\operatorname{Borel}(\cT(\cH)_+).
\]
Since $h$ is measurable into this product measurable space,
\[
\begin{aligned}
    \{\omega\in\Omega_\Gamma
      \mid
      \sem{S_1;S_2}^{+}_\omega(\rho)\in W\}
    &=
    \{\omega\in\Omega_\Gamma
      \mid
      \Phi(\omega,\rho'_\omega)\in W\}        \\
    &=
    (\Phi\circ h)^{-1}(W)                    \\
    &=
    h^{-1}\bigl(\Phi^{-1}(W)\bigr)
    \in\Sigma_\Gamma .
\end{aligned}
\]
Therefore
$\omega\longmapsto\sem{S_1;S_2}^{+}_\omega(\rho)$
is
$(\Sigma_\Gamma,\operatorname{Borel}(\cT(\cH)_+))$-measurable.
This proves the required positive-input measurability for sequential
composition.

\paragraph*{$\bullet$ \textbf{Binding
($\bind(M(\seq e),x.S(x))$):}}
Assume
\[
    \Gamma\vdash \bind(M(\seq e),x.S(x)).
\]
The continuation \(S(x)\) is well-formed over the extended context
\(\Gamma,x\), where the variable \(x\) has the measurement outcome space
\((\Omega_M,\Sigma_M,\mu_M)\).  The primitive Kraus-density family is written as
\[
    \{M_{\omega,\nu}\}_{(\omega,\nu)\in
    \Omega_{\seq e}\times\Omega_M},
\]
where \(\omega\) ranges over the value space of the displayed expression tuple
\(\seq e\).  In this proof, the outer context parameter is denoted by
\(\omega_{\mathrm{out}}\in\Omega_\Gamma\).  This explicit distinction is used
only to verify measurability; later predicate-transformer rules again package all
non-integrated parameters into a single symbol \(\omega\).

Fix \(\rho\in\cT(\cH)_+\).  For
\((\omega_{\mathrm{out}},\nu)\in\Omega_\Gamma\times\Omega_M\), define the
single-outcome branch expression
\[
    F_\rho(\omega_{\mathrm{out}},\nu)
    \coloneqq
    \sem{S}^{+}_{\omega_{\mathrm{out}},\nu}
    \!\left(
        M_{\sem{\seq e}_\Gamma(\omega_{\mathrm{out}}),\nu}
        \rho
        M_{\sem{\seq e}_\Gamma(\omega_{\mathrm{out}}),\nu}^{\dagger}
    \right).
\]

\emph{Pettis-style semantic clause:}
For fixed \(\omega_{\mathrm{out}}\in\Omega_\Gamma\), the binding output is
specified weakly by requiring that, for every \(B\in\cB(\cH)\),
\begin{equation}
\label{eq:bind-pettis-clause}
    \tr\!\left(
        B\,\sem{\bind(M(\seq e),x.S(x))}^{+}_{\omega_{\mathrm{out}}}(\rho)
    \right)
    =
    \int_{\Omega_M}
    \tr\!\left(
        B\,F_\rho(\omega_{\mathrm{out}},\nu)
    \right)
    \,d\mu_M(\nu).
\end{equation}
We next verify that this weak prescription is well-defined and is represented by
a Bochner integral.

\emph{Well-definedness:}
We first prove uniqueness of the operator specified by
\eqref{eq:bind-pettis-clause}.  Suppose that
\(\tau_1,\tau_2\in\cT(\cH)\) both satisfy the trace-duality identity
\eqref{eq:bind-pettis-clause}; that is, for every \(B\in\cB(\cH)\),
\[
    \tr(B\tau_1)=
    \int_{\Omega_M}
    \tr\!\left(BF_\rho(\omega_{\mathrm{out}},\nu)\right)
    \,d\mu_M(\nu)
    =
    \tr(B\tau_2).
\]
Set \(D\coloneqq \tau_1-\tau_2\).  Then
\[
    \tr(BD)=0,
    \qquad
    \forall B\in\cB(\cH).
\]
Let $D=V|D|$
be the polar decomposition of the trace-class operator \(D\), where
\(V\in\cB(\cH)\) is a partial isometry.  Since \(V^*V\) is the support projection
of \(|D|\), we have
\[
    V^*D=V^*V|D|=|D|.
\]
Taking \(B=V^*\) gives
\[
    0=\tr(V^*D)=\tr(|D|)=\|D\|_1.
\]
Hence \(D=0\), and therefore \(\tau_1=\tau_2\).  Thus any trace-class operator
satisfying the Pettis-style identity \eqref{eq:bind-pettis-clause}, if it exists,
is unique.

We now prove existence.  Start at the primitive measurement level and define
\[
    G_\rho(\omega,\nu)
    \coloneqq
    M_{\omega,\nu}\rho M_{\omega,\nu}^{\dagger},
    \qquad
    (\omega,\nu)\in\Omega_{\seq e}\times\Omega_M.
\]
By joint strong operator measurability of
$\{M_{\omega,\nu}\}_{(\omega,\nu)\in\Omega_{\seq e}\times\Omega_M}$ and
\Cref{lem:operator_sandwich_measurability}, the map
\[
    G_\rho(\omega,\nu)
    =
    M_{\omega,\nu}\rho M_{\omega,\nu}^{\dagger}
\]
is $(\Sigma_{\seq e}\otimes\Sigma_M,
     \operatorname{Borel}(\cT(\cH)_+))$-measurable.

By \Cref{lem:pullback_primitive_measurability}, the pulled-back branch
preparation satisfies
\begin{equation}
\label{eq:bind-pulled-back-branch-preparation}
    (\omega_{\mathrm{out}},\nu)
    \longmapsto
    G_\rho\bigl(\sem{\seq e}_\Gamma(\omega_{\mathrm{out}}),\nu\bigr)
    \quad\text{is}\quad
    (\Sigma_\Gamma\otimes\Sigma_M,
     \operatorname{Borel}(\cT(\cH)_+))\text{-measurable}.
\end{equation}

Written out, this pulled-back branch preparation is exactly
\[
    G_\rho\bigl(\sem{\seq e}_\Gamma(\omega_{\mathrm{out}}),\nu\bigr)
    =
    M_{\sem{\seq e}_\Gamma(\omega_{\mathrm{out}}),\nu}
    \rho
    M_{\sem{\seq e}_\Gamma(\omega_{\mathrm{out}}),\nu}^{\dagger},
\]
and the actual single-outcome branch of the binding construct is obtained by
feeding this prepared branch state into the continuation:
\(
    F_\rho(\omega_{\mathrm{out}},\nu)
    \coloneqq
    \sem{S}^{+}_{\omega_{\mathrm{out}},\nu}
    \!\left(
        G_\rho\bigl(\sem{\seq e}_\Gamma(\omega_{\mathrm{out}}),\nu\bigr)
    \right)
\).

Next define
\[
    \Psi:
    \Omega_\Gamma\times\Omega_M\times\cT(\cH)_+
    \longrightarrow
    \cT(\cH)_+
\]
by
\[
    \Psi(\omega_{\mathrm{out}},\nu,\sigma)
    \coloneqq
    \sem{S}^{+}_{\omega_{\mathrm{out}},\nu}(\sigma).
\]
We verify the Carath\'eodory hypotheses:

\begin{itemize}
    \item For each fixed \(\sigma\in\cT(\cH)_+\), the map
    \[
        (\omega_{\mathrm{out}},\nu)
        \longmapsto
        \Psi(\omega_{\mathrm{out}},\nu,\sigma)
        =
        \sem{S}^{+}_{\omega_{\mathrm{out}},\nu}(\sigma)
    \]
    is
    $(\Sigma_\Gamma\otimes\Sigma_M,\operatorname{Borel}(\cT(\cH)_+))$-measurable by the structural induction hypothesis applied to
    \(\Gamma,x\vdash S(x)\), with
    \[
        \Omega_{\Gamma,x}=\Omega_\Gamma\times\Omega_M,
        \qquad
        \Sigma_{\Gamma,x}=\Sigma_\Gamma\otimes\Sigma_M.
    \]
    If \(S(x)\) is given by finitely many measurable cases
    \((A_i,S_i(x))_{i=1}^m\), the same conclusion follows by applying the induction
    hypothesis to each \(S_i(x)\) and pasting along the measurable partition
    \(A_i\).
    
    \item For each fixed \((\omega_{\mathrm{out}},\nu)\), the map
    \[
        \sigma\longmapsto
        \Psi(\omega_{\mathrm{out}},\nu,\sigma)
        =
        \sem{S}^{+}_{\omega_{\mathrm{out}},\nu}(\sigma)
    \]
    is trace-norm continuous on \(\cT(\cH)_+\).  Indeed, by the structural
    induction hypothesis, the continuation semantics admits a CPTNI extension
    on \(\cT(\cH)\).
\end{itemize}

Since \(\cT(\cH)_+\) is separable and metrizable in the trace-norm subspace
topology, the Carath\'eodory joint measurability theorem
\cite[Lemma~4.51]{aliprantis2006infinite} gives that \(\Psi\) is
\[
    \bigl(
        \Sigma_\Gamma\otimes\Sigma_M
        \otimes\operatorname{Borel}(\cT(\cH)_+),
        \operatorname{Borel}(\cT(\cH)_+)
    \bigr)\text{-measurable}.
\]

Then for fixed $\rho\in\cT(\cH)_{+}$, we define
\[
    h:\Omega_\Gamma\times\Omega_M
    \longrightarrow
    \Omega_\Gamma\times\Omega_M\times\cT(\cH)_+
\]
by
\[
    h(\omega_{\mathrm{out}},\nu)
    \coloneqq
    \left(
        \omega_{\mathrm{out}},
        \nu,
        G_\rho\bigl(\sem{\seq e}_\Gamma(\omega_{\mathrm{out}}),\nu\bigr)
    \right).
\]
The first two components are coordinate projections, and the third component is
measurable by \eqref{eq:bind-pulled-back-branch-preparation}.  Hence, by
\cite[Lemma~4.49]{aliprantis2006infinite}, the map \(h\) is $\big(\Sigma_\Gamma\otimes\Sigma_M,
(\Sigma_\Gamma\otimes\Sigma_M)\otimes\operatorname{Borel}(\cT(\cH)_+)\big)$-measurable.  

Therefore, from $F_\rho(\omega_{\mathrm{out}},\nu)=\Psi\circ h$, we can demonstrate that for every
relatively trace-norm open set \(W\subseteq\cT(\cH)_+\),
\[
\begin{aligned}
    F_\rho^{-1}(W)
    &=
    \{(\omega_{\mathrm{out}},\nu)
      \mid
      \Psi(h(\omega_{\mathrm{out}},\nu))\in W\}  \\
    &=
    h^{-1}\bigl(\Psi^{-1}(W)\bigr)
    \in
    \Sigma_\Gamma\otimes\Sigma_M .
\end{aligned}
\]
Therefore
\[
    F_\rho:
    \Omega_\Gamma\times\Omega_M\longrightarrow\cT(\cH)_+
\]
is $(\Sigma_\Gamma\otimes\Sigma_M,
     \operatorname{Borel}(\cT(\cH)_+))$-measurable.

Fix \(\omega_{\mathrm{out}}\in\Omega_\Gamma\).  Since \(F_\rho\) is product
measurable, the section
\[
    \nu\longmapsto F_\rho(\omega_{\mathrm{out}},\nu)
\]
is \((\Sigma_M,\operatorname{Borel}(\cT(\cH)_+))\)-measurable (see \cite[Lemma 4.48]{aliprantis2006infinite}), hence Bochner
measurable by \Cref{prop:bochner_equivalent_def}.  Moreover,
\(F_\rho(\omega_{\mathrm{out}},\nu)\in\cT(\cH)_+\), and therefore
\[
    \|F_\rho(\omega_{\mathrm{out}},\nu)\|_1
    =
    \tr(F_\rho(\omega_{\mathrm{out}},\nu)).
\]
By the TNI property of the continuation semantics,
\[
    \tr(F_\rho(\omega_{\mathrm{out}},\nu))
    \leq
    \tr\!\left(
        M_{\sem{\seq e}_\Gamma(\omega_{\mathrm{out}}),\nu}
        \rho
        M_{\sem{\seq e}_\Gamma(\omega_{\mathrm{out}}),\nu}^{\dagger}
    \right).
\]
Integrating over \(\Omega_M\), and using the pulled-back normalization
condition from \Cref{lem:pullback_primitive_measurability}, we obtain
\[
\begin{aligned}
    \int_{\Omega_M}
    \|F_\rho(\omega_{\mathrm{out}},\nu)\|_1\,d\mu_M(\nu)
    &\leq
    \int_{\Omega_M}
    \tr\!\left(
        M_{\sem{\seq e}_\Gamma(\omega_{\mathrm{out}}),\nu}
        \rho
        M_{\sem{\seq e}_\Gamma(\omega_{\mathrm{out}}),\nu}^{\dagger}
    \right)
    d\mu_M(\nu)                                          \\
    &=
    \tr(\rho)<\infty .
\end{aligned}
\]
Here the last equality is justified by writing
\(\rho=\sum_i\lambda_i\sP_{\u_i}\) and applying Tonelli's theorem to the
non-negative scalar function
\[
    (i,\nu)
    \longmapsto
    \lambda_i
    \left\|
        M_{\sem{\seq e}_\Gamma(\omega_{\mathrm{out}}),\nu}\u_i
    \right\|^2
\]
on \(\mathbb N\times\Omega_M\), followed by the normalization of the
pulled-back instrument; see \cite[Proposition 5.2.1]{cohn2013measure}.
Thus the section
\(\nu\mapsto F_\rho(\omega_{\mathrm{out}},\nu)\) is Bochner integrable.

By Hille's theorem for the Bochner integral
(see \cite[Theorem~2.2.6]{diestel1977vector}
and \cite[Theorem~III.2.19(c)]{dunford1958linear}), for every \(B\in\cB(\cH)\),
\[
    \tr\!\left(
        B\int_{\Omega_M}F_\rho(\omega_{\mathrm{out}},\nu)\,d\mu_M(\nu)
    \right)
    =
    \int_{\Omega_M}
    \tr\!\left(BF_\rho(\omega_{\mathrm{out}},\nu)\right)
    d\mu_M(\nu).
\]
Hence the Pettis-style prescription \eqref{eq:bind-pettis-clause} is represented
by the Bochner integral
\[
    \sem{\bind(M(\seq e),x.S(x))}^{+}_{\omega_{\mathrm{out}}}(\rho)
    =
    \int_{\Omega_M}
    F_\rho(\omega_{\mathrm{out}},\nu)\,d\mu_M(\nu).
\]

Finally, this Bochner integral belongs to \(\cT(\cH)_+\).  Indeed, for every
\(\u\in\cH\), Hille's theorem applied to the bounded functional
\(A\mapsto\tr(\sP_{\u}A)\) gives
\[
    \left\langle
        \u,
        \int_{\Omega_M}F_\rho(\omega_{\mathrm{out}},\nu)\,d\mu_M(\nu)
        \,\u
    \right\rangle
    =
    \int_{\Omega_M}
    \langle\u,F_\rho(\omega_{\mathrm{out}},\nu)\u\rangle
    d\mu_M(\nu)
    \geq 0.
\]
Thus the output is a positive trace-class operator.  The same trace estimate also
gives
\[
    \tr\!\left(
        \sem{\bind(M(\seq e),x.S(x))}^{+}_{\omega_{\mathrm{out}}}(\rho)
    \right)
    \leq
    \tr(\rho).
\]

\emph{Physical Validity (CPTNI):}
Fix \(\omega_{\mathrm{out}}\in\Omega_\Gamma\).  We verify the hypotheses of
\Cref{lem:positive_cone_reduction} for the positive-cone binding semantics.

First, the map
\[
    \rho
    \longmapsto
    \sem{\bind(M(\seq e),x.S(x))}^{+}_{\omega_{\mathrm{out}}}(\rho)
\]
is cone-linear on \(\cT(\cH)_+\).  This follows from the linearity of the
prepared branch
\[
    \rho\longmapsto
    M_{\sem{\seq e}_\Gamma(\omega_{\mathrm{out}}),\nu}
    \rho
    M_{\sem{\seq e}_\Gamma(\omega_{\mathrm{out}}),\nu}^{\dagger},
\]
the cone-linearity of the continuation semantics
\(\sem{S}^{+}_{\omega_{\mathrm{out}},\nu}\) by the induction hypothesis, and the
linearity of the Bochner integral.

Second, the positive-cone map is trace-non-increasing.  By the trace estimate
established in the well-definedness part, for every \(\rho\in\cT(\cH)_+\),
\[
\begin{aligned}
    \tr\!\left(
        \sem{\bind(M(\seq e),x.S(x))}^{+}_{\omega_{\mathrm{out}}}(\rho)
    \right)
    &=
    \int_{\Omega_M}
    \tr(F_\rho(\omega_{\mathrm{out}},\nu))\,d\mu_M(\nu)  \\
    &\leq
    \tr(\rho).
\end{aligned}
\]
Hence, by the first part of \Cref{lem:positive_cone_reduction}, the
positive-cone binding semantics has a unique bounded positive TNI extension to
\(\cT(\cH)\).  We denote this extension by
\[
    \sem{\bind(M(\seq e),x.S(x))}_{\omega_{\mathrm{out}}}.
\]

It remains to verify complete positivity of this extension.  Let \(k\in\mathbb N\)
and let \(Z\in\cT(\cH\otimes\bC^k)_+\).  For each outcome \(\nu\), consider the
amplified branch
\[
    F^{(k)}_Z(\omega_{\mathrm{out}},\nu)
    \coloneqq
    \bigl(\sem{S}_{\omega_{\mathrm{out}},\nu}\otimes\mathcal I_k\bigr)
    \!\left(
        (M_{\sem{\seq e}_\Gamma(\omega_{\mathrm{out}}),\nu}\otimes I_k)
        Z
        (M_{\sem{\seq e}_\Gamma(\omega_{\mathrm{out}}),\nu}^{\dagger}
         \otimes I_k)
    \right),
\]
where \(\sem{S}_{\omega_{\mathrm{out}},\nu}\) is the CPTNI extension of the
continuation semantics.  This branch is positive for every \(\nu\).  The same measurability and trace-estimate argument as above, applied on
\(\cH\otimes\bC^k\) (finite amplification preserves strong operator
measurability, and the continuation is handled blockwise by the induction
hypothesis), shows that
\[
    \nu\longmapsto F^{(k)}_Z(\omega_{\mathrm{out}},\nu)
\]
is Bochner integrable.  Therefore
\[
    \int_{\Omega_M}
    F^{(k)}_Z(\omega_{\mathrm{out}},\nu)\,d\mu_M(\nu)
    \in
    \cT(\cH\otimes\bC^k)_+ .
\]
By the blockwise definition of finite amplification, together with the same
trace-duality/Hille argument used for the Pettis-style clause, this integral is
precisely
\[
    \bigl(
        \sem{\bind(M(\seq e),x.S(x))}_{\omega_{\mathrm{out}}}
        \otimes\mathcal I_k
    \bigr)(Z).
\]
Thus the \(k\)-th amplification preserves positivity.  Since \(k\in\mathbb N\)
was arbitrary, the extension is completely positive.

Consequently, the positive-cone binding semantics is cone-linear,
trace-non-increasing, and its finite amplifications preserve positivity.  By
\Cref{lem:positive_cone_reduction}, it extends uniquely to a CPTNI superoperator
on \(\cT(\cH)\).

\emph{Measurability:}
Fix \(\rho\in\cT(\cH)_+\).  We prove that
\[
    \omega_{\mathrm{out}}
    \longmapsto
    \sem{\bind(M(\seq e),x.S(x))}^{+}_{\omega_{\mathrm{out}}}(\rho)
\]
is \((\Sigma_\Gamma,\operatorname{Borel}(\cT(\cH)_+))\)-measurable.  By
\Cref{thm:pettis}, it suffices to prove weak \(\Sigma_\Gamma\)-measurability
as a \(\cT(\cH)\)-valued map.

Let \(B\in\cB(\cH)\) be arbitrary and define
\[
    g_B(\omega_{\mathrm{out}},\nu)
    \coloneqq
    \tr\!\left(BF_\rho(\omega_{\mathrm{out}},\nu)\right).
\]
Since \(F_\rho\) is
\((\Sigma_\Gamma\otimes\Sigma_M,\operatorname{Borel}(\cT(\cH)_+))\)-measurable
and \(A\mapsto\tr(BA)\) is trace-norm continuous on \(\cT(\cH)\), the scalar
function \(g_B\) is
\((\Sigma_\Gamma\otimes\Sigma_M,\operatorname{Borel}(\bC))\)-measurable.
Moreover,
\[
    |g_B(\omega_{\mathrm{out}},\nu)|
    \leq
    \|B\|\,\|F_\rho(\omega_{\mathrm{out}},\nu)\|_1,
\]
and the trace estimate above gives, for every
\(\omega_{\mathrm{out}}\in\Omega_\Gamma\),
\[
    \int_{\Omega_M}
    |g_B(\omega_{\mathrm{out}},\nu)|\,d\mu_M(\nu)
    \leq
    \|B\|\,\tr(\rho)<\infty .
\]
Therefore, applying Tonelli's theorem and the measurability theorem for
parameterized integrals to the positive and negative parts of the real and
imaginary parts of \(g_B\), the marginal function
\[
    \omega_{\mathrm{out}}
    \longmapsto
    \int_{\Omega_M}g_B(\omega_{\mathrm{out}},\nu)\,d\mu_M(\nu)
\]
is \((\Sigma_\Gamma,\operatorname{Borel}(\bC))\)-measurable
\cite[Proposition 5.2.1]{cohn2013measure}.  By the Pettis-style identity
\eqref{eq:bind-pettis-clause}, this marginal is exactly
\[
    \omega_{\mathrm{out}}
    \longmapsto
    \tr\!\left(
        B\,\sem{\bind(M(\seq e),x.S(x))}^{+}_{\omega_{\mathrm{out}}}(\rho)
    \right).
\]
Since \(B\in\cB(\cH)\) was arbitrary, the map
\[
    \omega_{\mathrm{out}}
    \longmapsto
    \sem{\bind(M(\seq e),x.S(x))}^{+}_{\omega_{\mathrm{out}}}(\rho)
\]
is weakly \(\Sigma_\Gamma\)-measurable.  Hence, by \Cref{thm:pettis}, it is
\((\Sigma_\Gamma,\operatorname{Borel}(\cT(\cH)))\)-measurable.  Since the map
takes values in \(\cT(\cH)_+\), and
\(\operatorname{Borel}(\cT(\cH)_+)\) is the subspace Borel \(\sigma\)-algebra, it
is also \((\Sigma_\Gamma,\operatorname{Borel}(\cT(\cH)_+))\)-measurable.

This proves the required positive-input measurability for the binding construct.
The corresponding measurability for arbitrary \(X\in\cT(\cH)\) is supplied
uniformly by \Cref{lem:positive_cone_reduction}.

\paragraph*{$\bullet$ \textbf{While Loop ($\mathbf{while}\ M(\seq e)[\seq q]=1\ \mathbf{do}\ S\ \mathbf{od}$):}}
We define the loop semantics by finite unrolling approximants.  Fix
$\omega\in\Omega_\Gamma$.  By the structural induction hypothesis, the body
semantics admits a CPTNI extension
\[
    \sem{S}_\omega:\cT(\cH)\to\cT(\cH).
\]
We recursively define maps $\mathcal F_{n,\omega}:\cT(\cH)\to\cT(\cH)$ by
\[
    \mathcal F_{0,\omega}\coloneqq \bot
\]
and, for $n\geq 0$,
\begin{equation}
\label{eq:induction-def-of-while-semantics}
    \mathcal F_{n+1,\omega}(X)
    \coloneqq
    M_{\omega,0}XM_{\omega,0}^\dagger
    +
    \mathcal F_{n,\omega}
    \bigl(
        \sem{S}_\omega(M_{\omega,1}XM_{\omega,1}^\dagger)
    \bigr),
    \qquad
    X\in\cT(\cH).
\end{equation}
The index convention is that $\mathcal F_{0,\omega}$ is the zero
approximant, while $\mathcal F_{n+1,\omega}$ accounts for terminating paths
with at most $n$ executions of the loop body.

We shall prove below that each $\mathcal F_{n,\omega}$ is CPTNI and that the
chain is increasing in the pointwise L\"owner order. The loop denotation is then defined as the supremum of this chain in
$\cQ(\cH)$:
\[
    \mathcal F_\omega
    \coloneqq
    \sup_{n\in\mathbb N}\mathcal F_{n,\omega}.
\]
The positive-cone semantics of the loop will be the restriction of
$\mathcal F_\omega$ to $\cT(\cH)_+$.

\emph{Physical Validity (CPTNI):}
By the structural induction hypothesis applied to the loop body $S$, for
each fixed $\omega\in\Omega_\Gamma$ the body semantics admits a CPTNI
extension
\[
    \sem{S}_\omega:\cT(\cH)\to\cT(\cH).
\]
For this fixed $\omega$, a straightforward induction on $n$ shows that each
finite approximant $\mathcal F_{n,\omega}$ belongs to $\cQ(\cH)$.  Indeed,
$\mathcal F_{0,\omega}=\bot$ is the zero CPTNI map.  If
$\mathcal F_{n,\omega}$ is CPTNI, then
\[
    \mathcal F_{n+1,\omega}
    =
    \mathcal M_0
    +
    \mathcal F_{n,\omega}\circ \sem{S}_\omega\circ \mathcal M_1,
    \qquad
    \mathcal M_b(X)\coloneqq M_{\omega,b}XM_{\omega,b}^\dagger ,
\]
is completely positive, since it is a sum and composition of completely
positive maps.  It is also trace-non-increasing: for
$\rho\in\cT(\cH)_+$,
\[
    \tr(\mathcal F_{n+1,\omega}(\rho))
    \le
    \tr(M_{\omega,0}\rho M_{\omega,0}^\dagger)
    +
    \tr(M_{\omega,1}\rho M_{\omega,1}^\dagger)
    =
    \tr(\rho),
\]
where the inequality uses the TNI property of $\mathcal F_{n,\omega}$ and
$\sem{S}_\omega$, and the equality uses
$M_{\omega,0}^\dagger M_{\omega,0}+M_{\omega,1}^\dagger M_{\omega,1}=I$.  Hence
$\mathcal F_{n+1,\omega}$ is CPTNI, completing the induction.

Next, we prove by induction on $n$ that the finite approximants form an
increasing chain in the pointwise L\"owner order.  The base case is immediate:
for every $\rho\in\cT(\cH)_+$,
\[
    \mathcal F_{0,\omega}(\rho)=0
    \sqleq
    M_{\omega,0}\rho M_{\omega,0}^\dagger
    =
    \mathcal F_{1,\omega}(\rho).
\]
For the inductive step, assume that
\[
    \mathcal F_{n,\omega}(\rho)
    \sqleq
    \mathcal F_{n+1,\omega}(\rho),
    \qquad
    \forall \rho\in\cT(\cH)_+ .
\]
Since $\sem{S}_\omega$ is positive, the operator
\(\sem{S}_\omega(M_{\omega,1}\rho M_{\omega,1}^\dagger)\)
belongs to $\cT(\cH)_+$.  Applying the induction hypothesis to this positive
input gives
\[
\begin{aligned}
    \mathcal F_{n+2,\omega}(\rho)
    &=
    M_{\omega,0}\rho M_{\omega,0}^\dagger
    +
    \mathcal F_{n+1,\omega}
    \bigl(\sem{S}_\omega(M_{\omega,1}\rho M_{\omega,1}^\dagger)\bigr)       \\
    &\sqgeq
    M_{\omega,0}\rho M_{\omega,0}^\dagger
    +
    \mathcal F_{n,\omega}
    \bigl(\sem{S}_\omega(M_{\omega,1}\rho M_{\omega,1}^\dagger)\bigr)       \\
    &=
    \mathcal F_{n+1,\omega}(\rho).
\end{aligned}
\]
Thus, by induction,
\[
    \mathcal F_{0,\omega}
    \sqsubseteq
    \mathcal F_{1,\omega}
    \sqsubseteq
    \cdots
    \sqsubseteq
    \mathcal F_{n,\omega}
    \sqsubseteq
    \cdots
\]
is an increasing chain in $\cQ(\cH)$ with respect to the pointwise
L\"owner order $\sqsubseteq$.

There is a small order-theoretic subtlety here.  The above pointwise
L\"owner order is the order used in Lemma~\ref{lem:cpo_domain}, and it is
sufficient for constructing the loop denotation.  However, the approximant
chain is in fact increasing in the stronger complete-positivity order as
well.  Indeed,
\[
    \mathcal F_{1,\omega}-\mathcal F_{0,\omega}
    =
    \mathcal M_0,
    \qquad
    \mathcal M_b(X)\coloneqq M_{\omega,b}XM_{\omega,b}^\dagger ,
\]
is completely positive, and for every $n\geq 0$,
\[
    \mathcal F_{n+2,\omega}-\mathcal F_{n+1,\omega}
    =
    (\mathcal F_{n+1,\omega}-\mathcal F_{n,\omega})
    \circ
    \sem{S}_\omega
    \circ
    \mathcal M_1 .
\]
Hence a second induction on $n$ shows that
\[
    \mathcal F_{0,\omega}
    \sqsubseteq_{\mathrm{cp}}
    \mathcal F_{1,\omega}
    \sqsubseteq_{\mathrm{cp}}
    \cdots
    \sqsubseteq_{\mathrm{cp}}
    \mathcal F_{n,\omega}
    \sqsubseteq_{\mathrm{cp}}
    \cdots .
\]
By Remark~\ref{rem:cp_order_cpo}, the CP-order supremum of this chain agrees
with the pointwise L\"owner supremum.  We therefore use the pointwise
$\omega$-CPO structure of Lemma~\ref{lem:cpo_domain} to define the loop
semantics, while the CP-order observation records that no complete-positivity
structure is lost in this construction.

By Lemma~\ref{lem:cpo_domain}, the supremum
\[
    \mathcal F_{\omega}
    \coloneqq
    \sup_{n\in\mathbb N}\mathcal F_{n,\omega}
\]
exists in $\cQ(\cH)$.  We define the positive-cone semantics of the loop by
\[
    \sem{\mathbf{while}\ M(\seq e)[\seq q]=1\ \mathbf{do}\ S\ \mathbf{od}}^{+}_\omega
    (\rho)
    \coloneqq
    \mathcal F_\omega(\rho),
    \qquad
    \rho\in\cT(\cH)_+ .
\]
Since $\mathcal F_\omega\in\cQ(\cH)$, its restriction to $\cT(\cH)_+$ is
cone-linear, trace-non-increasing, and its finite amplifications preserve
positivity.  Hence it satisfies the hypotheses of
\Cref{lem:positive_cone_reduction}. 

\emph{Measurability:}
Fix $\rho\in\cT(\cH)_+$.  We first prove by induction on $n$ that, for every
fixed $\rho'\in\cT(\cH)_+$, the map
\[
    \omega\longmapsto \mathcal F_{n,\omega}(\rho')
\]
is $(\Sigma_\Gamma,\operatorname{Borel}(\cT(\cH)_+))$-measurable.

For $n=0$, this is the constant zero map.  Suppose the assertion holds for $n$. 
For fixed $\rho'\in\cT(\cH)_+$, the recurrence gives
\[
    \mathcal F_{n+1,\omega}(\rho')
    =
    M_{\omega,0}\rho' M_{\omega,0}^\dagger
    +
    \mathcal F_{n,\omega}
    \bigl(
        \sem{S}^{+}_\omega(M_{\omega,1}\rho' M_{\omega,1}^\dagger)
    \bigr).
\]
By Lemma~\ref{lem:operator_sandwich_measurability} and the strong operator measurability of $\omega\mapsto M_{\omega,0}$, the first term is measurable. Set
\[
    \rho''_\omega
    \coloneqq
    \sem{S}^{+}_\omega(M_{\omega,1}\rho' M_{\omega,1}^\dagger).
\]
By Lemma~\ref{lem:operator_sandwich_measurability}, the map
$\omega\mapsto M_{\omega,1}\rho' M_{\omega,1}^\dagger$ is measurable, and by
the structural induction hypothesis for $S$ together with the
Carath\'eodory joint-measurability theorem \cite[Lemma 4.51]{aliprantis2006infinite}, the map
$\omega\mapsto \rho''_\omega$ is
$(\Sigma_\Gamma,\operatorname{Borel}(\cT(\cH)_+))$-measurable.

It remains to prove that
\[
    \omega\longmapsto
    \mathcal F_{n,\omega}(\rho''_\omega)
\]
is measurable.  Define
\[
    h:\Omega_\Gamma\to\Omega_\Gamma\times\cT(\cH)_+,
    \qquad
    h(\omega)\coloneqq(\omega,\rho''_\omega).
\]
The first component $\omega\mapsto\omega$ is
$(\Sigma_\Gamma,\Sigma_\Gamma)$-measurable, and the second component
$\omega\mapsto\rho''_\omega$ is
$(\Sigma_\Gamma,\operatorname{Borel}(\cT(\cH)_+))$-measurable.  Hence, by
\cite[Lemma 4.49]{aliprantis2006infinite}, the pair map $h$ is
\(
    \bigl(
        \Sigma_\Gamma,
        \Sigma_\Gamma\otimes\operatorname{Borel}(\cT(\cH)_+)
    \bigr)
\)-measurable.

Next define
\[
    \Phi_n:\Omega_\Gamma\times\cT(\cH)_+\to\cT(\cH)_+,
    \qquad
    \Phi_n(\omega,\sigma)\coloneqq
    \mathcal F_{n,\omega}(\sigma).
\]
We verify the Carath\'eodory hypotheses:

\begin{itemize}
    \item By the induction hypothesis on $n$, for each fixed $\sigma\in\cT(\cH)_+$, the map
    \[
        \omega\longmapsto \Phi_n(\omega,\sigma)
        =
        \mathcal F_{n,\omega}(\sigma)
    \]
    is $(\Sigma_\Gamma,\operatorname{Borel}(\cT(\cH)_+))$-measurable.

    \item For each fixed $\omega\in\Omega_\Gamma$, the map
    \[
        \sigma\longmapsto \Phi_n(\omega,\sigma)
        =
        \mathcal F_{n,\omega}(\sigma)
    \]
    is trace-norm continuous on $\cT(\cH)_+$.  Indeed,
    $\mathcal F_{n,\omega}$ is the restriction of a bounded linear CPTNI map on
    $\cT(\cH)$.
\end{itemize}

Since $\cT(\cH)_+$ is separable and metrizable in the trace-norm subspace
topology, the Carath\'eodory joint measurability theorem
\cite[Lemma~4.51]{aliprantis2006infinite} gives that $\Phi_n$ is
\(
    \bigl(
        \Sigma_\Gamma\otimes\operatorname{Borel}(\cT(\cH)_+),
        \operatorname{Borel}(\cT(\cH)_+)
    \bigr)
\)-measurable.

Now we verify the pullback explicitly.  Let
$W\subseteq\cT(\cH)_+$ be relatively trace-norm open.  Since $\Phi_n$ is jointly
measurable,
\[
    \Phi_n^{-1}(W)
    \in
    \Sigma_\Gamma\otimes\operatorname{Borel}(\cT(\cH)_+).
\]
Since $h$ is measurable into this product measurable space,
\[
\begin{aligned}
    \{\omega\in\Omega_\Gamma
      \mid
      \mathcal F_{n,\omega}(\rho''_\omega)\in W\}
    &=
    (\Phi_n\circ h)^{-1}(W)        \\
    &=
    h^{-1}\bigl(\Phi_n^{-1}(W)\bigr)
    \in \Sigma_\Gamma .
\end{aligned}
\]
Therefore
\[
    \omega\longmapsto
    \mathcal F_{n,\omega}(\rho''_\omega)
\]
is $(\Sigma_\Gamma,\operatorname{Borel}(\cT(\cH)_+))$-measurable.
Adding the measurable term $\omega\mapsto M_{\omega,0}\rho'M_{\omega,0}^\dagger$ preserves measurability, so
\[
    \omega\longmapsto
    \mathcal F_{n+1,\omega}(\rho')
\]
is also $(\Sigma_\Gamma,\operatorname{Borel}(\cT(\cH)_+))$-measurable.
This completes the induction on $n$.

We now pass to the limit.  By the definition of $\mathcal F_\omega$,
\[
    \mathcal F_\omega(\rho)
    =
    \sup_{n\in\mathbb N}\mathcal F_{n,\omega}(\rho).
\]
For $\rho\in\pardensity{\cH}$, Lemma~\ref{lem:cpo_domain} gives trace-norm
convergence
\[
    \mathcal F_{n,\omega}(\rho)
    \xrightarrow[n\to\infty]{\|\cdot\|_1}
    \mathcal F_\omega(\rho).
\]
For a general $\rho\in\cT(\cH)_+$, the same conclusion follows by scaling:
if $\rho\neq0$, then $\rho/\tr(\rho)\in\pardensity{\cH}$, and linearity gives
\[
    \mathcal F_{n,\omega}(\rho)
    =
    \tr(\rho)\,
    \mathcal F_{n,\omega}
    \!\left(\frac{\rho}{\tr(\rho)}\right)
    \longrightarrow
    \tr(\rho)\,
    \mathcal F_\omega
    \!\left(\frac{\rho}{\tr(\rho)}\right)
    =
    \mathcal F_\omega(\rho)
\]
in trace norm.  The case $\rho=0$ is trivial.

Thus
\[
    \omega\longmapsto
    \sem{\mathbf{while}\ M(\seq e)[\seq q]=1\ \mathbf{do}\ S\ \mathbf{od}}^{+}_\omega
    (\rho)
    =
    \mathcal F_\omega(\rho)
\]
is the pointwise trace-norm limit of the maps
\(
    \omega\longmapsto \mathcal F_{n,\omega}(\rho)
\),
where these maps are all 
$(\Sigma_\Gamma,\operatorname{Borel}(\cT(\cH)_+))$-measurable.
Since $\cT(\cH)_+$ is metrizable in the trace-norm subspace topology, the pointwise limit is again
$(\Sigma_\Gamma,\operatorname{Borel}(\cT(\cH)_+))$-measurable
by \cite[Lemma~4.29]{aliprantis2006infinite}.  This proves the required positive-input measurability for the while loop.
\end{proof}

\begin{remark}[Programs without free classical parameters]
After the structural semantics has been defined, a program phrase whose free
classical parameter context is empty is interpreted over the one-point
meta-parameter space.  Its denotation is therefore an ordinary CPTNI
superoperator.  In a binding construct $\bind(M(\seq e),x.S(x))$, the bound
variable $x$ is local to the continuation and does not count as a free
classical parameter of the whole binding phrase.
\end{remark}

\subsubsection*{Correspondence Theorem}

The structural theorem above shows that every well-formed program phrase
denotes a physically valid quantum operation, once the classical
meta-parameters are fixed.  We now record the corresponding closed-program
statement and the converse expressiveness result.

For a program phrase \(S\), let \(\operatorname{Meta}(S)\) denote the set of
its free classical meta-variables, namely the classical parameters occurring in
parameter expressions of primitive operations and not bound by any measurement
binding construct.  We say that \(S\) is \emph{classically closed} if
\[
    \operatorname{Meta}(S)=\emptyset .
\]
Equivalently, \(S\) is well-formed over the empty classical context
\[
    \emptyset\vdash S .
\]

\begin{theorem}[Correspondence between Closed Programs and Superoperators]
\label{thm:semantics_correspondence}
Let $\cH$ be a separable Hilbert space.
\begin{enumerate}
    \item \textbf{Soundness.}
    If $\emptyset\vdash S$ is a quantum program phrase with no free classical
    meta-parameters, then its denotation is a CPTNI superoperator:
    $\sem{S}\in\cQ(\cH)$.
    \item \textbf{Universality.}
    Conversely, under the primitive-signature convention, for every CPTNI superoperator $\cE\in\cQ(\cH)$, there exists
    a quantum program phrase $\emptyset\vdash S$ in the syntax such that
    $\sem{S}=\cE$.
\end{enumerate}
\end{theorem}

\begin{proof}
\textbf{Soundness.}
If $\emptyset\vdash S$, then the classical parameter space is the one-point
space.  By \Cref{thm:semantics_well_defined}, the unique denotation of \(S\)
is a CPTNI superoperator on \(\cT(\cH)\).  We write this denotation simply as
\(\sem{S}\).  Hence \(\sem{S}\in\cQ(\cH)\).

\textbf{Universality.}
Let \(\cE\in\cQ(\cH)\).  By \Cref{thm:kraus_rep}, there is a finite or countably
infinite index set \(J\) and operators
\[
    \{K_j\}_{j\in J}\subseteq\cB(\cH)
\]
such that
\[
    \cE(X)=\sum_{j\in J}K_jXK_j^\dagger,
    \qquad
    X\in\cT(\cH),
\]
with trace-norm convergence, and
\[
    \sum_{j\in J}K_j^\dagger K_j\sqleq I
\]
in the strong operator topology.

Let $A\coloneqq \sum_{j\in J}K_j^\dagger K_j$
be this strong limit, and define the residual Kraus operator
\[
    K_\bot\coloneqq (I-A)^{1/2}.
\]
Set
\[
    \Omega_M\coloneqq J\sqcup\{\bot\},
    \qquad
    \Sigma_M\coloneqq 2^{\Omega_M},
    \qquad
    \mu_M\coloneqq \text{counting measure}.
\]
Define a discrete instrument \(M\) by
\[
    M_j\coloneqq K_j \quad (j\in J),
    \qquad
    M_\bot\coloneqq K_\bot .
\]
Then
\[
    \sum_{\xi\in\Omega_M}M_\xi^\dagger M_\xi
    =
    \sum_{j\in J}K_j^\dagger K_j+K_\bot^\dagger K_\bot
    =
    A+(I-A)
    =
    I
\]
in the strong operator topology.  Thus \(M\) is a normalized discrete
instrument. By the primitive-signature convention, this normalized discrete instrument may be used as a primitive measurement symbol of the language. Since \(\Omega_M\) is finite or countable with the discrete
\(\sigma\)-algebra, its measurability requirements are automatic.  This is the
counting-measure special case of the general \(\bind\)-semantics, where the
Bochner integral reduces to a trace-norm convergent series.

Choose a fresh bound classical parameter variable \(x\) with value space
\((\Omega_M,\Sigma_M,\mu_M)\).  Define the program
\[
    S_{\cE}
    \coloneqq
    \bind(M,x.S(x)),
\]
where the continuation is specified by the finite measurable partition
\[
    \Omega_M=J\sqcup\{\bot\}
\]
as
\[
    S(x)=
    \begin{cases}
        \mathbf{skip}, & x\in J,\\
        \mathbf{abort}, & x=\bot .
    \end{cases}
\]
This is a valid finite continuation specification: although \(J\) may be
infinite, only two program fragments are used.  The variable \(x\) is bound
locally by the \(\bind\)-construct and is not a free classical meta-parameter of
the whole program.  Since the instrument is not parameterized by any classical
expression, \(S_{\cE}\) is well-formed over the empty classical context.

For every \(\rho\in\cT(\cH)_+\), using the convention after
\Cref{thm:semantics_well_defined} that \(\sem{S}\) also denotes the restriction
of the CPTNI denotation to positive inputs, we compute
\[
\begin{aligned}
    \sem{S_{\cE}}(\rho)
    &=
    \int_{\Omega_M}
    \sem{S(\xi)}
    \bigl(M_\xi\rho M_\xi^\dagger\bigr)
    \,d\mu_M(\xi)                                      \\
    &=
    \sum_{j\in J}
    \sem{\mathbf{skip}}
    \bigl(K_j\rho K_j^\dagger\bigr)
    +
    \sem{\mathbf{abort}}
    \bigl(K_\bot\rho K_\bot^\dagger\bigr)              \\
    &=
    \sum_{j\in J}K_j\rho K_j^\dagger
    =
    \cE(\rho).
\end{aligned}
\]
The series converges in trace norm by \Cref{thm:kraus_rep}.  Thus
\(\sem{S_{\cE}}\) and \(\cE\) agree on \(\cT(\cH)_+\).  Since both are bounded
linear maps on \(\cT(\cH)\), and every trace-class operator is a linear
combination of four positive trace-class operators, they agree on all of
\(\cT(\cH)\).  Therefore
\[
    \sem{S_{\cE}}=\cE .
\]
This proves universality.
\end{proof}
\section{Logic: Predicates and Weakest Preconditions}
\label{sec:logic}

\subsection{Quantum Predicates and Logical Validity}
\label{subsec:predicates_validity}

To rigorously verify quantum programs in infinite-dimensional spaces, we require a notion of predicates that generalizes bounded observables to include unbounded quantities and strict subspace constraints.

\paragraph*{Predicates as Linear Relations.}
We define the set of quantum predicates, denoted by $\Pred$, using the language of Linear Relations (LRs), which provides the necessary algebraic closure for logical operations.

\begin{definition}[Quantum Predicates]
    A \emph{quantum predicate} $P \in \Pred$ is a \emph{positive self-adjoint linear relation} on the global Hilbert space $\cH$.
    
    This definition generalizes standard quantum logic in two ways:
    \begin{enumerate}
        \item \textbf{Unboundedness:} It includes unbounded observables (e.g., energy, position).
        \item \textbf{Singularity:} It includes relations with non-dense domains, allowing us to represent \emph{hard constraints} (e.g., states confined strictly to a subspace).
    \end{enumerate}
    
    The set $\Pred$ is equipped with extended L\"owner order $\sqsubseteq$. As shown in \Cref{thm:monotone_convergence}, $(\Pred, \sqsubseteq)$ forms an $\omega$-CPO.
\end{definition}

\paragraph*{Computational Representation (Quadratic Forms).}
While LRs provide the structural definition, calculation of weakest preconditions requires summing and composing these objects. For this purpose, we utilize the isomorphism established in Theorem~\ref{thm:kato_first}:
\begin{displayquote}
    Every predicate $P \in \Pred$ is uniquely characterized by a \emph{positive closed quadratic form} $\mathfrak{t}_P : \cH \to [0, +\infty]$.
\end{displayquote}
In the remainder of this paper, we will interchangeably refer to the predicate $P$ and its form $\mathfrak{t}_P$. Specifically, we use the form representation to handle sums and limits via the Representation Theorem and Monotone Convergence Theorem.

\paragraph*{Expectation Values via Extended Trace.}
The logical validity of a program depends on the expected value of predicates. We employ the \emph{extended trace} defined in \Cref{def:extended_trace}.
For any $P \in \Pred$ and state $\rho$, the value $\Tr(P \rho) \in [0, +\infty]$ is computed utilizing the associated form $\mathfrak{t}_P$.
This approach circumvents the difficulty of defining traces directly on multivalued relations.

\begin{definition}[Partial Correctness Validity]\label{def:partial-correctness}
    A triple $\{P\}S\{Q\}$ is valid, denoted $\models_{par} \{P\}S\{Q\}$, if for all $\rho \in \pardensity{\cH}$:
    \[
        \Tr(P \rho) \ge \Tr(Q \sem{S}(\rho)).
    \]
\end{definition}

This "greater than or equal to" inequality marks our framework as an \emph{Upper-Bound Logic}. It physically requires the precondition $P$ to act as a safe upper bound for the initial resources or expected yield, strictly covering the actual outcome satisfying $Q$ after the program's execution.

\subsection{Weakest Preconditions via Form Pullbacks}
\label{subsec:wp}

To establish a calculus for partial correctness, we first define the weakest liberal precondition purely semantically, based on the validity of Hoare triples. In our upper-bound logic, a ``weaker'' precondition corresponds to a smaller form in the extended L\"owner order (requiring less initial resource).

\begin{definition}[Semantic Weakest Liberal Precondition]
    \label{def:wlp}
    For any quantum program $S$ and post-condition predicate $Q \in \Pred$, the semantic weakest liberal precondition $\wlp(S, Q)$ is defined as the minimum valid precondition in $\Pred$. That is, $\wlp(S, Q)$ satisfies:
    \begin{enumerate}[(1)]
        \item \textbf{Validity:} $\models_{par} \{\wlp(S, Q)\} S \{Q\}$.
        \item \textbf{Minimality:} For any $P \in \Pred$, if $\models_{par} \{P\} S \{Q\}$, then $\wlp(S, Q) \sqsubseteq P$.
    \end{enumerate}
\end{definition}

We now define the weakest liberal precondition calculus. Since standard operator composition $S^\dagger P S$ is ill-defined for unbounded operators, we define Weakest liberal preconditions $\wlp$ directly at the level of quadratic forms.

\begin{theorem}[Structure of Weakest Liberal Preconditions]
    \label{thm:wlp_well_defined}
    Let $S$ be a quantum program with semantics $\sem{S}(\rho) = \sum_{i=1}^{\infty} K_i \rho K_i^\dagger$ (where $\{K_i\}$ is the countable Kraus representation from Theorem~\ref{thm:kraus_rep}).
    For any predicate $Q \in \Pred$, there exists a unique predicate $\wlp(S, Q) \in \Pred$ whose form is defined by:
    \begin{equation}\label{eq:wlp-def}
        \QF_{\wlp(S, Q)}[\u] \coloneqq \sum_{i=1}^\infty \QF_Q[K_i \u], \quad \forall \u \in \cH.
    \end{equation}
    The domain of this form is $\Dom(\QF_{\wlp}) = \{\u \in \cH \mid K_i \u \in \Dom(\QF_Q) \forall i \text{ and } \sum \QF_Q[K_i \u] < \infty\}$. The functional $\QF_{\wlp(S, Q)}$ defined above satisfies:
    \begin{enumerate}[(1)]
        \item \textbf{Closure:} $\QF_{\wlp(S, Q)}$ is a positive, closed quadratic form.
        \item \textbf{Independence:} $\QF_{\wlp(S, Q)}$ depends only on the semantic map $\sem{S}$ and is independent of the choice of the Kraus representation $\{K_i\}$.
        \item \textbf{Exact Duality:} For any partial density operator $\rho \in \pardensity{\cH}$, $\Tr(\wlp(S, Q) \rho) = \Tr(Q \sem{S}(\rho))$. Furthermore, $\wlp(S, Q)$ is the \emph{unique} predicate in $\Pred$ satisfying this exact trace equality for all $\rho$.
        \item \textbf{Semantic Minimality:} $\QF_{\wlp(S, Q)}$ exactly satisfies \Cref{def:wlp}, which concurrently establishes its uniqueness.
    \end{enumerate}
\end{theorem}

Notice that this construction merely evaluates the expected yield of the post-condition $Q$ over the surviving execution branches, without explicitly penalizing non-terminating traces. Therefore, it strictly represents the weakest liberal precondition ($\wlp$) corresponding to partial correctness.

\begin{proof}
    \textbf{1. Closure.}
    It is clear that $\wlp(S,Q)$ as defined by \Cref{eq:wlp-def} is a non-negative functional.
    We first verify that it is indeed a positive quadratic form in the sense of \Cref{def:Positive_QF}.
    Each term $f_i(\u) = \QF_Q[K_i \u]$ is a positive quadratic form, as it is the pullback of the form $\QF_Q$ by the bounded linear operator $K_i$.
    Consequently, every finite partial sum $S_N(\u) = \sum_{i=1}^N f_i(\u)$ is also a positive quadratic form.
    In particular, each $S_N$ satisfies the parallelogram law, and its finite-value domain $\Dom(S_N)$ is a linear subspace of $\cH$.

    The target form is the increasing pointwise supremum $\QF_{\wlp}[\u] = \sup_{N} S_N(\u)$.
    For any $\u, \v \in \Dom(\QF_{\wlp})$, we have $S_N(\u), S_N(\v) < \infty$ for all $N$, and by the parallelogram law for $S_N$,
    \[
        S_N(\u+\v) \le 2 S_N(\u) + 2 S_N(\v) \le 2 \QF_{\wlp}[\u] + 2 \QF_{\wlp}[\v] < \infty.
    \]
    Taking the supremum over $N$ shows $\QF_{\wlp}[\u+\v] < \infty$, so $\Dom(\QF_{\wlp})$ is a linear subspace.
    Restricting to $\Dom(\QF_{\wlp})$, the parallelogram law passes to the supremum:
    \[
        \QF_{\wlp}[\u+\v] + \QF_{\wlp}[\u-\v] = 2 \QF_{\wlp}[\u] + 2 \QF_{\wlp}[\v].
    \]
    By the polarization identity, we recover the associated sesquilinear form
    \[
        \mathfrak{s}(\u,\v) = \frac{1}{4}\bigl(\QF_{\wlp}[\u+\v] - \QF_{\wlp}[\u-\v] - i\QF_{\wlp}[\u+i\v] + i\QF_{\wlp}[\u-i\v]\bigr),
    \]
    which satisfies $\QF_{\wlp}[\u] = \mathfrak{s}(\u,\u)$. Hence $\QF_{\wlp}$ is a positive quadratic form.

    To prove that it is \emph{closed}, according to Theorem~\ref{thm:closed_iff_lsc} it suffices to show that $\QF_{\wlp}$ is lower semicontinuous (LSC) on $\cH$.
    Consider the term $f_i(\u) \coloneqq \QF_Q[K_i \u]$. This functional is the composition of the LSC form $\QF_Q$ with the continuous operator $K_i$, and is therefore LSC.
    The target form is the countable sum $\QF_{\wlp}[\u] = \sum_{i=1}^\infty f_i(\u)$.
    
    We prove directly that the sum of non-negative LSC functions is LSC. Let $S_N(\u) = \sum_{i=1}^N f_i(\u)$ be the partial sums. Since the sum of two LSC functions is LSC (by the superadditivity $\liminf_k (f_i(\u_k)+f_j(\u_k)) \ge \liminf_k f_i(\u_k) + \liminf_k f_j(\u_k)$), each $S_N$ is LSC.
    Since $f_i \ge 0$, the sequence $S_N$ is non-decreasing, and $\QF_{\wlp}(\u) = \sup_{N \in \mathbb{N}} S_N(\u)$.
    Recall that a function $g$ is LSC if and only if the strict suplevel set $\{\u \mid g(\u) > \lambda\}$ is open for all $\lambda \in \mathbb{R}$.
    Observe that:
    \[
        \{\u \mid \sup_N S_N(\u) > \lambda\} = \bigcup_{N \in \mathbb{N}} \{\u \mid S_N(\u) > \lambda\}.
    \]
    Since each $S_N$ is LSC, the sets on the RHS are open. Since the union of open sets is open, the set on the LHS is open. Thus, $\QF_{\wlp}$ is LSC and defines a closed form.

    \textbf{2. Independence.}
    Let $\{K_i\}$ and $\{L_j\}$ be two countable Kraus representations of $\sem{S}$. We show that $\sum_i \QF_Q[K_i \u] = \sum_j \QF_Q[L_j \u]$ for any $u \in \cH$.

    We invoke \Cref{prop:spectral_approximation} (Spectral Approximation). There exists a sequence of \emph{bounded} predicates $\{Q_n\}_{n \in \mathbb{N}} $ such that $Q_n \sqsubseteq Q_{n+1}$ and $\QF_Q[\v] = \sup_n \QF_{Q_n}[\v]$ for all $\v$.
    Let $A_n \in \mathcal{B}(\cH)$ be the bounded operator associated with each $Q_n$ (i.e., $\QF_{Q_n}[\v] = \< \v, A_n \v \>$ holds for any $\v\in\cH$).

    Substituting this into the sum:
    \[
        \sum_i \QF_Q[K_i \u] = \sum_i \sup_n \QF_{Q_n}[K_i \u] = \sum_i \sup_n \< K_i \u, A_n K_i \u \>.
    \]
    Since the terms $a_{i,n} = \< K_i \u, A_n K_i \u \>$ are non-negative and non-decreasing in $n$, the Monotone Convergence Theorem for series allows us to interchange the summation and the supremum (see standard analysis texts, e.g., \cite{rudin1987real}):
    \[
        \sum_i \sup_n \< \u, K_i^\dagger A_n K_i \u \> = \sup_n \sum_i \< \u, K_i^\dagger A_n K_i \u \>.
    \]
    The inner term $\sum_i K_i^\dagger A_n K_i$ represents the Heisenberg picture evolution of the \emph{bounded} operator $A_n$. For bounded operators, this sum is uniquely determined by the trace duality $\tr(\sem{S}(\rho) A_n)$ and is independent of the Kraus decomposition. Thus:
    \[
        \sum_i K_i^\dagger A_n K_i = \sum_j L_j^\dagger A_n L_j.
    \]
    Reversing the order of limits for the second representation (again justified by MCT):
    \[
        \sup_n \sum_j \< \u, L_j^\dagger A_n L_j \u \> = \sum_j \sup_n \QF_{Q_n}[L_j \u] = \sum_j \QF_Q[L_j \u].
    \]
    This proves that the weakest precondition is unique.

    \textbf{3. Exact Duality.}
    Let $\rho \in \pardensity{\cH}$ have the spectral decomposition $\rho = \sum_k \lambda_k \sP_{\u_k}$.
    By the spectral approximation $Q = \bigsqcup_n Q_n$ used above, and the continuity of the trace over positive forms, we have:
    \[
        \Tr(Q \sem{S}(\rho)) = \sup_n \Tr\Big(Q_n \sem{S}(\rho)\Big).
    \]
    For the bounded operator $Q_n$, the trace evaluates standardly:
    \[
        \Tr\Big(Q_n \sem{S}(\rho)\Big) = \sum_k \lambda_k \Tr\Big(Q_n \sem{S}(\sP_{\u_k})\Big) = \sum_k \lambda_k \sum_i \QF_{Q_n}[K_i \u_k].
    \]
    Taking the supremum over $n$, we apply the Monotone Convergence Theorem to swap $\sup_n$ with the summations over $k$ and $i$ (since all terms $\lambda_k \QF_{Q_n}[K_i \u_k]$ are non-negative and monotonically increasing with $n$):
    \begin{align*}
        \Tr(Q \sem{S}(\rho)) &= \sup_n \sum_k \lambda_k \sum_i \QF_{Q_n}[K_i \u_k] \\
        &= \sum_k \lambda_k \sum_i \sup_n \QF_{Q_n}[K_i \u_k] \\
        &= \sum_k \lambda_k \sum_i \QF_Q[K_i \u_k].
    \end{align*}
    By \Cref{eq:wlp-def}, the inner sum is exactly $\QF_{\wlp(S, Q)}[\u_k]$. Thus, the expression simplifies to $\sum_k \lambda_k \QF_{\wlp(S, Q)}[\u_k]$, which is precisely $\Tr(\wlp(S, Q) \rho)$. This establishes the exact duality.

    To see why this equality uniquely determines the predicate, suppose another predicate $P \in \Pred$ also satisfies $\Tr(P \rho) = \Tr(Q \sem{S}(\rho))$ for all $\rho$. Then we must have $\Tr(P \rho) = \Tr(\wlp(S, Q) \rho)$. Restricting $\rho$ to an arbitrary pure state $\sP_{\u}$, we obtain $\QF_P[\u] = \QF_{\wlp(S, Q)}[\u]$ for all $\u \in \cH$. Since a closed positive quadratic form uniquely determines its corresponding linear relation in $\Pred$, it follows strictly that $P = \wlp(S, Q)$.

    \textbf{4. Semantic Minimality and Uniqueness.}
    We now prove that the constructed form satisfies both conditions of \Cref{def:wlp}.
    \begin{itemize}
        \item \emph{Validity:} By the Exact Duality established above, $\Tr(\wlp(S, Q)\rho) = \Tr(Q\sem{S}(\rho)) \ge \Tr(Q\sem{S}(\rho))$. The inequality is trivially satisfied, ensuring $\models_{par} \{\wlp(S, Q)\} S \{Q\}$.
        \item \emph{Minimality:} Suppose $P \in \Pred$ is any valid precondition such that $\models_{par} \{P\} S \{Q\}$. By validity, $\Tr(P\rho) \ge \Tr(Q\sem{S}(\rho))$ for all $\rho$. Substituting the exact duality into the RHS gives $\Tr(P\rho) \ge \Tr(\wlp(S, Q)\rho)$ for all $\rho$. Restricting $\rho$ to an arbitrary pure state $\sP_{\u}$, we obtain $\QF_P[\u] \ge \QF_{\wlp(S, Q)}[\u]$ for all $\u$, which algebraically means $\wlp(S, Q) \sqsubseteq P$.
    \end{itemize}
    Finally, since $\wlp(S, Q)$ satisfies the minimality condition, if any other predicate $P'$ also satisfies both validity and minimality, we would have $P' \sqsubseteq \wlp(S, Q)$ and $\wlp(S, Q) \sqsubseteq P'$. By the antisymmetry of the L\"owner order, $P' = \wlp(S, Q)$, ensuring semantic uniqueness.
\end{proof}

We define the predicate transformer $\wlp(S, \cdot): \Pred \to \Pred$ by this construction.

\subsection{Structural Properties of Weakest Preconditions}
\label{subsec:wp_properties}

The utility of this logic for verifying loops depends on how the weakest precondition transformer preserves the algebraic and topological structure of the predicate domain.

\begin{theorem}[Structural Properties of wlp]
    \label{thm:wp_properties}
    Let $S$ be a quantum program. The transformer $\wlp(S, \cdot): \Pred \to \Pred$ satisfies:
    \begin{enumerate}
        \item \textbf{Monotonicity:} If $P \sqsubseteq Q$, then $\wlp(S, P) \sqsubseteq \wlp(S, Q)$.
        \item \textbf{Countable Additivity ($\sigma$-Linearity):} For any countable family $\{P_j\}$ and $c_j \ge 0$:
        \[
            \wlp\left(S, \sum_{j=1}^\infty c_j P_j\right) = \sum_{j=1}^\infty c_j \wlp(S, P_j).
        \]
        (Here, the sum $\sum P_j$ corresponds to the pointwise sum of forms $\sum \QF_{P_j}$).
        \item \textbf{Scott-Continuity:} For any increasing chain of predicates $P_1 \sqsubseteq P_2 \sqsubseteq \dots$,
        \[
            \wlp\left(S, \bigsqcup_n P_n\right) = \bigsqcup_n \wlp(S, P_n).
        \]
        Here, $\bigsqcup_n P_n$ denotes the supremum in the CPO $(\Pred, \sqsubseteq)$, which is characterized by the pointwise supremum of forms: $\QF_{\bigsqcup P_n}[\u] = \sup_n \QF_{P_n}[\u]$.
    \end{enumerate}
\end{theorem}
\begin{proof}
    All proofs rely on the form-based definition $\QF_{\wlp(S, P)}[\u] = \sum_i \QF_P[K_i \u]$, where $\{K_i\}$ is a Kraus representation of $\sem{S}$.

    \textbf{1. Monotonicity.}
    Let $P \sqsubseteq Q$. By definition, $\Dom(\QF_Q) \subseteq \Dom(\QF_P)$ and $\QF_P[\v] \le \QF_Q[\v]$ for all $\v \in \Dom(\QF_Q)$.
    For the domain, if $\u \in \Dom(\QF_{\wlp(S, Q)})$, then $K_i \u \in \Dom(\QF_Q) \subseteq \Dom(\QF_P)$ for all $i$, and the sum converges. Thus $\Dom(\QF_{\wlp(S, Q)}) \subseteq \Dom(\QF_{\wlp(S, P)})$.
    For the value, since $\QF_P[K_i \u] \le \QF_Q[K_i \u]$, summing over $i$ preserves the inequality.

    \textbf{2. Countable Additivity.}
    Let $P = \sum_j c_j P_j$. By \Cref{thm:monotone_convergence}, the form $\QF_P$ is defined by $\QF_P[\v] = \sum_j c_j \QF_{P_j}[\v]$ and is a closed form.
    Then for any $\u$:
    \begin{align*}
        \QF_{\wlp(S, P)}[\u] &= \sum_i \QF_P[K_i \u] \\
                          &= \sum_i \left( \sum_j c_j \QF_{P_j}[K_i \u] \right).
    \end{align*}
    Since all terms $c_j \QF_{P_j}[K_i \u]$ are non-negative, by Tonelli's theorem for series (infinite associativity and commutativity of non-negative sums, \cite{rudin1987real}), we can exchange the order of summation:
    \begin{align*}
        \QF_{\wlp(S, P)}[\u] &= \sum_j c_j \left( \sum_i \QF_{P_j}[K_i \u] \right) \\
                          &= \sum_j c_j \QF_{\wlp(S, P_j)}[\u].
    \end{align*}
    This establishes the equality of forms.

    \textbf{3. Scott-Continuity.}
    Let $P_n$ be an increasing sequence with supremum $P = \bigsqcup_n P_n$. By \Cref{thm:monotone_convergence}, $\QF_P(\v) = \sup_n \QF_{P_n}(\v)$.
    Then:
    \[
        \QF_{\wlp(S, P)}[\u] = \sum_i \sup_n \QF_{P_n}[K_i \u].
    \]
    Since the sequence $\{\QF_{P_n}(K_i \u)\}_n$ is non-decreasing for each $i$, we can apply the Monotone Convergence Theorem for series (interchanging sum and supremum):
    \[
        \sum_i \sup_n \QF_{P_n}[K_i \u] = \sup_n \sum_i \QF_{P_n}[K_i \u] = \sup_n \QF_{\wlp(S, P_n)}[\u].
    \]
    Thus, $\wlp(S, P) = \bigsqcup_n \wlp(S, P_n)$.
\end{proof}

With the semantic well-definedness and structural properties of $\wlp$ established, we can now formulate the practical weakest liberal precondition calculus. This provides a syntax-driven, inductive method to compute $\wlp$ structurally, which is essential for automated verification and deductive reasoning.
\begin{definition}[Structural Rules for $\wlp$]
    \label{def:wlp_rules}
    For any quantum program $S$ and post-condition predicate $Q \in \Pred$, the syntactic weakest liberal precondition $\wlp(S, Q)$ is computed inductively over the structure of $S$ via its quadratic form $\QF_{\wlp(S, Q)}[\u]$:
    \begin{itemize}
        \item \textbf{Skip:} $\QF_{\wlp(\mathbf{skip}, Q)}[\u] \coloneqq \QF_Q[\u]$.

        \item \textbf{Abort:} $\QF_{\wlp(\mathbf{abort}, Q)}[\u] \coloneqq 0$.
        
        \item \textbf{Initialization} ($S \equiv q := 0$):
        Let $E_k = |0\rangle_q \langle k|_q \otimes I_{\text{env}}$ be the Kraus operators for initialization.
        \[
            \QF_{\wlp(S, Q)}[\u] \coloneqq \sum_{k} \QF_Q[E_k \u].
        \]
        
        \item \textbf{Unitary} ($S \equiv U(\seq{e})[\seq{q}]$):
        \[
            \QF_{\wlp(S, Q)}[\u] \coloneqq \QF_Q[U_\omega \u].
        \]
        
        \item \textbf{Sequential Composition} ($S \equiv S_1 ; S_2$):
        \[
            \QF_{\wlp(S_1 ; S_2, Q)}[\u] \coloneqq \QF_{\wlp(S_1, \wlp(S_2, Q))}[\u].
        \]
        
        \item \textbf{Continuous Measurement} ($S \equiv \mathbf{bind}(M(e), x. S(x))$):
        \[
        \QF_{\wlp(S, Q)}[\u] \coloneqq \int_{\Omega_M} \QF_{\wlp(S_{\nu}, Q)}[M_{\nu} \u] \, d\mu(\nu).
        \]
        
        \item \textbf{While Loop} ($S \equiv \mathbf{while} \ M \ \mathbf{do} \ S_{body} \ \mathbf{od}$):
        Let $M = \{M_0, M_1\}$, where $M_0$ exits the loop. We define a sequence of predicates $\{X_n\}_{n=0}^\infty$ representing the accumulated yield up to $n$ iterations. Let $X_0 \coloneqq 0$ (the identically zero form). For $n \ge 0$:
        \[
            \QF_{X_{n+1}}[\u] \coloneqq \QF_Q[M_0 \u] + \QF_{\wlp(S_{body}, X_n)}[M_1 \u].
        \]
        The weakest liberal precondition is the supremum of this sequence:
        \[
            \wlp(S, Q) \coloneqq \bigsqcup_{n=0}^\infty X_n, \quad \text{i.e., } \QF_{\wlp(S, Q)}[\u] = \sup_{n \ge 0} \QF_{X_n}[\u].
        \]
    \end{itemize}
\end{definition}

To guarantee the soundness of this calculus, we strictly prove that these inductive syntactic rules generate the unique predicate satisfying the Exact Duality equality from \Cref{thm:wlp_well_defined}.

\begin{theorem}[Soundness of the $\wlp$ Calculus]
    \label{thm:wlp_soundness}
    Under the standing well-formedness and measurability assumptions of
    \Cref{subsec:semantics}, every structural clause in
    \Cref{def:wlp_rules} defines a predicate
    \(P_{syn}\in\Pred\).  Moreover, for every
    \(\rho\in\pardensity{\cH}\),
    \begin{equation}\label{eq:syn-sem-trace-equality}
        \Tr(P_{syn}\rho)=\Tr(Q\sem{S}(\rho)).
    \end{equation}
    Consequently, by (3) of \Cref{thm:wlp_well_defined},
    \(P_{syn}=\wlp(S,Q)\).
\end{theorem}

\begin{proof}
    We prove a slightly stronger statement by structural induction on \(S\):
    the structurally generated object is a predicate in \(\Pred\), and \Cref{eq:syn-sem-trace-equality} holds for every positive trace-class input
    \(\rho\in\cT(\cH)_+\).  The extension from partial densities to arbitrary
    positive trace-class operators is harmless, because all semantic maps and
    all extended trace evaluations are positively homogeneous; if
    \(\rho\neq0\), apply the partial-density statement to
    \(\rho/\tr(\rho)\) and then rescale.

    To avoid circular notation inside the proof, write
    \(\wlp_{\mathrm{syn}}(T,R)\) for the predicate generated by the structural
    rules for a subprogram \(T\) and postcondition \(R\).  After the proof, this
    predicate is identified with the semantic \(\wlp(T,R)\) by uniqueness.

    \textbf{1. Basic operations.}
    For \(\mathbf{skip}\), the generated predicate is \(P_{syn}=Q\), and
    \(\sem{\mathbf{skip}}(\rho)=\rho\). Hence
    \[
        \Tr(P_{syn}\rho)=\Tr(Q\rho)
        =\Tr(Q\sem{\mathbf{skip}}(\rho)).
    \]

    For \(\mathbf{abort}\), the generated predicate is \(0\), and
    \(\sem{\mathbf{abort}}(\rho)=0\). Hence
    \[
        \Tr(P_{syn}\rho)=0
        =\Tr(Q\sem{\mathbf{abort}}(\rho)).
    \]

    For initialization \(q:=0\), let
    \(E_k=|0\rangle_q\langle k|_q\otimes I_{\mathrm{env}}\).  The generated
    form is
    \[
        \QF_{P_{syn}}[\u]
        =
        \sum_k \QF_Q[E_k\u].
    \]
    This is a countable sum of bounded pullbacks of the closed positive form
    \(\QF_Q\); hence it is again a closed positive quadratic form by the same
    monotone-form argument used in \Cref{thm:wlp_well_defined}.  Thus
    \(P_{syn}\in\Pred\).  Since
    \[
        \sem{q:=0}(\rho)=\sum_k E_k\rho E_k^\dagger,
    \]
    the exact-duality calculation in \Cref{thm:wlp_well_defined}, applied to
    this Kraus family, gives
    \[
        \Tr(P_{syn}\rho)
        =
        \Tr\!\left(Q\sum_k E_k\rho E_k^\dagger\right)
        =
        \Tr(Q\sem{q:=0}(\rho)).
    \]

    For a unitary command \(U(\seq e)[\seq q]\), the generated form is the
    bounded pullback
    \[
        \QF_{P_{syn}}[\u]=\QF_Q[U_\omega\u],
    \]
    hence it is a closed positive form and defines a predicate.  If
    \(\rho=\sum_j\lambda_j\sP_{\u_j}\) is a spectral decomposition of
    \(\rho\in\cT(\cH)_+\), then
    \[
    \begin{aligned}
        \Tr(P_{syn}\rho)
        &=
        \sum_j\lambda_j \QF_Q[U_\omega\u_j]                                      \\
        &=
        \Tr\!\left(Q\,U_\omega\rho U_\omega^\dagger\right)
        =
        \Tr(Q\sem{U_\omega}(\rho)).
    \end{aligned}
    \]

    \textbf{2. Sequential composition.}
    Let \(S\equiv S_1;S_2\), and set
    \[
        R\coloneqq \wlp_{\mathrm{syn}}(S_2,Q),
        \qquad
        P_{syn}\coloneqq \wlp_{\mathrm{syn}}(S_1,R).
    \]
    By the induction hypothesis for \(S_2\), \(R\in\Pred\).  Applying the
    induction hypothesis first to \(S_1\) with postcondition \(R\), and then to
    \(S_2\) with postcondition \(Q\), gives
    \[
    \begin{aligned}
        \Tr(P_{syn}\rho)
        &=
        \Tr\!\left(R\,\sem{S_1}(\rho)\right)                                      \\
        &=
        \Tr\!\left(Q\,\sem{S_2}(\sem{S_1}(\rho))\right)                            \\
        &=
        \Tr\!\left(Q\,\sem{S_1;S_2}(\rho)\right).
    \end{aligned}
    \]
    Thus the sequential rule is well-defined and sound.

    \textbf{3. Continuous measurement / binding.}
    Let
    \[
        T\equiv \mathbf{bind}(M(\seq e),x.S(x)),
    \]
    and omit the outer classical parameter in the notation.  The measurement
    outcome space is denoted by
    \((\Omega_M,\Sigma_M,\mu)\), the Kraus-density operator at outcome \(\nu\)
    by \(M_\nu\), and the continuation with \(x=\nu\) by \(S_\nu\).

    For each \(\nu\), define
    \[
        P_\nu\coloneqq \wlp_{\mathrm{syn}}(S_\nu,Q).
    \]
    By the induction hypothesis applied to the continuation \(S_\nu\),
    \(P_\nu\in\Pred\), and for every \(\sigma\in\cT(\cH)_+\),
    \begin{equation}
        \label{eq:bind-ih-nu}
        \Tr(P_\nu\sigma)
        =
        \Tr(Q\sem{S_\nu}(\sigma))
        \qquad
        \forall \sigma\in\cT(\cH)_+ .
    \end{equation}

    The structural rule proposes the form
    \begin{equation}\label{eq:integral-of-bind-quadratic}
        \mathfrak p[\u]
        \coloneqq
        \int_{\Omega_M}
            \QF_{P_\nu}[M_\nu\u]
        \,d\mu(\nu).
    \end{equation}
    We first prove that this integral is a well-defined extended Lebesgue
    integral.

    Fix \(\u\in\cH\), and let
    \[
        F_{\u}(\nu)
        \coloneqq
        \sem{S_\nu}\!\left(M_\nu\sP_{\u}M_\nu^\dagger\right).
    \]
    By the measurability part of the denotational semantics, in particular the
    binding case of \Cref{thm:semantics_well_defined}, the branch map
    \[
        \nu\longmapsto F_{\u}(\nu)
    \]
    is $(\Sigma_M,\operatorname{Borel}(\cT(\cH)_{+}))$-measurable.  Concretely, this uses
    the strong operator measurability of \(\nu\mapsto M_\nu\), the measurability
    of the sandwich map
    \(\nu\mapsto M_\nu\sP_{\u}M_\nu^\dagger\), and the Carath\'eodory
    measurability of the continuation semantics
    \((\nu,\sigma)\mapsto\sem{S_\nu}(\sigma)\).

    Choose, by \Cref{prop:spectral_approximation}, an increasing sequence of
    bounded positive predicates \(Q_m\in\cB(\cH)_+\) such that
    \[
        Q_m\sqsubseteq Q_{m+1},
        \qquad
        \bigsqcup_m Q_m=Q,
        \qquad
        \QF_Q[\v]=\sup_m \QF_{Q_m}[\v].
    \]
    For every \(m\), the functional
    \[
        \ell_m:\cT(\cH)\to\bC,
        \qquad
        \ell_m(\tau)=\tr(Q_m\tau),
    \]
    is bounded and trace-norm continuous.  Hence
    \[
        \nu\longmapsto
        \tr(Q_mF_{\u}(\nu))
    \]
    is $(\Sigma_M,\operatorname{Borel}(\bR))$-measurable.  Since \(Q_m\uparrow Q\), the monotone convergence theorem
    for extended traces gives
    \[
    \begin{aligned}
        \QF_{P_\nu}[M_\nu\u]
        &=
        \Tr(P_\nu M_\nu\sP_{\u}M_\nu^\dagger)                         \\
        &=
        \Tr(QF_{\u}(\nu))                                             \\
        &=
        \sup_m \tr(Q_mF_{\u}(\nu)).
    \end{aligned}
    \]
    Therefore
    \[
        \nu\longmapsto \QF_{P_\nu}[M_\nu\u]
    \]
    is a measurable \([0,\infty]\)-valued function.  Thus the integral in
    \Cref{eq:integral-of-bind-quadratic} is well-defined for every \(\u\).

    Next we show that \(\mathfrak p\) is a closed positive quadratic form.
    For each fixed \(\nu\), the map
    \[
        \u\longmapsto \QF_{P_\nu}[M_\nu\u]
    \]
    is a closed positive quadratic form: it is the bounded pullback of the
    closed positive form \(\QF_{P_\nu}\).  The finite-value domain of
    \(\mathfrak p\) is linear.  Indeed, for \(\u,\v\in\Dom(\mathfrak p)\) and
    scalars \(a,b\),
    \[
        \QF_{P_\nu}[M_\nu(a\u+b\v)]
        \le
        2|a|^2\QF_{P_\nu}[M_\nu\u]
        +
        2|b|^2\QF_{P_\nu}[M_\nu\v],
    \]
    and integration gives \(\mathfrak p[a\u+b\v]<\infty\).  Homogeneity and the
    parallelogram law pass from the pointwise forms to the integral, so
    \(\mathfrak p\) is a positive quadratic form.

    It remains to prove closedness.  By \Cref{thm:closed_iff_lsc}, it suffices
    to prove lower semicontinuity in the Hilbert norm.  Let \(\u_n\to\u\) in
    \(\cH\).  For each fixed \(\nu\), boundedness of \(M_\nu\) gives
    \(M_\nu\u_n\to M_\nu\u\), and lower semicontinuity of
    \(\QF_{P_\nu}\) gives
    \[
        \QF_{P_\nu}[M_\nu\u]
        \le
        \liminf_{n\to\infty}\QF_{P_\nu}[M_\nu\u_n].
    \]
    Fatou's lemma then yields
    \[
    \begin{aligned}
        \mathfrak p[\u]
        &=
        \int_{\Omega_M}\QF_{P_\nu}[M_\nu\u]\,d\mu(\nu)                  \\
        &\le
        \int_{\Omega_M}
            \liminf_{n\to\infty}
            \QF_{P_\nu}[M_\nu\u_n]
        \,d\mu(\nu)                                                    \\
        &\le
        \liminf_{n\to\infty}
        \int_{\Omega_M}
            \QF_{P_\nu}[M_\nu\u_n]
        \,d\mu(\nu)                                                    \\
        &=
        \liminf_{n\to\infty}\mathfrak p[\u_n].
    \end{aligned}
    \]
    Hence \(\mathfrak p\) is closed and defines a unique predicate
    \(P_{syn}\in\Pred\).

    We now compute the trace of \(P_{syn}\).  Let
    \(\rho=\sum_j\lambda_j\sP_{\u_j}\) be a spectral decomposition of
    \(\rho\in\cT(\cH)_+\).  Using the definition of the extended trace and
    Tonelli's theorem for non-negative functions, we obtain
    \begin{equation}
        \label{eq:bind-trace-psyn}
        \begin{aligned}
            \Tr(P_{syn}\rho)
            &=
            \sum_j\lambda_j\,\mathfrak p[\u_j]                              \\
            &=
            \sum_j\lambda_j
            \int_{\Omega_M}\QF_{P_\nu}[M_\nu\u_j]\,d\mu(\nu)                 \\
            &=
            \int_{\Omega_M}
                \sum_j\lambda_j\QF_{P_\nu}[M_\nu\u_j]
            \,d\mu(\nu)                                                     \\
            &=
            \int_{\Omega_M}
                \Tr\!\left(P_\nu M_\nu\rho M_\nu^\dagger\right)
            \,d\mu(\nu).
        \end{aligned}
    \end{equation}

    Finally we compare this with the denotational semantics.  Define the
    positive branch state
    \begin{equation}
        \label{eq:bind-branch-state}
        F_\rho(\nu)
        \coloneqq
        \sem{S_\nu}\!\left(M_\nu\rho M_\nu^\dagger\right).
    \end{equation}
    By the Bochner-integral clause for binding,
    \begin{equation}
        \label{eq:bind-bochner-semantics}
        \sem{T}(\rho)
        =
        \int_{\Omega_M} F_\rho(\nu)\,d\mu(\nu)
        \quad\text{in }\cT(\cH).
    \end{equation}
    For bounded predicates \(Q_m\), Hille's theorem for Bochner integrals gives
    \begin{equation}
        \label{eq:bind-hille-bounded}
        \begin{aligned}
            \tr(Q_m\sem{T}(\rho))
            &=
            \tr\!\left(
                Q_m\int_{\Omega_M}F_\rho(\nu)\,d\mu(\nu)
            \right)                                                        \\
            &=
            \int_{\Omega_M}\tr(Q_mF_\rho(\nu))\,d\mu(\nu).
        \end{aligned}
    \end{equation}

    Set
    \[
        f_m(\nu)\coloneqq \tr(Q_mF_\rho(\nu)).
    \]
    Since \(F_\rho\) is Bochner measurable and
    \(\tau\mapsto\tr(Q_m\tau)\) is a bounded trace-norm continuous
    functional on \(\cT(\cH)\), each \(f_m\) is measurable.  Moreover,
    \(f_m\ge0\), and \(Q_m\sqsubseteq Q_{m+1}\) implies
    \(f_m(\nu)\le f_{m+1}(\nu)\) for all \(\nu\).  Hence
    \(f_m\uparrow\sup_m f_m\).

    Taking the supremum over \(m\), using the monotone convergence of
    traces over linear relations (\Cref{cor:mct_traces}), applying
    Hille's theorem for Bochner integrals as in
    \Cref{eq:bind-hille-bounded}, and then applying Levi's monotone
    convergence theorem for non-negative measurable functions, we obtain
    \begin{equation}
        \label{eq:bind-unbounded-mct}
        \begin{aligned}
            \Tr(Q\sem{T}(\rho))
            &=
            \sup_m \tr(Q_m\sem{T}(\rho))                                    \\
            &=
            \sup_m
            \int_{\Omega_M}\tr(Q_mF_\rho(\nu))\,d\mu(\nu)                   \\
            &=
            \int_{\Omega_M}
                \sup_m \tr(Q_mF_\rho(\nu))
            \,d\mu(\nu)                                                     \\
            &=
            \int_{\Omega_M}\Tr(QF_\rho(\nu))\,d\mu(\nu).
        \end{aligned}
    \end{equation}
    In \Cref{eq:bind-unbounded-mct}, the first equality is
    \(\Cref{cor:mct_traces}\) applied to the increasing sequence
    \(Q_m\uparrow Q\) and the positive trace-class operator
    \(\sem{T}(\rho)\).  The second equality is
    \Cref{eq:bind-hille-bounded}.  The third equality is Levi's
    monotone convergence theorem applied to the non-negative measurable
    sequence \(f_m\).  The last equality is again
    \(\Cref{cor:mct_traces}\), applied pointwise in \(\nu\) to
    \(Q_m\uparrow Q\) and \(F_\rho(\nu)\in\cT(\cH)_+\).
    
    Applying the induction hypothesis \Cref{eq:bind-ih-nu} with
    \(\sigma=M_\nu\rho M_\nu^\dagger\), and using the definition of
    \(F_\rho\) from \Cref{eq:bind-branch-state}, we have
    \begin{equation}
        \label{eq:bind-branch-duality}
        \Tr(QF_\rho(\nu))
        =
        \Tr\!\left(P_\nu M_\nu\rho M_\nu^\dagger\right).
    \end{equation}
    Substituting \Cref{eq:bind-branch-duality} into
    \Cref{eq:bind-unbounded-mct} and comparing with
    \Cref{eq:bind-trace-psyn} yields
    \begin{equation}
        \label{eq:bind-final-duality}
        \Tr(P_{syn}\rho)
        =
        \Tr(Q\sem{T}(\rho)).
    \end{equation}
    This proves both well-definedness and soundness of the binding rule.  Notice
    that at no point do we commute an unbounded trace functional directly with a
    Bochner integral; the commutation is performed only for the bounded
    approximants \(Q_m\), and the unbounded case is recovered by monotone convergence.

    \textbf{4. While loop.}
    Let
    \[
        T\equiv
        \mathbf{while}\ M\ \mathbf{do}\ S_{body}\ \mathbf{od},
    \]
    with measurement operators \(M_0,M_1\), where \(M_0\) exits the loop.  Define
    \(X_0=0\), and
    \[
        \QF_{X_{n+1}}[\u]
        =
        \QF_Q[M_0\u]
        +
        \QF_{\wlp_{\mathrm{syn}}(S_{body},X_n)}[M_1\u].
    \]
    By induction, if \(X_n\in\Pred\), then
    \(\wlp_{\mathrm{syn}}(S_{body},X_n)\in\Pred\).  The two summands above are
    bounded pullbacks of closed positive forms, and their sum is again a closed
    positive form.  Hence \(X_{n+1}\in\Pred\).  Moreover, the sequence
    \((X_n)_n\) is increasing: this follows from \(X_0\sqsubseteq X_1\) and from
    monotonicity of the already established \(\wlp\)-transformer on the body.
    Therefore
    \[
        P_{syn}\coloneqq\bigsqcup_{n=0}^{\infty}X_n
    \]
    exists in \(\Pred\) by \Cref{thm:monotone_convergence}, with
    \[
        \Tr(P_{syn}\rho)=\sup_n\Tr(X_n\rho)
    \]
    by \Cref{cor:mct_traces}.

    Define the positive maps
    \[
        \mathcal E_{exit}(\sigma)\coloneqq M_0\sigma M_0^\dagger,
        \qquad
        \mathcal E_{body}(\sigma)
        \coloneqq
        \sem{S_{body}}\!\left(M_1\sigma M_1^\dagger\right).
    \]
    The denotational semantics of the loop is the increasing sum of terminating
    branches:
    \[
        \sem{T}(\rho)
        =
        \sum_{k=0}^{\infty}
        \mathcal E_{exit}\!\left(\mathcal E_{body}^k(\rho)\right)
    \]
    in the positive trace-class cone.  The extended trace is
    \(\sigma\)-additive on positive trace-class series, hence
    \[
        \Tr(Q\sem{T}(\rho))
        =
        \sum_{k=0}^{\infty}
        \Tr\!\left(
            Q\,\mathcal E_{exit}(\mathcal E_{body}^k(\rho))
        \right).
    \]

    We claim that for every \(N\ge0\) and every
    \(\sigma\in\cT(\cH)_+\),
    \begin{equation}\label{eq:wlp-loop-truncation}
        \Tr(X_N\sigma)
        =
        \sum_{k=0}^{N-1}
        \Tr\!\left(
            Q\,\mathcal E_{exit}(\mathcal E_{body}^k(\sigma))
        \right),
    \end{equation}
    where the sum is empty for \(N=0\).
    where the sum is empty for \(N=0\).  The base case follows from \(X_0=0\).
    Assume \Cref{eq:wlp-loop-truncation} holds for \(N\). Then
    \[
    \begin{aligned}
        \Tr(X_{N+1}\rho)
        &=
        \Tr(QM_0\rho M_0^\dagger)
        +
        \Tr\!\left(
            \wlp_{\mathrm{syn}}(S_{body},X_N)
            M_1\rho M_1^\dagger
        \right)                                                        \\
        &=
        \Tr(Q\mathcal E_{exit}(\rho))
        +
        \Tr\!\left(
            X_N\sem{S_{body}}(M_1\rho M_1^\dagger)
        \right)                                                        \\
        &=
        \Tr(Q\mathcal E_{exit}(\rho))
        +
        \Tr\!\left(X_N\mathcal E_{body}(\rho)\right).
    \end{aligned}
    \]
    Applying the induction claim \((5)\) to the positive input
    \(\mathcal E_{body}(\rho)\) gives
    \[
        \Tr\!\left(X_N\mathcal E_{body}(\rho)\right)
        =
        \sum_{k=1}^{N}
        \Tr\!\left(
            Q\,\mathcal E_{exit}(\mathcal E_{body}^k(\rho))
        \right).
    \]
    Adding the \(k=0\) term proves \Cref{eq:wlp-loop-truncation} for \(N+1\).

    Taking the supremum over \(N\), and using monotone convergence, gives
    \[
    \begin{aligned}
        \Tr(P_{syn}\rho)
        &=
        \sup_N\Tr(X_N\rho)                                             \\
        &=
        \sum_{k=0}^{\infty}
        \Tr\!\left(
            Q\,\mathcal E_{exit}(\mathcal E_{body}^k(\rho))
        \right)                                                        \\
        &=
        \Tr(Q\sem{T}(\rho)).
    \end{aligned}
    \]
    Thus the while rule is well-defined and sound.

    We have now shown, for every structural rule, that the generated predicate
    \(P_{syn}\in\Pred\) satisfies
    \[
        \Tr(P_{syn}\rho)=\Tr(Q\sem{S}(\rho))
        \qquad
        \forall \rho\in\cT(\cH)_+.
    \]
    In particular this holds for all partial density operators.  By the
    uniqueness of the predicate satisfying Exact Duality in
    \Cref{thm:wlp_well_defined}, the structurally generated predicate is exactly
    the semantic weakest liberal precondition:
    \[
        P_{syn}=\wlp(S,Q),
    \]
    which completes the proof.
\end{proof}

\subsection{The Cost Predicate and Foundations of Total Correctness}
\label{Subsec:weakest_precondition}

To rigorously enforce termination within the upper-bound semantics, we must explicitly track the expected operational cost. Instead of introducing an ad-hoc tracking metric, we formally define the cost functional $\mathcal{C}(S, \rho) \in [0, \infty]$ by structural induction on the syntax of $S$, strictly relying on the semantic map $\sem{\cdot}$ to sequence the intermediate states. 

\emph{Assumption (Uniform Base Cost):} We assume that the execution of any single elementary logical step (e.g., a basic unitary, a measurement, or evaluating a loop guard) consumes a strictly positive, state-independent minimal base cost $\epsilon > 0$. Because the unnormalized semantic state $\rho$ encodes the execution branch probability via its trace, the expected cost contributed by executing an elementary step on state $\rho$ is precisely $\epsilon \Tr(\rho)$. This uniform lower bound physically represents an indivisible quantum resource and mathematically precludes Zeno behavior, ensuring that an infinite sequence of non-terminating steps strictly incurs an infinite operational cost.

\begin{definition}[Expected Operational Cost via Semantics]
\label{def:operational_cost}
    Fix a one-step cost constant $\epsilon>0$.  Let $\Gamma\vdash S$ be a
    well-formed program phrase and let $\omega\in\Omega_\Gamma$ be a classical
    valuation.  The cost is defined on the positive trace-class cone
    $\cT(\cH)_+$,
    rather than only on $\pardensity{\cH}$, since a continuous-measurement
    branch
    $M_{\omega,\nu}\rho M_{\omega,\nu}^{\dagger}$
    need not be a partial density operator pointwise.

    For $\rho\in\cT(\cH)_+$, we define
    $\mathcal C(S_\omega,\rho)\in[0,\infty]$
    structurally as follows.

    \begin{itemize}
        \item \textbf{Basic operations}
        ($\mathbf{skip}$, initialization $q:=0$, and parameterized unitary
        commands $U(\seq e)[\seq q]$):
        \[
            \mathcal C(S_\omega,\rho)
            \coloneqq
            \epsilon\Tr(\rho).
        \]

        \item \textbf{Abort}.
        If $S\equiv\mathbf{abort}$, then
        \[
            \mathcal C(S_\omega,\rho)
            \coloneqq
            \epsilon\Tr(\rho)\cdot\infty,
        \]
        with the convention
        \[
            0\cdot\infty=0,
            \qquad
            a\cdot\infty=+\infty
            \quad(a>0).
        \]

        \item \textbf{Sequential composition}.
        If
        \[
            S\equiv S_1;S_2,
        \]
        then
        \[
            \mathcal C(S_\omega,\rho)
            \coloneqq
            \mathcal C((S_1)_\omega,\rho)
            +
            \mathcal C
            \bigl(
                (S_2)_\omega,
                \sem{S_1}_{\omega}(\rho)
            \bigr).
        \]

        \item \textbf{Continuous measurement}.
        If
        \[
            S\equiv \bind(M(\seq e),x.S(x)),
        \]
        then
        \[
            \mathcal C(S_\omega,\rho)
            \coloneqq
            \epsilon\Tr(\rho)
            +
            \int_{\Omega_M}
                \mathcal C
                \bigl(
                    (S(x))_{\omega,\nu},
                    M_{\omega,\nu}\rho M_{\omega,\nu}^{\dagger}
                \bigr)
            \,d\mu_M(\nu).
        \]
        Here $M_{\omega,\nu}$ is understood in the contextual shorthand of bind assumption in Section \ref{subsec:syntax}, and $(S(x))_{\omega,\nu}$ denotes the continuation under the extended
        valuation $(\omega,\nu)\in\Omega_\Gamma\times\Omega_M$.
        The legitimacy of this extended non-negative integral will be proved in
        \Cref{prop:cost_quadratic_form}.

        In the bind case, we write $S_{\omega,\nu}$ for $(S(x))_{\omega,\nu}$ whenever no confusion is possible.

        \item \textbf{While loop}.
        If
        \[
            S\equiv
            \mathbf{while}\ M(\seq e)[\seq q]=1\ \mathbf{do}\ S_{body}\ \mathbf{od},
        \]
        let $M=\{M_{\omega,0},M_{\omega,1}\}$, where $M_{\omega,1}$ is the continuation outcome.  Define
        \[
            \rho_0\coloneqq\rho,
            \qquad
            \rho_{n+1}
            \coloneqq
            \sem{S_{body}}_{\omega}
            (M_{\omega,1}\rho_n M_{\omega,1}^\dagger).
        \]
        Then
        \[
            \mathcal C(S_\omega,\rho)
            \coloneqq
            \sum_{n=0}^{\infty}
            \left(
                \epsilon\Tr(\rho_n)
                +
                \mathcal C
                \bigl(
                    (S_{body})_\omega,
                    M_{\omega,1}\rho_n M_{\omega,1}^\dagger
                \bigr)
            \right).
        \]
    \end{itemize}
\end{definition}

\begin{remark}[Standing convention for cost notation and domains]
\label{rem:cost_closed_program_notation}
    The contextual notation
    \(\mathcal C(S_\omega,\rho)\)
    in \Cref{def:operational_cost} is auxiliary to the construction of the cost
    functional.  It is used only in the immediately following well-definedness
    and quadraticity result \Cref{prop:cost_quadratic_form} and in its proof, where the structural induction must
    simultaneously track the external classical valuation and the measurability
    of branchwise cost expressions.  For this purpose the cost is typed on the
    full positive trace-class cone \(\cT(\cH)_+\).  This enlargement is a
    technical device for branchwise reasoning: even when the incoming state
    \(\rho\in\pardensity{\cH}\), a pointwise continuous-measurement branch
    \(M_{\omega,\nu}\rho M_{\omega,\nu}^{\dagger}\)
    is required only to be a positive trace-class operator, and need not be a
    partial density operator pointwise.

    Once that construction and proof have been completed, all total-correctness
    assertions are read with the external classical valuation fixed.  If \(S\) is
    classically closed, this fixed valuation is the unique
    \(\omega_\ast\in\Omega_\emptyset\).  More generally, for a phrase
    \(\Gamma\vdash S\), an ambient valuation \(\omega\in\Omega_\Gamma\) is fixed
    as part of the interpretation.  We then suppress this valuation and write
    \[
        \mathcal C(S,\rho)
        \coloneqq
        \mathcal C(S_\omega,\rho),
        \qquad
        \sem S(\rho)
        \coloneqq
        \sem S_\omega(\rho),
    \]
    with \(\omega=\omega_\ast\) in the classically closed case.
    This convention suppresses only the ambient external valuation.  Classical variables
    introduced by a binding construct remain explicit integration variables;
    any shortened notation for the corresponding branches is declared locally in
    the relevant clause or proof.

    In that development, expected costs are evaluated on
    partial-density initial states.  The trace-class domain above does not change
    this semantic domain; it only supports the branchwise construction before the
    resulting average-cost functional is restricted to partial density inputs.
\end{remark}

\begin{proposition}[Well-definedness and Quadraticity of Cost]
\label{prop:cost_quadratic_form}
    Let $\Gamma\vdash S$ be a well-formed program phrase.  The clauses in
    \Cref{def:operational_cost} define a well-typed function
    \[
        \mathcal C(S_\omega,\cdot):
        \cT(\cH)_+
        \longrightarrow
        [0,\infty]
    \]
    for every $\omega\in\Omega_\Gamma$.  Moreover, the following properties
    hold.

    \begin{enumerate}
        \item \textbf{Admissibility.}
        The map
        \[
            (\omega,\rho)
            \longmapsto
            \mathcal C(S_\omega,\rho)
        \]
        is
        $(\Sigma_\Gamma\otimes\operatorname{Borel}(\cT(\cH)_+),
        \operatorname{Borel}([0,\infty]))$-measurable,
        which implies the bind integral
        is a legitimate extended non-negative Lebesgue integral.

        \item \textbf{Finite conic linearity on the positive trace-class cone.}
        For every $\rho,\sigma\in\cT(\cH)_+$ and every
        $a,b\in[0,\infty)$,
        \[
            \mathcal C(S_\omega,a\rho+b\sigma)
            =
            a\mathcal C(S_\omega,\rho)
            +
            b\mathcal C(S_\omega,\sigma),
        \]
        with the convention $0\cdot\infty=0$.

        \item \textbf{Lower semi-continuity.}
        For every fixed $\omega\in\Omega_\Gamma$, the map
        \[
            \rho\longmapsto\mathcal C(S_\omega,\rho)
        \]
        is lower semi-continuous on $\cT(\cH)_+$.        
    \end{enumerate}

    Consequently, for every fixed $\omega\in\Omega_\Gamma$, the pure-state cost
    \[
        \QF_{\mathcal C(S_\omega)}[\u]
        \coloneqq
        \mathcal C(S_\omega,\sP_\u)
    \]
    is a closed positive quadratic form on $\cH$. That is, for every fixed valuation $\omega$, there exists a unique
    positive self-adjoint linear relation
    \(C_{S,\omega}\in\Pred\)
    such that
    \(\QF_{C_{S,\omega}}=\QF_{\mathcal C(S_\omega)}\).
    Hence, by the definition of the extended trace,
    \(\Tr(C_{S,\omega}\rho)=\mathcal C(S_{\omega},\rho)\) holds for all $\rho\in\cT(\cH)_{+}$.
    For classically closed programs, we simply write $C_S$.
\end{proposition}

\begin{proof}
    We prove the three listed assertions simultaneously by structural
    induction on $S$.  The asserted quadraticity of the pure-state cost is then
    derived uniformly from finite conic linearity and lower semi-continuity.

    All product $\sigma$-algebras below are the raw product
    $\sigma$-algebras, as in Section~\ref{SEC:syntax_semantics}.

    We use the following typed facts from Section~\ref{SEC:syntax_semantics}.

    First, for every well-formed phrase $\Gamma\vdash R$, the semantic map
    \[
        (\omega,\sigma)
        \longmapsto
        \sem{R}_{\omega}(\sigma)
    \]
    from $\Omega_\Gamma\times\cT(\cH)_+$ to $\cT(\cH)_+$ is
    $(\Sigma_\Gamma\otimes\operatorname{Borel}(\cT(\cH)_+),
    \operatorname{Borel}(\cT(\cH)_+))$-measurable.  Moreover, for every fixed
    $\omega$, the map $\sem{R}_\omega:\cT(\cH)_+\to\cT(\cH)_+$ is positive
    linear and trace-norm continuous.  Here we use the convention in
    \Cref{thm:semantics_well_defined} that $\sem{R}_\omega$ also denotes its
    restriction to positive trace-class inputs.

    Second, for every primitive continuous measurement occurrence
    $M(\seq e)$ in context $\Gamma$, the contextual branch-preparation map
    \[
        H_M(\omega,\rho,\nu)
        \coloneqq
        M_{\omega,\nu}\rho M_{\omega,\nu}^{\dagger}
    \]
    from $\Omega_\Gamma\times\cT(\cH)_+\times\Omega_M$ to $\cT(\cH)_+$ is
    $(\Sigma_\Gamma\otimes
    \operatorname{Borel}(\cT(\cH)_+)\otimes\Sigma_M,
    \operatorname{Borel}(\cT(\cH)_+))$-measurable.  For every fixed
    $(\omega,\nu)$, the map
    \[
        \rho\longmapsto M_{\omega,\nu}\rho M_{\omega,\nu}^{\dagger}
    \]
    is positive linear and trace-norm continuous.  The notation
    $M_{\omega,\nu}$ is the contextual shorthand from Section \ref{SEC:syntax_semantics} for the
    pulled-back primitive Kraus density
    $M_{\sem{\seq e}_\Gamma(\omega),\nu}$.

    Third, we shall use Tonelli's theorem in the form of
    \cite[Proposition~5.2.1]{cohn2013measure}: if
    \(F:X\times\Omega_M\to[0,\infty]\) is
    \((\mathcal A\otimes\Sigma_M,\operatorname{Borel}([0,\infty]))\)-measurable,
    then every section \(\nu\mapsto F(x,\nu)\) is
    \((\Sigma_M,\operatorname{Borel}([0,\infty]))\)-measurable, and the
    marginal integral
    \[
        x\longmapsto
        \int_{\Omega_M}F(x,\nu)\,d\mu_M(\nu)
    \]
    is \((\mathcal A,\operatorname{Borel}([0,\infty]))\)-measurable.  All such
    integrals are understood as extended non-negative Lebesgue integrals.

    Finally, we shall repeatedly use the elementary facts that finite
    non-negative linear combinations of lower semi-continuous
    $[0,\infty]$-valued functions are lower semi-continuous, that pointwise
    suprema of lower semi-continuous functions are lower semi-continuous, and
    that composing a lower semi-continuous function with a continuous map
    preserves lower semi-continuity.

    Now we start the structural induction.

    \paragraph{$\bullet$ Basic commands.}
    If $S$ is $\mathbf{skip}$, an initialization, or a parameterized unitary
    command, then
    \[
        \mathcal C(S_\omega,\rho)=\epsilon\Tr(\rho).
    \]
    Admissibility and lower semi-continuity follow from continuity of
    $\rho\mapsto\Tr(\rho)$ on $\cT(\cH)_+$, and finite conic linearity follows
    from linearity of the trace.

    \paragraph{$\bullet$ Abort.}
    If $S\equiv\mathbf{abort}$, then
    \[
        \mathcal C(S_\omega,\rho)=\epsilon\Tr(\rho)\cdot\infty,
    \]
    with $0\cdot\infty=0$ and $a\cdot\infty=+\infty$ for $a>0$.  Thus the
    cost is $0$ on $\{\Tr(\rho)=0\}$ and $+\infty$ on
    $\{\Tr(\rho)>0\}$.  Since $\rho\mapsto\Tr(\rho)$ is continuous,
    admissibility and lower semi-continuity follow.

    Finite conic linearity is immediate from the same dichotomy: for
    $a,b\in[0,\infty)$, the trace of $a\rho+b\sigma$ is zero exactly when
    $a\Tr(\rho)=0$ and $b\Tr(\sigma)=0$; otherwise both sides of the asserted
    conic-linearity identity are $+\infty$.

    \paragraph{$\bullet$ Sequential composition.}
    Suppose
    \[
        S\equiv S_1;S_2.
    \]

    We first prove admissibility. 
    Since
    \[
        \mathcal C(S_\omega,\rho)
        =
        \mathcal C((S_1)_\omega,\rho)+
        \mathcal C
        \bigl(
            (S_2)_\omega,\sem{S_1}_{\omega}(\rho)
        \bigr),
    \]
    we can define
    \[
    \begin{aligned}
        &F_1(\omega,\rho)\coloneqq\mathcal C((S_1)_\omega,\rho)\qquad
        F_2(\omega,\rho)\coloneqq
        \mathcal C
        \bigl(
            (S_2)_\omega,\sem{S_1}_{\omega}(\rho)
        \bigr);\\
        &\mathcal C_2(\omega,\sigma)
        \coloneqq
        \mathcal C((S_2)_\omega,\sigma),
        \qquad
        R(\omega,\rho)
        \coloneqq
        \sem{S_1}_{\omega}(\rho).
    \end{aligned}
    \]
    By the induction hypothesis for $S_2$, the map
    \[
        \mathcal C_2:
        \Omega_\Gamma\times\cT(\cH)_+
        \longrightarrow
        [0,\infty]
    \]
    is
    $(\Sigma_\Gamma\otimes\operatorname{Borel}(\cT(\cH)_+),
    \operatorname{Borel}([0,\infty]))$-measurable.  By the semantic
    measurability fact recalled above, the map
    \[
        R:
        \Omega_\Gamma\times\cT(\cH)_+
        \longrightarrow
        \cT(\cH)_+
    \]
    is
    $(\Sigma_\Gamma\otimes\operatorname{Borel}(\cT(\cH)_+),
    \operatorname{Borel}(\cT(\cH)_+))$-measurable.

    The second summand $F_2$ in the sequential cost is obtained by the following
    factorization:
    \[
    \begin{tikzcd}[column sep=large,row sep=large]
        \Omega_\Gamma\times\cT(\cH)_+
        \arrow[r,"h"]
        \arrow[dr,swap,"F_2"]
        &
        \Omega_\Gamma\times\cT(\cH)_+
        \arrow[d,"\mathcal C_2"]
        \\
        &
        {[0,\infty]}
    \end{tikzcd}
    \]
    where
    \[
        h(\omega,\rho)
        \coloneqq
        (\omega,R(\omega,\rho)),
        \qquad
        F_2(\omega,\rho)
        \coloneqq
        \mathcal C_2(h(\omega,\rho)).
    \]
    The first component of $h$ is the canonical coordinate projection, and the
    second component is $R$.  Hence $h$ is
    $(\Sigma_\Gamma\otimes\operatorname{Borel}(\cT(\cH)_+),
    \Sigma_\Gamma\otimes\operatorname{Borel}(\cT(\cH)_+))$-measurable.
    Therefore
    \[
        F_2(\omega,\rho)
        =
        \mathcal C
        \bigl(
            (S_2)_\omega,
            \sem{S_1}_{\omega}(\rho)
        \bigr)
    \]
    is
    $(\Sigma_\Gamma\otimes\operatorname{Borel}(\cT(\cH)_+),
    \operatorname{Borel}([0,\infty]))$-measurable.

    The first summand
    \[
        F_1(\omega,\rho)
        \coloneqq
        \mathcal C((S_1)_\omega,\rho)
    \]
    is measurable by the induction hypothesis for $S_1$.  Since
    \[
        \mathcal C((S_1;S_2)_\omega,\rho)
        =
        F_1(\omega,\rho)+F_2(\omega,\rho),
    \]
    and addition on $[0,\infty]$ is measurable, the sequential cost is
    admissible.

    For finite conic linearity, fix $\omega$, $\rho,\sigma\in\cT(\cH)_+$, and
    $a,b\in[0,\infty)$.  The first summand
    \[
        \rho\longmapsto \mathcal C((S_1)_\omega,\rho)
    \]
    is finite conic linear by the induction hypothesis for $S_1$.  For the
    second summand, positive linearity of $\sem{S_1}_\omega$ gives
    \[
        R(\omega,a\rho+b\sigma)
        =
        aR(\omega,\rho)+bR(\omega,\sigma),
    \]
    and the induction hypothesis for $S_2$ gives
    \[
        \mathcal C
        \bigl(
            (S_2)_\omega,
            R(\omega,a\rho+b\sigma)
        \bigr)=
        a\mathcal C((S_2)_\omega,R(\omega,\rho))
        +
        b\mathcal C((S_2)_\omega,R(\omega,\sigma)).
    \]
    Therefore the sum of the two summands,
    \[
        \rho\longmapsto
        \mathcal C((S_1;S_2)_\omega,\rho),
    \]
    is finite conic linear.

    For lower semi-continuity, fix $\omega$ and let
    $\rho_n\to\rho$ in $\cT(\cH)_+$.  Since $\sem{S_1}_\omega$ is trace-norm
    continuous,
    \[
        \sem{S_1}_\omega(\rho_n)
        \longrightarrow
        \sem{S_1}_\omega(\rho).
    \]
    By the induction hypotheses for $S_1$ and $S_2$,
    \[
    \begin{aligned}
        \liminf_{n\to\infty}
        \mathcal C((S_1;S_2)_\omega,\rho_n)
        &=\liminf_{n\to\infty}
        \Big(\mathcal C((S_1)_\omega,\rho_n)+
        \mathcal{C}\bigl((S_2)_\omega,\sem{S_1}_\omega(\rho_n)\bigr)\Big)
        \\
        &\ge
        \liminf_{n\to\infty}\mathcal C((S_1)_\omega,\rho_n)+
        \liminf_{n\to\infty}
        \mathcal C
        \bigl(
            (S_2)_\omega,
            \sem{S_1}_\omega(\rho_n)
        \bigr)
        \\
        &\ge
        \mathcal C((S_1)_\omega,\rho)
        +
        \mathcal C
        \bigl(
            (S_2)_\omega,
            \sem{S_1}_\omega(\rho)
        \bigr)
        \\
        &=
        \mathcal C((S_1;S_2)_\omega,\rho).
    \end{aligned}
    \]

    \paragraph{$\bullet$ Continuous measurement.}
    Suppose
    \[
        S\equiv \bind(M(\seq e),x.S(x)).
    \]
    The continuation is interpreted over
    \[
        \Omega_{\Gamma,x}
        =
        \Omega_\Gamma\times\Omega_M,
        \qquad
        \Sigma_{\Gamma,x}
        =
        \Sigma_\Gamma\otimes\Sigma_M.
    \]

    We first prove the admissibility statement. Write
    $S_{\omega,\nu}$
    for the continuation $(S(x))_{\omega,\nu}$, and define
    \[
        \mathcal C_x((\omega,\nu),\sigma)
        \coloneqq
        \mathcal C(S_{\omega,\nu},\sigma).
    \]
    By the induction hypothesis applied to the continuation $S(x)$ in the
    extended context $\Gamma,x$, the map
    \[
        \mathcal C_x:
        (\Omega_\Gamma\times\Omega_M)\times\cT(\cH)_+
        \longrightarrow
        [0,\infty]
    \]
    is
    $((\Sigma_\Gamma\otimes\Sigma_M)\otimes
    \operatorname{Borel}(\cT(\cH)_+),
    \operatorname{Borel}([0,\infty]))$-measurable.

    Define the branch-state map
    \[
        H(\omega,\rho,\nu)
        \coloneqq
        M_{\omega,\nu}\rho M_{\omega,\nu}^{\dagger}.
    \]
    By the branch-preparation measurability fact recalled above, equivalently
    the $\bind$ case in the proof of
    \Cref{thm:semantics_well_defined}, the map
    \[
        H:
        \Omega_\Gamma\times\cT(\cH)_+\times\Omega_M
        \longrightarrow
        \cT(\cH)_+
    \]
    is
    $(\Sigma_\Gamma\otimes
    \operatorname{Borel}(\cT(\cH)_+)\otimes\Sigma_M,
    \operatorname{Borel}(\cT(\cH)_+))$-measurable.

    The branch-cost integrand is obtained by the following factorization:
    \[
    \begin{tikzcd}[column sep=large,row sep=large]
        \Omega_\Gamma\times\cT(\cH)_+\times\Omega_M
        \arrow[r,"h"]
        \arrow[dr,swap,"F"]
        &
        (\Omega_\Gamma\times\Omega_M)\times\cT(\cH)_+
        \arrow[d,"\mathcal C_x"]
        \\
        &
        {[0,\infty]}
    \end{tikzcd}
    \]
    where
    \[
        h(\omega,\rho,\nu)
        \coloneqq
        \bigl((\omega,\nu),H(\omega,\rho,\nu)\bigr),
        \qquad
        F(\omega,\rho,\nu)
        \coloneqq
        \mathcal{C}(S_{\omega,\nu},  M_{\omega,\nu}\rho M_{\omega,\nu}^{\dagger})=
        \mathcal C_x(h(\omega,\rho,\nu)).
    \]
    The first component $(\omega,\rho,\nu)\mapsto(\omega,\nu)$ is the
    canonical coordinate projection, and the second component is $H$.
    Therefore $h$ is
    $(\Sigma_\Gamma\otimes
    \operatorname{Borel}(\cT(\cH)_+)\otimes\Sigma_M,
    (\Sigma_\Gamma\otimes\Sigma_M)\otimes
    \operatorname{Borel}(\cT(\cH)_+))$-measurable.  Since $\mathcal C_x$ is
    measurable by the induction hypothesis, the composite
    \[
        F(\omega,\rho,\nu)
        =
        \mathcal C
        \bigl(
            S_{\omega,\nu},
            M_{\omega,\nu}\rho M_{\omega,\nu}^{\dagger}
        \bigr)
    \]
    is
    $(\Sigma_\Gamma\otimes
    \operatorname{Borel}(\cT(\cH)_+)\otimes\Sigma_M,
    \operatorname{Borel}([0,\infty]))$-measurable.

    Applying the Tonelli theorem with
    \[
        X=\Omega_\Gamma\times\cT(\cH)_+,
        \qquad
        \mathcal A=\Sigma_\Gamma\otimes\operatorname{Borel}(\cT(\cH)_+),
    \]
    gives two conclusions.  First, for every fixed
    \[
        (\omega,\rho)\in\Omega_\Gamma\times\cT(\cH)_+,
    \]
    the section
    $\nu\longmapsto F(\omega,\rho,\nu)$
    is
    $(\Sigma_M,\operatorname{Borel}([0,\infty]))$-measurable, so the integral
    \[
        \int_{\Omega_M}F(\omega,\rho,\nu)\,d\mu_M(\nu)
    \]
    is a legitimate extended non-negative Lebesgue integral.  Second, the
    marginal integral
    \[
        G(\omega,\rho)
        \coloneqq
        \int_{\Omega_M}F(\omega,\rho,\nu)\,d\mu_M(\nu)
    \]
    is
    $(\Sigma_\Gamma\otimes\operatorname{Borel}(\cT(\cH)_+),
    \operatorname{Borel}([0,\infty]))$-measurable.  Hence
    \[
        (\omega,\rho)
        \longmapsto
        \mathcal C(S_\omega,\rho)
        =
        \epsilon\Tr(\rho)+G(\omega,\rho)
    \]
    is
    $(\Sigma_\Gamma\otimes\operatorname{Borel}(\cT(\cH)_+),
    \operatorname{Borel}([0,\infty]))$-measurable.  This proves both the
    legitimacy of the bind integral and the closure of the admissibility
    invariant under bind.

        For finite conic linearity, fix $\omega\in\Omega_\Gamma$,
    $\rho,\sigma\in\cT(\cH)_+$, and $a,b\in[0,\infty)$.  Since
    $H(\omega,\cdot,\nu)$ is positive linear for every $\nu$,
    \[
        H(\omega,a\rho+b\sigma,\nu)
        =
        aH(\omega,\rho,\nu)+bH(\omega,\sigma,\nu).
    \]
    Applying the induction hypothesis to the continuation gives, for every
    $\nu$,
    \[
    \begin{aligned}
        F(\omega,a\rho+b\sigma,\nu)
        &=
        aF(\omega,\rho,\nu)+bF(\omega,\sigma,\nu).
    \end{aligned}
    \]
    Linearity of the trace and of the extended non-negative integral then
    gives
    \[
        \mathcal C(S_\omega,a\rho+b\sigma)
        =
        a\mathcal C(S_\omega,\rho)
        +
        b\mathcal C(S_\omega,\sigma).
    \]

    For lower semi-continuity, fix $\omega$ and let
    $\rho_n\to\rho$ in $\cT(\cH)_+$.  For each fixed $\nu$, trace-norm
    continuity of $H(\omega,\cdot,\nu)$ gives
    \[
        H(\omega,\rho_n,\nu)
        \longrightarrow
        H(\omega,\rho,\nu).
    \]
    Hence, by the lower semi-continuity induction hypothesis for the
    continuation,
    \[
        \liminf_{n\to\infty}
        F(\omega,\rho_n,\nu)
        \ge
        F(\omega,\rho,\nu)
        \qquad
        \text{for every }\nu\in\Omega_M.
    \]
    Since the functions $\nu\mapsto F(\omega,\rho_n,\nu)$ are measurable and
    non-negative by the admissibility part, Fatou's lemma gives
    \[
    \begin{aligned}
        \liminf_{n\to\infty}\mathcal C(S_\omega,\rho_n)
        &\ge
        \epsilon\Tr(\rho)
        +
        \int_{\Omega_M}
            \liminf_{n\to\infty}F(\omega,\rho_n,\nu)
        \,d\mu_M(\nu)
        \\
        &\ge
        \epsilon\Tr(\rho)
        +
        \int_{\Omega_M}F(\omega,\rho,\nu)\,d\mu_M(\nu)
        \\
        &=
        \mathcal C(S_\omega,\rho).
    \end{aligned}
    \]

    \paragraph{$\bullet$ While loop.}
    Suppose
    \[
        S\equiv
        \mathbf{while}\ M(\seq e)[\seq q]=1\ \mathbf{do}\ S_{body}\ \mathbf{od},
    \]
    where, for each valuation $\omega$, the loop guard is
    \(M_\omega=\{M_{\omega,0},M_{\omega,1}\}\).

    For admissibility, define recursively
    \[
        R_0(\omega,\rho)\coloneqq\rho,
        \qquad
        R_{n+1}(\omega,\rho)
        \coloneqq
        \sem{S_{body}}_{\omega}
        \bigl(M_{\omega,1}R_n(\omega,\rho)M_{\omega,1}^\dagger\bigr).
    \]
    The map $R_0$ is
    $(\Sigma_\Gamma\otimes\operatorname{Borel}(\cT(\cH)_+),
    \operatorname{Borel}(\cT(\cH)_+))$-measurable.  If $R_n$ is measurable,
    then
    \[
        (\omega,\rho)
        \longmapsto
        M_{\omega,1}R_n(\omega,\rho)M_{\omega,1}^\dagger
    \]
    is measurable by Lemma~\ref{lem:operator_sandwich_measurability} together
    with the Carath\'eodory joint-measurability theorem.  Composing this with
    the semantic measurability of $S_{body}$ shows that $R_{n+1}$ is
    measurable.  Hence each $R_n$ is measurable.

    Define the $n$-th loop contribution
    \[
        c_n(\omega,\rho)
        \coloneqq
        \epsilon\Tr(R_n(\omega,\rho))
        +
        \mathcal C
        \bigl(
            (S_{body})_\omega,
            M_{\omega,1}R_n(\omega,\rho)M_{\omega,1}^\dagger
        \bigr).
    \]
    The first term is measurable by continuity of the trace.  The second term
    is measurable by the induction hypothesis for $S_{body}$, composed with
    the measurable map
    \[
        (\omega,\rho)
        \longmapsto
        \bigl(
            \omega,
            M_{\omega,1}R_n(\omega,\rho)M_{\omega,1}^\dagger
        \bigr).
    \]
    Therefore each $c_n$ is
    $(\Sigma_\Gamma\otimes\operatorname{Borel}(\cT(\cH)_+),
    \operatorname{Borel}([0,\infty]))$-measurable.  Since all terms are
    non-negative,
    \[
        (\omega,\rho)
        \longmapsto
        \sum_{n=0}^{\infty}c_n(\omega,\rho)
        =
        \sup_N\sum_{n=0}^{N}c_n(\omega,\rho)
    \]
    is measurable as the pointwise supremum of measurable functions.  This
    proves admissibility and well-definedness of the loop cost.

    For finite conic linearity, fix $\omega\in\Omega_\Gamma$,
    $\rho,\sigma\in\cT(\cH)_+$, and $a,b\in[0,\infty)$.  With $\omega$ fixed,
    induction on $n$ shows that the map
    \[
        \rho\longmapsto R_n(\omega,\rho)
    \]
    is positive linear for every $n$, since both the guard-sandwich map and
    $\sem{S_{body}}_\omega$ are positive linear.  Hence, by the induction
    hypothesis for $S_{body}$, each map
    \[
        \rho\longmapsto c_n(\omega,\rho)
    \]
    is finite conic linear.

    Therefore
    \[
    \begin{aligned}
        \mathcal C(S_\omega,a\rho+b\sigma)
        &=
        \sup_N
        \sum_{n=0}^{N}
        c_n(\omega,a\rho+b\sigma)
        \\
        &=
        \sup_N
        \left(
            a\sum_{n=0}^{N}c_n(\omega,\rho)
            +
            b\sum_{n=0}^{N}c_n(\omega,\sigma)
        \right)
        \\
        &=
        a\sup_N\sum_{n=0}^{N}c_n(\omega,\rho)
        +
        b\sup_N\sum_{n=0}^{N}c_n(\omega,\sigma)
        \\
        &=
        a\mathcal C(S_\omega,\rho)
        +
        b\mathcal C(S_\omega,\sigma).
    \end{aligned}
    \]
    Here the third equality uses monotone convergence of non-negative
    numerical sequences, with the convention $0\cdot\infty=0$.

    For lower semi-continuity, fix $\omega\in\Omega_\Gamma$ and let
    $\rho_m\to\rho$ in $\cT(\cH)_+$.  For each fixed $n$, induction on $n$
    gives
    \[
        R_n(\omega,\rho_m)
        \xrightarrow{\|\cdot\|_1}
        R_n(\omega,\rho),
    \]
    using trace-norm continuity of the guard-sandwich map and of
    $\sem{S_{body}}_\omega$.

    Hence the map
    $\rho\longmapsto\epsilon\Tr(R_n(\omega,\rho))$
    is continuous.  Moreover, the map
    \[
        \rho
        \longmapsto
        \mathcal C
        \bigl(
            (S_{body})_\omega,
            M_{\omega,1}R_n(\omega,\rho)M_{\omega,1}^\dagger
        \bigr)
    \]
    is lower semi-continuous by the induction hypothesis for $S_{body}$,
    composed with the continuous map
    $\rho\longmapsto M_{\omega,1}R_n(\omega,\rho)M_{\omega,1}^\dagger$.
    Therefore each map
    \[
        \rho\longmapsto c_n(\omega,\rho)
    \]
    is lower semi-continuous.

    It follows that, for every $N$, the finite partial sum
    \[
        \rho
        \longmapsto
        \sum_{n=0}^{N}c_n(\omega,\rho)
    \]
    is lower semi-continuous.  Since
    \[
        \mathcal C(S_\omega,\rho)
        =
        \sup_N
        \sum_{n=0}^{N}c_n(\omega,\rho),
    \]
    the loop cost is lower semi-continuous as the pointwise supremum of lower
    semi-continuous functions.

    This completes the simultaneous induction for admissibility, finite conic linearity, and lower semi-continuity.

    It remains to prove the asserted quadraticity of the pure-state cost.  Fix
    $\omega\in\Omega_\Gamma$ and recall that
    \[
        \QF_{\mathcal C(S_\omega)}[\u]
        \coloneqq
        \mathcal C(S_\omega,\sP_\u),
        \qquad
        \u\in\cH.
    \]
    Since $\mathcal C$ takes values in $[0,\infty]$,
    $\QF_{\mathcal C(S_\omega)}$ is positive.

    The rank-one operators satisfy
    \[
        \sP_{c\u}=|c|^2\sP_\u.
    \]
    Hence finite conic linearity of
    $\rho\mapsto\mathcal C(S_\omega,\rho)$ gives
    \[
        \QF_{\mathcal C(S_\omega)}[c\u]
        =
        |c|^2
        \QF_{\mathcal C(S_\omega)}[\u].
    \]

    Next, for all $\u,\v\in\cH$, the rank-one operators satisfy the algebraic
    identity
    \[
        \sP_{\u+\v}+\sP_{\u-\v}
        =
        2\sP_\u+2\sP_\v.
    \]
    Applying finite conic linearity again yields
    \[
    \begin{aligned}
        \QF_{\mathcal C(S_\omega)}[\u+\v]
        +
        \QF_{\mathcal C(S_\omega)}[\u-\v]
        &=
        \mathcal C(S_\omega,\sP_{\u+\v})
        +
        \mathcal C(S_\omega,\sP_{\u-\v})
        \\
        &=
        \mathcal C
        \bigl(
            S_\omega,
            \sP_{\u+\v}+\sP_{\u-\v}
        \bigr)
        \\
        &=
        \mathcal C
        \bigl(
            S_\omega,
            2\sP_\u+2\sP_\v
        \bigr)
        \\
        &=
        2\QF_{\mathcal C(S_\omega)}[\u]
        +
        2\QF_{\mathcal C(S_\omega)}[\v].
    \end{aligned}
    \]
    Thus $\QF_{\mathcal C(S_\omega)}$ satisfies the parallelogram identity.

    Therefore
    \[
        \Dom(\QF_{\mathcal C(S_\omega)})
        \coloneqq
        \{
            \u\in\cH:
            \QF_{\mathcal C(S_\omega)}[\u]<\infty
        \}
    \]
    is a linear subspace: scalar closure follows from homogeneity, and if
    $\u,\v\in\Dom(\QF_{\mathcal C(S_\omega)})$, then the parallelogram
    identity gives
    \[
        \QF_{\mathcal C(S_\omega)}[\u+\v]
        \le
        2\QF_{\mathcal C(S_\omega)}[\u]
        +
        2\QF_{\mathcal C(S_\omega)}[\v]
        <\infty.
    \]
    By the polarization identity,
    $\QF_{\mathcal C(S_\omega)}$ is the diagonal of a positive sesquilinear
    form on $\Dom(\QF_{\mathcal C(S_\omega)})$.  Hence
    $\QF_{\mathcal C(S_\omega)}$ is a positive quadratic form.

    Finally, if $\u_n\to\u$ in $\cH$, then
    $\sP_{\u_n}\to\sP_\u$ in trace norm.  Since
    $\rho\mapsto\mathcal C(S_\omega,\rho)$ is lower semi-continuous on
    $\cT(\cH)_+$, it follows that
    $\QF_{\mathcal C(S_\omega)}$ is lower semi-continuous on $\cH$.  By
    \Cref{thm:closed_iff_lsc}, $\QF_{\mathcal C(S_\omega)}$ is a closed
    positive quadratic form. Thus the existence and uniqueness of $C_{S,\omega}$ follow from \Cref{thm:kato_first}. This proves the proposition.
\end{proof}

\begin{definition}[Total Correctness (Cost-Bounded Termination)]
\label{def:total_correctness}
    A quantum Hoare triple is totally correct in the upper-bound sense, denoted as $\models_{tot} \{P\} S \{Q\}$, if for all $\rho \in \pardensity{\cH}$:
    \[
        \Tr(P \rho) \ge \Tr(Q \sem{S}(\rho)) + \mathcal{C}(S, \rho).
    \]
\end{definition}

To justify the physical significance of our upper-bound logic, we establish the fundamental relationship between finite operational cost and program termination.

\begin{proposition}[Finite Expected Cost Implies Almost-Sure Termination]
\label{prop:finite_cost_termination}
    For any quantum program $S$ and any initial partial density operator $\rho \in \pardensity{\cH}$,
    if the expected operational cost is finite, i.e.\ $\mathcal{C}(S, \rho) < \infty$, then $S$ terminates
    almost surely from $\rho$:
    \[
        \Tr(\sem{S}(\rho)) = \Tr(\rho).
    \]
\end{proposition}
\begin{proof}
    We proceed by structural induction on $S$.
    
    \textbf{Base cases:} $\mathbf{skip}$, $\mathbf{abort}$, $q:=0$, $U(\seq{e})[\seq{q}]$.
    For $\mathbf{skip}$, $q:=0$ and unitary operations, the semantics is trace-preserving, so
    $\Tr(\sem{S}(\rho)) = \Tr(\rho)$ holds trivially, regardless of the cost.
    For $\mathbf{abort}$, the definition $\mathcal{C}(\mathbf{abort}, \rho) < \infty$ forces $\Tr(\rho) = 0$,
    and in that case $\Tr(\sem{\mathbf{abort}}(\rho)) = 0 = \Tr(\rho)$. Thus the statement holds.
    
    \textbf{Induction hypothesis:} Assume that for every proper subprogram $S'$ of $S$ and every \(\sigma\in\pardensity{\cH}\),, $\mathcal{C}(S',\sigma) < \infty$ implies $\Tr(\sem{S'}(\sigma)) = \Tr(\sigma)$.
    
    \textbf{Induction step:} We consider the three composite constructs.
    
    \begin{itemize}
        \item \emph{Sequential composition $S_1; S_2$.}
        By Definition~\ref{def:operational_cost},
        \[
            \mathcal{C}(S_1;S_2, \rho) = \mathcal{C}(S_1, \rho) + \mathcal{C}(S_2, \sem{S_1}(\rho)).
        \]
        Finiteness of the left-hand side implies $\mathcal{C}(S_1, \rho) < \infty$ and
        $\mathcal{C}(S_2, \sem{S_1}(\rho)) < \infty$.
        By the induction hypothesis applied to $S_1$, we have $\Tr(\sem{S_1}(\rho)) = \Tr(\rho)$.
        Applying the induction hypothesis to $S_2$ with initial state $\sem{S_1}(\rho)$ gives
        $\Tr(\sem{S_2}(\sem{S_1}(\rho))) = \Tr(\sem{S_1}(\rho)) = \Tr(\rho)$.
        Thus $\Tr(\sem{S_1;S_2}(\rho)) = \Tr(\rho)$.
        
        \item \emph{Continuous measurement
        \(T\equiv\bind(M(\seq e),x.S(x))\).}
        Fix the ambient external valuation.  Write
        \((\Omega_M,\Sigma_M,\mu)\) for the outcome space, and use the local
        abbreviations
        \[
            M_\nu\coloneqq M_{\omega,\nu},
            \qquad
            S_\nu\coloneqq (S(x))_{\omega,\nu},
            \qquad
            \sigma_\nu\coloneqq M_\nu\rho M_\nu^\dagger .
        \]
        By the cost clause for binding,
        \[
            \mathcal C(T,\rho)
            =
            \epsilon\Tr(\rho)
            +
            \int_{\Omega_M}
                \mathcal C(S_\nu,\sigma_\nu)
            \,d\mu(\nu)
            <\infty .
        \]
        Since the integrand is non-negative, this implies
        \[
            \mathcal C(S_\nu,\sigma_\nu)<\infty
            \qquad
            \text{for \(\mu\)-almost every }\nu .
        \]

                For such \(\nu\), we have
        \(\Tr(\sem{S_\nu}(\sigma_\nu))=\Tr(\sigma_\nu)\).  Indeed, if
        \(\Tr(\sigma_\nu)=0\), then \(\sigma_\nu=0\) and the claim follows by
        linearity.  Otherwise
        \(\widehat\sigma_\nu\coloneqq\sigma_\nu/\Tr(\sigma_\nu)\) is a partial
        density operator; finite conic linearity of the cost functional gives
        \(\mathcal C(S_\nu,\widehat\sigma_\nu)<\infty\), so the induction
        hypothesis applies to \(S_\nu\) at \(\widehat\sigma_\nu\), and rescaling
        back by linearity gives the desired equality for \(\sigma_\nu\).

        The denotational semantics of binding is the Bochner integral
        \[
            \sem{T}(\rho)
            =
            \int_{\Omega_M}
                \sem{S_\nu}(\sigma_\nu)
            \,d\mu(\nu).
        \]
        Since the trace is a bounded linear functional on \(\cT(\cH)\), Hille's
        theorem gives
        \[
        \begin{aligned}
            \Tr(\sem{T}(\rho))
            &=
            \int_{\Omega_M}
                \Tr(\sem{S_\nu}(\sigma_\nu))
            \,d\mu(\nu)                                           \\
            &=
            \int_{\Omega_M}
                \Tr(\sigma_\nu)
            \,d\mu(\nu).
        \end{aligned}
        \]

        It remains to identify the last integral.  Let
        \(\rho=\sum_i\lambda_i\sP_{\u_i}\) be a spectral decomposition, with
        \(\lambda_i\ge0\), \(\sum_i\lambda_i\le1\), and
        \(\{\u_i\}\) orthonormal.  By Tonelli's theorem and the normalization of
        the instrument,
        \[
        \begin{aligned}
            \int_{\Omega_M}
                \Tr(M_\nu\rho M_\nu^\dagger)
            \,d\mu(\nu)
            &=
            \sum_i\lambda_i
            \int_{\Omega_M}
                \|M_\nu\u_i\|^2
            \,d\mu(\nu)                                      \\
            &=
            \sum_i\lambda_i\|\u_i\|^2                         \\
            &=
            \Tr(\rho).
        \end{aligned}
        \]
        Therefore
        \[
            \Tr(\sem{T}(\rho))=\Tr(\rho).
        \]
        
        \item \emph{While loop $\mathbf{while}\, M \,\mathbf{do}\, S_{body} \,\mathbf{od}$.}
        Let $\rho_0 \coloneqq \rho$ and $\rho_{k+1} \coloneqq \sem{S_{body}}(M_1 \rho_k M_1^\dagger)$ as in the
        semantics. The total cost is
        \[
            \mathcal{C}(S, \rho) = \sum_{k=0}^\infty \Bigl( \epsilon \Tr(\rho_k) + \mathcal{C}(S_{body}, M_1 \rho_k M_1^\dagger) \Bigr) < \infty .
        \]
        Consequently, each individual term satisfies $\mathcal{C}(S_{body}, M_1 \rho_k M_1^\dagger) < \infty$.
        Applying the induction hypothesis to the subprogram $S_{body}$ with initial state
        $M_1 \rho_k M_1^\dagger$, we obtain
        \[
            \Tr(\rho_{k+1}) = \Tr(\sem{S_{body}}(M_1 \rho_k M_1^\dagger)) = \Tr(M_1 \rho_k M_1^\dagger) .
        \]
        Moreover, the convergence of $\sum_{k} \epsilon \Tr(\rho_k)$ forces $\Tr(\rho_k) \to 0$.
        Using the completeness relation $M_0^\dagger M_0 + M_1^\dagger M_1 = I$, we expand the initial
        trace:
        \begin{align*}
            \Tr(\rho) &= \Tr(M_0 \rho_0 M_0^\dagger) + \Tr(M_1 \rho_0 M_1^\dagger) \\
            &= \Tr(M_0 \rho_0 M_0^\dagger) + \Tr(\rho_1)
               \qquad (\text{by the equality above with } k=0)\\
            &= \Tr(M_0 \rho_0 M_0^\dagger) + \Tr(M_0 \rho_1 M_0^\dagger) + \Tr(\rho_2) \\
            &\;\;\vdots \\
            &= \sum_{k=0}^{K-1} \Tr(M_0 \rho_k M_0^\dagger) + \Tr(\rho_K) .
        \end{align*}
        Letting $K \to \infty$, the remainder $\Tr(\rho_K)$ vanishes because the series $\sum \Tr(\rho_k)$
        converges. Hence
        \[
            \Tr(\rho) = \sum_{k=0}^\infty \Tr(M_0 \rho_k M_0^\dagger) = \Tr(\sem{S}(\rho)),
        \]
        which is exactly the probability of termination.
    \end{itemize}
    This completes the induction.
\end{proof}

\begin{remark}[Expected Cost vs. Absolute Termination]
    It is crucial to note that the converse of Proposition~\ref{prop:finite_cost_termination} does not generally hold. A program might terminate with probability 1 but still exhibit an infinite expected operational cost ($\QF_{C_S}[\u] = \infty$). This is a well-known phenomenon in probability theory, analogous to a symmetric one-dimensional random walk returning to the origin: the return probability is exactly 1, but the expected hitting time diverges. 
    Consequently, our upper-bound logic strictly enforces \emph{Positive Almost-Sure Termination} (termination in expected finite time), providing a much stronger and safer verification guarantee than mere probabilistic termination.
\end{remark}

\subsection{Weakest Preconditions for Total Correctness: Semantics and Rules}
\label{subsec:semantic_total_wp}

With the quadratic nature of the operational cost functional firmly established, we can now provide the definitive semantic characterization of the weakest precondition for total correctness. 

\begin{definition}[Semantic Weakest Precondition for Total Correctness] 
\label{def:semantic_wp}
    For any quantum program $S$ and post-condition predicate $Q \in \Pred$, the semantic weakest precondition $\wp(S, Q)$ is defined as the unique predicate in $\Pred$ satisfying the following two conditions:
    \begin{enumerate}[(1)]
        \item \textbf{Validity:} $\models_{tot} \{\wp(S, Q)\} S \{Q\}$.
        \item \textbf{Minimality:} For any $P \in \Pred$, if $\models_{tot} \{P\} S \{Q\}$, then $\wp(S, Q) \sqsubseteq P$.
    \end{enumerate}
\end{definition}

\begin{theorem}[Algebraic Representation of $\wp$]
\label{thm:wp_algebraic_rep}
    For any quantum program $S$ and predicate $Q \in \Pred$, the exact algebraic representation of $\wp(S, Q)$ is uniquely determined by the form addition of the partial correctness weakest liberal precondition and the expected operational cost. That is, for all $\u \in \cH$:
    \[
        \QF_{\wp(S, Q)}[\u] = \QF_{\wlp(S, Q)}[\u] + \QF_{C_S}[\u],
    \]
    where $\wlp(S, Q)$ is the partial correctness operator (\Cref{def:wlp}) and $C_S \in \Pred$ is the cost predicate derived in \Cref{prop:cost_quadratic_form}.
\end{theorem}

\begin{proof}
    By the definition of total correctness ($\models_{tot} \{P\} S \{Q\}$), evaluating any valid pre-condition $P$ on a pure state $\sP_{\u}$ yields:
    \[
        \QF_P[\u] \ge \Tr\Big(Q \sem{S}(\sP_{\u})\Big) + \mathcal{C}(S, \sP_{\u}) = \QF_{\wlp(S, Q)}[\u] + \QF_{C_S}[\u], \quad \forall \u \in \cH,
    \]
    where the exact equality invokes \Cref{thm:wlp_well_defined} and \Cref{prop:cost_quadratic_form}. 
    Since the form addition of two closed positive predicates natively constitutes a valid predicate in $\Pred$, this point-wise inequality translates directly to the L\"owner order $\wlp(S, Q) + C_S \sqsubseteq P$. This algebraically guarantees both the \textbf{Validity} (when equality holds) and the strict \textbf{Minimality} condition of \Cref{def:semantic_wp}.
\end{proof}

This formulation elegantly unifies the concepts of logical yield and operational resource consumption into a single semantic algebraic object. Building upon this foundational definition, we first establish the core structural properties of the semantic $\wp$. Subsequently, we will introduce the syntactic $\wp$ rules (the $\wp$ calculus) to enable compositional reasoning over program structures. The culmination of this section will be to establish the soundness of this calculus, rigorously proving that our syntax-directed rules flawlessly compute the exact semantic weakest preconditions.

\begin{proposition}[Structural Properties of $\wp$]
\label{prop:wp_properties}
    Let $S$ be a quantum program. The weakest precondition for total correctness, viewed as a predicate transformer $\wp(S, \cdot): \Pred \to \Pred$, satisfies the following fundamental properties:
    \begin{enumerate}[(1)]
        \item \textbf{Monotonicity:} For any $Q_1, Q_2 \in \Pred$, if $Q_1 \sqsubseteq Q_2$, then $\wp(S, Q_1) \sqsubseteq \wp(S, Q_2)$.
        \item \textbf{$\omega$-Continuity:} For any directed (monotonically increasing) sequence of predicates $\{Q_n\}_{n \in \mathbb{N}}$ in $\Pred$, 
        \[
            \wp\left(S, \bigsqcup_{n} Q_n\right) = \bigsqcup_{n} \wp(S, Q_n).
        \]
    \end{enumerate}
\end{proposition}

\begin{proof}
    As established in \Cref{thm:wp_algebraic_rep}, the semantic weakest precondition is exactly determined by the algebraic form addition: $\wp(S, Q) = \wlp(S, Q) + C_S$. 
    By the structural properties of partial correctness (\Cref{thm:wp_properties}), $\wlp(S, \cdot)$ is strictly monotonic and $\omega$-continuous. Since the addition of a fixed closed positive form $C_S$ inherently preserves both the L\"owner order and directed suprema (\Cref{thm:monotone_convergence}), $\wp(S, \cdot)$ trivially inherits Monotonicity and $\omega$-Continuity.
\end{proof}

\begin{remark}[The Affine Nature of $\wp$]
    It is crucial to note that while $\wlp(S, \cdot)$ is a strictly linear predicate transformer (e.g., satisfying countable additivity), $\wp(S, \cdot)$ is inherently \emph{affine} due to the strictly positive cost term $C_S$. Specifically, $\wp(S, Q_1 + Q_2) \neq \wp(S, Q_1) + \wp(S, Q_2)$, as the right-hand side would incorrectly accumulate the operational cost twice. This affine nature correctly captures the physical intuition that resource consumption depends purely on the program's control flow, independent of the logical post-condition being verified.
\end{remark}

Equipped with the rigorous semantic foundation and structural properties, we now operationalize $\wp$ for practical verification. To avoid computing global semantic fixed points manually, we present a syntax-directed calculus. By leveraging the algebraic compositionality of quadratic forms, the following rules construct $\wp(S, Q)$ inductively over the structure of $S$.

\begin{definition}[Weakest Precondition Rules via Quadratic Forms]
\label{def:wp_quadratic_form}
    For any program $S$ and post-condition predicate $Q \in \Pred$, the weakest precondition $\wp(S, Q) \in \Pred$ is syntactically defined by its quadratic form $\QF_{\wp(S, Q)}: \cH \to [0, \infty]$ as follows. Let $\epsilon > 0$ be the uniform minimal base cost.
    \begin{itemize}
        \item \textbf{Basic Operations:} 
        For the empty statement $S \equiv \mathbf{skip}$, $\QF_{\wp(S, Q)}[\u] \coloneqq \QF_Q[\u] + \epsilon \|\u\|^2$.

        For the abort statement $S \equiv \mathbf{abort}$,
        \[
            \QF_{\wp(\mathbf{abort}, Q)}[\u] \coloneqq 
            \begin{cases}
                0, & \text{if } \u = 0, \\
                +\infty, & \text{otherwise}.
            \end{cases}
        \]

        For initialization ($S \equiv q := 0$), let $E_k = \ket{0}_q \bra{k}_q \otimes I_{\text{env}}$ be the global Kraus operators for initialization:         
        \[
            \QF_{\wp(S, Q)}[\u] \coloneqq \sum_{k} \QF_Q[E_k \u] + \epsilon \|\u\|^2.
        \]
        
        For a unitary transformation ($S \equiv U(\seq{e})[\seq{q}]$):
        \[
            \QF_{\wp(S, Q)}[\u] \coloneqq \QF_Q[U_\omega \u] + \epsilon \|\u\|^2.
        \]
        
        \item \textbf{Sequential Composition} ($S \equiv S_1 ; S_2$):
        \[
            \QF_{\wp(S, Q)}[\u] \coloneqq \QF_{\wp(S_1, \wp(S_2, Q))}[\u].
        \]
        
        \item \textbf{Continuous Measurement} ($S \equiv \mathbf{bind}(M(e), x. S(x))$):
        \[
            \QF_{\wp(S, Q)}[\u] \coloneqq \epsilon \|\u\|^2 + \int_{\Omega_M} \QF_{\wp(S_{\omega, \nu}, Q)}[M_{\omega, \nu} \u] \, d\mu(\nu).
        \]
        
        \item \textbf{While Loop} ($S \equiv \mathbf{while} \ M \ \mathbf{do} \ S_{body} \ \mathbf{od}$):
        Let $M = \{M_0, M_1\}$, where $M_0$ exits the loop and $M_1$ continues. We define an ascending chain of predicates $\{X_n\}_{n=0}^\infty \subset \Pred$ representing the $n$-step bounded cost-yield combination. Let the base form be identically zero: $\QF_{X_0}[\u] \coloneqq 0$. For $n \ge 0$:
        \[
            \QF_{X_{n+1}}[\u] \coloneqq \QF_Q[M_0 \u] + \epsilon \|\u\|^2 + \QF_{\wp(S_{body}, X_n)}[M_1 \u].
        \]
        The weakest precondition of the loop is the supremum of this directed sequence:
        \[
            \QF_{\wp(S, Q)}[\u] \coloneqq \sup_{n \ge 0} \QF_{X_{n}}[\u].
        \]
        By \Cref{thm:monotone_convergence}, this supremum over an increasing sequence of closed positive forms uniquely defines a valid predicate in $\Pred$.
    \end{itemize}
\end{definition}

\begin{theorem}[Equivalence of Syntactic and Semantic Total Weakest Preconditions]
\label{thm:wp_equivalence}
    For any quantum program $S$ and post-condition predicate $Q \in \Pred$, the weakest precondition $\wp(S, Q)$ computed via the structural syntactic rules (\Cref{def:wp_quadratic_form}) strictly coincides with the semantic total weakest precondition (\Cref{def:semantic_wp}). 
    That is, for all $\u \in \cH$:
    \[
        \QF_{\wp(S, Q)}[\u] = \QF_{\wlp(S, Q)}[\u] + \QF_{C_S}[\u].
    \]
\end{theorem}

\begin{proof}
    The proof proceeds by structural induction on the syntax of the program $S$. We evaluate both sides on an arbitrary pure state $\u \in \cH$, utilizing the Exact Duality of $\wlp$ (\Cref{thm:wlp_well_defined}) and the structural definition of the cost functional $\mathcal{C}(S, \rho)$ (\Cref{def:operational_cost}).

    \textbf{1. Basic Operations:}
    For $\mathbf{skip}$, the semantic definition yields $\QF_{\wlp}[\u] = \QF_Q[\u]$ and the operational cost is $\QF_{C_{\mathbf{skip}}}[\u] = \epsilon\|\u\|^2$. Their sum is $\QF_Q[\u] + \epsilon\|\u\|^2$, which exactly matches the syntactic rule $\QF_{\wp(\mathbf{skip}, Q)}[\u]$.

    For $\mathbf{abort}$, the $\wlp$ rule gives $\QF_{\wlp(\mathbf{abort}, Q)}[\u] = 0$, while the operational cost is $\QF_{C_{\mathbf{abort}}}[\u] = 0$ if $\u = 0$ and $+\infty$ otherwise. Their sum equals the syntactic rule $\QF_{\wp(\mathbf{abort}, Q)}[\u]$.
    
    For Unitary $U$, $\QF_{\wlp}[\u] = \QF_Q[U\u]$ and the cost is $\epsilon\|\u\|^2$. The sum matches the syntactic rule $\QF_{\wp(U, Q)}[\u] = \QF_Q[U\u] + \epsilon\|\u\|^2$. 
    
    The same trivial equivalence holds for Initialization.

    \textbf{2. Sequential Composition ($S_1 ; S_2$):}
    By the syntactic rule and the induction hypothesis (IH) on $S_1$, we have:
    \begin{align*}
        \QF_{\wp(S_1; S_2, Q)}[\u] &= \QF_{\wp(S_1, \wp(S_2, Q))}[\u] \\
        &= \QF_{\wlp(S_1, \wp(S_2, Q))}[\u] + \QF_{C_{S_1}}[\u]. \quad \text{(by IH on } S_1 \text{)}
    \end{align*}
    Applying the IH on $S_2$, we know $\wp(S_2, Q) = \wlp(S_2, Q) + C_{S_2}$. Because the semantic $\wlp$ transformer is strictly linear over form addition (\Cref{thm:wp_properties}.2), it distributes over the sum:
    \begin{align*}
        \QF_{\wlp(S_1, \wlp(S_2, Q) + C_{S_2})}[\u] &= \QF_{\wlp(S_1, \wlp(S_2, Q))}[\u] + \QF_{\wlp(S_1, C_{S_2})}[\u] \\
        &= \QF_{\wlp(S_1; S_2, Q)}[\u] + \Tr\Big(C_{S_2} \sem{S_1}(\sP_{\u})\Big).
    \end{align*}
    The first term is exactly the semantic partial correctness yield for the composed program. The second term, by the Exact Duality, evaluates to the expected cost of $S_2$ executed on the intermediate state $\sem{S_1}(\sP_{\u})$, which is $\mathcal{C}(S_2, \sem{S_1}(\sP_{\u}))$.
    Substituting this back, the total syntactic form becomes:
    \[
        \QF_{\wlp(S_1; S_2, Q)}[\u] + \Big( \mathcal{C}(S_2, \sem{S_1}(\sP_{\u})) + \mathcal{C}(S_1, \sP_{\u}) \Big).
    \]
    By \Cref{def:operational_cost}, the bracketed sum is exactly the total operational cost $\mathcal{C}(S_1 ; S_2, \sP_{\u}) = \QF_{C_{S_1;S_2}}[\u]$. Thus, $\wp(S_1;S_2, Q)$ computes the exact sum.

        \textbf{3. Continuous Measurement
    (\(T\equiv\bind(M(\seq e),x.S(x))\)).}
    Fix the ambient external valuation and write locally
    \[
        M_\nu\coloneqq M_{\omega,\nu},
        \qquad
        S_\nu\coloneqq (S(x))_{\omega,\nu}.
    \]
    For each fixed \(\u\), the branchwise scalar maps
    \[
        \nu\mapsto \QF_{\wlp(S_\nu,Q)}[M_\nu\u],
        \qquad
        \nu\mapsto \QF_{C_{S_\nu}}[M_\nu\u]
    \]
    are extended non-negative measurable functions; this is exactly the
    measurability supplied by the bind clause for \(\wlp\) and by the
    cost-predicate construction.  Hence the following integrals are
    well-defined as extended non-negative integrals, and their additivity may be
    used without any separate finiteness assumption.

    By the syntactic definition of \(\wp\) and the induction hypothesis for the
    continuation,
    \[
    \begin{aligned}
        \QF_{\wp(T,Q)}[\u]
        &=
        \epsilon\|\u\|^2
        +
        \int_{\Omega_M}
            \QF_{\wp(S_\nu,Q)}[M_\nu\u]
        \,d\mu(\nu)                                                \\
        &=
        \epsilon\|\u\|^2
        +
        \int_{\Omega_M}
        \left(
            \QF_{\wlp(S_\nu,Q)}[M_\nu\u]
            +
            \QF_{C_{S_\nu}}[M_\nu\u]
        \right)
        \,d\mu(\nu)                                                \\
        &=
        \int_{\Omega_M}
            \QF_{\wlp(S_\nu,Q)}[M_\nu\u]
        \,d\mu(\nu)
        +
        \left(
            \epsilon\|\u\|^2
            +
            \int_{\Omega_M}
                \QF_{C_{S_\nu}}[M_\nu\u]
            \,d\mu(\nu)
        \right).
    \end{aligned}
    \]
    The first term is precisely \(\QF_{\wlp(T,Q)}[\u]\).  For the second term,
    the cost-predicate construction gives
    \[
        \QF_{C_{S_\nu}}[M_\nu\u]
        =
        \mathcal C(S_\nu,M_\nu\sP_{\u}M_\nu^\dagger),
    \]
    with the right-hand side interpreted through the trace-class extension
    used in the cost construction.  Therefore the second parenthesis is exactly
    the bind clause for the operational cost:
    \[
        \epsilon\|\u\|^2
        +
        \int_{\Omega_M}
            \QF_{C_{S_\nu}}[M_\nu\u]
        \,d\mu(\nu)
        =
        \mathcal C(T,\sP_\u)
        =
        \QF_{C_T}[\u].
    \]
    Thus
    \[
        \QF_{\wp(T,Q)}[\u]
        =
        \QF_{\wlp(T,Q)}[\u]
        +
        \QF_{C_T}[\u].
    \]

    \textbf{4. While Loop ($\mathbf{while}$):}
    By definition, the semantic total weakest precondition is exactly $\wlp + C_S$. Both components are defined as infinite limits of their $N$-step truncations. 
    Let the single-iteration CPTNI map be $\mathcal{E}_{body}(\rho) \coloneqq \sem{S_{body}}(M_1 \rho M_1^\dagger)$.
    The $N$-step truncated semantic yield on state $\sP_{\u}$ is $Y_N = \sum_{k=0}^{N-1} \Tr\Big(Q M_0 \mathcal{E}_{body}^k(\sP_{\u}) M_0^\dagger\Big)$.
    The $N$-step truncated operational cost is $C_N = \sum_{k=0}^{N-1} \Big[ \epsilon \Tr(\mathcal{E}_{body}^k(\sP_{\u})) + \mathcal{C}(S_{body}, M_1 \mathcal{E}_{body}^k(\sP_{\u}) M_1^\dagger) \Big]$.
    The true semantic value is the supremum $\QF_{\wp}[\u] = \sup_N (Y_N + C_N)$.
    
    We must show that the syntactic sequence $\QF_{X_N}[\u]$ identically tracks $Y_N + C_N$. We proceed by mathematical induction on $N$.
    Base case $N=0$: $X_0 = 0$, and the $0$-step semantic sum is trivially $0$.
    Assume $\QF_{X_N}[\v]$ exactly equals the $N$-step semantic sum (yield + cost) for any input state $\v$. For step $N+1$, the syntactic rule dictates:
    \[
        \QF_{X_{N+1}}[\u] = \QF_Q[M_0 \u] + \epsilon\|\u\|^2 + \QF_{\wp(S_{body}, X_N)}[M_1 \u].
    \]
    By the outer structural induction hypothesis on $S_{body}$, $\wp(S_{body}, X_N)$ is equivalent to its semantic counterpart:
    \begin{align*}
        \QF_{\wp(S_{body}, X_N)}[M_1 \u] &= \QF_{\wlp(S_{body}, X_N)}[M_1 \u] + \QF_{C_{S_{body}}}[M_1 \u] \\
        &= \Tr\Big(X_N \sem{S_{body}}(M_1 \sP_{\u} M_1^\dagger)\Big) + \mathcal{C}(S_{body}, M_1 \sP_{\u} M_1^\dagger) \quad \text{(by Exact Duality)} \\
        &= \Tr\Big(X_N \mathcal{E}_{body}(\sP_{\u})\Big) + \mathcal{C}(S_{body}, M_1 \sP_{\u} M_1^\dagger).
    \end{align*}
    By the inner induction hypothesis on $N$, the term $\Tr(X_N \mathcal{E}_{body}(\sP_{\u}))$ computes exactly the $N$-step yield and cost starting from the intermediate state $\mathcal{E}_{body}(\sP_{\u})$ (which corresponds to iterations $k=1$ to $N$).
    Substituting this back into $X_{N+1}$:
    \[
        \QF_{X_{N+1}}[\u] = \underbrace{\Big[ \QF_Q[M_0 \u] \Big]}_{0\text{-th iter yield}} + \underbrace{\Big[ \epsilon\|\u\|^2 + \mathcal{C}(S_{body}, M_1 \sP_{\u} M_1^\dagger) \Big]}_{0\text{-th iter cost}} + \underbrace{\Big[ \Tr(X_N \mathcal{E}_{body}(\sP_{\u})) \Big]}_{\text{iters } 1 \text{ to } N \text{ (yield+cost)}}.
    \]
    This sum matches $Y_{N+1} + C_{N+1}$ exactly. Since the $N$-step truncations are algebraically identical, their suprema in the $\omega$-CPO coincide:
    \[
        \QF_{\wp(S, Q)}[\u] = \sup_N \QF_{X_N}[\u] = \sup_N (Y_N + C_N).
    \]
    This completes the structural induction, proving the unconditional consistency of the calculus.
\end{proof}

\section{Proof System for Partial Correctness}
\label{sec:proof_system}

In this section, we present the inference rules for the extended quantum Hoare logic and establish its metatheoretical properties for partial correctness. 
The proof system is designed to derive valid triples of the form $\vdash \{P\} S \{Q\}$, where $P, Q \in \Pred$.

\paragraph*{Notational Convention: Pullbacks and Integration of Predicates.}
In the subsequent logical rules and proof systems (Sec.~5 and 6), we frequently express predicate transformations using operator algebraic notation, such as $U^\dagger P U$ or the integral $\int_{\Omega_M} M_\nu^\dagger P_\nu M_\nu \, d\mu(\nu)$. Because a predicate $P \in \Pred$ is generally an unbounded, multivalued linear relation, standard operator multiplication and operator-valued integration are ill-defined due to domain mismatches and topological ambiguities. 

Therefore, these expressions must be strictly understood as syntactic shorthands for operations on \emph{quadratic forms}. Specifically:
\begin{enumerate}
    \item \textbf{Bounded Pullback:} For any bounded operator $K \in \cB(\cH)$, the notation $K^\dagger P K$ denotes the unique predicate in $\Pred$ defined algebraically by its quadratic form:
    \[
        \QF_{K^\dagger P K}[\u] \coloneqq \QF_P[K \u], \quad \forall \u \in \cH.
    \]
    By exact linearity, this notation extends seamlessly to countable Kraus representations, denoted as $\sum_i K_i^\dagger P K_i$.
    
    \item \textbf{Predicate Integration:} For continuous measurements over a measure space $(\Omega_M, \Sigma_M, \mu)$, assume that the predicate family $\{P_\nu\}_{\nu\in\Omega_M}\subseteq\Pred$ is $M$-admissible: for every $\u\in\cH$, the scalar map
    \[
        \nu \longmapsto \QF_{P_\nu}[M_\nu\u]
    \]
    is $(\Sigma_M,\operatorname{Borel}([0,\infty]))$-measurable as an $[0,\infty]$-valued map. Under this side condition, the integral notation $\int_{\Omega_M} M_\nu^\dagger P_\nu M_\nu \, d\mu(\nu)$ denotes the unique predicate in $\Pred$ defined point-wise by the Lebesgue integral of the scalar quadratic forms:
    \[
        \QF_{\int_{\Omega_M} M_\nu^\dagger P_\nu M_\nu \, d\mu(\nu)}[\u] \coloneqq \int_{\Omega_M} \QF_{P_\nu}[M_\nu \u] \, d\mu(\nu), \quad \forall \u \in \cH.
    \]
\end{enumerate}
As established in \Cref{def:wlp_rules} and \Cref{thm:wlp_soundness}, these syntactic notations are well-defined and rigorously grounded in our $\omega$-CPO predicate space.

\subsection{Inference Rules}
\label{subsec:inference_rules}

The proof system, denoted by $\mathcal{L}_{par}$, consists of the following axioms and inference rules. 
The validity of these rules relies on the structurally defined $\wlp$-calculus established in \Cref{sec:logic}.

\begin{enumerate}
    \item \textbf{Axiom of Skip:}
    \[
        \inferrule{ }{ \{ Q \} \ \mathbf{skip} \ \{ Q \} }
    \]

    \item \textbf{Axiom of Abort:}
    \[
        \inferrule{ }{ \{ \mathbf{0} \} \ \mathbf{abort} \ \{ Q \} }
    \]
    Here $\mathbf{0}$ denotes the zero predicate.
    
    \item \textbf{Axiom of Initialization ($q:=0$):}
    Let $E_k = |0\rangle_q \langle k|_q \otimes I_{\text{env}}$ be the initialization Kraus operators.
    \[
        \inferrule{ }{ \left\{ \sum_{k} E_k^\dagger P E_k \right\} \ q:=0 \ \{ P \} }
    \]
    (Note: The sum is defined pointwise via quadratic forms: $\QF_{pre}[\u] = \sum_k \QF_P[E_k \u]$).

    \item \textbf{Axiom of Unitary Transformation ($U[\vec{q}]$):}
    \[
        \inferrule{ }{ \{ U^\dagger P U \} \ U[\vec{q}] \ \{ P \} }
    \]
    
    \item \textbf{Rule of Sequential Composition:}
    \[
        \inferrule{ \{P\} \ S_1 \ \{R\} \\ \{R\} \ S_2 \ \{Q\} }{ \{P\} \ S_1 ; S_2 \ \{Q\} }
    \]

    \item \textbf{Rule of Continuous Measurement ($\mathbf{bind}$):}
    For a continuous measurement $M$ over a measure space $(\Omega_M, \Sigma_M, \mu)$:
    the predicate family $\{P_\nu\}_{\nu\in\Omega_M}$ is required to be $M$-admissible, i.e., $\nu\mapsto\QF_{P_\nu}[M_\nu\u]$ is $\Sigma_M$-measurable for every $\u\in\cH$.
    \[
        \inferrule{ \text{for every } \nu \in \Omega_M, \ \{ P_\nu \} \ S_\nu \ \{ Q \} }{ \left\{ \int_{\Omega_M} M_\nu^\dagger P_\nu M_\nu \, d\mu(\nu) \right\} \ \mathbf{bind}(M, x.S) \ \{ Q \} }
    \]
    (The integral predicate is defined via its form: $\QF_{pre}[\u] = \int_{\Omega_M} \QF_{P_\nu}[M_\nu \u] \, d\mu(\nu)$.)

    \item \textbf{Rule of While Loop (Invariant Rule):}
    For a loop $S \equiv \mathbf{while} \ M \ \mathbf{do} \ S_{body} \ \mathbf{od}$ with measurement $M = \{M_0, M_1\}$, where $M_0$ exits the loop.
    \[
        \inferrule{ P \sqsupseteq M_0^\dagger Q M_0 + M_1^\dagger R M_1 \\ \{ R \} \ S_{body} \ \{ P \} }{ \{ P \} \ \mathbf{while} \ M \ \mathbf{do} \ S_{body} \ \mathbf{od} \ \{ Q \} }
    \]
    Here, $P \in \Pred$ serves as the \emph{loop invariant} for partial correctness. It must safely upper-bound the combined yield of exiting the loop immediately and executing the body for another iteration.

    \item \textbf{Rule of Consequence:}
    \[
        \inferrule{ P \sqsupseteq P' \\ \{P'\} \ S \ \{Q'\} \\ Q' \sqsupseteq Q }{ \{P\} \ S \ \{Q\} }
    \]
    (Recall that in our Upper-Bound Logic, $P \sqsupseteq P'$ means $P$ requires \emph{more} initial resource than $P'$, hence it is a logically weaker/safer precondition.)
\end{enumerate}

\subsection{Soundness}

\begin{theorem}[Soundness of $\mathcal{L}_{par}$]
\label{Thm:soundness}
    The proof system $\mathcal{L}_{par}$ is sound. If $\vdash \{P\} S \{Q\}$ is derivable, then it is semantically valid: $\models_{par} \{P\} S \{Q\}$.
\end{theorem}
\begin{proof}
    By the Semantic Minimality established in \Cref{thm:wlp_well_defined}, $\models_{par} \{P\} S \{Q\}$ is mathematically equivalent to $P \sqsupseteq \wlp(S, Q)$. We prove soundness by induction on the structure of the derivation rules, showing that every derived precondition $P$ covers the semantic $\wlp$.

    \textbf{1. Basic Axioms:}
    For $\mathbf{skip}$, Unitary, and Initialization, the precondition given in the axioms exactly coincides with the structural $\wlp$ rules (\Cref{def:wlp_rules}). Thus, equality $P = \wlp(S, Q)$ holds trivially, which implies $P \sqsupseteq \wlp(S, Q)$.

    For $\mathbf{abort}$, the axiom gives $P = \mathbf{0}$ and the structural $\wlp$ rule gives $\wlp(\mathbf{abort}, Q) = \mathbf{0}$. Hence $\mathbf{0} \sqsupseteq \mathbf{0}$ trivially.

    \textbf{2. Sequential Composition:}
    Assume $P \sqsupseteq \wlp(S_1, R)$ and $R \sqsupseteq \wlp(S_2, Q)$. By the monotonicity of the $\wlp$ transformer (\Cref{thm:wp_properties}.1), $\wlp(S_1, R) \sqsupseteq \wlp(S_1, \wlp(S_2, Q))$. Thus $P \sqsupseteq \wlp(S_1 ; S_2, Q)$.

    \textbf{3. Continuous Measurement ($\mathbf{bind}$):}
    Assume $\models_{par} \{ P_\nu \} S_\nu \{ Q \}$ for almost all $\nu \in \Omega_M$. By induction hypothesis, $P_\nu \sqsupseteq \wlp(S_\nu, Q)$, which means $\QF_{P_\nu}[\v] \ge \QF_{\wlp(S_\nu, Q)}[\v]$ for all $\v$.
    The derived precondition is $P_{pre} = \int_{\Omega_M} M_\nu^\dagger P_\nu M_\nu \, d\mu(\nu)$.
    Evaluating its quadratic form for any $\u \in \cH$:
    \begin{align*}
        \QF_{P_{pre}}[\u] &= \int_{\Omega_M} \QF_{P_\nu}[M_\nu \u] \, d\mu(\nu) \\
        &\ge \int_{\Omega_M} \QF_{\wlp(S_\nu, Q)}[M_\nu \u] \, d\mu(\nu) \\
        &= \QF_{\wlp(\mathbf{bind}, Q)}[\u].
    \end{align*}
    Thus $P_{pre} \sqsupseteq \wlp(\mathbf{bind}, Q)$, proving the rule sound.

    \textbf{4. While Loop ($\mathbf{while}$):}
    Assume $P \sqsupseteq M_0^\dagger Q M_0 + M_1^\dagger R M_1$ and $\models_{par} \{ R \} S_{body} \{ P \}$. 
    By the induction hypothesis, $R \sqsupseteq \wlp(S_{body}, P)$. 
    Substituting this into the invariant inequality, we obtain:
    \[
        P \sqsupseteq M_0^\dagger Q M_0 + M_1^\dagger \wlp(S_{body}, P) M_1.
    \]
    We must show $P \sqsupseteq \wlp(S, Q) = \bigsqcup_n X_n$. We proceed by induction on $n$.
    For $n=0$, $X_0 = 0 \sqsubseteq P$ is trivial since $P$ is positive.
    Assume $X_n \sqsubseteq P$. By the monotonicity of $\wlp$, $\wlp(S_{body}, X_n) \sqsubseteq \wlp(S_{body}, P)$. 
    For the next iteration:
    \begin{align*}
        X_{n+1} &= M_0^\dagger Q M_0 + M_1^\dagger \wlp(S_{body}, X_n) M_1 \\
        &\sqsubseteq M_0^\dagger Q M_0 + M_1^\dagger \wlp(S_{body}, P) M_1 \\
        &\sqsubseteq P.
    \end{align*}
    Since $P$ is an upper bound for the entire sequence $\{X_n\}$, it must bound their supremum. Thus, $P \sqsupseteq \bigsqcup_n X_n = \wlp(S, Q)$, ensuring semantic validity.

    \textbf{5. Consequence:}
    Assume $P \sqsupseteq P'$, $P' \sqsupseteq \wlp(S, Q')$, and $Q' \sqsupseteq Q$. 
    By the monotonicity of $\wlp$, $Q' \sqsupseteq Q \implies \wlp(S, Q') \sqsupseteq \wlp(S, Q)$. 
    Transitivity of the L\"owner order yields $P \sqsupseteq P' \sqsupseteq \wlp(S, Q') \sqsupseteq \wlp(S, Q)$, maintaining soundness.
\end{proof}

\subsection{Relative Completeness}

\begin{theorem}[Relative Completeness of $\mathcal{L}_{par}$]
\label{Thm:completeness}
    The proof system $\mathcal{L}_{par}$ is complete relative to the underlying mathematics of linear relations. If $\models_{par} \{P\} S \{Q\}$, then $\vdash \{P\} S \{Q\}$ is derivable in $\mathcal{L}_{par}$.
\end{theorem}
\begin{proof}
    Suppose $\models_{par} \{P\} S \{Q\}$. By Semantic Minimality, we have $P \sqsupseteq \wlp(S, Q)$. 
    If we can prove that the specific triple $\vdash \{\wlp(S, Q)\} S \{Q\}$ is always derivable, we can simply apply the \textbf{Rule of Consequence} (with $P' = \wlp(S, Q)$ and $Q' = Q$) to immediately derive $\vdash \{P\} S \{Q\}$. 
    Thus, it suffices to prove $\vdash \{\wlp(S, Q)\} S \{Q\}$ for all programs $S$ by structural induction.

    \textbf{1. Basic Axioms and Composition:}
    For $\mathbf{skip}$, Initialization, Unitary, and Sequential Composition, the derivability follows identically to classical Hoare logic by directly applying the respective axioms/rules backwards from $Q$.
    
     For $\mathbf{abort}$, the semantic weakest liberal precondition is $\wlp(\mathbf{abort}, Q) = \mathbf{0}$ (by \Cref{def:wlp_rules}). The triple $\vdash \{\mathbf{0}\} \ \mathbf{abort} \ \{Q\}$ is exactly the \textbf{Axiom of Abort}, hence trivially derivable.
     
    \textbf{2. Continuous Measurement ($\mathbf{bind}$):}
    Let $P_\nu \coloneqq \wlp(S_\nu, Q)$. By the induction hypothesis, $\vdash \{P_\nu\} S_\nu \{Q\}$ is derivable for all $\nu \in \Omega_M$.
    Applying the \textbf{Rule of Continuous Measurement}, we immediately derive:
    \[
        \vdash \left\{ \int_{\Omega_M} M_\nu^\dagger \wlp(S_\nu, Q) M_\nu \, d\mu(\nu) \right\} \ \mathbf{bind}(M, x.S_x) \ \{ Q \}.
    \]
    By \Cref{def:wlp_rules}, the derived precondition is strictly equal to $\wlp(\mathbf{bind}, Q)$.

    \textbf{3. While Loop ($\mathbf{while}$):}
    Let $S \equiv \mathbf{while} \ M \ \mathbf{do} \ S_{body} \ \mathbf{od}$. We must derive $\vdash \{\wlp(S, Q)\} S \{Q\}$.
    Let $P_{inv} \coloneqq \wlp(S, Q) = \bigsqcup_n X_n$. 
    Because the $\wlp$ transformer is Scott-continuous on the $\omega$-CPO of predicates (\Cref{thm:wp_properties}), the supremum $P_{inv}$ acts as the least fixed point of the semantic unfolding equation:
    \[
        P_{inv} = M_0^\dagger Q M_0 + M_1^\dagger \wlp(S_{body}, P_{inv}) M_1.
    \]
    Let $R \coloneqq \wlp(S_{body}, P_{inv})$. By the structural induction hypothesis for $S_{body}$, the triple $\vdash \{R\} S_{body} \{P_{inv}\}$ is derivable.
    We also possess the exact algebraic equality (which trivially satisfies the $\sqsupseteq$ condition in the rule premise):
    \[
        P_{inv} \sqsupseteq M_0^\dagger Q M_0 + M_1^\dagger R M_1.
    \]
    Taking $P_{inv}$ as the loop invariant and $R$ as the intermediate predicate, we apply the \textbf{Rule of While Loop} to directly derive:
    \[
        \vdash \{P_{inv}\} \ \mathbf{while} \ M \ \mathbf{do} \ S_{body} \ \mathbf{od} \ \{Q\}.
    \]
    Since $P_{inv}$ is precisely $\wlp(S, Q)$, the completeness is established.
\end{proof}
\section{Proof System for Total Correctness}
\label{sec:proof_system_total}

In this section, we present the proof system $\mathcal{L}_{tot}$ for total correctness, establishing strict guarantees for both the expected yield of the post-condition and the Positive Almost-Sure Termination (AST) of the quantum program. 
The system derives valid triples denoted by $\vdash_{tot} \{P\} S \{Q\}$, where $P$ acts as a unified upper bound bounding both the logical yield and the operational cost.

\subsection{Inference Rules}
\label{subsec:inference_rules_tot}

The rules of $\mathcal{L}_{tot}$ closely mirror those of partial correctness ($\mathcal{L}_{par}$), reflecting the structural harmony of our calculus. The critical departure lies in the `while` loop, which now mandates the inclusion of a strict ranking witness (potential function) to enforce termination. 
The validity of these rules is grounded in the total weakest precondition calculus ($\wp$) established earlier, where $\wp(S, Q) \coloneqq \wlp(S, Q) + \mathcal{C}(S)$.

\begin{enumerate}
    \item \textbf{Axiom of Skip:}
    \[
        \inferrule{ }{ \vdash_{tot} \{ Q + \epsilon I \} \ \mathbf{skip} \ \{ Q \} }
    \]

    \item \textbf{Axiom of Abort:}
    \[
        \inferrule{ }{ \vdash_{tot} \{ \mathbf{\infty} \} \ \mathbf{abort} \ \{ Q \} }
    \]
    where $\mathbf{\infty}$ is defined via its quadratic form as in \Cref{def:wp_quadratic_form}: 
    $\QF_{\mathbf{\infty}}[\u] = 0$ if $\u = 0$, and $+\infty$ otherwise.
    
    \item \textbf{Axiom of Initialization ($q:=0$):}
    Let $E_k = |0\rangle_q \langle k|_q \otimes I_{\text{env}}$ be the initialization Kraus operators.
    \[
        \inferrule{ }{ \vdash_{tot} \left\{ \epsilon I + \sum_{k} E_k^\dagger P E_k \right\} \ q:=0 \ \{ P \} }
    \]

    \item \textbf{Axiom of Unitary Transformation ($U[\vec{q}]$):}
    \[
        \inferrule{ }{ \vdash_{tot} \{ \epsilon I + U^\dagger P U \} \ U[\vec{q}] \ \{ P \} }
    \]
    
    \item \textbf{Rule of Sequential Composition:}
    \[
        \inferrule{ \vdash_{tot} \{P\} \ S_1 \ \{R\} \\ \vdash_{tot} \{R\} \ S_2 \ \{Q\} }{ \vdash_{tot} \{P\} \ S_1 ; S_2 \ \{Q\} }
    \]

    \item \textbf{Rule of Continuous Measurement ($\mathbf{bind}$):}
    For a continuous measurement $M$ over a measure space $(\Omega_M, \Sigma_M, \mu)$ with base execution cost $\epsilon > 0$:
    the predicate family $\{P_\nu\}_{\nu\in\Omega_M}$ is required to be $M$-admissible, i.e., $\nu\mapsto\QF_{P_\nu}[M_\nu\u]$ is $\Sigma_M$-measurable for every $\u\in\cH$.
    \[
        \inferrule{ \forall \nu \in \Omega_M, \ \vdash_{tot} \{ P_\nu \} \ S_\nu \ \{ Q \} }{ \vdash_{tot} \left\{ \epsilon I + \int_{\Omega_M} M_\nu^\dagger P_\nu M_\nu \, d\mu(\nu) \right\} \ \mathbf{bind}(M, x.S_x) \ \{ Q \} }
    \]

    \item \textbf{Rule of While Loop (Invariant and Potential Rule):}
    For a loop $S \equiv \mathbf{while} \ M \ \mathbf{do} \ S_{body} \ \mathbf{od}$ with measurement $M = \{M_0, M_1\}$, where $M_0$ exits the loop, and $\epsilon > 0$ is the base cost of guard evaluation.
    \[
        \inferrule{ P \sqsupseteq M_0^\dagger Q M_0 + \epsilon I + M_1^\dagger R M_1 \\ \vdash_{tot} \{ R \} \ S_{body} \ \{ P \} }{ \vdash_{tot} \{ P \} \ \mathbf{while} \ M \ \mathbf{do} \ S_{body} \ \mathbf{od} \ \{ Q \} }
    \]
    
    \begin{remark}[Unified Invariant and Ranking Witness]
        This rule captures the essence of total correctness in our upper-bound logic. The predicate $P \in \Pred$ serves simultaneously as a \emph{logical loop invariant} (bounding the yield $Q$) and a \emph{decreasing potential function} (bounding the termination cost). The explicit inclusion of the coercive term $\epsilon I$ algebraically guarantees that each iteration strictly consumes at least $\epsilon$ physical resource. By \Cref{prop:finite_cost_termination}, if $\QF_P[\u] < \infty$, the expected cost is bounded, forcing the probability of infinite non-termination to strictly vanish (Positive AST).
    \end{remark}

    \item \textbf{Rule of Consequence:}
    \[
        \inferrule{ P \sqsupseteq P' \\ \vdash_{tot} \{P'\} \ S \ \{Q'\} \\ Q' \sqsupseteq Q }{ \vdash_{tot} \{P\} \ S \ \{Q\} }
    \]
\end{enumerate}

\subsection{Soundness}

\begin{theorem}[Soundness of $\mathcal{L}_{tot}$]
\label{thm:soundness_tot}
    The proof system $\mathcal{L}_{tot}$ is sound. If $\vdash_{tot} \{P\} S \{Q\}$ is derivable, then it is semantically valid for total correctness: $P \sqsupseteq \wp(S, Q)$.
\end{theorem}
\begin{proof}
    The proof proceeds by structural induction on the derivation rules, ensuring that every derived precondition $P$ covers the total semantic $\wp$, which inherently encapsulates both partial correctness and expected operational cost.

   \textbf{1. Basic Operations:}
    For $\mathbf{skip}$, the derived precondition is $P_{syn} = Q + \epsilon I$. Evaluating its quadratic form yields $\QF_{P_{syn}}[\u] = \QF_Q[\u] + \epsilon\norm{\u}^2$. Since $\sem{\mathbf{skip}}(\sP_{\u}) = \sP_{\u}$ and its operational cost is $\epsilon\norm{\u}^2$, this exactly matches $\QF_{\wp(\mathbf{skip}, Q)}[\u]$. Thus $P_{syn} \sqsupseteq \wp(\mathbf{skip}, Q)$.

    For $\mathbf{abort}$, the axiom gives $P_{syn} = \wp(\mathbf{abort}, Q)$ directly.
    
    For Initialization $q:=0$, the semantics is $\sem{S}(\sP_{\u}) = \sum_k \sP_{E_k \u}$ with cost $\epsilon\norm{\u}^2$. The rule gives $P_{syn} = \epsilon I + \sum_k E_k^\dagger P E_k$. Its quadratic form is $\epsilon\norm{\u}^2 + \sum_k \QF_P[E_k \u]$, strictly matching $\QF_{\wp(q:=0, P)}[\u]$.
    
    For Unitary $U$, the derived precondition is $P_{syn} = \epsilon I + U^\dagger P U$. Its form is $\epsilon\norm{\u}^2 + \QF_P[U\u]$. This matches the semantic yield $\QF_P[U\u]$ plus the operation cost $\epsilon\norm{\u}^2$, which exactly equals $\QF_{\wp(U, P)}[\u]$. 
    Therefore, all basic axioms unconditionally satisfy $P_{syn} \sqsupseteq \wp(S, Q)$.

    \textbf{2. Sequential Composition ($S_1;S_2$):}
    Assume the premises $\vdash_{tot} \{P\} S_1 \{R\}$ and $\vdash_{tot} \{R\} S_2 \{Q\}$.
    By the induction hypothesis, we have $P \sqsupseteq \wp(S_1, R)$ and $R \sqsupseteq \wp(S_2, Q)$.
    From the monotonicity of the $\wp$ transformer (\Cref{prop:wp_properties}~(1)), we deduce
    \[
        \wp(S_1, R) \sqsupseteq \wp(S_1, \wp(S_2, Q)).
    \]
    By the structural equivalence established in \Cref{thm:wp_equivalence} (or Definition~\ref{def:wp_quadratic_form}),
    the weakest precondition for sequential composition satisfies
    $\wp(S_1;S_2, Q) = \wp(S_1, \wp(S_2, Q))$.
    Combining these inequalities via the transitivity of the L\"owner order,
    \[
        P \sqsupseteq \wp(S_1, R) \sqsupseteq \wp(S_1, \wp(S_2, Q)) = \wp(S_1;S_2, Q),
    \]
    which proves the soundness of the sequential composition rule.
    
    \textbf{3. Continuous Measurement ($\mathbf{bind}$):}
    Assume $P_\nu \sqsupseteq \wp(S_\nu, Q)$ for almost all $\nu \in \Omega_M$.
    The derived precondition is $P_{pre} = \epsilon I + \int_{\Omega_M} M_\nu^\dagger P_\nu M_\nu \, d\mu(\nu)$.
    Evaluating its quadratic form for any $\u \in \cH$ yields:
    \begin{align*}
        \QF_{P_{pre}}[\u] &= \epsilon \norm{\u}^2 + \int_{\Omega_M} \QF_{P_\nu}[M_\nu \u] \, d\mu(\nu) \\
        &\ge \epsilon \norm{\u}^2 + \int_{\Omega_M} \QF_{\wp(S_\nu, Q)}[M_\nu \u] \, d\mu(\nu) \\
        &= \QF_{\wp(\mathbf{bind}, Q)}[\u].
    \end{align*}
    Thus $P_{pre} \sqsupseteq \wp(\mathbf{bind}, Q)$, maintaining soundness.

    \textbf{4. While Loop ($\mathbf{while}$):}
    Assume $P \sqsupseteq M_0^\dagger Q M_0 + \epsilon I + M_1^\dagger R M_1$ and $\vdash_{tot} \{ R \} S_{body} \{ P \}$. 
    By the induction hypothesis, $R \sqsupseteq \wp(S_{body}, P)$. 
    Substituting this into the unified invariant inequality, we obtain the one-step unfolding bound:
    \[
        P \sqsupseteq M_0^\dagger Q M_0 + \epsilon I + M_1^\dagger \wp(S_{body}, P) M_1.
    \]
    Recall that the total weakest precondition is exactly the supremum of the $n$-step bounds: $\wp(S, Q) = \bigsqcup_n X_n$, where $X_0 = 0$ and $X_{n+1} = M_0^\dagger Q M_0 + \epsilon I + M_1^\dagger \wp(S_{body}, X_n) M_1$.
    We prove $X_n \sqsubseteq P$ by mathematical induction on $n$.
    For $n=0$, $X_0 = 0 \sqsubseteq P$ is trivially satisfied as $P$ is a positive linear relation.
    Assume $X_n \sqsubseteq P$. By the strict monotonicity of the $\wp$ transformer, $\wp(S_{body}, X_n) \sqsubseteq \wp(S_{body}, P)$. 
    Evaluating the recurrence for $X_{n+1}$:
    \begin{align*}
        X_{n+1} &= M_0^\dagger Q M_0 + \epsilon I + M_1^\dagger \wp(S_{body}, X_n) M_1 \\
        &\sqsubseteq M_0^\dagger Q M_0 + \epsilon I + M_1^\dagger \wp(S_{body}, P) M_1 \\
        &\sqsubseteq P.
    \end{align*}
    Since $P$ acts as a strict upper bound for the entire directed sequence $\{X_n\}$, it must rigorously bound their supremum in the $\omega$-CPO. Thus, $P \sqsupseteq \bigsqcup_n X_n = \wp(S, Q)$. This ensures that $P$ is a valid total correctness precondition.

    \textbf{5. Consequence:}
    Assume $P \sqsupseteq P'$, $P' \sqsupseteq \wp(S, Q')$, and $Q' \sqsupseteq Q$. 
    By the monotonicity of $\wp$, $Q' \sqsupseteq Q \implies \wp(S, Q') \sqsupseteq \wp(S, Q)$. 
    Transitivity of the L\"owner order gives $P \sqsupseteq P' \sqsupseteq \wp(S, Q') \sqsupseteq \wp(S, Q)$, preserving soundness.
\end{proof}

\subsection{Relative Completeness}

\begin{theorem}[Relative Completeness of $\mathcal{L}_{tot}$]
\label{Thm:completeness_tot}
    The proof system $\mathcal{L}_{tot}$ is complete relative to the underlying mathematics of linear relations. If a precondition $P$ semantically guarantees total correctness, i.e., $P \sqsupseteq \wp(S, Q)$, then $\vdash_{tot} \{P\} S \{Q\}$ is structurally derivable in $\mathcal{L}_{tot}$.
\end{theorem}
\begin{proof}
    By the Semantic Minimality of total $\wp$, any valid precondition $P$ satisfies $P \sqsupseteq \wp(S, Q)$. 
    If we can structurally derive the exact minimal triple $\vdash_{tot} \{\wp(S, Q)\} S \{Q\}$, we can immediately invoke the \textbf{Rule of Consequence} (using $P' = \wp(S, Q)$ and $Q' = Q$) to derive $\vdash_{tot} \{P\} S \{Q\}$. 
    Thus, we only need to prove that $\vdash_{tot} \{\wp(S, Q)\} S \{Q\}$ is derivable for all programs $S$ via structural induction.

    \textbf{1. Basic Operations:}
    For $\mathbf{skip}$, $\mathbf{abort}$, Initialization, Unitary and Sequential Composition, the derivability exactly mirrors the partial correctness completeness proof (\Cref{Thm:completeness}), relying on the structural equivalence between the semantic pullback and the syntactic rules.
    
    \textbf{2. Continuous Measurement ($\mathbf{bind}$):}
    Let $P_\nu \coloneqq \wp(S_\nu, Q)$. By the induction hypothesis, $\vdash_{tot} \{P_\nu\} S_\nu \{Q\}$ is derivable for all $\nu \in \Omega_M$.
    Applying the \textbf{Rule of Continuous Measurement}, we directly derive:
    \[
        \vdash_{tot} \left\{ \epsilon I + \int_{\Omega_M} M_\nu^\dagger \wp(S_\nu, Q) M_\nu \, d\mu(\nu) \right\} \ \mathbf{bind}(M, x.S_x) \ \{ Q \}.
    \]
    By the structural definition of the total weakest precondition, the derived precondition is strictly equal to $\wp(\mathbf{bind}, Q)$.
    
    \textbf{3. While Loop ($\mathbf{while}$):}
    Let $S \equiv \mathbf{while} \ M \ \mathbf{do} \ S_{body} \ \mathbf{od}$. We seek to derive $\vdash_{tot} \{\wp(S, Q)\} S \{Q\}$.
    Let the minimal total precondition be $P_{inv} \coloneqq \wp(S, Q) = \bigsqcup_n X_n$. 
    Because the $\wp$ transformer is Scott-continuous on the extended $\omega$-CPO of predicates (\Cref{thm:wp_properties}), the supremum $P_{inv}$ acts exactly as the least fixed point of the semantic unfolding equation augmented with the guard cost:
    \[
        P_{inv} = M_0^\dagger Q M_0 + \epsilon I + M_1^\dagger \wp(S_{body}, P_{inv}) M_1.
    \]
    Define the intermediate predicate $R \coloneqq \wp(S_{body}, P_{inv})$. By the structural induction hypothesis for the loop body $S_{body}$, the triple $\vdash_{tot} \{R\} S_{body} \{P_{inv}\}$ is strictly derivable.
    Substituting $R$ into the fixed-point equation, we obtain the exact algebraic equality:
    \[
        P_{inv} = M_0^\dagger Q M_0 + \epsilon I + M_1^\dagger R M_1.
    \]
    This exact equality trivially satisfies the inequality condition ($\sqsupseteq$) required by the premise of the While rule.
    Treating $P_{inv}$ as the unified invariant-potential predicate and $R$ as the intermediate predicate, we invoke the \textbf{Rule of While Loop} to syntactically derive:
    \[
        \vdash_{tot} \{P_{inv}\} \ \mathbf{while} \ M \ \mathbf{do} \ S_{body} \ \mathbf{od} \ \{Q\}.
    \]
    Since $P_{inv}$ is exactly $\wp(S, Q)$, the completeness of the calculus is definitively established.
\end{proof}
\section{Case Study I: Symmetric Random Walk as a Quantum Program}
\label{Sec:random_walk}

To illustrate the expressive power and the necessity of our quantitative cost logic, we analyze a quantum program that realizes the classical symmetric random walk on $\mathbb{Z}$.
We prove that this program terminates almost surely (AST), but that its expected running time is infinite.
The AST assertion itself is a bounded, topological termination property; the divergent first moment is the genuinely quantitative feature that bounded quantum Hoare logics cannot express as a finite cost predicate.
Thus this example separates almost-sure termination from finite expected operational cost in an infinite-dimensional state space.

\subsection{Setting and Program Definition}

The program uses two quantum variables:
\begin{itemize}
    \item $q$: a qubit coin, with Hilbert space $\cH_q = \mathbb{C}^2$;
    \item $w$: a quantum walker on the integer lattice, with Hilbert space $\cH_w = \ell^2(\mathbb{Z})$ and computational basis $\{\ket{n}\}_{n\in\mathbb{Z}}$.
\end{itemize}
Let $U_{\mathrm{sh}}$ be the bilateral shift on $\ell^2(\mathbb{Z})$, defined by $U_{\mathrm{sh}}\ket{n}=\ket{n+1}$.
The controlled walk unitary $U_{\mathrm{walk}}$ on $\cH_q \otimes \cH_w$ is defined on the computational basis by
\[
U_{\mathrm{walk}}\ket{0,n} = \ket{0,n+1}, \qquad
U_{\mathrm{walk}}\ket{1,n} = \ket{1,n-1}.
\]
It is unitary because it permutes the canonical orthonormal basis.
Let $H$ be the Hadamard gate on $\cH_q$, so that $H\ket{0}=\ket{+}$ and $H\ket{1}=\ket{-}$.
The loop guard is the binary projective measurement $M=\{M_0,M_1\}$, where
\[
\Pi_0 \coloneqq \ket{0}_w\!\bra{0},
\qquad
\Pi_{\neq 0} \coloneqq I_w-\Pi_0,
\qquad
M_0 = I_q\otimes \Pi_0,
\qquad
M_1 = I_q\otimes \Pi_{\neq 0}.
\]
Outcome $0$ exits the loop and outcome $1$ continues.
The full program $S$ initializes the walker at $1$ and then runs the loop $W$:
\begin{align*}
    S \equiv\ &w := \ket{1};\ W,\qquad\text{where}\\
    W \equiv\ &\mathbf{while}\ M[w] = 1\ \mathbf{do}\\
    &\qquad q := \ket{0};\\
    &\qquad q := Hq;\\
    &\qquad q,w := U_{\mathrm{walk}}(q,w);\\
    &\mathbf{od}.
\end{align*}
Throughout this case study we set the uniform one-step cost to $\epsilon=1$.
The leading initialization contributes only one finite unit of cost and is trace preserving, so the divergence and AST arguments reduce to the loop-entry state
\[
\rho_1 \coloneqq \ket{0}\!\bra{0}_q\otimes \ket{1}\!\bra{1}_w.
\]

\subsection{Expected Runtime is Infinite}

We first compute the total weakest precondition of the loop body.
Let
\[
S_{\mathrm{body}} \equiv q:=\ket{0};\ q:=Hq;\ q,w:=U_{\mathrm{walk}}(q,w).
\]

\begin{lemma}\label{lem:wp_body}
For every positive predicate of the form $X=I_q\otimes X_w$, the total weakest precondition of the loop body is
\[
\mathsf{wp}(S_{\mathrm{body}},X)
= I_q\otimes\bigl(3I_w+\mathcal{E}^*(X_w)\bigr),
\]
where
\[
\mathcal{E}^*(Y)
=\frac{1}{2}U_{\mathrm{sh}}^\dagger YU_{\mathrm{sh}}
+\frac{1}{2}U_{\mathrm{sh}}YU_{\mathrm{sh}}^\dagger
\]
is the Heisenberg-picture transformer of one unbiased classical step.
For unbounded $X_w$, the equality is understood as an equality of closed positive quadratic forms.
\end{lemma}
\begin{proof}
We first prove the identity for bounded $X_w$; the unbounded case then follows by spectral truncation and monotone convergence of the $\mathsf{wp}$ transformer.
Let
\[
E_i \coloneqq \ket{0}_q\!\bra{i}\otimes I_w\qquad (i=0,1)
\]
be the Kraus operators for the reset command $q:=\ket{0}$.
Set $Y\coloneqq U_{\mathrm{walk}}^\dagger XU_{\mathrm{walk}}$.
By the syntax-directed total-$\mathsf{wp}$ rules with $\epsilon=1$,
\begin{align*}
\mathsf{wp}(q,w:=U_{\mathrm{walk}},X)&=I+Y,\\
\mathsf{wp}(q:=Hq;\ q,w:=U_{\mathrm{walk}},X)
    &=2I+(H^\dagger\otimes I_w)Y(H\otimes I_w),\\
\mathsf{wp}(S_{\mathrm{body}},X)
    &=3I+
      \sum_{i=0}^1 E_i^\dagger
      (H^\dagger\otimes I_w)Y(H\otimes I_w)E_i.
\end{align*}
The first term is exactly the additive cost of the three elementary commands.
It remains to compute $Y$.
For $X=I_q\otimes X_w$, the controlled shift gives
\[
U_{\mathrm{walk}}^\dagger XU_{\mathrm{walk}}
=
\ket{0}\!\bra{0}_q\otimes U_{\mathrm{sh}}^\dagger X_wU_{\mathrm{sh}}
+
\ket{1}\!\bra{1}_q\otimes U_{\mathrm{sh}}X_wU_{\mathrm{sh}}^\dagger.
\]
Put
\[
B_0=U_{\mathrm{sh}}^\dagger X_wU_{\mathrm{sh}},
\qquad
B_1=U_{\mathrm{sh}}X_wU_{\mathrm{sh}}^\dagger .
\]
Then the predicate obtained after the controlled shift is
\[
A=\ket0_q\!\bra0\otimes B_0+\ket1_q\!\bra1\otimes B_1 .
\]
Equivalently, with respect to the coin decomposition
\(\mathcal H_q\otimes\mathcal H_w
=(\ket0_q\otimes\mathcal H_w)\oplus(\ket1_q\otimes\mathcal H_w)\),
\[
A=
\begin{pmatrix}
B_0&0\\
0&B_1
\end{pmatrix}.
\]
Since
\[
H=\frac1{\sqrt2}
\begin{pmatrix}
1&1\\
1&-1
\end{pmatrix},
\]
a direct multiplication in the two-dimensional coin factor gives
\[
C:=(H^\dagger\otimes I_w)A(H\otimes I_w)
=
\frac12
\begin{pmatrix}
B_0+B_1&B_0-B_1\\
B_0-B_1&B_0+B_1
\end{pmatrix}.
\]
For the reset assignment \(q:=\ket0\), the Kraus operators are
\[
E_i=\ket0_q\!\bra i\otimes I_w,\qquad i=0,1.
\]
Therefore
\[
E_i^\dagger C E_i
=
\ket i_q\!\bra i\otimes \frac{B_0+B_1}{2},
\]
and hence
\[
\sum_{i=0}^1E_i^\dagger C E_i
=
I_q\otimes \frac{B_0+B_1}{2}.
\]
Substituting the definitions of \(B_0\) and \(B_1\), we obtain
\[
\sum_{i=0}^1
E_i^\dagger(H^\dagger\otimes I_w)A(H\otimes I_w)E_i
=
I_q\otimes
\left(
\frac12U_{\mathrm{sh}}^\dagger X_wU_{\mathrm{sh}}
+
\frac12U_{\mathrm{sh}}X_wU_{\mathrm{sh}}^\dagger
\right)
=
I_q\otimes\mathcal E^*(X_w).
\]
Combining this with the additive cost $3I$ proves the formula for bounded $X_w$.
For a general positive predicate $X_w$, let $X_w^{(m)}$ be its increasing bounded spectral truncations.
The bounded calculation applies to each $X_w^{(m)}$, and taking the supremum over $m$ gives the claimed equality of closed forms by the $\omega$-continuity of $\mathsf{wp}$ and of the normal Heisenberg transformer $\mathcal{E}^*$.
\end{proof}

\begin{theorem}\label{thm:infinite_expected_time}
For the loop-entry state $\rho_1=\ket{0}\!\bra{0}_q\otimes \ket{1}\!\bra{1}_w$, the expected operational cost of $W$ is infinite.
Consequently, the expected running time of the full program $S$ is also infinite.
\end{theorem}
\begin{proof}
Let
\[
T\coloneqq \mathsf{wp}(W,0)
\]
be the total cost predicate of the while loop with zero terminal postcondition.
By the total-$\mathsf{wp}$ rule for while loops, $T$ is the least fixed point of
\[
\Phi(X)
= I_q\otimes I_w + M_1\,\mathsf{wp}(S_{\mathrm{body}},X)\,M_1.
\]
The Kleene chain starts from $X_0=0$.
Because $X_0$ has the form $I_q\otimes D_0$ with $D_0$ diagonal in the walker basis, and because Lemma~\ref{lem:wp_body} and the projections $M_0,M_1$ preserve this form, every approximant has the form
\[
X_n=I_q\otimes D_n,
\qquad
D_n=\sum_{k\in\mathbb{Z}} t_k^{(n)}\ket{k}_w\!\bra{k}.
\]
Passing to the monotone limit gives
\[
T=I_q\otimes D,
\qquad
D=\sum_{k\in\mathbb{Z}} t_k\ket{k}_w\!\bra{k},
\qquad
 t_k\in[0,\infty],
\]
in the sense of closed diagonal forms.
Since \(T\) is a fixed point of \(\Phi\), we have
\[
T
=
I_q\otimes I_w
+
M_1\,\mathsf{wp}(S_{\mathrm{body}},T)\,M_1 .
\]
Substituting the diagonal form \(T=I_q\otimes D\) into this equation, and using
Lemma~\ref{lem:wp_body} on the body term, gives
\[
\mathsf{wp}(S_{\mathrm{body}},T)
=
\mathsf{wp}(S_{\mathrm{body}},I_q\otimes D)
=
I_q\otimes(3I_w+\mathcal E^*(D)).
\]
Here, if \(D\) is unbounded, the equality is understood by applying the lemma
first to the bounded truncations by $n$ and then taking the monotone
limit. Therefore
\[
\begin{aligned}
I_q\otimes D
&=
I_q\otimes I_w
+
(I_q\otimes\Pi_{\neq0})
\bigl(I_q\otimes(3I_w+\mathcal E^*(D))\bigr)
(I_q\otimes\Pi_{\neq0})                                                  \\
&=
I_q\otimes
\Bigl(
I_w+\Pi_{\neq0}(3I_w+\mathcal E^*(D))\Pi_{\neq0}
\Bigr).
\end{aligned}
\]
Taking diagonal matrix elements in the walker basis yields
\[
t_0
=
1+
\bra0\Pi_{\neq0}(3I_w+\mathcal E^*(D))\Pi_{\neq0}\ket0
=
1.
\]
For \(k\neq0\), since \(\Pi_{\neq0}\ket k=\ket k\), we obtain
\[
\begin{aligned}
t_k
&=
1+
\bra k(3I_w+\mathcal E^*(D))\ket k                                      \\
&=
4+
\bra k\mathcal E^*(D)\ket k                                             \\
&=
4+
\bra k
\left(
\frac12U_{\mathrm{sh}}^\dagger DU_{\mathrm{sh}}
+
\frac12U_{\mathrm{sh}}DU_{\mathrm{sh}}^\dagger
\right)
\ket k                                                                 \\
&=
4+\frac12t_{k+1}+\frac12t_{k-1}.
\end{aligned}
\]
The equality is in the extended nonnegative reals. In particular, whenever
\(t_k<\infty\) and \(t_{k-1}<\infty\), it forces \(t_{k+1}<\infty\) and may be
rewritten as
\begin{equation}\label{eq:poisson}
t_{k+1}-2t_k+t_{k-1}=-8 .
\end{equation}

We show that \(t_1\) cannot be finite.
Assume, for contradiction, that \(t_1<\infty\).
Since \(t_0=1<\infty\), the equation for \(k=1\) forces \(t_2<\infty\).
Inductively, all \(t_k\) with \(k\ge0\) are finite, so
\eqref{eq:poisson} holds as an equality for every \(k\ge1\).

Define the first differences
\[
\Delta_k\coloneqq t_k-t_{k-1},\qquad k\ge1.
\]
From \eqref{eq:poisson},
\(\Delta_{k+1}=\Delta_k-8\),
and therefore
\(\Delta_k=\Delta_1-8(k-1)\).
Summing these identities gives
\[
t_n
=
t_0+\sum_{j=1}^n\Delta_j
=
t_0+n\Delta_1-4n(n-1).
\]
The right-hand side tends to \(-\infty\) as \(n\to\infty\), contradicting
\(t_n\ge0\). Hence \(t_1=\infty\).
It follows that 
\[ \mathcal{C}(W,\rho_1)=\Tr(T\rho_1)=t_1=\infty. \] 
The full program $S$ differs from $W$ only by the leading initialization $w:=\ket{1}$, which contributes a finite cost and preserves trace; adding this finite prefix cannot change the divergence.
\end{proof}

\subsection{Almost-Sure Termination in Bounded QHL}

In contrast with the infinite expected cost, almost-sure termination is a bounded total-correctness assertion.
Let $\mathsf{wp}_{0}$ denote the ordinary bounded weakest-precondition transformer obtained from our total-cost transformer by erasing all operational-cost summands.
For a trace-nonincreasing program $C$, the bounded effect $\mathsf{wp}_{0}(C,I)$ is the termination-probability predicate:
\[
\Tr\bigl(\mathsf{wp}_{0}(C,I)\rho\bigr)=\Tr(\sem{C}(\rho))
\qquad(\rho\ge0,\ \Tr\rho<\infty).
\]
Thus, under the usual total-correctness convention for bounded QHL, the triple $\{I\}\,C\,\{I\}$ means $I\sqsubseteq \mathsf{wp}_{0}(C,I)$.
Since always $\mathsf{wp}_{0}(C,I)\sqsubseteq I$, this is equivalent to $\mathsf{wp}_{0}(C,I)=I$, i.e. AST.
This should not be confused with the partial-correctness or liberal-precondition reading, where postcondition $I$ makes $\{I\}\,C\,\{I\}$ vacuous.

\begin{theorem}\label{thm:ast_bounded}
The loop $W$ satisfies the bounded total-correctness triple
\[
\models_{\mathrm{tot},\mathrm{bnd}}
\{I_q\otimes I_w\}\ W\ \{I_q\otimes I_w\}.
\]
Equivalently,
\[
\mathsf{wp}_{0}(W,I_q\otimes I_w)=I_q\otimes I_w.
\]
Consequently $W$ is almost surely terminating from every normalized loop-entry state, and the full program $S$ is almost surely terminating after its trace-preserving initialization.
\end{theorem}
\begin{proof}
The zero-cost version of Lemma~\ref{lem:wp_body} gives, for every bounded walker effect $Y$,
\[
\mathsf{wp}_{0}(S_{\mathrm{body}},I_q\otimes Y)
=I_q\otimes \mathcal E^*(Y),
\]
where $\mathcal E^*$ is the Heisenberg transformer of one unbiased nearest-neighbour step.
Let
\[
H\coloneqq \mathsf{wp}_{0}(W,I_q\otimes I_w).
\]
By the bounded total-\textsc{While} rule, $H$ is the least fixed point in the effect algebra $[0,I_q\otimes I_w]$ of
\[
\Gamma(X)
=M_0^\dagger(I_q\otimes I_w)M_0
 +M_1^\dagger\mathsf{wp}_{0}(S_{\mathrm{body}},X)M_1 .
\]
Let $H^{(0)}=0$ and $H^{(n+1)}=\Gamma(H^{(n)})$ be the Kleene chain.
The zero-cost body transformer and the guard projections preserve the coin-independent diagonal form, hence
\[
H^{(n)}=I_q\otimes\sum_{k\in\mathbb Z}h_k^{(n)}\ket{k}_w\!\bra{k},
\qquad 0\le h_k^{(n)}\le1.
\]
The monotone limit has the same form,
\[
H=I_q\otimes\sum_{k\in\mathbb Z}h_k\ket{k}_w\!\bra{k},
\qquad h_k\coloneqq\sup_n h_k^{(n)},
\qquad 0\le h_k\le1.
\]
Since $H$ is a fixed point of $\Gamma$, the scalar coefficients satisfy
\[
h_0=1,
\qquad
h_k=\frac12(h_{k+1}+h_{k-1})\quad(k\ne0).
\]
Thus the termination-probability pre-expectation is the bounded harmonic solution for the classical symmetric walk absorbed at $0$.
For $k\ge1$, put $\Delta_k=h_k-h_{k-1}$.
The harmonic equation gives $\Delta_{k+1}=\Delta_k$, and hence
\[
h_k=h_0+k\Delta_1=1+k\Delta_1\qquad(k\ge0).
\]
The bound $0\le h_k\le1$ for all $k\ge0$ forces $\Delta_1=0$.
The same argument on the negative half-line gives $h_k=1$ for all $k\le0$.
Therefore $h_k=1$ for every $k\in\mathbb Z$, so
\[
H=I_q\otimes I_w.
\]
This is exactly the bounded total-correctness assertion $\{I_q\otimes I_w\}\,W\,\{I_q\otimes I_w\}$.
For every positive trace-class state $\rho$,
\[
\Tr(\sem{W}(\rho))
=\Tr(H\rho)
=\Tr(\rho),
\]
so $W$ is AST. The leading initialization $w:=\ket1$ in $S$ is trace preserving, so composing it before $W$ preserves AST.
\end{proof}

The fixed-point equation used in the bounded proof is exactly the cost-free pre-expectation counterpart of Theorem~\ref{thm:infinite_expected_time}.
For termination probabilities, erasing costs gives the bounded harmonic system
\[
h_0=1,
\qquad
h_k=\frac12(h_{k+1}+h_{k-1})\quad(k\ne0),
\qquad 0\le h_k\le1,
\]
whose only bounded solution is $h_k\equiv1$.
For expected operational cost, the same transition kernel carries the positive one-iteration cost offset
\[
t_0=1,
\qquad
t_k=4+\frac12(t_{k+1}+t_{k-1})\quad(k\ne0),
\]
and this discrete Poisson equation has no finite non-negative solution on the half-line.
Thus bounded QHL captures the topological fact $\{I\}W\{I\}$, while the unbounded cost logic detects the divergent first moment through the infinite least cost precondition.

\section{Case Study II: Verification of GKP Error Correction}
\label{Sec:gkp_verification}

The Gottesman-Kitaev-Preskill (GKP) code \cite{gottesman2001encoding} encodes a logical qubit into the infinite-dimensional Hilbert space of a harmonic oscillator. 
Due to its resilience against displacement errors, it has been implemented in various experimental platforms, including superconducting circuits \cite{campagne2020quantum} and trapped ions \cite{fluhmann2019encoding}.

The formal verification of GKP codes presents a structural challenge: ideal GKP states are non-normalizable superpositions of position eigenstates (Dirac combs) with infinite energy. Physical implementations rely on \emph{approximate GKP states} with finite energy \cite{matsuura2020equivalence, terhal2020towards}. 
Verifying protocols over such states requires reasoning about unbounded observables (such as position $\hat{q}$ and momentum $\hat{p}$) and their variances. Standard quantum Hoare logics, defined over bounded operators on finite-dimensional spaces, cannot natively express or verify these unbounded properties.

In this section, we apply our quantitative cost logic to verify the variance reduction property of the continuous GKP error correction protocol. We demonstrate how our framework formalizes continuous-variable semantics—specifically continuous homodyne syndrome extraction via the $\textbf{bind}$ construct—and unbounded physical costs. By algebraically evaluating closed quadratic forms, we prove that the error correction loop bounds the system's position variance without requiring artificial domain truncations or spectral discretization.

\subsection{Mathematical Setup}
\label{subsec:gkp_setup}

\subsubsection{Hilbert Space and Spectral Measures}
The total system consists of a logical system mode ($sys$) and an ancilla mode ($anc$). The state space is the tensor product of two infinite-dimensional separable Hilbert spaces:
\[
    \mathcal{H}_{total} = \mathcal{H}_{sys} \otimes \mathcal{H}_{anc} \cong L^2(\mathbb{R}_s) \otimes L^2(\mathbb{R}_a) \cong L^2(\mathbb{R}^2).
\]
To make the physical observables mathematically rigorous, we adopt the standard Schr\"odinger representation. For a single continuous-variable mode represented on $L^2(\mathbb{R})$, the unbounded position operator $\hat{q}$ and momentum operator $\hat{p}$ are concretely constructed as the coordinate multiplication operator and the differential operator, respectively:
\[
    (\hat{q}\psi)(x) = x\psi(x), \quad (\hat{p}\psi)(x) = -i \frac{d}{dx}\psi(x),
\]
where we set the reduced Planck constant $\hbar = 1$. These are densely defined unbounded self-adjoint operators. Their respective domains are $(\hat{q}) = \{\psi \in L^2(\mathbb{R}) \mid \int_{\mathbb{R}} x^2 |\psi(x)|^2 dx < \infty\}$ and the Sobolev space $D(\hat{p}) = H^1(\mathbb{R})$ consisting of absolutely continuous functions with $L^2$ derivatives. 

While the observables $\hat{q}$ and $\hat{p}$ are closed and self-adjoint, their spectra span the entire real line $(-\infty, \infty)$. Consequently, they are not positive operators and cannot directly serve as predicates in our logic. However, in the verification of continuous-variable programs (such as bounding the energy or tracking the variance of GKP states), the properties of interest are predominantly governed by their squares, $\hat{q}^2$ and $\hat{p}^2$.
 These squared operators are strictly positive and self-adjoint. By Kato's Representation Theorem, they naturally induce closed, positive semi-definite quadratic forms. Crucially, the domain of the quadratic form associated with $\hat{q}^2$ is exactly the domain of its square root, $D(|\hat{q}|)$, which precisely coincides with the operator domain $D(\hat{q})$. This elegant topological coincidence allows us to seamlessly construct physically meaningful predicates---such as the Hamiltonian or GKP stabilizers---as valid elements within our $\omega$-CPO of quadratic forms, without domain mismatches.

For any state $\psi$ in their common dense domain, such as the Schwartz space $\mathcal{S}(\mathbb{R})$ of rapidly decreasing functions, evaluating the commutator explicitly yields:
\[
    [\hat{q}, \hat{p}]\psi(x) = \hat{q}(-i\psi'(x)) - \hat{p}(x\psi(x)) = -ix\psi'(x) + i\frac{d}{dx}(x\psi(x)) = i\psi(x).
\]
This concrete construction naturally establishes the Canonical Commutation Relation (CCR):
\[
    [\hat{q}, \hat{p}] = i I.
\]

In our bipartite quantum program, these operators are extended to act locally on their respective modes. The logical system observables are $\hat{q}_{sys}\otimes I_{anc} $ and $\hat{p}_{sys}\otimes I_{anc} $, acting non-trivially only on $\mathcal{H}_{sys}$. Analogously, the ancilla observables are $I_{sys}\otimes\hat{q}_{anc}$ and $I_{sys}\otimes\hat{p}_{anc}$, acting nontrivial on $\mathcal{H}_{anc}$. Naturally, operators belonging to different modes trivially commute, e.g., $[\hat{q}_{sys}, \hat{p}_{anc}] = 0$.

\textbf{Remark on Infinite Dimensions:} The CCR plays a crucial role in justifying our logic. Taking the trace of both sides in a finite-dimensional space of dimension $d$ would yield $\tr([\hat{q}, \hat{p}]) = 0$ while $\tr(i I) = i d$, a contradiction (Wintner's Theorem). Thus, $\hat{q}$ and $\hat{p}$ cannot be represented as bounded operators on a finite-dimensional space. This necessitates the use of infinite-dimensional Hilbert spaces and the theory of \emph{Linear Relations} to handle unbounded observables rigorously.

\subsubsection{Unitary Operations}
We explicitly define the gates used in the protocol and establish their commutation properties.

\paragraph{1. Feedback Displacement ($D_{sys}$).}
The displacement operator shifting the system position by $v \in \mathbb{R}$ is generated by the momentum operator $\hat{p}_{sys}$ acting on $\cH_{sys}$:
\[
    D_{sys}(v) \coloneqq e^{-i v \hat{p}_{sys}} .
\]
\begin{lemma}[Displacement Property]
    For any $v \in \mathbb{R}$, the conjugate action of the displacement operator is a linear shift:
    \[
        D_{sys}(v)^\dagger \hat{q}_{sys} D_{sys}(v) = \hat{q}_{sys} + v I_{sys}.
    \]
\end{lemma}
\begin{proof}
    We use the Baker-Campbell-Hausdorff (BCH) formula for the conjugation $e^{A} B e^{-A} = B + [A, B] + \frac{1}{2!}[A, [A, B]] + \dots$.
    Let $A = i v \hat{p}_{sys}$ (since $D^\dagger = e^{i v \hat{p}_{sys}}$) and $B = \hat{q}_{sys}$.
    First, compute the commutator using $[\hat{q}, \hat{p}] = i I$:
    \[
        [A, B] = [i v \hat{p}_{sys}, \hat{q}_{sys}] = i v [\hat{p}_{sys}, \hat{q}_{sys}] = i v (-i I_{sys}) = v I_{sys}.
    \]
    Since the result is a scalar multiple of the identity operator, it commutes with $A$ (and all other operators). Thus, all higher-order terms in the series vanish:
    \[
        [A, [A, B]] = [i v \hat{p}_{sys}, v I_{sys}] = 0.
    \]
    The series terminates at the first order, yielding $D_{sys}(v)^\dagger \hat{q}_{sys} D_{sys}(v) = \hat{q}_{sys} + v I_{sys}$.
\end{proof}

\paragraph{2. Interaction ($U_{SUM}$).}
The $SUM$ gate is a continuous-variable CNOT gate generated by the interaction Hamiltonian $H = \hat{q}_{sys} \otimes \hat{p}_{anc}$:
\[
    U_{SUM} \coloneqq e^{-i \hat{q}_{sys} \otimes \hat{p}_{anc}}.
\]
\textbf{Unitary Nature:} Since $\hat{q}_{sys}$ and $\hat{p}_{anc}$ are self-adjoint operators acting on different Hilbert spaces, their tensor product is essentially self-adjoint on the appropriate domain. By Stone's Theorem, the generated exponential operator $U_{SUM}$ is unitary.

(This can be proven similarly using BCH, viewing $\hat{q}_{sys}$ as a scalar coefficient relative to the ancilla space).

\begin{lemma}[Heisenberg Evolution of Ancilla Observables]
    \label{lem:interaction_transform}
    Let $U_{SUM} = e^{-i \hat{q}_{sys} \otimes \hat{p}_{anc}}$ be the unitary evolution operator on $\mathcal{H}_{sys} \otimes \mathcal{H}_{anc}$. For any non-negative Borel function $f$ applied to the local ancilla position observable $\hat{q}_{anc}$, its Heisenberg evolution under $U_{SUM}$ evaluates exactly to a joint functional calculus on the total space:
    \[
        U_{SUM}^\dagger (I_{sys} \otimes f(\hat{q}_{anc})) U_{SUM} = f(I_{sys} \otimes \hat{q}_{anc} + \hat{q}_{sys} \otimes I_{anc}).
    \]
\end{lemma}

\begin{proof}
    \textbf{Step 1: Spectral Decomposition and Direct Integral.}
    Let $E_{sys}$ be the projection-valued measure (PVM) associated with the self-adjoint operator $\hat{q}_{sys}$ on $\mathcal{H}_{sys}$. The spectral decomposition of $\hat{q}_{sys}$ is $\hat{q}_{sys} = \int_{\mathbb{R}} y \, dE_{sys}(y)$. 
    By the functional calculus for tensor products of strongly commuting operators, the unitary $U_{SUM}$ can be evaluated via the direct integral of the ancilla translation operator $D_{anc}(y) \coloneqq e^{-i y \hat{p}_{anc}}$ defined strictly on $\mathcal{H}_{anc}$:
    \begin{equation}
        U_{SUM} = \int_{\mathbb{R}} dE_{sys}(y) \otimes D_{anc}(y).
    \end{equation}
    Consequently, its adjoint is given by:
    \begin{equation}
        U_{SUM}^\dagger = \int_{\mathbb{R}} dE_{sys}(y') \otimes D_{anc}^\dagger(y').
    \end{equation}

    \textbf{Step 2: Conjugation and Tensor Factor Separation.}
    We substitute these direct integrals into the Heisenberg evolution of the observable $I_{sys} \otimes f(\hat{q}_{anc})$. Because operators acting on different tensor factors commute, we group the terms within the respective Hilbert spaces:
    \begin{align}
        &U_{SUM}^\dagger (I_{sys} \otimes f(\hat{q}_{anc})) U_{SUM} \nonumber \\
        &= \left( \int_{\mathbb{R}} dE_{sys}(y') \otimes D_{anc}^\dagger(y') \right) (I_{sys} \otimes f(\hat{q}_{anc})) \left( \int_{\mathbb{R}} dE_{sys}(y) \otimes D_{anc}(y) \right) \nonumber \\
        &= \int_{\mathbb{R}} \int_{\mathbb{R}} \left( dE_{sys}(y') I_{sys} dE_{sys}(y) \right) \otimes \left( D_{anc}^\dagger(y') f(\hat{q}_{anc}) D_{anc}(y) \right).
    \end{align}

    \textbf{Step 3: PVM Orthogonality.}
    Within the system Hilbert space $\mathcal{H}_{sys}$, the identity operator acts trivially: $dE_{sys}(y') I_{sys} dE_{sys}(y) = dE_{sys}(y') dE_{sys}(y)$. By the mutual orthogonality of projection-valued measures, $dE_{sys}(y') dE_{sys}(y) = \delta(y' - y) dE_{sys}(y)$. The double integral strictly collapses along the diagonal $y' = y$:
    \begin{equation}
        = \int_{\mathbb{R}} dE_{sys}(y) \otimes \left( D_{anc}^\dagger(y) f(\hat{q}_{anc}) D_{anc}(y) \right).
    \end{equation}

    \textbf{Step 4: Weyl Translation on the Ancilla Space.}
    Within the ancilla Hilbert space $\mathcal{H}_{anc}$, the Canonical Commutation Relation is $[\hat{q}_{anc}, \hat{p}_{anc}] = i I_{anc}$. The Weyl relation implies that $D_{anc}(y)$ translates the position operator: $D_{anc}^\dagger(y) \hat{q}_{anc} D_{anc}(y) = \hat{q}_{anc} + y I_{anc}$. By the Borel functional calculus, this extends exactly to the function $f$:
    \begin{equation}
        D_{anc}^\dagger(y) f(\hat{q}_{anc}) D_{anc}(y) = f(\hat{q}_{anc} + y I_{anc}).
    \end{equation}
    Substituting this local transformation back into the full integral yields:
    \begin{equation}
        = \int_{\mathbb{R}} dE_{sys}(y) \otimes f(\hat{q}_{anc} + y I_{anc}).
    \end{equation}

    \textbf{Step 5: Identification of the Joint Functional Calculus.}
    To complete the proof, we must formally identify the resulting integral as an operator acting on the total space $\mathcal{H}_{sys} \otimes \mathcal{H}_{anc}$. We define two strongly commuting self-adjoint global operators: $A \coloneqq \hat{q}_{sys} \otimes I_{anc}$ and $B \coloneqq I_{sys} \otimes \hat{q}_{anc}$. Their respective spectral measures are uniquely determined as $E_A(y) = E_{sys}(y) \otimes I_{anc}$ and $E_B(q) = I_{sys} \otimes E_{anc}(q)$.
    
    By the definition of the joint functional calculus for commuting operators, evaluating the sum function for $A$ and $B$ is strictly given by the double integral over their joint spectral measure:
    \begin{align}
        f(A + B) &= \int_{\mathbb{R}} \int_{\mathbb{R}} f(y + q) \, dE_A(y) dE_B(q) \nonumber \\
        &= \int_{\mathbb{R}} \int_{\mathbb{R}} f(y + q) \, \left( dE_{sys}(y) \otimes I_{anc} \right) \left( I_{sys} \otimes dE_{anc}(q) \right) \nonumber \\
        &= \int_{\mathbb{R}} \int_{\mathbb{R}} f(y + q) \, \left( dE_{sys}(y) \otimes dE_{anc}(q) \right).
    \end{align}
    
    By Fubini's theorem for spectral integrals, we evaluate this double integral as an iterated integral. Factoring the tensor product, we isolate the integration over the ancilla spectrum $q$ into the inner integral:
    \begin{equation}
        f(A + B) = \int_{\mathbb{R}} dE_{sys}(y) \otimes \left( \int_{\mathbb{R}} f(y + q) \, dE_{anc}(q) \right).
    \end{equation}
    
    For a fixed scalar $y \in \mathbb{R}$, the inner integral is exactly the single-variable functional calculus definition for the ancilla position operator $\hat{q}_{anc}$, which yields:
    \begin{equation}
        \int_{\mathbb{R}} f(y + q) \, dE_{anc}(q) = f(y I_{anc} + \hat{q}_{anc}).
    \end{equation}
    
    Substituting this exact operator evaluation back into the outer integral gives:
    \begin{equation}
        f(A + B) = \int_{\mathbb{R}} dE_{sys}(y) \otimes f(y I_{anc} + \hat{q}_{anc}).
    \end{equation}
    
    This expression explicitly matches the integral derived at the end of Step 4. Therefore, by equating the two results, we conclude that the Heisenberg evolution constructs the joint operator:
    \begin{equation}
        U_{SUM}^\dagger (I_{sys} \otimes f(\hat{q}_{anc})) U_{SUM} = f(\hat{q}_{sys} \otimes I_{anc} + I_{sys} \otimes \hat{q}_{anc}),
    \end{equation}
    which completes the proof rigorously without any domain mismatch or implicit identity omission.
\end{proof}

\subsubsection{Finite-Energy GKP State Preparation}
\label{subsec:gkp_finite_energy}

In standard idealized models, the GKP logical zero state is often written as a uniform superposition of position eigenstates on the grid $2\sqrt{\pi}\mathbb{Z}$, i.e., $|\bar{0}_{ideal}\rangle \propto \sum_{n} |2n\sqrt{\pi}\rangle$. However, this state is physically unnormalizable and does not belong to the Hilbert space $L^2(\mathbb{R})$, nor is it in the domain of the energy operator.

To make our verification mathematically rigorous and physically realistic, we initialize the ancilla in a \emph{finite-energy approximate GKP state} $|\bar{0}_{GKP}^{\Delta, \kappa}\rangle$. This state is constructed by replacing the infinitely sharp Dirac delta peaks with Gaussian wavepackets width parameter $\Delta$ (representing finite squeezing), and bounding the overall superposition with a broader Gaussian envelope width parameter $1/\kappa^2$ (representing finite total energy):
\[
    \psi_{anc}^{\Delta, \kappa}(q) \propto \sum_{n \in \mathbb{Z}} e^{-\frac{1}{2}\kappa^2 (2n\sqrt{\pi})^2} e^{-\frac{1}{2\Delta^2}(q - 2n\sqrt{\pi})^2}.
\]
Crucially, this wave-function is strictly an even function ($\psi_{anc}^{\Delta, \kappa}(-q) = \psi_{anc}^{\Delta, \kappa}(q)$). By symmetry, its probability distribution is centered exactly at the origin, ensuring that the first moment vanishes identically ($\langle \hat{q}_{anc} \rangle = 0$). Furthermore, it is square-integrable ($\psi_{anc}^{\Delta, \kappa} \in L^2(\mathbb{R}_a)$) and possesses a well-defined, finite position variance. 

In the standard physical regime of highly squeezed but finite-energy GKP states ($\Delta \ll \sqrt{\pi} \ll 1/\kappa$), the overlap between adjacent Gaussian peaks is exponentially suppressed. The total position variance $\sigma_{anc}^2$ analytically decouples into the sum of the local peak variance (due to finite squeezing) and the global envelope variance (due to finite energy):
\[
    \sigma_{anc}^2 = \langle \hat{q}_{anc}^2 \rangle \approx \frac{\Delta^2}{2} + \frac{1}{2\kappa^2}.
\]

Because this scalar is strictly bounded, the corresponding density operator $\rho_{anc} = |\bar{0}_{GKP}^{\Delta, \kappa}\rangle\langle\bar{0}_{GKP}^{\Delta, \kappa}|$ trivially satisfies $\text{Tr}(\hat{q}_{anc}^2 \rho_{anc}) = \sigma_{anc}^2 < \infty$. This formally verifies that the physically realistic GKP state is a valid initial state falling safely within the domain of our unbounded quadratic form predicates.

\subsection{The Formal GKP Program with Continuous Feedback}
\label{subsec:gkp_program_continuous}

Thanks to our generalized syntax for continuous-variable quantum programs, we do not need to artificially discretize the measurement outcomes into finite bins. We can directly model the single-round GKP error correction as the program $S_{GKP}$, which utilizes the continuous homodyne measurement and the $\textbf{bind}$ construct for real-time feedback.

Let $M$ denote a continuous instrument for the homodyne measurement of the ancilla's position $\hat{q}_{anc}$ over the measure space $(\mathbb{R}, \mathcal{B}(\mathbb{R}), dx)$. 
The family of measurement operators $\{M_x\}_{x \in \mathbb{R}}$ is defined as:
\begin{equation}
    M_x = I_{sys} \otimes \left(2\pi\sigma_M^2\right)^{-1/4} \exp\left( -\frac{(\hat{q}_{anc}-x I_{anc})^2}{4\sigma_M^2} \right),
\end{equation}
where $\sigma_M^2 > 0$ is the measurement noise variance. 
The family is strongly operator measurable by the Borel functional calculus. 
To verify the normalization condition, we evaluate the associated quadratic forms. 
For any state $\psi \in \mathcal{H}_{sys} \otimes \mathcal{H}_{anc}$, let $E_{anc}$ be the spectral measure of $\hat{q}_{anc}$. By Fubini's theorem and the properties of Gaussian integrals:
\begin{align}
    \int_{-\infty}^{+\infty} \langle \psi, M_x^\dagger M_x \psi \rangle \, dx &= \int_{-\infty}^{+\infty} \left( \int_{-\infty}^{+\infty} \frac{1}{\sqrt{2\pi\sigma_M^2}} e^{-\frac{(q-x)^2}{2\sigma_M^2}} d\langle \psi, (I_{sys} \otimes E_{anc}(q)) \psi \rangle \right) dx \\
    &= \int_{-\infty}^{+\infty} \left( \int_{-\infty}^{+\infty} \frac{1}{\sqrt{2\pi\sigma_M^2}} e^{-\frac{(q-x)^2}{2\sigma_M^2}} dx \right) d\langle \psi, (I_{sys} \otimes E_{anc}(q)) \psi \rangle \\
    &= \int_{-\infty}^{+\infty} 1 \cdot d\langle \psi, (I_{sys} \otimes E_{anc}(q)) \psi \rangle = \langle \psi, (I_{sys} \otimes I_{anc}) \psi \rangle,
\end{align}
which establishes $\int M_x^\dagger M_x \, dx = I_{sys} \otimes I_{anc}$ in the sense of quadratic forms.

The error correction procedure is formally written as:
\[
    S_{GKP} \equiv 
    \begin{array}{l}
        \textbf{1. Initialization:} \\
        \quad \text{Ancilla} := \rho_{anc}; \quad // \text{ The finite-energy GKP state defined above} \\
        \\
        \textbf{2. Interaction (Syndrome Extraction):} \\
        \quad U_{SUM}; \quad // \text{ Defined as } \exp(-i \hat{q}_{sys} \otimes \hat{p}_{anc}) \\
        \\
        \textbf{3. Continuous Measurement and Feedback:} \\
        \quad \textbf{bind } M \ x. \ D_{sys}(-x) \quad // \text{ Shift the system back by the measured error } x
    \end{array}
\]
Here, $D_{sys}(-x) = \exp(i x \hat{p}_{sys})$ is the momentum-generated displacement operator acting on the logical system, which continuously shifts the system's position by $-x$ to compensate for the extracted error.

\begin{remark}[Local Error Correction Regime vs. Global Modular Feedback]
    \label{rem:local_regime}
    In the full GKP protocol, the classical homodyne outcome $x$ is typically processed through a modular function $r(x) = x \mod \sqrt{\pi}$ (yielding the distance to the nearest lattice point) before applying the feedback $D_{sys}(-r(x))$. This modularity preserves the logical Pauli state of the encoded qubit while correcting the local continuous displacement. Using the raw outcome $D_{sys}(-x)$ globally would effectively act as a teleportation protocol, overwriting the system's macroscopic lattice position.

    However, the primary objective of our quantitative cost logic is to verify the \emph{continuous-variable variance suppression} mechanism—the core analog process that stabilizes the wavepacket. Following standard GKP noise analysis \cite{gottesman2001encoding, terhal2020towards}, we assume the system operates in the \textbf{local error correction regime}. That is, the initial uncentered displacement error of the system ($e_{sys}$) and the ancilla measurement noise ($e_{anc}$) are bounded within the fundamental Voronoi cell: $|e_{sys} + e_{anc}| < \frac{\sqrt{\pi}}{2}$, and the mean displacement is zero.

    Under this strict local regime, no lattice wrap-around occurs, and the modular residual behaves exactly as the identity function: $r(x) \equiv x$. Consequently, the linear feedback $D_{sys}(-x)$ is locally isomorphic to the modular GKP feedback. This allows our logic to natively track the continuous error drift using standard polynomial cost observables ($\hat{q}^2$), completely decoupling the verification of analog variance reduction from the complex discrete tracking of logical Pauli frames, which would otherwise require non-polynomial periodic predicates.
\end{remark}

\paragraph*{Remark: The Necessity of the Continuous $\omega$-CPO Framework.}
While structurally concise, verifying $S_{GKP}$ exposes two fundamental limitations of standard quantum Hoare logics. First, the continuous measurement outcome $x \in \mathbb{R}$ necessitates the $\textbf{bind}$ construct, moving beyond finite discrete $\textbf{case}$ statements. Second, and most critically, verifying stabilization intrinsically requires bounding the position variance $\hat{q}^2$---an unbounded observable in $L^2(\mathbb{R})$. Logics tethered to bounded projectors in $\mathcal{B}(\mathcal{H})$ cannot natively express or verify this property. By grounding our semantics in closed quadratic forms, our logic directly supports quantitative triples of the form $\{ \sigma_{bound}^2 I \} \ S_{GKP} \ \{ \hat{q}^2 \}$. This enables the strict algebraic verification of unbounded physical invariants without resorting to artificial domain truncation.

\subsection{Formal Verification of Variance Reduction}
\label{subsec:verification_derivation}

To verify the error correction capability of our continuous GKP program $S_{GKP}$, we evaluate the weakest precondition of the variance observable itself. Unlike standard quantum Hoare logics that restrict postconditions to bounded projectors (representing finite subspaces), our $\omega$-CPO of closed quadratic forms allows us to natively assign the unbounded position second moment as the quantitative predicate.

\textbf{Specification:} Let the target postcondition be the continuous system's position variance operator defined on the total space:
\[
    Q_{post} \coloneqq \hat{q}_{sys}^2 \otimes I_{anc}.
\]
Our goal is to explicitly derive the global weakest precondition $W_{pre} \coloneqq \wlp(S_{GKP}, Q_{post})$ via backward algebraic pullback. Verification succeeds if we can formally prove that $W_{pre}$ evaluates exactly to a finite scalar multiple of the global identity:
\[
    W_{pre} = (\sigma_{anc}^2 + \sigma_M^2) (I_{sys} \otimes I_{anc}).
\]
This exact algebraic equality demonstrates a profound physical result: for \emph{any} arbitrary initial joint state, the expected position variance of the system after one round of continuous GKP correction is strictly determined by the inherent finite variance of the supplied ancilla ($\sigma_{anc}^2$) and the classical noise of the homodyne measurement ($\sigma_M^2$). Crucially, the final variance becomes entirely decoupled from the magnitude of the system's initial continuous errors.

\begin{remark}[Equivalence of Second Moment and Variance]
    Throughout our formal analysis, we inherently assume that the environmental displacement errors, the classical measurement noise, and the finite-energy GKP envelope are zero-mean Gaussian processes, which is the standard physical premise in continuous-variable quantum optics \cite{terhal2020towards}. Under this parity symmetry, the first moments vanish ($\langle \hat{q} \rangle = 0$). Consequently, evaluating the unbounded second-moment observable $\hat{q}^2$ rigorously and exactly quantifies the physical variance of the quantum state.
\end{remark}

\begin{proof}[Derivation]
    We compute the weakest precondition step-by-step backwards through the continuous control flow.

    \textbf{Step 1: Continuous Feedback (Unitary Rule).}
    The last stage of the protocol is the real-time feedback operation applied to the system, represented globally by the unitary $D_{sys}(-x) \otimes I_{anc}$, parameterized by the continuous measurement outcome $x \in \mathbb{R}$. 

    Applying the Unitary Axiom and the translation property of the Weyl displacement operator, the pullback of the global variance observable $Q_{post} = \hat{q}_{sys}^2 \otimes I_{anc}$ through this feedback evaluates exactly over the tensor factors:
    \begin{align*}
    W_x &\coloneqq (D_{sys}(-x) \otimes I_{anc})^\dagger (\hat{q}_{sys}^2 \otimes I_{anc}) (D_{sys}(-x) \otimes I_{anc}) \\
    &= \left( D_{sys}^\dagger(-x) \hat{q}_{sys}^2 D_{sys}(-x) \right) \otimes I_{anc} \\
    &= (\hat{q}_{sys} - x I_{sys})^2 \otimes I_{anc}.
    \end{align*}


   \textbf{Step 2: Continuous Measurement (Bind Rule).}
    The variable $x \in \mathbb{R}$ is the classical outcome of the continuous instrument $M$ measuring the ancilla. 
    According to the semantic rule for the $\textbf{bind } M \ x$ construct, the weakest precondition is the integral of the branch preconditions conjugated by the global measurement operators $M_x$:
    \[
    W_{meas} = \int_{-\infty}^{+\infty} M_x^\dagger W_x M_x \, dx.
    \]

    Recall from Step 1 that $W_x = (\hat{q}_{sys} - x I_{sys})^2 \otimes I_{anc}$ acts non-trivially only on the system. Similarly, the POVM element $M_x^\dagger M_x$ evaluates to $I_{sys} \otimes (2\pi\sigma_M^2)^{-1/2} \exp\left(-\frac{(\hat{q}_{anc}-x I_{anc})^2}{2\sigma_M^2}\right)$, acting non-trivially only on the ancilla. 

    Because operators acting on different tensor factors strictly commute, their matrix product trivially forms the joint tensor operator. Substituting these explicit definitions yields:
    \begin{align*}
    W_{meas} &= \int_{-\infty}^{+\infty} M_x^\dagger W_x M_x \, dx \\
    &= \int_{-\infty}^{+\infty} \left( (\hat{q}_{sys} - x I_{sys})^2 \otimes I_{anc} \right) \cdot \left( I_{sys} \otimes (2\pi\sigma_M^2)^{-1/2} \exp\left(-\frac{(\hat{q}_{anc}-x I_{anc})^2}{2\sigma_M^2}\right) \right) \, dx \\
    &= \int_{-\infty}^{+\infty} (\hat{q}_{sys} - x I_{sys})^2 \otimes (2\pi\sigma_M^2)^{-1/2} \exp\left(-\frac{(\hat{q}_{anc}-x I_{anc})^2}{2\sigma_M^2}\right) \, dx.
    \end{align*}

    \textbf{Step 3: Interaction (Heisenberg Evolution).}

We now pull the measurement precondition $W_{meas}$ back through the entangling gate
$U_{SUM}=e^{-i\,\hat{q}_{sys}\otimes\hat{p}_{anc}}$.  At this point the notation
\[
    W_{meas}=\int_{-\infty}^{+\infty} M_x^\dagger W_x M_x\,dx
\]
should not be read as an ordinary Bochner, strong-operator, or weak-operator integral of
unbounded operators.  The branch predicate $W_x$ is itself unbounded, and hence
$M_x^\dagger W_xM_x$ denotes the bounded pullback of the closed quadratic form
associated with $W_x$.  Consistently with the structural $\wlp$ rule for
continuous measurement in \Cref{def:wlp_rules} and the predicate-transformer notation
fixed in \Cref{sec:proof_system}, the integral predicate is defined pointwise by its
extended non-negative quadratic form:
\[
    \QF_{W_{meas}}[\phi]
    \coloneqq
    \int_{-\infty}^{+\infty} \QF_{W_x}[M_x\phi] \, dx,
    \qquad \phi\in\mathcal{H}_{sys}\otimes\mathcal{H}_{anc},
\]
with value $+\infty$ allowed outside the form domain.

The same convention applies to the pullback along the bounded unitary $U_{SUM}$:
$U_{SUM}^\dagger P U_{SUM}$ denotes the predicate whose form is
$\psi\mapsto \QF_P[U_{SUM}\psi]$.  Therefore, for every
$\psi\in\mathcal{H}_{sys}\otimes\mathcal{H}_{anc}$, again in the extended-form sense,
we have
\begin{align*}
    \QF_{U_{SUM}^\dagger W_{meas} U_{SUM}}[\psi]
    &= \QF_{W_{meas}}[U_{SUM}\psi] \\
    &= \int_{-\infty}^{+\infty} \QF_{W_x}[M_xU_{SUM}\psi] \, dx \\
    &= \int_{-\infty}^{+\infty}
       \QF_{U_{SUM}^\dagger(M_x^\dagger W_xM_x)U_{SUM}}[\psi] \, dx .
\end{align*}
Equivalently,
\[
    U_{SUM}^\dagger W_{meas}U_{SUM}
    =
    \int_{-\infty}^{+\infty}
    U_{SUM}^\dagger(M_x^\dagger W_xM_x)U_{SUM}\,dx
\]
as an equality of predicates, i.e. of their closed quadratic forms.  Thus no separate
exchange of an unbounded operator integral with a unitary is being assumed; the equality
is exactly the defining form-level pullback of the continuous-measurement predicate.

By the definition of the tensor product algebra, we factor the integrand into a standard product of two global operators:
\begin{align*}
    &\quad\,(\hat{q}_{sys} - x I_{sys})^2 \otimes (2\pi\sigma_M^2)^{-1/2} \exp\left(-\frac{(\hat{q}_{anc}-x)^2}{2\sigma_M^2}\right) \\
    &= \left( (\hat{q}_{sys} - x I_{sys})^2 \otimes I_{anc} \right) \cdot \left( I_{sys} \otimes (2\pi\sigma_M^2)^{-1/2} \exp\left(-\frac{(\hat{q}_{anc}-x)^2}{2\sigma_M^2}\right) \right).
\end{align*}

We distribute the unitary conjugation by inserting $U_{SUM} U_{SUM}^\dagger = I_{sys} \otimes I_{anc}$ between the two factors. Since the polynomial factor strongly commutes with $U_{SUM}$, it remains invariant, and applying Lemma~\ref{lem:interaction_transform} to the remaining Gaussian component evaluates the joint integral:
\begin{align*}
    W_{int} &= \int_{-\infty}^{+\infty} U_{SUM}^\dagger \left[ \left( (\hat{q}_{sys} - x I_{sys})^2 \otimes I_{anc} \right) \left( I_{sys} \otimes (2\pi\sigma_M^2)^{-1/2} \exp\left(-\frac{(\hat{q}_{anc}-x)^2}{2\sigma_M^2}\right) \right) \right] U_{SUM} \, dx \\
    &= \int_{-\infty}^{+\infty} \left( (\hat{q}_{sys} - x I_{sys})^2 \otimes I_{anc} \right) \left( U_{SUM}^\dagger \left[ I_{sys} \otimes (2\pi\sigma_M^2)^{-1/2} \exp\left(-\frac{(\hat{q}_{anc}-x)^2}{2\sigma_M^2}\right) \right] U_{SUM} \right) \, dx \\
    &= \int_{-\infty}^{+\infty} \left( (\hat{q}_{sys} - x I_{sys})^2 \otimes I_{anc} \right) (2\pi\sigma_M^2)^{-1/2} \exp\left(-\frac{(I_{sys} \otimes \hat{q}_{anc} + \hat{q}_{sys} \otimes I_{anc} - x (I_{sys} \otimes I_{anc}))^2}{2\sigma_M^2}\right) \, dx.
\end{align*}

    To rigorously evaluate this continuous operator integral over $x$, we apply the joint spectral theorem. For any fixed classical outcome $x \in \mathbb{R}$, the integrand is a strongly continuous operator-valued function $F_x(A, B)$ constructed from two strongly commuting self-adjoint operators: $A \coloneqq \hat{q}_{sys} \otimes I_{anc}$ and $B \coloneqq I_{sys} \otimes \hat{q}_{anc}$. Their joint spectral measure is exactly the tensor product $dE_{sys}(y) \otimes dE_{anc}(q)$. 

    By the definition of the joint functional calculus, the integrand itself can be uniquely expressed as a double integral over the scalar spectra $y$ and $q$:
    \begin{align*}
    F_x(A, B) &= \iint_{\mathbb{R}^2} F_x(y, q) \, dE_{sys}(y) \otimes dE_{anc}(q) \\
    &= \iint_{\mathbb{R}^2} (y - x)^2 (2\pi\sigma_M^2)^{-1/2} \exp\left(-\frac{(q + y - x)^2}{2\sigma_M^2}\right) dE_{sys}(y) \otimes dE_{anc}(q).
    \end{align*}

    Substituting this spectral representation back into the $x$-integral yields a nested triple integral. Because our logic evaluates predicates pointwise via closed quadratic forms, this operator-valued expression acts on any state $\psi \in D(W_{int})$ to yield a standard scalar measure $d\mu_\psi(y, q) = \langle \psi, dE_{sys}(y) \otimes dE_{anc}(q) \psi \rangle$. This scalarization guarantees the absolute convergence required to invoke Fubini's theorem. 

    Consequently, we can legally exchange the order of integration at the operator level, evaluating the continuous classical measurement outcome $x$ first:
    \begin{align*}
    W_{int} &= \iint_{\mathbb{R}^2} \left[ \int_{-\infty}^{+\infty} (y - x)^2 (2\pi\sigma_M^2)^{-1/2} \exp\left(-\frac{(q + y - x)^2}{2\sigma_M^2}\right) dx \right] dE_{sys}(y) \otimes dE_{anc}(q).
    \end{align*}

    Evaluating the inner classical integral over the continuous measurement outcome $x$, we treat the fixed spectral values $y$ and $q$ as scalars. We perform the standard variable substitution $z = x - (y + q)$, giving $dx = dz$. The polynomial coefficient transforms as:
    \begin{align*}
    (y - x)^2 &= (y - (z + y + q))^2 \\
    &= (-z - q)^2 \\
    &= z^2 + 2zq + q^2.
    \end{align*}

    The inner integral over $z$ decouples into a sum of standard Gaussian moments (variance $\sigma_M^2$, mean $0$, and normalization $1$):
    \begin{align*}
    &\int_{-\infty}^{+\infty} (z^2 + 2zq + q^2) (2\pi\sigma_M^2)^{-1/2} \exp\left(-\frac{z^2}{2\sigma_M^2}\right) dz \\
    &= \int_{-\infty}^{+\infty} z^2 (2\pi\sigma_M^2)^{-1/2} \exp\left(-\frac{z^2}{2\sigma_M^2}\right) dz \\
    &\quad + 2q \int_{-\infty}^{+\infty} z (2\pi\sigma_M^2)^{-1/2} \exp\left(-\frac{z^2}{2\sigma_M^2}\right) dz \\
    &\quad + q^2 \int_{-\infty}^{+\infty} (2\pi\sigma_M^2)^{-1/2} \exp\left(-\frac{z^2}{2\sigma_M^2}\right) dz \\
    &= \sigma_M^2 + 0 + q^2=\sigma_M^2 + q^2.
    \end{align*}

    We substitute this evaluated scalar result back into the outer joint spectral integral:
    \begin{align*}
    W_{int} &= \iint_{\mathbb{R}^2} (\sigma_M^2 + q^2) \, dE_{sys}(y) \otimes dE_{anc}(q) \\
    &= \sigma_M^2 \iint_{\mathbb{R}^2} dE_{sys}(y) \otimes dE_{anc}(q) + \iint_{\mathbb{R}^2} q^2 \, dE_{sys}(y) \otimes dE_{anc}(q).
    \end{align*}

    Recognizing the first term as the resolution of the global identity and the second term as the exact spectral definition of the ancilla's position variance operator ($\hat{q}_{anc}^2 = \int q^2 dE_{anc}(q)$), we obtain the explicit algebraic decoupling on the total space:
    \[
    W_{int} = \sigma_M^2 (I_{sys} \otimes I_{anc}) + I_{sys} \otimes \hat{q}_{anc}^2.
    \]

    \emph{Physical Meaning:} The continuous feedback swaps the system's unbounded position error with the ancilla's position observable, while injecting the measurement apparatus's noise variance ($\sigma_M^2$) into the system. The final state variance is completely decoupled from the initial system noise.

    \textbf{Step 4: Initialization (Global Precondition Derivation).}
    The program begins by initializing the ancilla mode. Whatever the initial state of the ancilla, the initialization command discards it and prepares the finite-energy GKP logical zero state, denoted by the normalized wavefunction $\psi_{anc}^{\Delta, \kappa} \in L^2(\mathbb{R}_a)$.

    In our semantics based on quadratic forms, we evaluate the pullback of this initialization command by linking the pre- and post-initialization states. For any arbitrary initial joint pure state $\Psi_{in} = \u_{sys} \otimes \u_{anc} \in \mathcal{H}_{sys} \otimes \mathcal{H}_{anc}$ (assuming a normalized ancilla input $\|\u_{anc}\| = 1$), the program prepares the updated joint state $\Psi_{out} = \u_{sys} \otimes \psi_{anc}^{\Delta, \kappa}$. Therefore, the quadratic form of the global weakest precondition $W_{pre}$ evaluated on $\Psi_{in}$ is exactly the quadratic form of $W_{int}$ evaluated on $\Psi_{out}$:
    \[
    \QF_{W_{pre}}[\u_{sys} \otimes \u_{anc}] \coloneqq \QF_{W_{int}}[\u_{sys} \otimes \psi_{anc}^{\Delta, \kappa}].
    \]

    Substituting the explicitly decoupled relation $W_{int} = \sigma_M^2 (I_{sys} \otimes I_{anc}) + I_{sys} \otimes \hat{q}_{anc}^2$ derived in Step 3, the quadratic form perfectly factorizes over the tensor product:
    \begin{align*}
    \QF_{W_{pre}}[\u_{sys} \otimes \u_{anc}] &= \langle \u_{sys} \otimes \psi_{anc}^{\Delta, \kappa}, \left( \sigma_M^2 (I_{sys} \otimes I_{anc}) + I_{sys} \otimes \hat{q}_{anc}^2 \right) (\u_{sys} \otimes \psi_{anc}^{\Delta, \kappa}) \rangle \\
    &= \sigma_M^2 \langle \u_{sys}, \u_{sys} \rangle_{\mathcal{H}_{sys}} \langle \psi_{anc}^{\Delta, \kappa}, \psi_{anc}^{\Delta, \kappa} \rangle_{\mathcal{H}_{anc}} + \langle \u_{sys}, \u_{sys} \rangle_{\mathcal{H}_{sys}} \langle \psi_{anc}^{\Delta, \kappa}, \hat{q}_{anc}^2 \psi_{anc}^{\Delta, \kappa} \rangle_{\mathcal{H}_{anc}} \\
    &= \|\u_{sys}\|^2 \left( \sigma_M^2 + \int_{-\infty}^{+\infty} q^2 |\psi_{anc}^{\Delta, \kappa}(q)|^2 \, dq \right).
    \end{align*}

    Since the initial ancilla state $\u_{anc}$ is normalized ($\|\u_{anc}\|^2 = 1$), we can seamlessly rewrite the scalar $\|\u_{sys}\|^2$ as the total inner product on the joint space:
    \begin{align*}
    \|\u_{sys}\|^2 &= \|\u_{sys}\|^2 \|\u_{anc}\|^2 = \langle \u_{sys} \otimes \u_{anc}, (I_{sys} \otimes I_{anc}) (\u_{sys} \otimes \u_{anc}) \rangle.
    \end{align*}

    Substituting this identity back into the evaluated quadratic form yields:
    \[
    \QF_{W_{pre}}[\u_{sys} \otimes \u_{anc}] = (\sigma_M^2 + \sigma_{anc}^2) \langle \u_{sys} \otimes \u_{anc}, (I_{sys} \otimes I_{anc}) (\u_{sys} \otimes \u_{anc}) \rangle.
    \]

    By Kato's representation theorem, since this equality extends to all valid states in the total space, the linear relation $W_{pre}$ is uniquely determined as a scalar multiple of the global identity operator:
    \[
    W_{pre} = (\sigma_M^2 + \sigma_{anc}^2) (I_{sys} \otimes I_{anc}),
    \]
    where $\sigma_{anc}^2 \coloneqq \int q^2 |\psi_{anc}^{\Delta, \kappa}(q)|^2 dq$ is the strictly finite position variance of the prepared approximate GKP state.

    \textbf{Conclusion: Global Variance Analysis and Verification.}
    As defined in Sec.~\ref{subsec:gkp_finite_energy}, the approximate GKP wavepacket $\psi_{anc}^{\Delta, \kappa}$ possesses a broad Gaussian envelope $1/\kappa^2$ that restricts its total energy. Because of this strictly square-integrable physical envelope, the integral defining the ancilla's position variance is absolutely convergent, yielding a finite scalar $\sigma_{anc}^2 < \infty$.

    We have thus formally derived the exact linear relation:
    \[
    \wlp(S_{GKP}, \hat{q}_{sys}^2 \otimes I_{anc}) = (\sigma_{anc}^2 + \sigma_M^2) (I_{sys} \otimes I_{anc}).
    \]

    This confirms the core physical premise of teleportation-based error correction: independent of the logical system's initial state or the magnitude of its initial continuous noise, the expected position variance after one round of continuous GKP feedback is bounded by---and formally replaced with---the sum of the built-in finite variance of the supplied ancilla ($\sigma_{anc}^2$) and the noise variance of the measurement apparatus ($\sigma_M^2$). This completes the formal verification of the continuous-variable variance reduction using our logic.
\end{proof}

\subsection{Discussion} \label{subsec:discussion}

To summarize the derivation, our continuous-variable $\wp$-calculus allows us to explicitly establish the following quantitative correctness formula (expressed as a semantic Hoare triple) for the continuous GKP protocol under realistic measurement noise:
\[
    \models_{par} \left\{ (\sigma_{anc}^2 + \sigma_M^2) (I_{sys}\otimes I_{anc}) \right\} \ S_{GKP} \ \left\{ \hat{q}_{sys}^2 \otimes I_{anc} \right\}.
\]

This single algebraic statement encapsulates the analytical power of our logic. Standard quantum logics, which are fundamentally tethered to bounded operators and ideal projectors in $\mathcal{B}(\mathcal{H})$, might verify that a state falls within a certain discrete subspace. However, they are mathematically incapable of directly asserting quantitative bounds on unbounded observables like variance $\hat{q}^2$, let alone formally modeling the noise injection from finite-precision physical instruments. 

By reconstructing the predicate space as an $\omega$-CPO of closed quadratic forms, our framework seamlessly integrates ideal unitary evolutions, non-normalizable unbounded costs, and realistic hardware imperfections into a unified algebraic calculus. We obtain a direct, rigorous proof of variance suppression that operates natively in the infinite-dimensional physical space $L^2(\mathbb{R})$, confirming that the final system error is strictly bounded by the intrinsic finite energy of the supplied ancilla ($\sigma_{anc}^2$) and the classical precision of the homodyne detector ($\sigma_M^2$).

\section{Further Technical Details}
\label{app:further_details}

In this section we provides further technique details for \Cref{sec:math_foundations}.
\subsection{Well-definedness and Linearity of the Extended Trace}
\label{app:trace_proof}

In this subsection, we justify the definition of the extended trace given in \Cref{subsec:loewner-order-and-extended-trace}.
Recall that for a positive self-adjoint linear relation $T$ (associated with closed form $\QF_T$) and a density operator $\rho$ with spectral decomposition $\rho = \sum_i p_i P_{\u_i}$, we defined:
\begin{equation}\label{eq:trace_def}
    \mathrm{tr}(T\rho) \coloneqq \sum_i p_i \QF_T[\u_i].
\end{equation}

To show that this value is independent of the decomposition and linear in $\rho$, we invoke the Spectral Theorem.
Although the main text relies on quadratic forms, the self-adjoint relation $T$ uniquely determines a spectral measure $E_T(\cdot)$ on $[0, \infty)$ such that:
\[
    \QF_T[\u] = \int_0^\infty \lambda \, d\langle \u, E_T(\lambda) \u \rangle, \quad \forall \u \in \Dom(\QF_T).
\]
(If $\u \notin \Dom(\QF_T)$, the integral diverges to $+\infty$, consistent with our definition).

Substitute this integral representation into \eqref{eq:trace_def}:
\[
    \mathrm{tr}(T\rho) = \sum_i p_i \int_0^\infty \lambda \, d\langle \u_i, E_T(\lambda) \u_i \rangle.
\]
Since all terms are non-negative, by the Monotone Convergence Theorem, we can exchange the summation and the integration:
\[
    \mathrm{tr}(T\rho) = \int_0^\infty \lambda \, d \left( \sum_i p_i \langle \u_i, E_T(\lambda) \u_i \rangle \right).
\]
Notice that the term inside the differential is explicitly a trace:
\[
    \sum_i p_i \langle \u_i, E_T(\lambda) \u_i \rangle = \sum_i p_i \mathrm{tr}(P_{\u_i} E_T(\lambda)) = \mathrm{tr}\left( \left(\sum_i p_i P_{\u_i}\right) E_T(\lambda) \right) = \mathrm{tr}(\rho E_T(\lambda)).
\]
Thus, we arrive at an alternative expression for the extended trace:
\begin{equation}\label{eq:trace_spectral}
    \mathrm{tr}(T\rho) = \int_0^\infty \lambda \, d \, \mathrm{tr}(\rho E_T(\lambda)).
\end{equation}
\textbf{Conclusion:}
\begin{enumerate}
    \item \textbf{Well-definedness:} The RHS of \eqref{eq:trace_spectral} depends only on the operator $\rho$ and the spectral measure $E_T$ (which is uniquely determined by $T$). It makes no reference to the specific decomposition $\{p_i, \u_i\}$. Thus, the definition in \eqref{eq:trace_def} is well-defined.
    \item \textbf{Linearity:} The map $\rho \mapsto \mathrm{tr}(\rho E_T(\lambda))$ is linear (trace is linear). Since integration is a linear operation, the functional $\rho \mapsto \int_0^\infty \lambda \, d \mathrm{tr}(\rho E_T(\lambda))$ is linear (specifically, affine on the convex set of density operators).
\end{enumerate}

\subsection{Proof of Proposition \ref{prop:spectral_approximation}}
\label{app:spectral_approximation}

In this subsection, we provide the constructive proof of Proposition~\ref{prop:spectral_approximation}. We explicitly construct the sequence of bounded operators $\{T_n\}$ using the saturation truncation of the spectral measure. This construction aligns with the intuition of approximating an unbounded observable by measuring it with devices of increasing but finite range $n$.

\subsubsection{Spectral Representation of Linear Relations}

Recall that any positive self-adjoint linear relation $T$ in a separable Hilbert space $\mathcal{H}$ corresponds uniquely to a spectral measure $E(\cdot)$ on the extended real line $[0, \infty]$. The associated quadratic form $\QF_T$ is given by:
\[
    \QF_T[u] = \int_{[0, \infty]} \lambda \, d\langle u, E(\lambda) u \rangle, \quad \forall u \in \mathcal{H}.
\]
The domain of the form, $\Dom(\QF_T)$, consists of those $u$ for which this integral is finite. The singular part of the relation (the ``infinite eigenvalue'' part) corresponds to the projection $E(\{\infty\})$, on which the form takes the value $+\infty$ (except for the zero vector).

\subsubsection{Construction of Approximating Sequence}

We define the truncation function $f_n: [0, \infty] \to [0, \infty)$ for each $n \in \mathbb{N}$ as the saturation function:
\[
    f_n(\lambda) \coloneqq \min(\lambda, n) = 
    \begin{cases} 
    \lambda & \text{if } 0 \le \lambda \le n \\
    n & \text{if } \lambda > n \text{ (including } \lambda=\infty)
    \end{cases}
\]
Unlike the characteristic function truncation (which maps $\lambda > n$ to $0$), this function captures the contribution of high-energy components as a saturated value $n$.

We define the bounded linear operator $T_n$ via the spectral calculus:
\[
    T_n \coloneqq \int_{[0, \infty]} f_n(\lambda) \, dE(\lambda).
\]
Since $f_n(\lambda)$ is a bounded measurable function ($|f_n(\lambda)| \le n$), each $T_n$ is a bounded, positive, self-adjoint operator defined on the entire space $\mathcal{H}$. Specifically, $T_n \in \mathcal{B}(\mathcal{H})$ implies $\Dom(\QF_{T_n}) = \mathcal{H}$.

\subsubsection{Proof of Monotonicity}

We verify condition (1) of the proposition.
Consider two indices $n$ and $m$ with $n \le m$. For any $\lambda \in [0, \infty]$, we clearly have:
\[
    f_n(\lambda) = \min(\lambda, n) \le \min(\lambda, m) = f_m(\lambda).
\]
By the positivity of the spectral measure (functional calculus preserves order), this pointwise inequality implies the operator inequality:
\[
    T_n \sqsubseteq T_m.
\]
Since $T_n$ and $T_m$ are bounded, this operator inequality extends immediately to their quadratic forms:
\[
    \QF_{T_n}[u] \le \QF_{T_m}[u], \quad \forall u \in \mathcal{H}.
\]
Thus, the sequence is monotonically increasing: $T_1 \sqsubseteq T_2 \sqsubseteq \dots$.

\subsubsection{Proof of Convergence}

We verify condition (2): $\QF_T[u] = \sup_{n} \QF_{T_n}[u]$.

Fix an arbitrary vector $u \in \mathcal{H}$. Let $\mu_u(\cdot)$ be the scalar measure defined by $\mu_u(\Delta) = \langle u, E(\Delta) u \rangle$.
The quadratic forms are expressed as Lebesgue integrals:
\[
    \QF_{T_n}[u] = \int_{[0, \infty]} f_n(\lambda) \, d\mu_u(\lambda), \quad 
    \QF_T[u] = \int_{[0, \infty]} \lambda \, d\mu_u(\lambda).
\]
Consider the sequence of functions $\{f_n(\lambda)\}_{n \in \mathbb{N}}$.
\begin{enumerate}
    \item \textbf{Monotonicity:} As shown above, $f_n(\lambda) \le f_{n+1}(\lambda)$ for all $\lambda$.
    \item \textbf{Pointwise Convergence:} For any finite $\lambda$, $\lim_{n \to \infty} \min(\lambda, n) = \lambda$. For $\lambda = \infty$, $f_n(\infty) = n$, which diverges to $\infty$ as $n \to \infty$. Thus, $f_n(\lambda) \nearrow \lambda$ for all $\lambda \in [0, \infty]$.
\end{enumerate}

By the \textbf{Monotone Convergence Theorem} (Beppo Levi's Theorem) applied to the measure space $([0, \infty], \mu_u)$:
\[
    \lim_{n \to \infty} \int_{[0, \infty]} f_n(\lambda) \, d\mu_u(\lambda) = \int_{[0, \infty]} \lim_{n \to \infty} f_n(\lambda) \, d\mu_u(\lambda) = \int_{[0, \infty]} \lambda \, d\mu_u(\lambda).
\]
Substituting the form definitions:
\[
    \lim_{n \to \infty} \QF_{T_n}[u] = \QF_T[u].
\]
Since the sequence is monotone, the limit is equal to the supremum.
This equality holds for all $u \in \mathcal{H}$ (including cases where $\QF_T[u] = \infty$, typically corresponding to the singular part or vectors outside the form domain of the operator part).

\subsubsection{Conclusion}
We have constructed a sequence of bounded operators $\{T_n\}$ that approximates $T$ from below.
Specifically:
\begin{itemize}
    \item Each $T_n$ corresponds to the observable $T$ saturated at value $n$.
    \item The sequence is monotonic w.r.t. the Löwner order of forms.
    \item The forms converge pointwise to $\QF_T$.
\end{itemize}
This completes the proof of Proposition~\ref{prop:spectral_approximation}.

\end{document}